\documentclass{llncs} % Aptara syntax
\usepackage{amssymb,amsmath}
\usepackage{booktabs}
\usepackage{pifont}
\usepackage{morefloats}
\usepackage{enumerate}

\usepackage[ruled]{algorithm2e}
\renewcommand{\algorithmcfname}{ALGORITHM}
\SetAlFnt{\small}
\SetAlCapFnt{\small}
\SetAlCapNameFnt{\small}
\SetAlCapHSkip{0pt}
\IncMargin{-\parindent}

\SetKw{Guess}{guess}%
\SetKw{Check}{check}%
\SetKw{Choose}{choose}%
\SetKw{KwAnd}{and}%
\SetKw{KwSuchThat}{such that}%
\SetKwFunction{canMap}{canMap}%
\SetKwFunction{canMapTail}{canMapTail}%
\SetKwFunction{isGenerated}{isGenerated}%
\SetKwProg{func}{Function}{}{}%

\usepackage{tikz}
\usetikzlibrary{positioning,arrows,fit,calc}
\usetikzlibrary{decorations.pathmorphing}
\usetikzlibrary{shapes,snakes}
\usepackage{pgfplots}
\pgfplotsset{compat=newest} %=1.10}

\tikzset{>=latex} 
\tikzstyle{point}=[circle,draw=black,thick,minimum size=1.8mm,inner sep=0pt,fill=white] % quantified variables, witness elements
\tikzstyle{bpoint}=[circle,draw=black,thick,minimum size=1.5mm,inner sep=0pt,fill=black] % free variables, ABox elements
\tikzstyle{lrgpoint}=[circle,draw=black,thick,minimum size=3.6mm,inner sep=0pt,fill=white] % quantified variables, witness elements
\tikzstyle{lrgbpoint}=[circle,draw=black,thick,minimum size=3mm,inner sep=0pt,fill=black] % free variables, ABox elements
\tikzstyle{hom}=[densely dotted,thin] % homomorphism 
\tikzstyle{query}=[thick] % arrows in the query
\tikzstyle{can}=[thick] % arrows in the canonical model

\tikzstyle{or-gate}=[rectangle,draw,inner sep=4pt,thick]
\tikzstyle{and-gate}=[rectangle,draw,inner sep=4pt,thick]
\tikzstyle{input}=[circle,draw,minimum size=4mm,text width=3mm,thick]

\usepackage{stmaryrd}

\newcommand{\OWL}{\textsl{OWL\,2}}
\newcommand{\OWLQL}{\textsl{OWL\,2\,QL}}

\newcommand{\DL}{\textsl{DL-Lite}}

\newcommand{\zo}{\{0,1\}}

\newcommand{\NC}{\mathsf{NC}}
\newcommand{\mNC}{\mathsf{mNC}}

\newcommand{\SAC}{\ensuremath{\mathsf{SAC}}}

\renewcommand{\P}{\mathsf{P}}
\newcommand{\mP}{\mathsf{mP}}

\newcommand{\NP}{\ensuremath{\mathsf{NP}}}

\newcommand{\mNP}{\mathsf{mNP}}

\newcommand{\NL}{\ensuremath{\mathsf{NL}}}

\newcommand{\mNL}{\mathsf{mNL}}
\newcommand{\comNL}{\mathsf{co}\text{-}\mathsf{mNL}}

\newcommand{\LOGCFL}{\ensuremath{\mathsf{LOGCFL}}}
\newcommand{\mLOGCFL}{\mathsf{mLOGCFL}}

\newcommand{\ACz}{{\ensuremath{\mathsf{AC}^0}}}
\newcommand{\mACz}{{\ensuremath{\mathsf{mAC}^0}}}
\newcommand{\Pitr}{{\ensuremath{\mathsf{\Pi}_3}}}
\newcommand{\mPitr}{{\ensuremath{\mathsf{m\Pi}_3}}}

\newcommand{\PSpace}{\textsc{PSpace}}

\newcommand{\T}{\mathcal{T}}
\newcommand{\A}{\mathcal{A}}
\newcommand{\C}{\mathcal{C}}
\newcommand{\I}{\mathcal{I}}

\newcommand{\Tmc}{\mathcal{T}}
\newcommand{\Amc}{\mathcal{A}}
\newcommand{\Cmc}{\mathcal{C}}

\newcommand{\q}{{\boldsymbol q}}
\renewcommand{\r}{\boldsymbol r}
\newcommand{\li}{\boldsymbol l}

\newcommand{\FO}{\text{FO}}
\newcommand{\PE}{\text{PE}}
\newcommand{\NDL}{\text{NDL}}

\newcommand{\cli}{\textsc{Clique}}
\newcommand{\reach}{\textsc{Reachability}}

\newcommand{\avec}[1]{\boldsymbol{#1}}

\newcommand{\ind}{\mathsf{ind}}
\newcommand{\tw}{\mathsf{tw}}

\newcommand{\omq}{{\ensuremath{\boldsymbol{Q}}}}

\renewcommand{\t}{\mathfrak{t}}
\newcommand{\tr}{\mathfrak{t}_\mathsf{r}}
\newcommand{\ti}{\mathfrak{t}_\mathsf{i}}
\newcommand{\qtw}{\q_{\mathsf{tw}}}

\newcommand{\Cir}{\boldsymbol{C}}

\newcommand\twfn{\ensuremath{f^\triangledown_{\omq}}}
\newcommand\homfn{\ensuremath{f^\blacktriangledown_{\omq}}}
\newcommand\homfnprime{\ensuremath{f^{\blacktriangledown\prime}_{\omq}}}

\newcommand{\fun}{f^\triangledown_{\omq}}
\newcommand{\prim}{f^\vartriangle_{\omq}}

\newcommand{\primfnP}{\ensuremath{f^{\vartriangle}_{\OMQT{H}}}}

\newcommand\twset{\Theta_{\omq}}

\newcommand\atom{{S(\avec{z})}}
\newcommand\rsz{{S(\avec{z})}}

\newcommand{\HG}[1]{\mathcal{H}(#1)}%{\mathcal{H}_{#1}}
\newcommand{\OMQI}[1]{\boldsymbol{S}_{#1}}
\newcommand{\OMQT}[1]{\boldsymbol{T}_{\hspace*{-0.2em}#1}}

\newcommand\canmod{\mathcal{C}_{\T,\A}}

\newcommand{\AND}{\textsc{and}}
\newcommand{\OR}{\textsc{or}}
\newcommand{\NOT}{\textsc{not}}

\newcommand{\ANDOP}{\land}
\newcommand{\OROP}{\lor}
\newcommand{\NOTOP}{\neg}

\newcommand{\degree}{\mathsf{deg}}

\newcommand{\HGP}{\mathsf{HGP}}
\newcommand{\THGP}{\mathsf{THGP}}
\newcommand{\mTHGP}{\mathsf{mTHGP}}
\newcommand{\mHGP}{\mathsf{mHGP}}

\newcommand{\poly}{\mathsf{poly}}

\newcommand{\emptyword}{\varepsilon}
\newcommand{\numtypes}{M}
\newcommand{\twidth}{m}

\newcommand{\leftt}{\mathsf{{left}}}
\newcommand{\rightt}{\mathsf{{right}}}
\newcommand{\Cover}{\mathsf{{cover}}}

\newenvironment{nitemize}{\begin{list}{--}{\itemsep=1pt\leftmargin=1.5em\labelwidth=1.5em\topsep=4pt}}{\end{list}}

\newcommand{\frontier}{\mathsf{frontier}}
\newcommand{\stack}{\mathsf{stack}}
\newcommand{\deepest}{\mathsf{deepest}}
\newcommand{\topof}{\mathsf{top}}

\def\bbbdalgo{\ensuremath{\mathsf{TreeQuery}}}
\def\bbarbalgo{\ensuremath{\mathsf{BLQuery}}}
\def\lspace{\ensuremath{\mathsf{L}}}
\def\llred{\ensuremath{\lspace^{\LOGCFL}}}

\newcommand{\false}{\ensuremath{\mathsf{false}}}
\newcommand{\true}{\ensuremath{\mathsf{true}}}

\newcommand{\dom}{\mathsf{dom}}

\newcommand{\aor}{U}%{\mathsf{or}}
\newcommand{\leftand}{L}%{\mathsf{la}}
\newcommand{\rightand}{R}%{\mathsf{ra}}
\newcommand{\trueleaf}{A}

\newcommand{\qclin}{\q'}
\newcommand{\dcx}{D(\avec{\alpha})}
\newcommand{\kbC}{\Tmc_{\avec{\alpha}},\Amc}

\newcommand{\assign}{\longleftarrow}

\allowdisplaybreaks

\begin{document}
%\hsize=15cm
%\vsize=21cm

% Page heads
\markboth{M. Bienvenu et al.}{Ontology-Mediated Queries: Complexity and Succinctness}

% Title portion
\title{Ontology-Mediated Queries: Combined Complexity and\\ Succinctness of Rewritings via Circuit Complexity}
%Query Rewriting and Circuit Complexity
\author{MEGHYN BIENVENU$^1$, STANISLAV KIKOT$^2$, ROMAN KONTCHAKOV$^2$, VLADIMIR PODOLSKII$^3$,
MICHAEL ZAKHARYASCHEV$^2$}
\institute{$^1$ LIRMM, CNRS\\ $^2$ Birkbeck, University of London\\$^3$ Steklov Mathematical Institute, Moscow}
\maketitle

\begin{abstract}
We give solutions to two fundamental computational problems in ontology-based data access with the W3C standard ontology language \OWLQL: the succinctness problem for first-order rewritings of ontology-mediated queries (OMQs), and the complexity problem for OMQ answering. We classify OMQs according to the shape of their conjunctive queries (treewidth, the number of leaves) and the existential depth of their ontologies. For each of these classes, we determine the combined complexity of OMQ answering, and  whether all OMQs in the class have polynomial-size first-order, positive existential, and nonrecursive datalog rewritings. We obtain the  succinctness results using hypergraph programs, a new computational model for Boolean functions, which makes it possible to connect the size of OMQ rewritings and circuit complexity. 
\end{abstract}

%\category{Artificial Intelligence}{Knowledge Representation and Reasoning}{Description Logics}
%\category{Information Systems}{Web Data Description Languages}{Web Ontology Language (OWL)}
%\category{Theory of Computation}{Database Theory}{Database Query Languages (Principles)}


%************************

\section{Introduction}\label{intro}

\subsection{Ontology-based data access}

Ontology-based data access (OBDA) via query rewriting was proposed by Poggi et al.\ \cite{PLCD*08} with the aim of facilitating query answering over complex, possibly incomplete and heterogeneous data sources. In an OBDA system (see Fig.~\ref{fig:obda}), the user does not have to be aware of the structure of data sources, which can be relational databases, spreadsheets, RDF triplestores, etc. Instead, the system provides the user with an ontology that serves as a high-level conceptual view of the data, gives a convenient vocabulary for user  queries, and enriches incomplete data with background knowledge. A snippet, $\T$, of such an ontology is shown below in the syntax of first-order (FO) logic:
\begin{align*}
& \forall x \, \big(\textit{ProjectManager}(x) \to \exists y\, (\textit{isAssistedBy}(x,y) \land \textit{PA}(y))\big),\\
& \forall x\, \big(\exists y\, \textit{managesProject}(x,y) \to \textit{ProjectManager}(x)\big),\\
& \forall x \, \big(\textit{ProjectManager}(x) \to \textit{Staff}(x)\big),\\ 
& \forall x \, \big(\textit{PA}(x) \to \textit{Secretary}(x)\big). 
\end{align*}
User queries are formulated in the signature of the ontology. For example, the conjunctive query (CQ)
\begin{align*}
\q(x) \ = \ \exists y \, (\textit{Staff}(x) \land \textit{isAssistedBy}(x,y) \land \textit{Secretary}(y)))
\end{align*}
\begin{figure}[t]%
\centering%
\begin{tikzpicture}[xscale=1.25]
\draw[thick,fill=gray!30] (-4,-1.7) rectangle +(11,4.4);
\begin{scope}[->,black!70,line width=2mm]
\draw[out=0,in=120] (0,4.21) to (4,3.4);
\draw[out=0,in=180] (0.3,2.25) to (1.4,2.25);
\draw[out=180,in=-90,looseness=1.5] (-0.5,-1.1) to (-2.3,0.7);
\draw[out=180,in=0] (3.5,-0.45) to (0.3,0.8);
\end{scope}
\draw[thick,fill=gray!30] (-4,3.1) rectangle +(4,1.2);
\node at (-2,3.7) {\itshape\tiny\begin{tabular}{l}SELECT ?s \{\\[-1pt]\hspace*{1em}?s a :Staff . \\[-1pt]\hspace*{1em}?s a $[$ a owl:restriction\textup{;} \\[-1pt]\hspace*{3.5em} owl:onProperty :assistedBy\textup{;}\\[-1pt]\hspace*{3.5em} owl:someValuesFrom :Secretary$]$ . \} \end{tabular}};
\node at (-0.75,3.95) {\bf query};
\draw[fill=gray!5] (-3.7,0.7) rectangle +(4,1.7);
\node at (-1.7,1.6) {\tiny\ttfamily\begin{tabular}{l}[] rdf:type rr:TriplesMap ;\\[-1pt]
\hspace*{1em}rr:logicalTable "SELECT  * FROM PROJECT";\\[-1pt]
\hspace*{1em}rr:subjectMap [ a rr:BlankNodeMap ;\\[-1pt]\hspace*{8em} rr:column "PRJ\_ID" ; ] ;\\[-1pt]
\hspace*{1em}rr:propertyObjectMap [ rr:property a:name;\\[-1pt]\hspace*{8em} rr:column "PRJ\_NAME" ] ;\\[-1pt]
\hspace*{1em}\dots\end{tabular}};
\node at (-0.8,0.95) {\small\bf mappings};
\draw[thick,fill=white] (1.4,1.1) rectangle +(5.3,2.3);
\node at (5.7,1.5) {\bf ontology};
\begin{scope}\tiny\itshape
\node[draw=black,fill=gray!5,rectangle] (staff) at (2.4, 3.1) {Staff};
\node[draw=black,fill=gray!5,rectangle] (projman) at (2.4, 2.5) {ProjectManager};
\node[draw=black,fill=gray!5,rectangle] (proj) at (2.4, 1.4) {\hspace*{1em}Project\hspace*{1em}};
\node[draw=black,fill=gray!5,diamond,aspect=2.5,inner ysep=0pt] (man) at (2.4, 1.95) {manages};
\node[draw=black,fill=gray!5,rectangle] (pa) at (5.7, 2.5) {\hspace*{2em}PA\hspace*{2em}};
\node[draw=black,fill=gray!5,diamond,aspect=2.5,inner ysep=0pt] (assisted) at (4.15, 2.5) {isAssistedBy};
\node[draw=black,fill=gray!5,rectangle] (sec) at (5.7, 3.1) {\hspace*{1em}Secretary\hspace*{1em}};
\draw(man) -- (projman);
\draw(man) -- (proj);
\draw(assisted) -- (pa);
\draw(assisted) -- (projman);
\draw[] (pa) -- (sec) node[sloped,midway,rotate=-90] {$\cup$};
\draw[] (projman) -- (staff) node[sloped,midway,rotate=-90] {$\cup$};
\end{scope}
\begin{scope}[shift={(1,0)}]
\draw[fill=gray!5] (0,-1.3) ellipse (1.5 and 0.2);
\fill[gray!5] (-1.5,-1.3) rectangle +(3,1.3);
\draw(-1.5,-1.3) -- +(0,1.3);
\draw(1.5,-1.3) -- +(0,1.3);
\draw[fill=gray!5] (0,0) ellipse (1.5 and 0.2);
\node at (0,-0.8) {\texttt{\tiny\begin{tabular}{l}CREATE TABLE PROJECT (\\[-1pt]\hspace*{1em}PRJ\_ID INT NOT NULL,\\[-1pt]\hspace*{1em}PRJ\_NAME VARCHAR(60) NOT NULL,\\[-1pt]\hspace*{1em}PRJ\_MANAGER\_ID INT NOT NULL\\[-1pt]\hspace*{1em}\dots\\[-1pt])\end{tabular}}};
\end{scope}
\begin{scope}[shift={(1,0.8)}]
\draw[fill=gray!5] (2.5,-1.6) rectangle +(3,1.6);
\draw (2.5,-0.2) -- +(3,0);
\draw (2.7,0) -- +(0,-1.6);
\foreach \x in {3.4,4.1,4.8} {\draw[ultra thin,gray] (\x,0) -- +(0,-1.6); }
\foreach \y in {-0.4,-0.6,-0.8,-1,-1.2,-1.4} {\draw[ultra thin,gray] (2.5,\y) -- +(3,0); }
\node at (3.05,-0.1) {\tiny\ttfamily\bfseries A};
\node at (3.75,-0.1) {\tiny\ttfamily\bfseries B};
\node at (4.45,-0.1) {\tiny\ttfamily\bfseries C};
\node at (5.15,-0.1) {\tiny\ttfamily\bfseries D};
\node at (2.6,-0.3) {\tiny\ttfamily\bfseries 1};
\node at (2.6,-0.5) {\tiny\ttfamily\bfseries 2};
\node at (2.6,-0.7) {\tiny\ttfamily\bfseries 3};
\node at (2.6,-0.9) {\tiny\ttfamily\bfseries 4};
\node at (2.6,-1.1) {\tiny\ttfamily\bfseries 5};
\node at (2.6,-1.3) {\tiny\ttfamily\bfseries 6};
\node at (2.6,-1.5) {\tiny\ttfamily\bfseries 7};
\end{scope}
\node at (4.2,-1.2) {\small\bf data sources};
\end{tikzpicture}%
\caption{Ontology-based data access.}\label{fig:obda}
\end{figure}
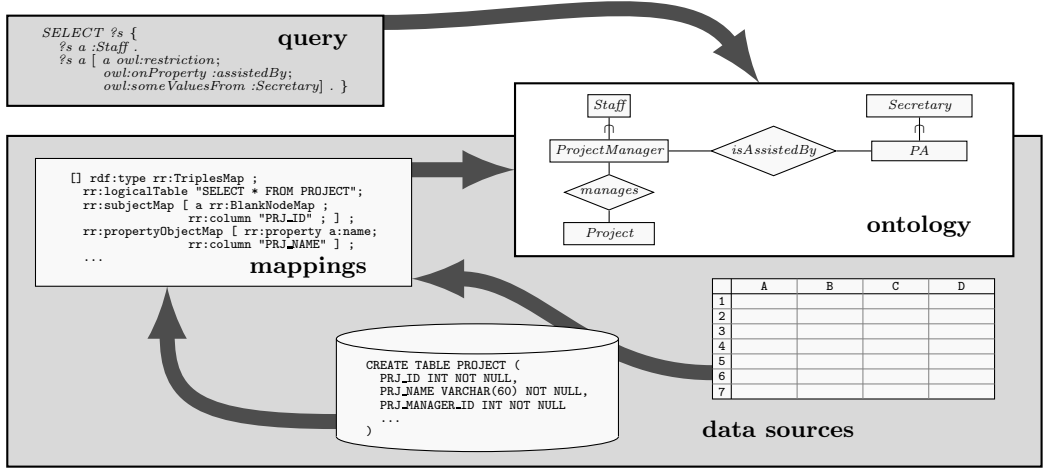%
is supposed to find the staff assisted by secretaries. The ontology signature and data schemas are related by mappings designed by the ontology engineer and invisible to the user. The mappings allow the system to view the data sources as a single RDF graph (a finite set of unary and binary atoms), $\A$, in the signature of the ontology. For example, the global-as-view (GAV) mappings 
\begin{align*}
& \forall x,y,z\,\big( {\small\texttt{PROJECT}}(x,y,z) \to \textit{managesProject}(x,z) \big),\\
& \forall x,y\, \big( {\small\texttt{STAFF}}(x,y) \land (y = 2) \to \textit{ProjectManager}(x)\big)
\end{align*}
populate the ontology predicates $\textit{managesProject}$ and $\textit{ProjectManager}$ with values from the database relations ${\small\texttt{PROJECT}}$ and ${\small\texttt{STAFF}}$. 
In the query rewriting approach of Poggi et al.\ \cite{PLCD*08}, the OBDA system employs the ontology and mappings in order to transform the user query into a query over the data sources, and then delegates the actual query evaluation to the underlying database engines and triplestores. 

For example, the first-order query 
\begin{multline*}
\q'(x) \ = \ \exists y \, \big[\textit{Staff}(x) \land \textit{isAssistedBy}(x,y) \land (\textit{Secretary}(y) \lor \textit{PA}(y))\big] \lor{}\\
 \textit{ProjectManager}(x) \lor \exists z\, \textit{managesProject}(x,z) 
\end{multline*}
is an \emph{FO-rewriting} of the \emph{ontology-mediated query} (OMQ) $\omq = (\T,\q)$  over any RDF graph $\A$ in the sense that $a$ is an answer to $\q'(x)$ over $\A$ iff $\q(a)$ is a logical consequence of $\T$ and $\A$.  As the system is not supposed to materialise $\A$, it uses the mappings to unfold the rewriting $\q'$ into an SQL (or SPARQL) query over the data sources.

Ontology languages suitable for OBDA via query rewriting have been identified by the Description Logic, Semantic Web, and Database/Datalog communities. The \DL{} family of description logics, first proposed by Calvanese et al.~\cite{CDLLR07} and later extended by Artale et al.~\cite{ACKZ09}, was specifically designed to ensure the existence of FO-rewritings for all conjunctive queries (CQs). Based on this family,
the W3C defined a profile \OWLQL\footnote{\url{http://www.w3.org/TR/owl2-overview/\#Profiles}}\ of the Web Ontology Language \OWL{}
`so that data [\ldots] % (assertions) that is 
stored in a standard relational database system can be queried through an ontology via a simple rewriting mechanism.\!' 
Various dialects of tuple-generating dependencies (tgds) that admit FO-rewritings of CQs and extend \OWLQL{} have also been identified~\cite{DBLP:journals/ai/BagetLMS11,DBLP:journals/ai/CaliGP12,DBLP:conf/datalog/CiviliR12}. 
We note in passing that while most work on OBDA (including the present paper) assumes that the user query is given as a CQ, 
other query languages, allowing limited forms of recursion and/or negation, have also been investigated \cite{DBLP:conf/icdt/Rosati07,DBLP:journals/ws/Gutierrez-Basulto15,DBLP:journals/jair/BienvenuOS15,DBLP:conf/aaai/KostylevRV15}.
SPARQL~1.1, the standard query language for RDF graphs, contains negation, aggregation and other features beyond first-order logic. The entailment regimes of SPARQL~1.1\footnote{\url{http://www.w3.org/TR/2013/REC-sparql11-entailment-20130321}}  also bring inferencing capabilities in the setting, which are, however, necessarily limited for efficient implementations.

%Ontology and query languages suitable for OBDA via query rewriting have been identified by the Description Logic, Datalog/Database and Semantic Web communities. Thus, W3C designed a special profile \OWLQL\footnote{\url{http://www.w3.org/TR/owl2-overview/#Profiles}}\ of the Web Ontology Language \OWL, which `enables conjunctive queries to be answered [\dots] using standard relational database technology.' In turn, \OWLQL{} is based on description logics of the \DL{} family~\cite{CDLLR07} and its extensions~\cite{ACKZ09}. Various dialects of tuple-generating dependencies (tgds) that admit FO-rewritings of conjunctive queries (CQs) and extend \OWLQL{} have also been identified~\cite{DBLP:journals/ai/BagetLMS11,DBLP:journals/ai/CaliGP12,DBLP:conf/datalog/CiviliR12}. 
%%%DBLP:journals/ws/CaliGL12, -- it contains only linear Datalog+- + negation + =; linear datalog is covered by Baget et al. 

By reducing OMQ answering to standard database query evaluation, which is generally regarded to be very efficient, OBDA via query rewriting has quickly become a hot topic in both theory and practice.
A number of rewriting techniques have been proposed and implemented for \OWLQL{} (PerfectRef~\cite{PLCD*08}, Presto/Prexto~\cite{DBLP:conf/kr/RosatiA10,DBLP:conf/esws/Rosati12}, tree witness rewriting~\cite{DBLP:conf/kr/KikotKZ12}), sets of tuple-generating dependencies (Nyaya~\cite{DBLP:conf/icde/GottlobOP11}, PURE~\cite{DBLP:journals/semweb/KonigLMT15}), and more expressive ontology languages that require recursive datalog rewritings (Requiem~\cite{DBLP:conf/dlog/Perez-UrbinaMH09}, Rapid~\cite{DBLP:conf/cade/ChortarasTS11}, Clipper~\cite{DBLP:conf/aaai/EiterOSTX12} and Kyrie~\cite{kyrie2}). 
A few mature OBDA systems have also recently emerged: pioneering MASTRO~\cite{DBLP:journals/semweb/CalvaneseGLLPRRRS11}, commercial Stardog~\cite{Perez-Urbina12} and Ultrawrap~\cite{DBLP:conf/semweb/SequedaAM14}, and the Optique platform~\cite{optique} based on the query answering engine Ontop~\cite{ISWC13,DBLP:conf/semweb/KontchakovRRXZ14}.
By providing a semantic end-to-end connection between users and multiple distributed data sources (and thus making the IT expert middleman redundant), OBDA has attracted the attention of industry,  with companies such as Siemens~\cite{DBLP:conf/semweb/KharlamovSOZHLRSW14} and Statoil~\cite{DBLP:conf/semweb/KharlamovHJLLPR15} experimenting with OBDA technologies to streamline the process of data access for their engineers.\!\footnote{See, e.g., \url{http://optique-project.eu}.} 

\subsection{Succinctness and complexity}

In this paper, our concern is two fundamental theoretical problems whose solutions 
will elucidate the computational costs required for answering OMQs with \OWLQL{} ontologies. The succinctness problem for FO-rewritings is to understand how difficult it is to construct FO-rewritings for OMQs in a given class and, in particular, to determine whether OMQs in the class have polynomial-size FO rewritings or not. In other words, the succinctness problem clarifies the computational costs of the \emph{reduction} of OMQ answering to database query evaluation. On the other hand, it is also important to measure the resources required to answer OMQs by a \emph{best possible algorithm}, not necessarily a reduction to database query evaluation. Thus, we are interested in the combined complexity of the OMQ answering problem: given an OMQ $\omq = (\T,\q(\avec{x}))$ from a certain class, a data instance $\A$ and a tuple $\avec{a}$ of constants from $\A$, decide whether $\T,\A \models \q(\avec{a})$.  The combined complexity of CQ evaluation has been thoroughly investigated in database theory; cf.~\cite{DBLP:conf/stoc/GroheSS01,Libkin} and references therein. To slightly simplify the setting for our problems, we assume that data is given in the form of an RDF graph and leave mappings out of the picture (in fact, GAV mappings only polynomially increase the size of FO-rewritings over RDF graphs).

We suggest a `two-dimensional' classification of OMQs. One dimension  takes account of the shape of the CQs in OMQs by quantifying their treewidth (as in classical database theory) and the number of leaves in tree-shaped CQs. Note that, in  SPARQL~1.1, the sub-queries that require rewriting under the \OWLQL{} entailment regime are always tree-shaped (they are, in essence, complex class expressions). 
The second dimension is the existential depth of ontologies, that is, the length of the longest chain of labelled nulls in the chase on any data. 
Thus, the NPD FactPages ontology,\!\footnote{http://sws.ifi.uio.no/project/npd-v2/} which was designed to facilitate  querying the datasets of the Norwegian Petroleum Directorate,\!\footnote{http://factpages.npd.no/factpages/} is of depth 5. A typical example of an ontology axiom causing infinite depth is \mbox{$\forall x\, \bigl(\textit{Person}(x) \to \exists y\, (\textit{ancestor}(y,x) \land \textit{Person}(y))\bigr)$}. 

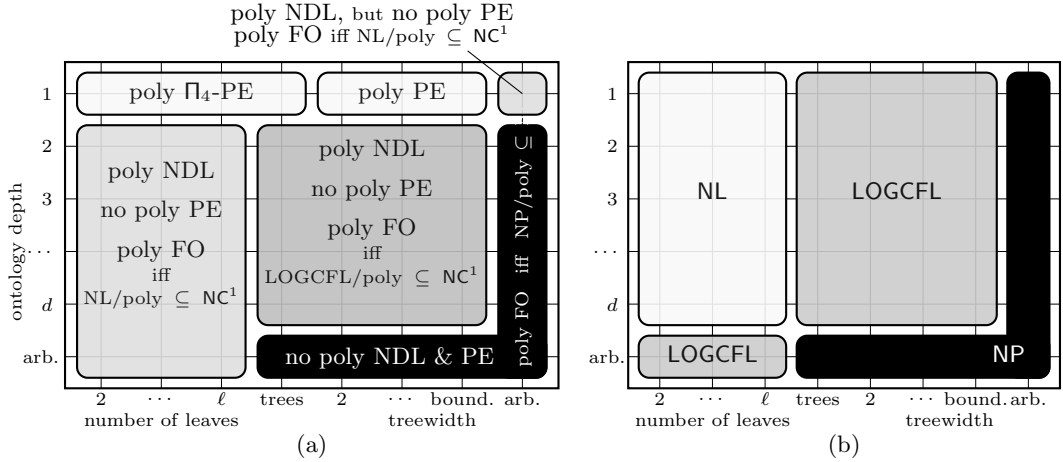
\begin{figure}[t]%
\tikzset{cmplx/.style={draw,thick,rounded corners,inner sep=0mm}}%
\mbox{}\hspace*{-0.3em}\begin{tikzpicture}[x=8mm,y=7mm]
\draw[thick] (0.4,-0.4) rectangle +(8.2,-6.2);
\node[rotate=90] at (-0.4,-4.0) {\scriptsize ontology depth};
\begin{scope}[ultra thin]
\draw (0.4,-1) -- +(8.2,0); \node at (0.1,-1) {\scriptsize 1};
\draw (0.4,-2) -- +(8.2,0); \node at (0.1,-2) {\scriptsize 2};
\draw (0.4,-3) -- +(8.2,0); \node at (0.1,-3) {\scriptsize 3};
\draw (0.4,-4) -- +(8.2,0); \node at (0,-4) {\scriptsize \dots};
\draw (0.4,-5) -- +(8.2,0); \node at (0.1,-5) {\scriptsize $d$};
\draw (0.4,-6) -- +(8.2,0); \node at (0,-6) {\scriptsize arb.};
\draw (1,-0.4) -- +(0,-6.2); \node at (1,-6.8) {\scriptsize 2}; 
\draw (2,-0.4) -- +(0,-6.2); \node at (2,-6.8) {\scriptsize \dots}; 
\draw (3,-0.4) -- +(0,-6.2); \node at (3,-6.8) {\scriptsize $\ell$}; 
\draw (4,-0.4) -- +(0,-6.2); \node at (4,-6.8) {\scriptsize trees}; 
\draw (5,-0.4) -- +(0,-6.2); \node at (5,-6.8) {\scriptsize 2}; 
\draw (6,-0.4) -- +(0,-6.2); \node at (6,-6.8) {\scriptsize \dots}; 
\draw (7,-0.4) -- +(0,-6.2); \node at (7,-6.8) {\scriptsize bound.}; 
\draw (8,-0.4) -- +(0,-6.2); \node at (8,-6.8) {\scriptsize arb.}; 
\end{scope}
\node at (2,-7.2) {\scriptsize number of leaves};
\node at (6.5,-7.2) {\scriptsize treewidth};
\node [fill=gray!25,cmplx,fill opacity=0.9,fit={(0.6,-1.6) (3.4,-6.4)}] 
{\raisebox{-5ex}{\begin{tabular}{c}poly NDL\\[4pt] no poly PE\\[4pt] 
\footnotesize poly FO\\[-2pt]\scriptsize iff\\[-2pt]\scriptsize NL/poly  $\,\subseteq\,$ $\NC^1$\end{tabular}}};
\node [fill=gray!50,cmplx,fill opacity=0.9,fit={(3.6,-1.6) (7.4,-5.4)}] 
{\raisebox{-6ex}{\begin{tabular}{c}poly NDL\\[4pt] no poly PE\\[4pt] 
\footnotesize poly FO\\[-2pt] \scriptsize iff\\[-2pt] \scriptsize LOGCFL/poly $\,\subseteq\,$ $\NC^1$\end{tabular}}};
\node [fill=black,cmplx,fit={(3.6,-5.6) (8.4,-6.4)}] 
{\hspace*{-1em}\raisebox{-9pt}{\textcolor{white}{\begin{tabular}{c}\small no poly NDL \& PE
%\\[-3pt]\footnotesize poly FO  \scriptsize iff \nppoly  $\,\subseteq\,$ \ncone
\end{tabular}}}}; 
\node [fill=black,cmplx,fit={(7.6,-1.6) (8.4,-6.4)}]  
{\rotatebox{90}{\hspace*{-4.9em}\textcolor{white}{\begin{tabular}{c}
%\footnotesize no poly PE or NDL\\[-5pt]\tiny 
\scriptsize poly FO \ iff \  NP/poly $\subseteq$ NC$^1$
\end{tabular}}}};
\node [fill=gray!5,cmplx,fill opacity=0.9,fit={(0.6,-0.6) (4.4,-1.4)}]  {\raisebox{-1.5ex}{poly $\mathsf{\Pi}_4$-PE}};
\node [fill=gray!5,cmplx,fill opacity=0.9,fit={(4.6,-0.6) (7.4,-1.4)}]  {\raisebox{-1.5ex}{poly PE}};
\node [fill=gray!25,cmplx,fill opacity=0.9,fit={(7.6,-0.6) (8.4,-1.4)}] { };
\node[inner sep=0pt] (test) at (5.5,0.3) {\begin{tabular}{c}poly NDL, {\scriptsize but} no poly PE\\[-3pt]\footnotesize poly FO \scriptsize iff NL/poly  $\,\subseteq\,$ $\NC^1$\end{tabular}};
\draw (8,-1) -- (test.-15);
\node at (4.5,-7.7) {\footnotesize (a)};
\end{tikzpicture}
\hspace*{0.2em}
\begin{tikzpicture}[x=7mm,y=7mm]
\draw[thick] (0.4,-0.4) rectangle +(8.2,-6.2);
\begin{scope}[ultra thin]
\draw (0.4,-1) -- +(8.2,0); \node at (0.1,-1) {\scriptsize 1};
\draw (0.4,-2) -- +(8.2,0); \node at (0.1,-2) {\scriptsize 2};
\draw (0.4,-3) -- +(8.2,0); \node at (0.1,-3) {\scriptsize 3};
\draw (0.4,-4) -- +(8.2,0); \node at (0,-4) {\scriptsize \dots};
\draw (0.4,-5) -- +(8.2,0); \node at (0.1,-5) {\scriptsize $d$};
\draw (0.4,-6) -- +(8.2,0); \node at (0,-6) {\scriptsize arb.};
\draw (1,-0.4) -- +(0,-6.2); \node at (1,-6.8) {\scriptsize 2}; 
\draw (2,-0.4) -- +(0,-6.2); \node at (2,-6.8) {\scriptsize \dots}; 
\draw (3,-0.4) -- +(0,-6.2); \node at (3,-6.8) {\scriptsize $\ell$}; 
\draw (4,-0.4) -- +(0,-6.2); \node at (4,-6.8) {\scriptsize trees}; 
\draw (5,-0.4) -- +(0,-6.2); \node at (5,-6.8) {\scriptsize 2}; 
\draw (6,-0.4) -- +(0,-6.2); \node at (6,-6.8) {\scriptsize \dots}; 
\draw (7,-0.4) -- +(0,-6.2); \node at (7,-6.8) {\scriptsize bound.}; 
\draw (8,-0.4) -- +(0,-6.2); \node at (8,-6.8) {\scriptsize arb.}; 
\end{scope}
\node at (2,-7.2) {\scriptsize number of leaves};
\node at (6.5,-7.2) {\scriptsize treewidth};
\node [fill=gray!5,cmplx,fill opacity=0.9,fit={(3.4,-5.4) (0.6,-0.6)}]  {$\NL$};
\node [fill=gray!40,cmplx,fill opacity=0.9,fit={(3.6,-0.6) (7.4,-5.4)}]  {$\LOGCFL$};
\node [fill=black,cmplx,fit={(8.4,-6.4) (7.6,-0.6)}]  {}; % {\textcolor{white}{\NP}}
\node [fill=black,cmplx,fit={(3.6,-5.6) (8.4,-6.4)}]  {\hspace*{7em}\raisebox{-1ex}{\textcolor{white}{\NP}}};
\node [fill=gray!40,cmplx,fill opacity=0.9,fit={(0.6,-5.6) (3.4,-6.4)}]  {\raisebox{-1ex}{\LOGCFL}};
\node at (4.5,-7.7) {\footnotesize (b)};
\end{tikzpicture}%
\caption{(a) Succinctness of OMQ rewritings, and (b) combined complexity of OMQ answering (tight bounds).}
\label{pic:results}
\end{figure}

\subsection{Results}\label{sec:results}

The results of our investigation are summarised in the succinctness and complexity landscapes of Fig.~\ref{pic:results}. In what follows, we discuss these results in more detail. 

The \emph{succinctness problem} we consider can be formalised as follows: given a sequence $\omq_n$ ($n<\omega$) of OMQs whose size is polynomial in $n$, determine whether the size of minimal rewritings of $\omq_n$ can be bounded by a polynomial function in $n$. We distinguish between three types of rewritings: arbitrary FO-rewritings, positive existential (PE-) rewritings (in which only $\land$, $\lor$ and $\exists$ are allowed), and non-recursive datalog (NDL-) rewritings.\!\footnote{Domain-independent FO-rewritings correspond to SQL queries, PE-rewritings to \textsc{Select}-\textsc{Project}-\textsc{Join}-\textsc{Union} (or SPJU) queries, and NDL-rewritings to SPJU queries with views; see also Remark~\ref{dom-ind}.}  
This succinctness problem was first considered by Kikot et al.~\cite{DBLP:conf/icalp/KikotKPZ12} and Gottlob and Schwentick~\cite{DBLP:conf/kr/GottlobS12}. The former constructed a sequence $\omq_n$ of OMQs (with tree-shaped CQs) whose PE- and NDL-rewritings are of exponential size, while FO-rewritings are superpolynomial unless $\NP \subseteq \P/\poly$. Gottlob and Schwentick~\cite{DBLP:conf/kr/GottlobS12} and Gottlob et al.~\cite{DBLP:journals/ai/GottlobKKPSZ14} showed that PE- (and so all other) `rewritings' can be made polynomial under the condition that all relevant data instances contain two special constants. The `succinctification' trick involves polynomially many extra existential quantifiers over these constants to guess a derivation of the given CQ in the chase, which makes such rewritings impractical (cf.~NFAs vs DFAs, and~\cite{DBLP:journals/tocl/Avigad03}). In this paper, we stay within the classical OBDA setting that does not impose any extra conditions on the data and does not allow any special constants in rewritings. 

Figure~\ref{pic:results}~(a) gives a summary of the succinctness results obtained in this paper.
It turns out that polynomial-size PE-rewritings are guaranteed to exist---in fact, can be constructed in polynomial time---only for the class of OMQs with ontologies of depth~1 and CQs of bounded treewidth; moreover, tree-shaped OMQs have polynomial-size $\mathsf{\Pi}_4$-PE-rewritings (with matrices of the form ${\land}{\lor}{\land}{\lor}$). Polynomial-size NDL-rewritings can be efficiently constructed for all tree-shaped OMQs with a bounded number of leaves, all OMQs with ontologies of bounded depth and CQs of bounded treewidth, and all OMQs with ontologies of depth 1. For OMQs with ontologies of depth 2 and arbitrary CQs, and OMQs with arbitrary ontologies and tree-shaped CQs, we have an exponential lower bound on the size of NDL- (and so PE-) rewritings. The  existence of polynomial-size FO rewritings for all OMQs in each of  these classes (save the first one) turns out to be equivalent to one of the major open problems in computational complexity such as $\NC^1 = \NP/\poly$.\!\footnote{$\mathsf{C}/\poly$ is the non-uniform analogue of a complexity class $\mathsf{C}$.}

We obtain these results by establishing a connection between succinctness of rewritings and circuit complexity, a branch of computational complexity theory that classifies Boolean functions according to the size of circuits computing them. Our starting point is the observation that the tree-witness PE-rewriting of an OMQ $\omq = (\T,\q)$ introduced by~\cite{DBLP:conf/kr/KikotKZ12} defines a hypergraph whose vertices are the atoms in $\q$ and whose hyperedges correspond to connected sub-queries of $\q$ that can be homomorphically mapped to labelled nulls of some chases for $\T$.
Based on this observation, we introduce a new computational model for Boolean functions by treating any hypergraph~$H$, whose vertices are labelled by (possibly negated) Boolean variables or constants~0 and~1, as a program computing a Boolean function $f_H$ that returns 1 on a valuation for the variables iff there is an independent subset of hyperedges covering all vertices labelled by 0 (under the valuation). We show that constructing short FO- (respectively,  PE- and NDL-) rewritings of $\omq$ is (nearly) equivalent to finding short Boolean formulas (respectively, monotone formulas and monotone circuits) computing the hypergraph function for~$\omq$. 

For each of the OMQ classes in Fig.~\ref{pic:results}~(a), we characterise the computational power of the corresponding hypergraph programs and employ results from circuit complexity to identify the size of rewritings. For example, we show that OMQs with ontologies of depth 1 correspond to hypergraph programs of degree $\le 2$ (in which every vertex belongs to at most two hyperedges), and that the latter are polynomially equivalent to nondeterministic branching programs (NBPs). Since NBPs compute the Boolean functions in the class $\NL/\poly \subseteq \P/\poly$, the tree-witness rewritings for OMQs with ontologies of depth 1 can be equivalently transformed into polynomial-size NDL-rewritings. On the other hand, there exist monotone Boolean functions computable by polynomial-size NBPs but not by polynomial-size monotone Boolean formulas, which establishes a superpolynomial lower bound for PE-rewritings. It also follows that all such OMQs have polynomial-size FO-rewritings iff \mbox{$\NC^1 = \NL/\poly$}.  
 
The succinctness results in Fig.~\ref{pic:results}~(a), characterising the complexity of the reduction to plain database query evaluation,  are complemented by the combined complexity results in Fig.~\ref{pic:results}~(b).
\emph{Combined complexity} measures the time and space required for a best possible algorithm to answer an OMQ $\omq = (\T,\q)$ from the given class over a data instance $\A$,  as a function of the size of $\omq$ and $\A$. 
It is known~\cite{CDLLR07,ACKZ09} that the general OMQ answering problem is \NP-complete for combined complexity---that is, of the same complexity as standard CQ evaluation in databases. However, answering tree-shaped OMQs turns out to be \NP-hard~\cite{DBLP:conf/dlog/KikotKZ11} in contrast to the well-known tractability of evaluating tree-shaped and bounded-treewidth CQs~\cite{DBLP:conf/vldb/Yannakakis81,DBLP:journals/tcs/ChekuriR00,DBLP:conf/icalp/GottlobLS99}. Here, we prove that, surprisingly, answering OMQs with ontologies of bounded depth and CQs of bounded treewidth is no harder than evaluating CQs of bounded treewidth, that is, \LOGCFL-complete. By restricting further the class of CQs to trees with a bounded number of leaves, we obtain an even better \NL-completeness result, which matches the complexity of evaluating the underlying CQs.   
If we consider bounded-leaf tree-shaped CQs coupled with arbitrary \OWLQL{} ontologies, then the OMQ answering problem remains tractable, \LOGCFL-complete to be more precise.
 
The plan of the paper is as follows. Section~\ref{sec:definitions} gives formal definitions of \OWLQL{}, OMQs and rewritings. Section~\ref{sec:TW} defines the tree-witness rewriting. Section~\ref{sec:Bfunctions} reduces the succinctness problem for OMQ rewritings to the succinctness problem for hypergraph Boolean functions associated with tree-witness rewritings, and introduces hypergraph programs for computing these functions. Section~\ref{sec:OMQs&hypergraphs} establishes a correspondence between the  OMQ classes in Fig.~\ref{pic:results} and the structure of the corresponding hypergraph functions and programs. Section~\ref{sec:circuit_complexity} characterises the computational power of hypergraph programs in these classes by relating them to standard models of computation for Boolean functions. Section~\ref{sec:7} uses  the results of the previous three sections and some known facts from circuit complexity to obtain the upper and lower bounds on the size of PE-, NDL- and FO-rewritings in Fig.~\ref{pic:results}~(a). Section~\ref{sec:complexity} establishes the combined complexity results in Fig.~\ref{pic:results}~(b). We conclude in Section~\ref{sec:conclusions} by discussing the obtained succinctness and complexity results and formulating a few open problems.
All omitted proofs can be found in the appendix.

%**********************

\section{\OWLQL{} ontology-mediated queries and first-order rewritability}\label{sec:definitions}

In first-order logic, any \OWLQL{} \emph{ontology} (or \emph{TBox} in description logic parlance), $\T$, can be given as a finite set of sentences (often called \emph{axioms}) of the following forms 
\begin{align*}
& \forall x\,\big(\tau(x) \to \tau'(x)\big), &
& \forall x\, \big(\tau(x) \land \tau'(x) \to \bot \big),\\
& \forall x,y\,\big(\varrho(x,y) \to \varrho'(x,y)\big), &
& \forall x,y\,\big(\varrho(x,y) \land \varrho'(x,y) \to \bot\big),\\
& \forall x\, \varrho(x,x) , &
& \forall x\,\big(\varrho(x,x) \to \bot\big),
\end{align*}
where the formulas $\tau(x)$ (called \emph{classes} or \emph{concepts}) and $\varrho(x,y)$ (called \emph{properties} or \emph{roles}) are defined, using unary predicates $A$ and binary predicates $P$, by the grammars
\begin{equation}\label{syntax}
\tau(x) \ ::= \ \top \ \mid \ A(x) \ \mid \ \exists y\,\varrho(x,y) \qquad \text{and} \qquad \varrho(x,y) \ ::= \ \top \ \mid \ P(x,y) \ \mid \ P(y,x).
\end{equation}
(Strictly speaking, \OWLQL{} ontologies can also contain inequalities $a \ne b$, for constants $a$ and $b$. However, they do not have any impact on the problems considered in this paper, and so will be ignored.) 
\begin{example}\label{ex:NPDontology}
To illustrate, we show a snippet of the NPD FactPages ontology:
\begin{align*}
& \forall x \, (\textit{GasPipeline}(x) \to \textit{Pipeline}(x)),\\
& \forall x \, (\textit{FieldOwner}(x) \leftrightarrow \exists y \, \textit{ownerForField}(x,y)),\\
& \forall y \, (\exists x \, \textit{ownerForField}(x,y) \to  \textit{Field}(y)),\\
& \forall x,y \, (\textit{shallowWellboreForField}(x,y) \to  \textit{wellboreForField}(x,y)),\\
& \forall x,y \, (\textit{isGeometryOfFeature}(x,y) \leftrightarrow  \textit{hasGeometry}(y,x)).
\end{align*}
\end{example}

To simplify presentation, in our ontologies we also use sentences of the form
\begin{equation}\label{eq:sugar}
\forall x\, \big(\tau(x) \to \zeta(x)\big),
\end{equation}
where
\begin{equation*}
\zeta(x) \ ::= \ \tau(x) \ \mid \ \zeta_1(x) \land \zeta_2(x) \ \mid \ \exists y\, \big(\varrho_1(x,y) \land \dots \land \varrho_k(x,y) \land \zeta(y)\big).
\end{equation*}
It is readily seen that such sentences are just syntactic sugar and can be eliminated by means of polynomially many fresh predicates. Indeed, any axiom of the form~\eqref{eq:sugar} with 
\begin{equation*}
\zeta(x) = \exists y\, \big(\varrho_1(x,y) \land \dots \land \varrho_k(x,y) \land \zeta'(y)\big)
\end{equation*}
can be (recursively) replaced by the following axioms, for a fresh $P_\zeta$ and $i=1,\dots,k$:
\begin{equation}\label{eq:replacement}
\forall x \, \bigl(\tau(x) \to \exists y\, P_\zeta(x,y)\bigr),\quad
\forall x,y\,\bigl(P_\zeta(x,y) \to \varrho_i (x,y)\bigr),   \quad
\forall y\, \bigl(\exists x\, P_\zeta(x,y) \to \zeta'(y)\bigr)
\end{equation}
because any first-order structure is a model of~\eqref{eq:sugar} iff it is a  restriction of some model of~\eqref{eq:replacement} to the signature of~\eqref{eq:sugar}. The result of eliminating the syntactic sugar from an ontology $\T$  is called the \emph{normalisation} of $\T$. We always assume that all of our ontologies are normalised even though this is not done explicitly; however, we stipulate (without loss of generality) that the normalisation predicates $P_\zeta$ never occur in the data.

When writing ontology axioms, we usually omit the universal quantifiers. We typically use the characters $P$, $R$ to denote binary predicates, $A$, $B$, $C$ for unary predicates, and $S$  for either of them. For a binary predicate $P$, we write $P^-$ to denote its inverse; that is, $P(x,y) = P^-(y,x)$, for any $x$ and $y$, and $P^{--}=P$.

A \emph{conjunctive query} (CQ) $\q(\avec{x})$ is a formula of the form $\exists \avec{y}\, \varphi(\avec{x}, \avec{y})$, where $\varphi$ is a conjunction of atoms $S(\avec{z})$ all of whose variables are among $\avec{x}$, $\avec{y}$.

\begin{example}
Here is a (fragment of a) typical CQ from the NPD FactPages:
\begin{multline*}
\q(x_1,x_2,x_3) ~=~ \exists y,z \, \big[\textit{ProductionLicence}(x_1) \land
\textit{ProductionLicenceOperator}(y) \land{} \\
\textit{dateOperatorValidFrom}(y, x_2) \land
\textit{licenceOperatorCompany}(y, z) \land{}\\
\textit{name}(z, x_3) \land \textit{operatorForLicence}(y, x_1) \big].
\end{multline*}
\end{example}

To simplify presentation and without loss of generality, we assume that CQs do not contain constants. Where convenient, we regard a CQ as the \emph{set} of its atoms; in particular, $|\q|$ is the \emph{size of $\q$}. The variables in $\avec{x}$ are called the \emph{answer variables} of a CQ $\q(\avec{x})$.
A CQ without answer variables is called \emph{Boolean}.
With every CQ $\q$, we associate its \emph{Gaifman graph} $G_{\q}$ whose vertices are the variables of $\q$ and whose edges are the  pairs $\{u,v\}$ such that $P(u,v)\in\q$, for some $P$. A CQ $\q$ is \emph{connected} if the graph $G_\q$ is connected.
We call $\q$ \emph{tree-shaped} if $G_{\q}$ is a tree\footnote{Tree-shaped CQs also go by the name of \emph{acyclic queries}~\cite{DBLP:conf/vldb/Yannakakis81,DBLP:conf/ijcai/BienvenuOSX13}.}\!\!, and if $G_{\q}$ is a tree with at most two leaves, then $\q$ is said to be \emph{linear}. 

An \OWLQL{} \emph{ontology-mediated query} (OMQ) is a pair $\omq(\avec{x}) = (\T,\q(\avec{x}))$, where $\T$ is an \OWLQL{} ontology and $\q(\avec{x})$ a CQ. The \emph{size of $\omq$} is defined as $|\omq| = |\T| + |\q|$, where $|\T|$ is the number of symbols in $\T$.  

A \emph{data instance}, $\A$, is a finite set of unary or binary ground atoms (called an \emph{ABox} in description logic).
We denote by $\ind(\A)$ the set of individual constants in $\A$.
Given an OMQ $\omq(\avec{x})$ and a data instance $\A$, a tuple $\avec{a}$ of constants from $\ind(\A)$ of length $|\avec{x}|$ is called a \emph{certain answer to $\omq(\avec{x})$ over} $\A$ if $\mathcal{I} \models \q(\avec{a})$ for all models $\mathcal{I}$ of $\T\cup\A$; in this case we write \mbox{$\T,\A \models \q(\avec{a})$}. If $\q(\avec{x})$ is Boolean, a certain answer to $\omq$ over $\A$ is `yes' if $\T,\A \models \q$, and `no' otherwise. We remind the reader~\cite{Libkin} that, for any CQ $\q(\avec{x}) = \exists \avec{y}\, \varphi(\avec{x}, \avec{y})$, any first-order structure $\I$  and any tuple $\avec{a}$ from its domain $\Delta$, we have $\I \models \q(\avec{a})$ iff there is a map $h \colon \avec{x} \cup \avec{y} \to \Delta$ such that
(\emph{i}) if $S(\avec{z}) \in \q$ then $\I \models S(h(\avec{z}))$, and (\emph{ii}) $h(\avec{x}) = \avec{a}$.
If (\emph{i}) is satisfied then $h$ is called a \emph{homomorphism} from $\q$ to $\I$, and we write $h \colon \q \to \I$; if (\emph{ii}) also holds, we write $h \colon \q(\avec{a}) \to \I$.

Central to OBDA is the notion of OMQ rewriting that reduces the problem of finding certain answers to standard query evaluation.
More precisely, an FO-formula $\q'(\avec{x})$, possibly with equality, $=$, is an \emph{FO-rewriting of an OMQ $\omq(\avec{x}) = (\T,\q(\avec{x}))$} if, for \emph{any}  data instance $\A$ (without the normalisation predicates for $\T$) and any tuple $\avec{a}$ in $\ind(\A)$, 
\begin{equation}\label{def:rewriting}
\T, \A \models \q(\avec{a}) \qquad \text{iff} \qquad \I_\A \models \q'(\avec{a}),
\end{equation}
where $\I_\A$ is the first-order structure over the domain $\ind(\A)$ such that $\I_\A \models S(\avec{a})$ iff $S(\avec{a}) \in \A$, for any ground atom $S(\avec{a})$. 
As $\A$ is arbitrary, this definition implies, in particular, that the rewriting must be \emph{constant-free}. 
If $\q'(\avec{x})$ is a positive existential formula---that is, $\q'(\avec{x}) = \exists \avec{y}\, \varphi(\avec{x}, \avec{y})$ with $\varphi$  constructed from atoms (possibly with equality) using $\land$ and $\lor$ only---we call it a \emph{PE-rewriting of $\omq(\avec{x})$}. A PE-rewriting whose matrix $\varphi$ is a disjunction of conjunctions is known as a \emph{UCQ-rewriting}; if $\varphi$ takes the form ${\land}{\lor}{\land}{\lor}$ we call it a $\mathsf{\Pi}_4$-\emph{rewriting}. 
The size $|\q'|$ of  $\q'$ is the number of symbols in it. 

We also consider rewritings in the form of nonrecursive datalog  queries.
Recall~\cite{Abitebouletal95} that a \emph{datalog program}, $\Pi$, is a finite set of Horn clauses
$\forall \avec{x}\, (\gamma_1 \land \dots \land \gamma_m \to \gamma_0)$,
where each $\gamma_i$ is an atom $P(x_1,\dots,x_l)$ with  $x_i \in \avec{x}$. The atom $\gamma_0$ is the \emph{head} of the clause, and $\gamma_1,\dots,\gamma_m$ its (possibly empty) \emph{body}.  
A predicate $S$ \emph{depends} on $S'$ in $\Pi$ if $\Pi$ has a clause with $S$ in the head and $S'$ in the body; $\Pi$ is  \emph{nonrecursive} if this dependence relation is acyclic. 

Let $\omq = (\T,\q(\avec{x}))$ be an OMQ, $\Pi$ a constant-free nonrecursive program,  and $G(\avec{x})$ a predicate. The pair $\q'(\avec{x}) = (\Pi,G(\avec{x}))$ is an \emph{\NDL-rewriting of $\omq$} if, for any data instance $\A$ and any tuple $\avec{a}$ in $\ind (\A)$, we have $\T,\A\models \q(\avec{a})$ iff $\Pi(\I_\A) \models G(\avec{a})$, where 
$\Pi(\I_\A)$ is the structure with domain $\ind(\A)$ obtained by closing $\I_\A$ under the clauses in~$\Pi$. Every \PE-rewriting can clearly be represented as an \NDL-rewriting of linear size. 

\begin{remark}\label{dom-ind}
As defined, \FO- and \PE-rewritings are not necessarily domain-independent queries, while \NDL-rewritings are not necessarily safe~\cite{Abitebouletal95}. For example, $(x=x)$ is a PE-rewriting of the OMQ $(\{\forall x \, P(x,x)\},P(x,x))$, and the program $(\{\top \to A(x)\}, A(x))$ is an \NDL-rewriting of the OMQ $(\{\top\to A(x)\}, A(x))$.  Rewritings can easily be made domain-independent and safe by relativising their variables to the predicates in the data signature (relational schema). For instance, if this signature is $\{A, P\}$, then a domain-independent relativisation of $(x=x)$ is the \PE-rewriting $\bigl(A(x) \lor \exists y\,P(x,y) \lor \exists y\, P(y,x)\bigr) \land (x = x)$. Note that if we exclude from \OWLQL{} reflexivity statements and axioms with $\top$ on the left-hand side, then rewritings are guaranteed to be domain-independent, and no relativisation is required. In any case, rewritings are always interpreted under the \emph{active domain semantics} adopted in databases; see~\eqref{def:rewriting}.
\end{remark}

As mentioned in the introduction, the \OWLQL{} profile of \OWL{} was designed to ensure FO-rewritability of all OMQs with ontologies in the profile or, equivalently, OMQ answering in \ACz{} for data complexity.
It should be clear, however, that for the OBDA approach to work in practice, the  rewritings of OMQs must be of `reasonable shape and size'\!. Indeed, it was observed experimentally~\cite{DBLP:journals/semweb/CalvaneseGLLPRRRS11} and also established theoretically~\cite{DBLP:conf/icalp/KikotKPZ12} that sometimes the rewritings are prohibitively large---exponentially-large in the size of the original CQ, to be more precise. 
These facts imply that, in the context of OBDA, we should actually be interested not in arbitrary but in \emph{polynomial-size} rewritings. In complexity-theoretic terms, the focus should not only be on the data complexity of OMQ answering, which is an appropriate measure for database query evaluation (where queries are indeed usually small)~\cite{Vardi82}, but also on the combined complexity that takes into account the contribution of ontologies and queries.

%***************************

\section{Tree-Witness Rewriting}\label{sec:TW}

Now we define one particular rewriting of \OWLQL{} OMQs that will play a key role in the succinctness and complexity analysis later on in the paper. This rewriting is a modification of the tree-witness PE-rewriting originally introduced by Kikot et al.~\cite{DBLP:conf/kr/KikotKZ12}
(cf.~\cite{Lutz-IJCAR08,KR10our,DBLP:journals/semweb/KonigLMT15} for rewritings based on similar ideas).

We begin with two simple observations that will help us remove unneeded clutter from the definitions. Every \OWLQL{} ontology $\T$ consists of two  parts: $\T^-$, which contains all the sentences with $\bot$, and the remainder, $\T^+$, which is consistent with every data instance. For any $\psi(\avec{z}) \to \bot$ in $\T^-$, consider the Boolean CQ $\exists \avec{z} \, \psi(\avec{z})$. It is not hard to see that, for any OMQ $(\T,\q(\avec{x}))$ and data instance $\A$, a tuple $\avec{a}$ is a certain answer to $(\T,\q(\avec{x}))$ over $\A$ iff either $\T^+,\A \models \q(\avec{a})$ or $\T^+,\A \models \exists \avec{z} \, \psi(\avec{z})$, for some $\psi(\avec{z}) \to \bot$ in $\T^-$; see~\cite{DBLP:journals/ws/CaliGL12} for more details. Thus, from now on we will assume that, in all our ontologies~$\T$, the `negative' part $\T^-$ is empty, and so  they are \emph{consistent} with all data instances.

The second observation will allow us to restrict the class of data instances we need to consider when rewriting OMQs.
In general, if we only require condition \eqref{def:rewriting} to hold for any data instance $\A$ from some class $\mathfrak A$, then we call $\q'(\avec{x})$ a \emph{rewriting of $\omq(\avec{x})$ over~$\mathfrak A$}. Such  classes of data instances  can be defined, for example, by the integrity constraints in the database schema and the mapping~\cite{ISWC13}.
We say that a data instance $\A$ is \emph{complete}\footnote{Rodriguez-Muro et al.~\cite{ISWC13} used the term `H-completeness'; see also~\cite{DBLP:conf/ijcai/KonigLM15}.} for an ontology $\T$ if $\T,\A \models S(\avec{a})$ implies $S(\avec{a}) \in \A$, for any ground atom $S(\avec{a})$ with $\avec{a}$ from $\ind(\A)$. The following proposition means that from now on we will only consider rewritings over complete data instances.

\begin{proposition}\label{complete}
If $\q'(\avec{x})$ is an \NDL-rewriting of $\omq(\avec{x}) = (\T,\q(\avec{x}))$ over com\-plete data instances, then there is an \NDL-rewriting $\q''(\avec{x})$ of $\omq(\avec{x})$ over arbitrary data instances with $|\q''| \leq |\q'| \cdot |\T|$. A similar result holds for \PE- and \FO-rewritings.
\end{proposition}
\begin{proof}
Let $(\Pi, G(\avec{x}))$ be an \NDL-rewriting of $\omq(\avec{x})$ over complete data instances.
Denote by $\Pi^*$ the result of replacing each predicate $S$ in $\Pi$ with a fresh predicate $S^*$. Define $\Pi'$ to be the union of $\Pi^*$ and the following clauses for predicates in $\Pi$:  
\begin{align*}
\tau(x) & \to A^*(x),\qquad \text{ if } \ \T \models \tau(x) \to A(x),\\
\varrho(x,y) & \to P^*(x,y), \quad\text{ if } \ \T \models \varrho(x,y) \to P(x,y),\\
\top & \to P^*(x,x), \quad\text{ if } \ \T\models P(x,x)
\end{align*}
(the empty body is denoted by $\top$).
It is readily seen that $(\Pi',G^*(\avec{x}))$ is an \NDL-rewriting of $\omq(\avec{x})$  over arbitrary data instances. 
The cases of \PE- and \FO-rewritings are similar except that we replace $A(x)$ and $P(x,y)$ with  
\begin{equation*}
\bigvee_{\T\models \tau(x) \to A(x)} \hspace*{-2em}\tau(x) \qquad\text{ and }\qquad \bigvee_{\T\models \varrho(x,y) \to P(x,y)} \hspace*{-2em}\varrho(x,y) \quad \vee \bigvee_{\T\models P(x,x)} \hspace*{-1em}(x = y),
\end{equation*}
respectively (the empty disjunction is, by definition,  $\bot$).
%
%To make the above rewritings domain-independent, we need to introduce a new unary predicate, $D(x)$, that would represent the active domain. This predicate can be defined by the following clauses:
%%
%\begin{align*}
%& A(x) \to D(x), && \text{for all unary predicates $A$ in the signature},\\
%& P(x,y) \to D(x) \text{ and } P(y,x) \to D(x), && \text{for all binary predicates $P$ in the signature}.
%\end{align*}
%%
%The domain predicate, $D$, is then used to make all occurrences of variables in an \NDL-program safe: if a variable $x$ occurs in the head but does not occur in any of the atoms in the body, then we extend the body with $D(x)$.  In \PE- and \FO-rewritings, we modify the quantifiers: every $\exists x\,\varphi(x)$ is replaced by $\exists x\,(D(x) \land \varphi(x))$ and every $\forall x\,\varphi(x)$ by $\forall x\,(D(x) \to \varphi(x))$. 
\end{proof}

As is well-known~\cite{Abitebouletal95}, every pair $(\T,\A)$ of an ontology $\T$ and data instance $\A$ possesses a \emph{canonical model} (or \emph{chase}) $\mathcal{C}_{\T,\A}$ such that \mbox{$\T,\A \models \q(\avec{a})$} iff $\mathcal{C}_{\T,\A} \models \q(\avec{a})$, for all CQs $\q(\avec{x})$ and $\avec{a}$ in $\ind(\A)$. In our proofs, we use the following definition of $\C_{\T,\A}$, where without loss of generality we assume that $\T$ does not contain binary predicates $P$ such that $\T \models \forall x,y \, P(x,y)$. Indeed, occurrences of such $P$ in $\T$ can be replaced by $\top$ and occurrences of $P(x,y)$ in CQs can simply be removed without changing  certain answers over any data instance  (provided that $x$ and $y$ occur in the remainder  of the query).

The domain $\Delta^{\mathcal{C}_{\T,\A}}$ of the canonical model $\mathcal{C}_{\T,\A}$ consists of $\ind(\A)$ and the \emph{witnesses}, or \emph{labelled nulls}, introduced by the existential quantifiers in (the normalisation of) $\T$. More precisely, the labelled nulls in $\mathcal{C}_{\T,\A}$ are finite words of the form $w = a \varrho_1 \dots \varrho_n$ ($n \geq 1$) such that 
\begin{nitemize}
\item $a \in \ind(\A)$ and  $\T, \A \models \exists y\, \varrho_1 (a,y)$, but  $\T, \A \not\models \varrho_1(a,b)$  for any $b \in \ind(\A)$;
\item $\T\not\models \varrho_i(x,x)$ for $1 \le i \le n$;
\item $\T \models \exists x\, \varrho_i(x,y) \to \exists z\, \varrho_{i+1}(y,z)$  and $\T \not \models \varrho_i(y,x) \to \varrho_{i+1}(x,y)$ for $1 \leq i < n$.
\end{nitemize}
Every individual name $a \in \ind(\A)$ is interpreted in $\mathcal{C}_{\T,\A}$ by itself, and unary and binary predicates are interpreted as follows: for any $u,v \in \Delta^{\mathcal{C}_{\T,\A}}$,
\begin{nitemize}
\item $\mathcal{C}_{\T,\A} \models A(u)$ iff either $u \in \ind(\A)$ and $\T,\A \models A(u)$, or $u = w\varrho$, for some $w$ and $\varrho$ with $\T \models \exists y\,\varrho(y,x) \to A(x)$;
\item $\mathcal{C}_{\T,\A} \models P(u,v)$ iff one of the following holds: (\emph{i}) $u,v \in \ind(\A)$ and $\T,\A \models P(u,v)$; (\emph{ii}) $u=v$ and $\T\models P(x,x)$; (\emph{iii}) $v = u\varrho$ and $\T \models \varrho(x,y) \to P(x,y)$; (\emph{iv}) $u = v\varrho^-$ and $\T \models \varrho(x,y) \to P(x,y)$.
\end{nitemize}

\begin{example}\label{ex:canon}
Consider the following ontologies: 
\begin{align*}
& \T_1 \ = \ \{ \ A(x) \to \exists y\, \bigl(R(x,y) \land Q(y,x)\bigr)\ \},\\
& \T_2 \ = \ \{\ A(x) \to \exists y\, R(x,y), \ \ \exists x\,R(x,y) \to \exists z\, Q(z,y)\  \},\\
& \T_3 \ = \ \{\ A(x) \to \exists y\, R(x,y), \ \ \exists x\,R(x,y) \to \exists z\, R(y,z)\  \}.
\end{align*}
The canonical models of $(\T_i,\A)$, for $\A = \{A(a)\}$, $i=1,2,3$, are shown in Fig.~\ref{fig:canonical}, where $\zeta(x) = \exists y\, (R(x,y) \land Q(y,x))$ and $P_\zeta$ is the corresponding normalisation predicate. When depicting canonical models, we use \begin{tikzpicture}\node[bpoint] at (0,0) {};\end{tikzpicture} for constants and \begin{tikzpicture}\node[point] at (0,0) {};\end{tikzpicture} for labelled nulls.
\begin{figure}[t]%
\centering%
\begin{tikzpicture}\footnotesize
\node[bpoint, label=right:{$A$}, label=left:{\scriptsize $a$}] (a0) at (0,0) {};
\node at (-1.3,0.1) {\normalsize $\C_{\T_1,\A}$};
\node[point, label=left:{\scriptsize $aP_\zeta$}] (a1) at (0,1) {};
\draw[can,->] (a0)  to node[label=left:{\textcolor{gray}{$P_\zeta$}}, label=right:{$R,Q^-$}]{} (a1);
\node[bpoint, label=right:{$A$}, label=left:{\scriptsize $a$}] (b0) at (5,0) {};
\node at (3.8,0.1) {\normalsize $\C_{\T_2,\A}$};
\node[point, label=left:{\scriptsize $aR$}] (b1) at (5,1) {};
\node[point, label=left:{\scriptsize $aRQ^-$}] (b2) at (5,2) {};
\draw[can,->] (b0)  to node[label=right:{$R$}]{} (b1);
\draw[can,->] (b1)  to node[label=right:{$Q^-$}]{} (b2);
\node[bpoint, label=right:{$A$}, label=left:{\scriptsize $a$}] (c0) at (10,0) {};
\node at (8.8,0.1) {\normalsize $\C_{\T_3,\A}$};
\node[point, label=left:{\scriptsize $aR$}] (c1) at (10,1) {};
\node [point, label=left:{\scriptsize $aRR$}] (c2) at (10,2) {};
\draw[can,->] (c0)  to node[label=right:{$R$}]{} (c1);
\draw[can,->] (c1)  to node[label=right:{$R$}]{} (c2);
\draw[dotted,thick] (c2) -- ++(0,0.4);
\end{tikzpicture}%
\caption{Canonical models in Example~\ref{ex:canon}.}\label{fig:canonical}
\end{figure}
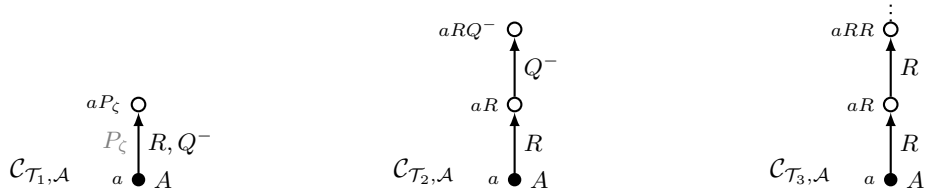%
\end{example}

For any ontology $\T$ and any formula $\tau(x)$ given by~\eqref{syntax}, we  denote by $\C_\T^{\smash{\tau(a)}}$ the canonical model of $(\T \cup \{ A(x) \to \tau(x)\} , \{ A(a) \})$, for a fresh unary predicate $A$. We say that $\T$ is \emph{of depth $k$}, $1 \le k < \omega$, if (\emph{i}) there is no $\varrho$ with $\T \models \varrho(x,x)$, (\emph{ii}) at least one of the $\C_\T^{\smash{\tau(a)}}$ contains a word $a \varrho_1 \dots \varrho_k$, but (\emph{iii}) none of the $\C_\T^{\smash{\tau(a)}}$ has such a word of greater length.
Thus, $\T_1$ in Example~\ref{ex:canon} is of depth 1, $\T_2$ of depth 2, while $\T_3$ is not of any finite depth.

Ontologies of infinite depth generate infinite canonical models. However, \OWLQL{} has the \emph{polynomial derivation depth property}  (PDDP) in the sense that there is a polynomial $p$ such that, for any OMQ $\omq(\avec{x}) = (\T, \q(\avec{x}))$, data instance $\A$ and $\avec{a}$ in $\ind(\A)$, we have $\T,\A \models \q(\avec{a})$ iff $\q(\avec{a})$ holds in the sub-model of $\mathcal{C}_{\T,\A}$ whose domain consists of words of the form $a \varrho_1 \dots \varrho_n$ with $n \le p(|\omq|)$~\cite{DBLP:conf/pods/JohnsonK82,DBLP:journals/ws/CaliGL12}. (In general, the bounded derivation depth property of an ontology language is a necessary and sufficient condition of FO-rewritability~\cite{DBLP:journals/ai/GottlobKKPSZ14}.)

We call a set $\Omega_{\omq}$ of words of the form $w=\varrho_1 \dots \varrho_n$ \emph{fundamental for $\omq$} if, for any  $\A$ and $\avec{a}$ in $\ind(\A)$, we have $\T,\A \models \q(\avec{a})$ iff $\q(\avec{a})$ holds in the sub-model of $\mathcal{C}_{\T,\A}$ with the domain $\{aw \mid a \in \ind(\A),\ w \in \Omega_{\omq}\}$.
We say that a class $\mathcal{Q}$ of OMQs has the \emph{polynomial  fundamental set property} (PFSP) if there is a polynomial $p$  such that every $\omq \in \mathcal{Q}$ has a  fundamental set~$\Omega_{\omq}$ with $|\Omega_{\omq}| \le p(|\omq|)$.
The class of \emph{all} OMQs (even with ontologies of finite depth and tree-shaped CQs) does not have the PFSP~\cite{DBLP:conf/icalp/KikotKPZ12}. On the other hand, it should be clear that the class of OMQs with ontologies of bounded depth does enjoy the PFSP. A less trivial example is given by the following theorem, which is an immediate consequence of Theorem~\ref{thm:noroles} below: 

\begin{theorem}\label{role-inc}
The class of OMQs whose ontologies do not contain axioms of the form $\varrho(x,y) \to \varrho'(x,y)$ \textup{(}and syntactic sugar~\eqref{eq:sugar}\textup{)} enjoys the PFSP.
\end{theorem}
%
%\begin{proof}
%Consider sub-queries of the form $R_1(y_1,y_2),R_2(y_2,y_3),\dots R_n(y_n,y_{n+1})$. They give rise to words $R_1,R_2,\dots,R_n$. We simplify them by applying the rule $R,R^- \to \varepsilon$. Take all words of the form $aww'$ such that $w'$ is one of the words obtained in this way and $|aww'| \le f(|\T|,|\q|)$ (with $f$ given by the PDDP).
%\end{proof}

We are now in a position to define the tree-witness PE-rewriting of \OWLQL{} OMQs.
Suppose we are given an OMQ $\omq(\avec{x}) = (\T,\q(\avec{x}))$ with $\q(\avec{x}) = \exists \avec{y}\, \varphi(\avec{x}, \avec{y})$. For a pair $\t = (\tr, \ti)$ of disjoint sets of variables in $\q$, with $\ti\subseteq \avec{y}$\footnote{We (ab)use set-theoretic notation for lists and, for example, write $\ti\subseteq \avec{y}$ to say that every element of $\ti$ is an element of $\avec{y}$.} and $\ti \ne\emptyset$ ($\tr$ can be empty), set
\begin{equation*}
\q_\t \ = \ \bigl\{\, S(\avec{z}) \in \q \mid \avec{z} \subseteq \tr\cup \ti \text{ and } \avec{z}\not\subseteq \tr\,\bigr\}.
\end{equation*}
%
%We call $\t = (\tr, \ti)$ a \emph{tree witness for $\omq$ generated by $\tau$} if $\q_\t$ is a minimal subset of $\q$ for which there is a homomorphism $h \colon \q_\t  \to \C_\T^{\tau(a)}$ such that $\tr = h^{-1}(a)$ and $\q_\t$ contains every atom of $\q$ with at least one variable from $\ti$. 
%
%\tocheck{If $\tau(x)=\exists y \varrho(x,y)$ and the homomorphism $h$ is such that $h(v) \in \{a \varrho w \mid w \in (\role_\T)^*\}$ for every variable $v\in \t_i$, then we say that $\varrho$ \emph{strongly generates} $\t$.}
%
If $\q_\t$ is a minimal subset of $\q$ for which there is a homomorphism $h \colon \q_\t  \to \C_\T^{\smash{\tau(a)}}$ such that $\tr = h^{-1}(a)$ and $\q_\t$ contains every atom of $\q$ with at least one variable from $\ti$, then we call $\t = (\tr, \ti)$ a \emph{tree witness for $\omq$ generated by $\tau$} (and \emph{induced by $h$}). 
Observe that if $\tr = \emptyset$ then $\q_\t$ is a connected component of $\q$; in this case we call $\t$ \emph{detached}.
Note also that the same tree witness $\t = (\tr, \ti)$ can be generated by different $\tau$.
Now, we set 
\begin{equation}\label{tw-formula}
\tw_{\t}(\tr)  ~=~   \exists z\,\bigl(\bigwedge_{x \in \tr} (x=z) \ \ \land  \hspace*{-1mm}\bigvee_{\t \text{ generated by } \tau} \hspace*{-1.5em}\tau(z)\bigr).
\end{equation}
The variables in $\ti$ do not occur in $\tw_{\t}$ and are called \emph{internal}. The variables in $\tr$, if any, are called \emph{root variables}. Note that no answer variable in $\q(\avec{x})$ can be internal. The length $|\tw_\t|$ of $\tw_\t$ is $O(|\omq|)$. Tree witnesses $\t$ and $\t'$ are \emph{conflicting} if $\q_\t \cap \q_{\t'} \ne \emptyset$. Denote by $\twset$ the set of tree witnesses for $\omq(\avec{x})$. A subset $\Theta\subseteq \twset$ is  \emph{independent} if no pair of distinct tree witnesses in it is conflicting. Let $\q_\Theta = \bigcup_{\t\in\Theta} \q_\t$.
The following PE-formula is called the \emph{tree-witness rewriting of $\omq(\avec{x})$ over complete data instances}: 
\begin{equation}\label{rewriting0}
\qtw(\avec{x}) \ \ =  \hspace*{-0.3em}  \bigvee_{\Theta \subseteq \twset \text{ independent}} \hspace*{-0.2em} \exists\avec{y}\  \bigl(\hspace*{-1mm}
\bigwedge_{S(\avec{z}) \in \q \setminus \q_\Theta}\hspace*{-1mm} S(\avec{z})
 \ \land \ \bigwedge_{\t\in\Theta} \tw_\t(\tr) \,\bigr).
\end{equation}

\begin{remark}\label{ignored}
As the normalisation predicates $P_\zeta$  cannot occur in data instances, we can omit from~\eqref{tw-formula} all the disjuncts with $P_\zeta$. For the same reason, the tree witnesses generated only by concepts with normalisation predicates will be ignored in the sequel. 
\end{remark}

\begin{example}\label{ex:conf}
Consider the OMQ $\omq(x_1,x_2) = (\T,\q(x_1,x_2))$ with
\begin{align*}
\T & \ = \ \bigl\{A_1(x) \to \underbrace{\exists y\, \bigl(R_1(x,y) \land Q(x,y)\bigr)}_{\zeta_1(x)},\
A_2(x) \to \underbrace{\exists y\, \bigl(R_2(x,y) \land Q(y,x)\bigr)}_{\zeta_2(x)}\bigr\},\\
\q(x_1, x_2)  & \ = \ \exists y_1,y_2\, \bigl(R_1(x_1,y_1)\land Q(y_2,y_1)\land R_2(x_2,y_2)\bigr).
\end{align*}
The CQ $\q$ is shown in Fig.~\ref{fig:conflicting-tws} alongside $\C_{\T}^{\smash{A_1(a)}}$ and $\C_{\T}^{\smash{A_2(a)}}$. When  depicting CQs, we use \begin{tikzpicture}\node[bpoint] at (0,0) {};\end{tikzpicture} for answer variables and \begin{tikzpicture}\node[point] at (0,0) {};\end{tikzpicture} for existentially quantified variables.
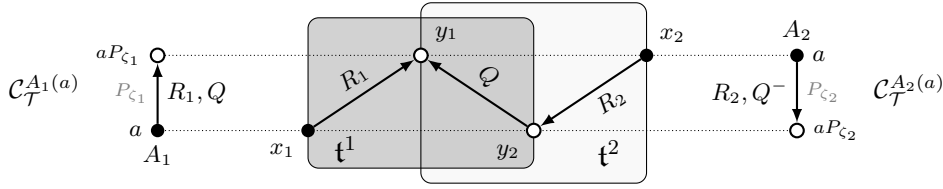
\begin{figure}[t]%
\centering%
\begin{tikzpicture}\footnotesize
\coordinate (c1) at (0,0);
\coordinate (c2) at (1.5,1);
\coordinate (c3) at (3,0);
\coordinate (c4) at (4.5,1);
\node [fill=gray!5,thin,rounded corners,inner xsep=0mm,inner ysep=7mm,fit=(c2) (c3) (c4)]{};
\node [draw,fill=gray!70,thin,fill opacity=0.5,rounded corners,inner xsep=0mm,inner ysep=5mm,fit=(c1) (c2) (c3)]{};
\node [draw,thin,rounded corners,inner xsep=0mm,inner ysep=7mm,fit=(c2) (c3) (c4)]{};
\draw[hom] (-2,1) -- (6.5,1);
\draw[hom] (-2,0) -- (6.5,0);
\node at (0.5,-0.25) {\large $\t^1$};
\node at (4,-0.3) {\large $\t^2$};
\node[bpoint, label=below left:{$x_1$}] (t1) at (c1) {};
\node[point, label=above right:{$y_1$}] (t2) at (c2) {};
\node[point, label=below left:{$y_2$}] (t3) at (c3) {};
\node[bpoint, label=above right:{$x_2$}] (t4) at (c4) {};
\draw[->,query] (t1)  to node [above, sloped]{$R_1$} (t2);
\draw[->,query] (t3)  to node [above, sloped]{$Q$} (t2);
\draw[->,query] (t4)  to node [pos=0.4, below, sloped]{$R_2$} (t3);
\node[bpoint, label=above:{$A_2$}, label=right:{$a$}] (cd1) at (6.5,1) {};
\node[point,label=right:{\scriptsize $aP_{\zeta_2}$}] (cd2) at (6.5,0) {};
\draw[can,->] (cd1)  to node [left]{$R_2,Q^-$} node [right]{\scriptsize\textcolor{gray}{$P_{\zeta_2}$}} (cd2);
\node at (8,0.5) {\normalsize $\C_{\T}^{\smash{A_2(a)}}$};
\node[bpoint, label=below:{$A_1$}, label=left:{$a$}] (cc1) at (-2,0) {};
\node[point, label=left:{\scriptsize $aP_{\zeta_1}$}] (cc2) at (-2,1) {};
\draw[can,->] (cc1)  to node [right]{$R_1,Q$} node [left]{\scriptsize\textcolor{gray}{$P_{\zeta_1}$}} (cc2);
\node at (-3.5,0.5) {\normalsize $\C_{\T}^{\smash{A_1(a)}}$};
\end{tikzpicture}%
\caption{Tree witnesses in Example~\ref{ex:conf}.}\label{fig:conflicting-tws}
\end{figure}
There are two tree witnesses, $\t^1$ and $\t^2$, for $\omq$ with
\begin{equation*}
\q_{\t^1} = \bigl\{\,R_1(x_1,y_1), Q(y_2,y_1)\,\bigr\}  \quad\text{ and }\quad \q_{\t^2} = \bigl\{\, Q(y_2,y_1), R_2(x_2,y_2)\,\bigr\}
\end{equation*}
shown in Fig.~\ref{fig:conflicting-tws} by the dark and light shading, respectively. 
The tree witness $\t^1 = (\tr^1, \ti^1)$
with $\tr^1 = \{x_1,y_2\}$ and $\ti^1 =\{y_1\}$ is generated by $A_1(x)$,  which gives
\begin{equation*}
\tw_{\t^1}(x_1,y_2)  = \exists z \, \bigl(A_1(z)\land(x_1 = z) \land (y_2 =z)\bigr). % \lor \exists y\, R_k(z,y)
\end{equation*}
(Recall that although $\t^1$ is also generated by $\exists y\,P_{\zeta_1}(y,z)$, we do not include it in the disjunction in $\tw_{\t^1}$ because $P_{\zeta_1}$ cannot occur in data instances.)
Symmetrically, the tree witness $\t^2$ gives
\begin{equation*}
\tw_{\t^2}(x_2,y_1)  = \exists z \, \bigl(A_2(z)\land(x_2 = z) \land (y_1 =z)\bigr). % \lor \exists y\, R_k(z,y)
\end{equation*}
As $\t^1$ and $\t^2$ are conflicting, $\twset$ contains three independent subsets: $\emptyset$, $\{\t^1\}$ and $\{\t^2\}$.  Thus, we obtain the following tree-witness rewriting $\q_\tw(x_1,x_2)$ of $\omq$ over complete data instances:
\begin{equation*}
\exists y_1,y_2 \, \big[\big(R_1(x_1,y_1) \land Q(y_2,y_1) \land R_2(x_2,y_2)\big) \lor{}
\big(\tw_{\t^1} \land R_2(x_2,y_2)\big) \lor \big(R_1(x_1,y_1) \land \tw_{\t^2} \big) \big].
\end{equation*}
\end{example}

\begin{theorem}[\cite{DBLP:conf/kr/KikotKZ12}]
For any OMQ $\omq(\avec{x}) = (\T,\q(\avec{x}))$, any data instance $\A$, which is  complete for $\T$, and any tuple $\avec{a}$ from $\ind(\A)$, we have $\T,\A \models \q(\avec{a})$ iff $\I_\A \models \qtw(\avec{a})$.
In other words, $\q_\tw$ is a rewriting of $\omq(\avec{x})$ over complete data instances.
\end{theorem}

Intuitively, for every homomorphism $h \colon \q(\avec{a}) \to \C_{\T,\A}$, the sub-CQs of $\q$ mapped by~$h$ to sub-models of the form $\C^{\smash{\tau(a)}}_\T$ define an independent set $\Theta$ of tree witnesses; see Fig.~\ref{fig:twr}. Conversely, if $\Theta$ is such a set, then the homomorphisms corresponding to the tree witnesses in $\Theta$ can be pieced together into a homomorphism from $\q(\avec{a})$ to $\C_{\T,\A}$---provided that the $S(\avec{z})$ from $\q \setminus \q_\Theta$ and the $\tw_\t(\tr)$ for $\t \in \Theta$ hold in $\I_\A$.

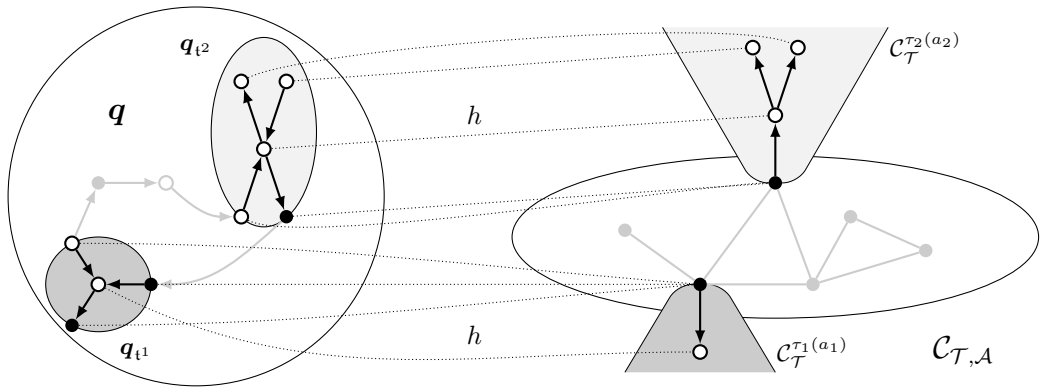
\begin{figure}[b]%
\centering%
\begin{tikzpicture}[yscale=0.9]
\draw (1.3,1.3) ellipse (2.5 and 2.8);
\draw[thin,fill=gray!40] (0,0) circle (0.7);
\node[point] (a0) at (0,0) {};
\node[bpoint] (a1) at (0:0.7) {};
\node[point] (a2) at (120:0.7) {};
\node[bpoint] (a3) at (-120:0.7) {};
\draw[query,->] (a1) -- (a0);
\draw[query,->] (a2) -- (a0);
\draw[query,<-] (a3) -- (a0);
\draw[thin,fill=gray!10] (2.2,2.25) ellipse (0.7 and 1.4);
\node[point] (a4) at (2.2,2) {};
\node[point] (a5) at (1.9,1) {};
\node[bpoint] (a6) at (2.5,1) {};
\node[point] (a7) at (2.5,3) {};
\node[point] (a8) at (1.9,3) {};
\draw[query,->] (a5) -- (a4);
\draw[query,<-] (a6) -- (a4);
\draw[query,->] (a7) -- (a4);
\draw[query,<-] (a8) -- (a4);
\node[bpoint,gray!40] (a9) at (0,1.5) {};
\node[point,gray!40,fill=white] (a10) at (0.9,1.5) {};
\draw[query,->,gray!40] (a9) -- (a10);
\draw[query,->,gray!40,out=-45,in=180] (a10) to (a5);
\draw[query,<-,gray!40] (a9) -- (a2);
\draw[query,<-,gray!40,out=0,in=-135] (a1) to (a6);
\node at (0.25,2.5) {\large $\q$};
\node at (0.5,-1) {\small $\q_{\t^1}$};
\node at (1.3,3.5) {\small $\q_{\t^2}$};
\begin{scope}[shift={(-1,0)}]
\draw (10,0.7) ellipse (3.5 and 1.2);
\draw[fill=gray!40,rounded corners=2mm,ultra thin] (8,-1.3) -- (8.7,0) -- (9.3,0) -- (10,-1.3);  
\node[bpoint] (b0) at (9,0) {};
\node[point] (b1) at (9,-1) {};
\draw[can,->] (b0) -- (b1);
\draw[fill=gray!10,rounded corners=2mm,ultra thin] (8.5,3.8) -- (9.7,1.5) -- (10.3,1.5) -- (11.5,3.8);  
\node[bpoint] (b2) at (10,1.5) {};
\node[point] (b3) at (10,2.5) {};
\draw[can,->] (b2) -- (b3);
\node[point] (b4) at (9.7,3.5) {};
\node[point] (b5) at (10.3,3.5) {};
\draw[can,->] (b3) -- (b4);
\draw[can,->] (b3) -- (b5);
\draw[hom,out=-30,in=180] (a0) to (b1);
\draw[hom] (a1) to (b0);
\draw[hom,out=0,in=175] (a2) to (b0);
\draw[hom,out=0,in=-175] (a3) to (b0);
\draw[hom] (a4) to (b3);
\draw[hom] (a7) to (b4);
\draw[hom,out=30,in=150,looseness=0.3] (a8) to (b5);
\draw[hom,out=-30,in=180,looseness=0.3] (a5) to (b2);
\draw[hom] (a6) to (b2);
\draw[can,gray!40] (b0) to (b2);
\node[bpoint,gray!40] (b6) at (11,1) {};
\node[bpoint,gray!40] (b7) at (12,0.5) {};
\node[bpoint,gray!40] (b8) at (10.5,0) {};
\node[bpoint,gray!40] (b9) at (8,0.8) {};
\draw[can,gray!40] (b6) to (b7);
\draw[can,gray!40] (b7) to (b8);
\draw[can,gray!40] (b2) to (b8);
\draw[can,gray!40] (b8) to (b6);
\draw[can,gray!40] (b0) to (b9);
\draw[can,gray!40] (b0) to (b8);
\node at (12.5,-1) {\large $\C_{\T,\A}$}; 
\node at (10.5,-1) {\small $\C_\T^{\tau_1(a_1)}$}; 
\node at (12,3.5) {\small $\C_\T^{\tau_2(a_2)}$}; 
\node at (6,2.5) {\normalsize $h$};
\node at (6,-0.75) {\normalsize $h$};
\end{scope}
\end{tikzpicture}%
\caption{Tree-witness rewriting.}\label{fig:twr}
\end{figure}

The \emph{size} of the tree-witness \PE-rewriting $\qtw$ depends on the number of tree witnesses in the given OMQ $\omq = (\T,\q)$ and, more importantly, on the cardinality of $\twset$ as we have $|\qtw| = O(2^{|\twset|} \cdot |\omq|^2)$ with $|\twset| \le 3^{|\q|}$. 

\begin{theorem}\label{thm:noroles}
OMQs $\omq = (\T,\q)$, in which $\T$ does not contain axioms of the form $\varrho(x,y) \to \varrho'(x,y)$ \textup{(}and syntax sugar~\eqref{eq:sugar}\textup{)}, have at most $3|\q|$ tree witnesses.
\end{theorem}
\begin{proof}
As observed above, there can be only one detached tree witness for each connected component of $\q$.  As $\T$ has no axioms of the form $\varrho(x,y) \to \varrho'(x,y)$, any two points in $\C_\T^{\smash{\tau(a)}}$ can be $R$-related for at most one $R$, and so no point can have more than one $R$-successor, for any $R$. It follows that, for every atom $P(x,y)$ in $\q$, there can be at most one tree witness $\t = (\tr, \ti)$ with $P(x,y)\in\q_\t$, $x\in \tr$ and $y \in \ti$ ($P^-(y,x)$ may give another tree witness). 
\end{proof}

OMQs with arbitrary axioms can have \emph{exponentially many} tree witnesses:

\begin{example}\label{ex:exp-tws}
Consider the OMQ $\omq_n = (\T,\q_n(\avec{x}^0))$, where 
\begin{align*}
\T \ & = \ \big\{ A(x) \to \exists y\, \big( R(y,x) \land \exists z\,  (R(y,z) \land B(z))\big)\big\},\\
\q_n(\avec{x}^0) \ & = \ \exists y,\avec{y}^1,\avec{x}^1,\avec{y}^2 \, \bigl( B(y) \ \ \land \bigwedge_{1\leq k\leq n} \hspace*{-0.5em}\big(R(y_k^1,y) \land R(y_k^1, x_k^1) \land R(y_k^2,x_k^1) \land R(y_k^2,x^0_k)\big)\bigr)
\end{align*}
and $\avec{x}^i$ and $\avec{y}^i$ denote vectors of $n$ variables $x^i_k$ and $y^i_k$, for $1\leq k \leq n$, respectively.
The CQ is shown in Fig.~\ref{fig:exp-tws} alongside the canonical model $\C_\T^{\smash{A(a)}}$.
\begin{figure}[t]%
\centering%
\begin{tikzpicture}\footnotesize
\node at (-1,0) {\normalsize $\q_n(\avec{x}^0)$};
\node[point, label=right:{$B$}, label=below:{$y$}] (a0) at (2,2) {};
\node[point] (a1) at (-1,1) {};
\node[point, label=below:{\normalsize $x^1_0$}] (a2) at (0,1) {};
\node[point] (a3) at (1,1) {};
\node[bpoint, label=below:{\normalsize $x^0_1$}] (a4) at (1,0) {};
\node[point] (b1) at (3,1) {};
\node[point, label=below:{\normalsize $x^1_n$}] (b2) at (4,1) {};
\node[point] (b3) at (5,1) {};
\node[bpoint, label=below:{\normalsize $x^0_n$}] (b4) at (5,0) {};
\draw[query,->] (a1)  to  (a0); % node[label=above:{\scriptsize $R$}]{}
\draw[query,->] (a1)  to  (a2); % node[label=below:{\scriptsize $R$}]{}
\draw[query,->] (a3)  to (a2); % node[label=below:{\scriptsize $R$}]{} 
\draw[query,->] (a3)  to (a4); % node[label=right:{\scriptsize $R$}]{}
\draw[query,->] (b1)  to  (a0); % node[pos=0.3,label=above:{\scriptsize $R$}]{}
\draw[query,->] (b1)  to  (b2); % node[label=below:{\scriptsize $R$}]{}
\draw[query,->] (b3)  to  (b2); % node[label=below:{\scriptsize$R$}]{}
\draw[query,->] (b3)  to  (b4); % node[label=right:{\scriptsize $R$}]{}
\node at (2,1) {\Large $\dots$};
\node[bpoint, label=left:{$A$}, label=below:{\scriptsize $a$}] (c0) at (8,0) {};
\node at (7,1.8) {\normalsize $\C_{\T}^{\smash{A(a)}}$};
\node[point] (c1) at (8,1) {};
\node[point, label=left:{$B$}] (c2) at (8,2) {};
\draw[can,->] (c0)  to node[label=left:{$R^-$}]{} (c1);
\draw[can,->] (c1)  to node[label=left:{$R$}]{} (c2);
\draw[hom] (11.5,0) -- (c0);
\draw[hom] (11.5,1) -- (c1);
\draw[hom] (11.5,2) -- (c2);
\node[point,label=left:{$B$},label=above:{$y$}] (d0) at (9.5,2) {};
\node[point] (d1) at (9.5,1) {};
\node[point] (d2) at (10,2) {};
\node[point] (d3) at (10,1) {};
\node[bpoint,label=below:{\normalsize $x_i^0$}] (d4) at (10,0) {};
\draw[query,->] (d1) to (d0);
\draw[query,->] (d1) to (d2);
\draw[query,->] (d3) to (d2);
\draw[query,->] (d3) to (d4);
\node[point,label=right:{$B$},label=above:{$y$}] (e0) at (11.5,2) {};
\node[point] (e1) at (11.5,1) {};
\node[point,label=below:{\normalsize $x_i^1$}] (e2) at (11.5,0) {};
\draw[query,->] (e1) to (e0);
\draw[query,->] (e1) to (e2);
\end{tikzpicture}%
\caption{The query $\q_n(\avec{x}^0)$ (all edges are labelled by $R$), the canonical model $\C_\T^{\smash{A(a)}}$ (the normalisation predicates are not shown) and two ways of mapping a branch of the query to the canonical model in Example~\ref{ex:exp-tws}.}\label{fig:exp-tws}
\end{figure}
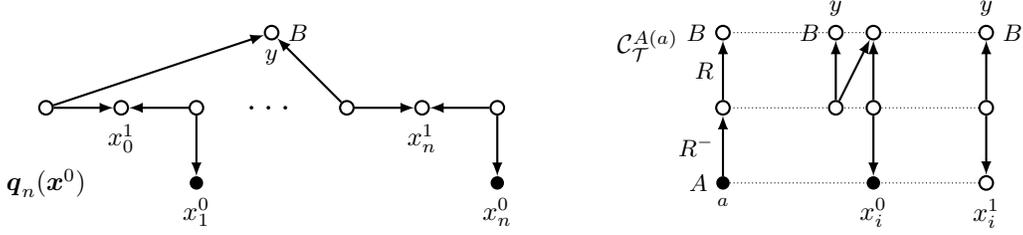
OMQ $\omq_n$ has at least $2^n$ tree witnesses: for any $\avec{\alpha} = (\alpha_1, \dots, \alpha_n) \in \{0,1\}^n$, there is a tree witness $(\tr^{\avec{\alpha}},\ti^{\avec{\alpha}})$ with $\tr^{\avec{\alpha}} = \{x_k^{\alpha_k} \mid 1 \le k \le n\}$. 
Note, however, that the tree-witness rewriting of $\omq_n$ can be equivalently transformed into the following \emph{polynomial-size} PE-rewriting: 
\begin{equation*}
\q_n(\avec{x}^0) \ \ \lor \ \ \exists z \, \big[A(z) \land \bigwedge_{1 \leq i\leq n} \big((x^0_i=z) \lor \exists y \, (R(y,x^0_i) \land R(y,z))\big)\big].
\end{equation*}
\end{example}

If any two tree witnesses for an OMQ $\omq$ are compatible in the sense that either they are non-conflicting or one is included in the other, then $\qtw$ can be equivalently transformed to the \PE-rewriting
\begin{equation*}
\exists \avec{y}\,\bigwedge_{S(\avec{z})\in \q} \bigl(\, S(\avec{z})  \ \ \lor \bigvee_{\t\in\twset \text{ with } S(\avec{z})\in\q_{\t}} \hspace*{-2em}\tw_\t(\tr)\,\bigr) \qquad \text{of size   $O(|\twset| \cdot |\omq|^2 )$.}
\end{equation*}
We now analyse transformations of this kind in the setting of Boolean functions.

%****************

\section{OMQ Rewritings as Boolean Functions}\label{sec:Bfunctions}

For any OMQ $\omq(\avec{x})=(\T,\q(\avec{x}))$, we define Boolean functions $\fun$ and $\prim$ such that:
\begin{nitemize}\itemsep=6pt
\item[--] if $\fun$ is computed by a Boolean formula (monotone  formula or monotone circuit) $\Phi$, then $\omq$ has an FO- (respectively, PE- or NDL-) rewriting of size $O(|\Phi| \cdot |\omq|)$;

\item[--] if $\q'$ is an FO- (PE- or NDL-) rewriting
of $\omq$, then $\prim$ is computed by a Boolean formula (respectively, monotone formula or monotone circuit) of size $O(|\q'|)$.
\end{nitemize}
We remind the reader (for details see, e.g.,~\cite{Arora&Barak09,Jukna12}) that an \emph{$n$-ary Boolean function}, for $n\ge 1$, is any function from $\{0,1\}^n$ to $\{0,1\}$. A Boolean function $f$ is \emph{monotone} if $f(\avec{\alpha}) \leq f(\avec{\beta})$ for all $\avec{\alpha}\leq \avec{\beta}$, where $\leq$ is the component-wise $\leq$ on vectors of $\{0,1\}$.
A \emph{Boolean circuit}, $\Cir$, is a directed acyclic graph whose vertices are called \emph{gates}. Each gate is labelled with a  propositional variable, a constant $0$ or~$1$, or with $\NOT$, $\AND$ or $\OR$. Gates  labelled with variables and constants have in-degree~$0$ and are called \emph{inputs}; $\NOT$-gates  have in-degree~$1$, while $\AND$- and $\OR$-gates have in-degree~$2$ (unless otherwise specified).
One of the gates in $\Cir$ is distinguished as the \emph{output gate}. Given an assignment $\avec{\alpha} \in \{0,1\}^n$ to the variables, we compute the value of each gate in $\Cir$ under $\avec{\alpha}$ as usual in Boolean logic. The \emph{output $\Cir(\avec{\alpha})$ of $\Cir$ on} $\avec{\alpha} \in \{0,1\}^n$ is the value of the output gate. We usually assume that the gates $g_1,\dots,g_m$ of $\Cir$ are ordered in such a way that $g_1,\dots,g_n$  are input gates; each gate $g_i$, for $i \geq n$, gets inputs from gates $g_{j_1},\dots,g_{j_k}$ with $j_1,\dots,j_k < i$, and $g_m$ is the output gate. We say that $\Cir$ \emph{computes} an $n$-ary Boolean function $f$ if \mbox{$ \Cir(\avec{\alpha})=f(\avec{\alpha})$} for all $\avec{\alpha} \in \{0,1\}^n$. The \emph{size} $|\Cir|$ of $\Cir$ is the number of gates in $\Cir$. A circuit is \emph{monotone} if it contains only inputs, $\AND$- and $\OR$-gates. \emph{Boolean formulas} can be thought of as circuits in which every logic gate has at most one outgoing edge.
Any monotone circuit computes a monotone function, and any monotone Boolean function can be computed by a monotone circuit.

%***************

\subsection{Hypergraph Functions}\label{hyper-functions}

Let $H = (V,E)$ be a hypergraph with \emph{vertices} $v \in V$ and \emph{hyperedges} $e \in E\subseteq 2^V$.
A subset $E' \subseteq E$  is said to be \emph{independent} if $e \cap e' = \emptyset$, for any distinct $e,e' \in E'$. The set of vertices that occur in the hyperedges of $E'$ is denoted by $V_{E'}$.
For each vertex $v \in V$ and each hyperedge $e \in E$, we take  propositional variables $p_v$ and $p_e$, respectively. The \emph{hypergraph function $f_H$ for $H$} is given by the monotone Boolean formula
\begin{equation}\label{hyper}
f_H \ \ \ = \ \bigvee_{
E' \text{ independent}}
\Big( \bigwedge_{v \in V \setminus V_{E'}}
\hspace*{-0.5em} p_v \ \land \ \bigwedge_{e \in E'} p_e
\Big).
\end{equation}

The tree-witness \PE-rewriting $\qtw$ of any OMQ $\omq(\avec{x})=(\T,\q(\avec{x}))$ defines a hypergraph whose vertices are the atoms of $\q$ and hyperedges are the sets $\q_\t$, where $\t$ is a tree witness for $\omq$. We denote this hypergraph by $\HG{\omq}$ and call $f_{\HG{\omq}}$ the \emph{tree-witness hypergraph function for $\omq$}. To simplify notation, we write $\twfn$ instead of $f_{\HG{\omq}}$. 
Note that formula~\eqref{hyper} defining $\twfn$ is obtained from rewriting~\eqref{rewriting0} by regarding the atoms $S(\avec{z})$ in $\q$ and tree-witness formulas $\tw_\t$ as propositional variables. We denote these variables by $p_{S(\avec{z})}$ and $p_{\t}$ (rather than $p_v$ and $p_e$), respectively.

\begin{example}\label{ex:simple hyper}
For the OMQ $\omq$ in Example~\ref{ex:conf}, the hypergraph $\HG{\omq}$ has 3~vertices (one for each atom in the query) and 2~hyperedges (one for each tree witness) shown in Fig.~\ref{fig:hyper-example}.
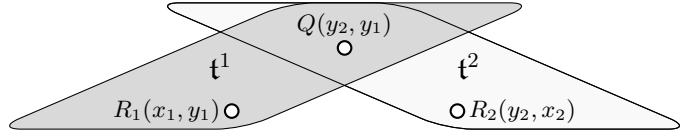
\begin{figure}[t]%
\centering%
\begin{tikzpicture}[label distance=-2pt,yscale=1.2,xscale=1.5]
\coordinate (c1) at (0,0);
\coordinate (c2) at (1,0.7);
\coordinate (c3) at (2,0);
\draw[fill=gray!5,rounded corners=10] (4.1,-0.2) -- (1.5,1.2) -- (-0.7,1.2) -- (1.9,-0.2) -- cycle;
\draw[fill=gray!30,rounded corners=10] (-2.1,-0.2) -- (0.5,1.2) -- (2.7,1.2) -- (0.1,-0.2) -- cycle;
\draw[rounded corners=10] (4.1,-0.2) -- (1.5,1.2) -- (-0.7,1.2) -- (1.9,-0.2) -- cycle;
\node at (c1) [point,fill=white,label=left:{\small $R_1(x_1,y_1)$}] {};
\node at (c2) [point,fill=white,label=above:{\small $Q(y_2,y_1)$}] {};
\node at (c3) [point,fill=white,label=right:{\small $R_2(y_2,x_2)$}] {};
\node at (-0.1,0.5) {\large $\t^1$};
\node at (2.1,0.5) {\large $\t^2$};
\end{tikzpicture}%
\caption{The hypergraph $\HG{\omq}$ for $\omq$ from Example~\ref{ex:conf}.}\label{fig:hyper-example}
\end{figure}
The tree-witness hypergraph function for $\omq$ is as follows: 
\begin{equation*}
\twfn \ \ =  \ \ \bigl(p_{R_1(x_1,y_1)} \land p_{Q(y_2,y_1)} \land p_{R_2(x_2,y_2)}\bigr) \lor
\bigl(p_{\t^1} \land p_{R_2(x_2,y_2)}  \bigr) \lor \bigl(p_{R_1(x_1,y_1)} \land p_{\t^2}\bigr).
\end{equation*}
\end{example}

Suppose the function $\twfn$ for an OMQ $\omq(\avec{x})$ is computed by a Boolean formula $\Phi$. Consider the first-order formula $\Phi^*(\avec{x})$ obtained by replacing each $p_{S(\avec{z})}$ in $\Phi$ with $S(\avec{z})$, each $p_{\t}$ with $\tw_\t$, and adding the appropriate prefix $\exists \avec{y}$. By comparing~\eqref{hyper} and~\eqref{rewriting0}, we see that $\Phi^*(\avec{x})$ is an \FO-rewriting of~$\omq(\avec{x})$ over  data instances that are complete over $\T$. This gives claim (\emph{i}) of the following theorem:

\begin{theorem}\label{TW2rew}
\textup{(}i\textup{)} If $\twfn$ is computed by a \textup{(}monotone\textup{)} Boolean formula $\Phi$, then
there is a \textup{(}\PE-\textup{)} \FO-rewriting of $\omq(\avec{x})$ of size $O(|\Phi| \cdot |\omq|)$.

\textup{(}ii\textup{)} If $\twfn$ is computed by a monotone Boolean circuit $\Cir$, then there is an \NDL-rewriting of
$\omq(\avec{x})$ of size $O(|\Cir|\cdot |\omq|)$.
\end{theorem}
\begin{proof}
(\emph{ii}) Let $\t^1,\dots,\t^{l}$ be tree witnesses for  $\omq(\avec{x}) = (\T,\q(\avec{x}))$, where  $\q(\avec{x}) = \exists \avec{y}\, \bigwedge_{i = 1}^n S_i(\avec{z}_i)$.
We assume that the gates $g_1, \dots, g_n$ of $\Cir$ are the inputs $p_{S_1(\avec{z}_1)}, \dots, p_{S_n(\avec{z}_n)}$ for the atoms, the gates $g_{n+1},\dots,g_{n+l}$ are the inputs $p_{\t^1},\dots,p_{\t^{l}}$ for the tree witnesses and $g_{n+l+1},\dots,g_m$ are $\AND$- and $\OR$-gates. Denote by $\Pi$ the following NDL-program, where $\avec{z} = \avec{x} \cup\avec{y}$:
\begin{nitemize}
\item[--] $S_i(\avec{z}_i) \to G_i(\avec{z})$, for $0 < i \le n$;

\item[--] $\tau(u) \to G_{i+m}(\avec{z}[\tr^j/u])$, for $0 < j \le l$ and $\tau$ generating $\t^j$, where $\avec{z}[\tr^j/u]$ is the result of replacing each $z \in \tr^j$ in $\avec{z}$ with $u$;

\item[--] $\begin{cases}
G_j(\avec{z}) \land G_k(\avec{z})  \to G_i(\avec{z}), &\text{if } g_i = g_j \ANDOP g_k,\\
G_j(\avec{z}) \to G_i(\avec{z})\text{ and  } G_k(\avec{z}) \to G_i(\avec{z}), &\text{if } g_i = g_j \OROP g_k,
\end{cases}$\quad for $n+l < i \leq m$;
\item[--] $G_m(\avec{z}) \to G(\avec{x})$.
\end{nitemize}
It is not hard to see that $(\Pi, G(\avec{x}))$ is an NDL-rewriting of $\omq(\avec{x})$.
\end{proof}

Thus, the problem of constructing polynomial-size rewritings of OMQs reduces to finding polynomial-size (monotone) formulas or monotone circuits for the corresponding functions $\twfn$.
Note, however, that $\twfn$ contains a variable $p_\t$ for every tree witness~$\t$, which makes this reduction useless for OMQs with exponentially many tree witnesses. To be able to deal with such OMQs, we slightly modify the tree-witness rewriting. 
 
Suppose $\t = (\tr,\ti)$ is a tree witness for $\omq = (\T,\q)$  induced by a homomorphism $h \colon \q_\t  \to \smash{\C_\T^{\tau(a)}}$. We say that $\t$ is $\varrho$-\emph{initiated} if $h(z)$ is of the form  $a \varrho w$, for every (equivalently, some) variable $z\in \ti$. For such $\varrho$, we define a formula $\varrho^*(x)$ by taking the disjunction of $\tau(x)$  with $\T\models \tau(x) \to \exists y\,\varrho(x,y)$. Again, the disjunction includes only those $\tau(x)$ that do not contain normalisation predicates (even though $\varrho$ itself can be one).

\begin{example}\label{ex:initated}
Consider the OMQ $\omq(x) = (\T,\q(x))$ with 
\begin{equation*}
\T = \bigl\{\,\exists y\, Q(x,y) \to \exists y\, P(x,y), \ \ P(x,y) \to R(x,y)\, \bigr\}\quad \text{and} \quad
\q(x) = \exists y \, R(x,y).
\end{equation*}
As shown in Fig.~\ref{fig:ex:initiated}, the tree witness $\t = (\{x\},\{y\})$ for $\omq(x)$ is generated by $\exists y\, Q(x,y)$, $\exists y\, P(x,y)$ and $\exists y\, R(x,y)$; it is also $P$- and $R$-initiated, but not $Q$-initiated. We have:
\begin{equation*}
P^*(x) = \exists y\,Q(x,y) \lor \exists y\,P(x,y) \text{ and } R^*(x) = \exists y\,Q(x,y) \lor \exists y\,P(x,y) \lor \exists y\, R(x,y).
\end{equation*}
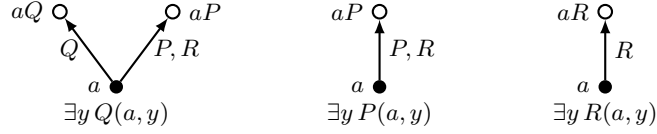
\begin{figure}[t]%
\centering%
\begin{tikzpicture}\footnotesize
\node (a0) at (-0.5,0) [bpoint, label=below:{$\exists y\, Q(a,y)$}, label=left:{$a$}]{};
\node (a1) at (-1.25,1) [point, label=left:{$aQ$}]{};
\node (a2) at (0.25,1) [point, label=right:{$aP$}]{};
\draw[can,->] (a0)  to node[label=left:{$Q$}]{} (a1);
\draw[can,->] (a0)  to node[label=right:{$P,R$}]{} (a2);
\node (b0) at (3,0) [bpoint, label=below:{$\exists y\, P(a,y)$}, label=left:{$a$}]{};
\node (b1) at (3,1) [point, label=left:{$aP$}]{};
\draw[can,->] (b0)  to node[label=right:{$P,R$}]{} (b1);
\node (c0) at (6,0) [bpoint, fill=black, label=below:{$\exists y\, R(a,y)$}, label=left:{$a$}]{};
\node (c1) at (6,1) [point, label=left:{$aR$}]{};
\draw[can,->] (c0)  to node[label=right:{$R$}]{} (c1);
\end{tikzpicture}%
\caption{Canonical models in Example~\ref{ex:initated}.}\label{fig:ex:initiated}
\end{figure}%
\end{example}

The modified tree-witness rewriting for $\omq(\avec{x}) = (\T,\q(\avec{x}))$, denoted $\qtw'(\avec{x})$, is obtained by replacing~\eqref{tw-formula} in~\eqref{rewriting0} with the formula
\begin{equation*}\label{tw-formula'}\tag{\ref{tw-formula}$'$}
\tw'_\t(\tr,\ti) ~=~ \bigwedge_{R(z,z')\in \q_\t}\hspace*{-0.5em} (z=z') \quad \land \bigvee_{\t \text{ is $\varrho$-initiated}} \ \bigwedge_{z \in \tr \cup \ti} \varrho^*(z).
\end{equation*}
Note that unlike~\eqref{tw-formula}, this formula contains the variables in both $\ti$ and $\tr$, which must be equal under every satisfying assignment. 
We associate with $\qtw'(\avec{x})$ the monotone Boolean function~$f^\blacktriangledown_{\omq}$ 
given by the formula obtained from~\eqref{hyper} by replacing each variable $p_v$ with the respective $p_{S(\avec{z})}$, for $S(\avec{z})\in \q$, and  each variable $p_e$ with the formula
\begin{equation}\label{subst}
\bigwedge_{R(z,z')\in \q_\t} \hspace*{-0.5em} p_{z=z'} \ \ \ \wedge \bigvee_{\t \text{ is $\varrho$-initiated}} \,\, \bigwedge_{z \in \tr\cup\ti} p_{\varrho^*(z)},
\end{equation}
for the respective tree witness $\t=(\tr,\ti)$ for $\omq(\avec{x})$, where $p_{z=z'}$ and $p_{\varrho^*(z)}$ are propositional variables. 
Clearly, the number of variables in $\homfn$ is \emph{polynomial} in $|\omq|$. 

\begin{example}
For the OMQ $\omq(x)$ in the Example~\ref{ex:conf}, we have: 
\begin{multline*}
\homfn \ \ =  \ \ \bigl(p_{R_1(x_1,y_1)} \land p_{Q(y_2,y_1)} \land p_{R_2(x_2,y_2)}\bigr)\ \  \lor {} \\ 
\bigl(
%\bigwedge_{z,z' \in \{x_1,y_1,y_2\}} \hspace*{-1.5em}p_{z=z'} 
\bigl(p_{x_1 = y_1}\! \land p_{y_2 = y_1} \ \ \wedge \hspace*{-0.5em}\bigwedge_{z \in \{x_1, y_1, y_2\}} \hspace*{-1.5em} p_{P^*_{\zeta_1}(z)}\bigr) \ \ \land \ \ p_{R_2(x_2,y_2)} \bigr) \ \ \lor {} \\
\bigl(p_{R_1(x_1,y_1)} \ \ \land \ \ %\bigwedge_{z,z' \in \{y_1, y_2, x_2\}} \hspace*{-1.5em}p_{z=z'} 
\bigl(p_{y_2 = y_1} \!\land p_{x_2 = y_2}
\ \ \wedge \hspace*{-0.5em}\bigwedge_{z \in\{y_1, y_2, x_2\}} \hspace*{-1.5em} p_{P^*_{\zeta_2}(z)}\bigr)\bigr).
\end{multline*}
\end{example}

The proof of the following theorem is given in Appendix~\ref{app:proof:Th4.5}:

\begin{theorem}\label{Hom2rew}
\textup{(}i\textup{)} For any OMQ $\omq(\avec{x})$, the formulas $\qtw(\avec{x})$ and $\qtw'(\avec{x})$ are equivalent, and so $\qtw'(\avec{x})$ is a PE-rewriting of $\omq(\avec{x})$ over complete data instances.

\textup{(}ii\textup{)} Theorem~\ref{TW2rew} continues to hold for $\twfn$ replaced by $\homfn$.
\end{theorem}

Finally, we observe that although $\twfn$ and  $\homfn$ are defined by exponential-size formulas, each of these functions can be computed by a nondeterministic polynomial algorithm (in the number of propositional variables). Indeed, given truth-values for the $p_{S(\avec{z})}$ and~$p_{\t}$ in $\twfn$,  guess a set $\Theta$ of at most $|\q|$ tree witnesses and check whether (\emph{i}) $\Theta$ is independent, (\emph{ii}) $p_{\t} = 1$ for all $\t \in \Theta$, and (\emph{iii}) every $S(\avec{z})$ with $p_{S(\avec{z})} = 0$ belongs to some $\t \in \Theta$. The function $\homfn$ is computed similarly except that, in (\emph{ii}), we check whether the polynomial-size formula~\eqref{subst} is true under the given truth-values for every $\t\in\Theta$.

%******************

\subsection{Primitive Evaluation Functions}

To obtain lower bounds on the size of rewritings, we associate with every OMQ $\omq$ a third Boolean function $\prim$ that describes the result of evaluating $\omq$ on data instances with a single individual constant. Let $\avec{\gamma}\in \{0,1\}^n$ be a vector assigning the truth-value $\avec{\gamma}(S_i)$ to each unary  or binary predicate $S_i$ in $\omq$. We associate with $\avec{\gamma}$  the data instance
\begin{equation*}
\A(\avec{\gamma})  \ \ = \ \ \bigl\{\, A_i(a) \mid \avec{\gamma}(A_i) = 1\,\bigr\}  \ \ \cup \ \ \bigl\{\, P_i(a,a) \mid \avec{\gamma}(P_i) = 1\,\bigr\}
\end{equation*}
and set $\prim(\avec{\gamma}) = 1$ iff $\T, \A(\avec{\gamma}) \models \q(\avec{a})$, where $\avec{a}$ is the $|\avec{x}|$-tuple of $a$s. We call $\prim$ the \emph{primitive evaluation function} for $\omq$.
\begin{theorem}\label{rew2prim}
\textup{(}i\textup{)} If $\q'$ is a \textup{(}\PE-\textup{)} \FO-rewriting of $\omq$, then $\prim$ can be computed by a \textup{(}monotone\textup{)} Boolean formula of size $O(|\q'|)$.

\textup{(}ii\textup{)} If $\q'$ is an \NDL-rewriting of $\omq$, then $\prim$ can be computed by a monotone Boolean circuit of size $O(|\q'|)$.
\end{theorem}
\begin{proof}
(\emph{i}) Let $\q'$ be an \FO-rewriting of $\omq$. We eliminate the quantifiers in $\q'$ by replacing each subformula of the form $\exists x\, \psi(x)$ and $\forall x\, \psi(x)$ in $\q'$ with $\psi(a)$. We then
replace each $a=a$ with $\top$ and each atom of the form $A(a)$ and $P(a,a)$ with the corresponding propositional variable. The resulting  Boolean formula clearly computes $\prim$.
If  $\q'$ is a \PE-rewriting of $\omq$, then the result is a monotone Boolean formula computing~$\prim$.

(\emph{ii}) If  $(\Pi, G)$ is an \NDL-rewriting of $\omq$, then
we replace all variables in $\Pi$ with $a$ and then perform the replacement described in (\emph{i}). We now turn the resulting propositional \NDL-program  $\Pi'$ into a monotone circuit computing $\prim$. For every (propositional) variable $p$ occurring in the head of a rule in $\smash{\Pi'}$, we take an appropriate number of $\OR$-gates whose output is $p$ and inputs are the bodies of the rules with head $p$; for every such body, we introduce an appropriate number of $\AND$-gates whose inputs are the propositional variables in the body, or, if the body is empty, we take the gate for constant~$1$.
\end{proof}

%*****************

\subsection{Hypergraph Programs}

We introduced hypergraph functions as Boolean abstractions of the tree-witness rewritings. Our next aim is to define a model of computation for these functions.

A \emph{hypergraph program} (HGP) $P$ is a hypergraph $H = (V,E)$ each of whose vertices is labelled by $0$, $1$ or a literal over a list $p_1,\dots,p_n$ of propositional variables. (As usual, a \emph{literal}, $\li$, is a propositional variable or its negation.) An \emph{input} for $P$ is a tuple \mbox{$\avec{\alpha} \in \{0,1\}^n$}, which is regarded as a valuation for $p_1,\dots,p_n$. The output  $P(\avec{\alpha})$ of $P$ on~$\avec{\alpha}$ is 1 iff there is an independent subset of $E$ that \emph{covers all zeros}---that is, contains every vertex in $V$ whose label evaluates to $0$ under $\avec{\alpha}$. We say that $P$ \emph{computes} an $n$-ary Boolean function~$f$ if $f(\avec{\alpha})=P(\avec{\alpha})$, for all $\avec{\alpha}\in\{0,1\}^n$. An HGP is \emph{monotone} if its vertex labels do not have negated variables. The \emph{size} $|P|$ of an HGP $P$ is the size $|H|$ of the underlying hypergraph $H = (V,E)$, which is $|V| + |E|$.

The following observation shows that monotone HGPs capture the computational power of hypergraph functions. We remind the reader that a \emph{subfunction} of a Boolean function $f$ is obtained from $f$ using two operations: (1) fixing some of its variables to $0$ or $1$, and (2) renaming (in particular, identifying) some of the variables in $f$. A hypergraph $H$ is said to be of \emph{degree at most} $d$ if every vertex in it belongs to at most $d$ hyperedges; $H$ is of \emph{degree} $d$ if every vertex in it belongs to exactly $d$ hyperedges.

\begin{proposition}\label{hyper:program}
\textup{(}i\textup{)} Any monotone HGP based on a hypergraph $H$ computes a subfunction of the hypergraph function $f_H$.

\textup{(}ii\textup{)} For any hypergraph $H$ of degree at most $d$, there is a monotone HGP of size $O(|H|)$ that computes $f_H$ and such that its hypergraph is of degree at most $\max(2,d)$.
\end{proposition}
\begin{proof}
To show (\emph{i}), it is enough to replace the vertex variables $p_v$ in $f_H$ by the corresponding vertex labels of the given HGP and fix all the edge variables $p_e$ to $1$.

For (\emph{ii}), given a hypergraph $H = (V,E)$, we label each $v \in V$ by the variable $p_v$. For each $e \in E$, we add  a fresh vertex $a_e$ labelled by $1$ and a fresh vertex $b_e$ labelled by~$p_e$; then we create a new hyperedge $e' = \{a_e, b_e\}$ and add $a_e$ to the hyperedge $e$. We claim that the resulting HGP $P$
computes $f_H$. Indeed, for any input $\avec{\alpha}$ with $\avec{\alpha}(p_e) = 0$, we have to include the edge $e'$ into the cover, and so cannot include the edge $e$ itself. Thus, $P(\avec{\alpha})=1$ iff there is an independent set $E$ of hyperedges with $\avec{\alpha}(p_e)=1$, for all $e\in E$, covering all zeros of the variables $p_v$. 
\end{proof}

In some cases, it will be convenient to use \emph{generalised HGPs} that allow hypergraph vertices to be labelled by conjunctions $\bigwedge_{i} \li_i$ of literals $\li_i$. 
The following proposition shows that this generalisation does not increase the computational power of HGPs.

\begin{proposition}\label{lem:genhgps}
For every generalised HGP $P$ over $n$ variables, there is an HGP $P^\prime$ computing the same function and such that $|P'| \le n \cdot |P|$.
\end{proposition}
\begin{proof}
To construct $P^\prime$, we split every vertex $v$ of $P$ labelled with $\bigwedge_{i=1}^{k} \li_i$ into $k$ new vertices $v_1, \dots, v_k$ and label $v_i$ with $\li_i$, for $1 \leq i \leq k$ (without  loss of generality, we can assume that $\li_i$ and $\li_j$ have distinct variables for $i \ne j$); each hyperedge containing $v$ will now contain all the $v_i$. It is easy to see that $P(\avec{\alpha}) = {P^\prime}(\avec{\alpha})$, for any input $\avec{\alpha}$. Since $k \leq n$, we have $|P'| \le n \cdot |P|$.
\end{proof}

%************************

\section{OMQs, hypergraphs and monotone hypergraph programs}\label{sec:OMQs&hypergraphs}

We now establish a correspondence between the structure of OMQs and hypergraphs.

\subsection{OMQs with ontologies of depth 2}\label{sec5.1}

To begin with, we show that every hypergraph $H = (V, E)$ can be represented by a polynomial-size OMQ $\omq_H = (\T,\q)$ with $\T$ of depth~$2$. With every vertex $v\in V$ we associate a unary predicate $A_v$, and with every hyperedge $e \in E$ a unary predicate $B_e$ and a binary predicate $R_e$. We define $\T$ to be the set of the following axioms, for $e \in E$: 
\begin{equation*}
B_e(x) \ \ \to \ \ \exists y \,\bigl[ \hspace*{-0.5em}\bigwedge_{e \cap e' \neq \emptyset,\  e \ne e'} \hspace*{-1.5em}R_{e'}(x,y)  \ \land \
%A'_e(y) \big),\qquad
%A'_e(x) \to 
\bigwedge_{v \in e} A_v(y) %, \quad A'_e(x) \to
\ \land \  \exists z\, R_{e}(z,y)\bigr].
\end{equation*}
Clearly, $\T$ is of depth 2. We also take the Boolean CQ $\q$ with variables $y_v$, for $v\in V$, and $z_e$, for $e\in E$:
\begin{equation*}
\q  \ \ = \ \  \bigl\{\, A_v(y_v) \mid v \in V \,\bigr\} \ \ \cup \ \ \bigl\{\,
R_e(z_e, y_v) \mid v \in e, \text{ for } v \in V \text{ and } e \in E \,\bigr\}.
\end{equation*}
\begin{example}
Consider again the hypergraph from Example~\ref{ex:simple hyper}, which we now denote by $H = (V,E)$ with $V = \{v_1,v_2,v_3\}$, $E = \{e_1,e_2\}$, $e_1 = \{v_1,v_2\}$ and $e_2 = \{v_2,v_3\}$. The CQ $\q$ and the canonical models $\C_{\T}^{B_{\smash{e_i}}(a)}$, for $i=1,2$, are shown in Fig.~\ref{fig:QH:simple hyper} along with four tree witnesses for $\omq_H$ (as explained in Remark~\ref{ignored}, we ignore the two extra tree witnesses generated only by normalisation predicates).
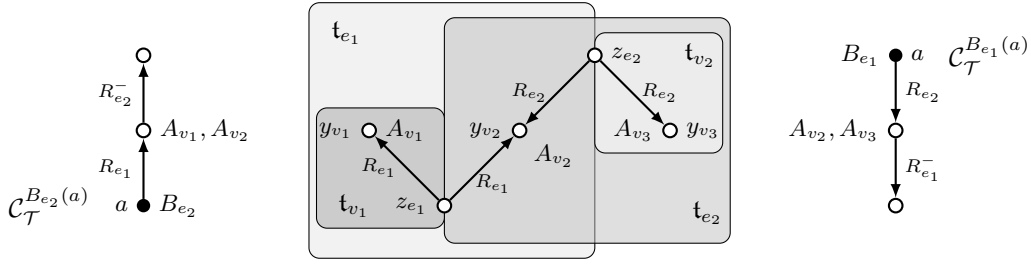
\begin{figure}[t]%
\centering%
\begin{tikzpicture}\footnotesize
\node[bpoint, label=left:{$B_{e_1}$}, label=right:{$a$}] (a) at (7,2) {};
\node[point, label=left:{$A_{v_2},A_{v_3}$}] (d1) at (7,1) {};
\node[point, fill=white] (d2) at (7,0) {};
\draw[can,->] (a)  to node [right]{\scriptsize $R_{e_2}$} (d1);
\draw[can,->] (d1)  to node [right]{\scriptsize $R_{e_1}^-$} (d2);
\node at (8.25,2) {\normalsize $\C_{\T}^{B_{e_1}(a)}$};
\node[bpoint, label=right:{$B_{e_2}$}, label=left:{$a$}] (a) at (-3,0) {};
\node[point, label=right:{$A_{v_1},A_{v_2}$}] (d1) at (-3,1) {};
\node[point, fill=white] (d2) at (-3,2) {};
\draw[can,->] (a)  to node [left]{\scriptsize $R_{e_1}$} (d1);
\draw[can,->] (d1)  to node [left]{\scriptsize $R_{e_2}^-$} (d2);
\node at (-4.25,0) {\normalsize $\C_{\T}^{B_{e_2}(a)}$};
\coordinate (czv1) at (0,1);
\coordinate (czv1p) at (-0.7,1);
\coordinate (czv1pp) at (-0.8,1);
\coordinate (cze1) at (1,0);
\coordinate (czv2) at (2,1);
\coordinate (cze2) at (3,2);
\coordinate (czv3) at (4,1);
\coordinate (czv3p) at (4.7,1);
\coordinate (czv3pp) at (4.8,1);
\node[draw,thin,fill=gray!10,rounded corners,inner xsep=0mm,inner ysep=7mm,fit=(cze2) (cze1) (czv1pp)] {};
\node[draw,thin,fill=gray!50,fill opacity=0.5,rounded corners,inner xsep=0mm,inner ysep=5mm,fit=(cze1) (cze2) (czv3pp)] {};
\node[draw,thin,fill=gray!40,rounded corners,inner xsep=0mm,inner ysep=3mm,fit=(cze1) (czv1p)] {};
\node[draw,thin,fill=gray!15,rounded corners,inner xsep=0mm,inner ysep=3mm,fit=(cze2) (czv3p)] {};
\node at (-0.2,0) {\normalsize $\t_{v_1}$};
\node at (-0.3,2.3) {\normalsize $\t_{e_1}$};
\node at (4.4,2) {\normalsize $\t_{v_2}$};
\node at (4.5,-0.1) {\normalsize $\t_{e_2}$};
\node[point,label=left:{$y_{v_1}$},label=right:$A_{v_1}$] (zv1) at (czv1) {};
\node[point,label=left:{$z_{e_1}$}] (ze1) at (cze1) {};
\node[point,label=left:{$y_{v_2}$},label=below right:{$A_{v_2}$}] (zv2) at (czv2) {};
\node[point,label=right:$z_{e_2}$] (ze2) at (cze2) {};
\node[point,label=right:{$y_{v_3}$},label=left:{$A_{v_3}$}] (zv3) at (czv3) {};
\draw[<-,query] (zv1) to node[left] {\scriptsize$R_{e_1}$} (ze1);
\draw[<-,query] (zv2) to node[right,near end] {\scriptsize$R_{e_1}$} (ze1);
\draw[->,query] (ze2) to node[left] {\scriptsize$R_{e_2}$} (zv2);
\draw[->,query] (ze2) to node[right] {\scriptsize$R_{e_2}$} (zv3);
\end{tikzpicture}%
\caption{The OMQ $\omq_H$ for $H$ from Example~\ref{ex:simple hyper} and its tree witnesses.}\label{fig:QH:simple hyper}
\end{figure}
\end{example}

It is not hard to see that the number of tree witnesses for $\omq_H$ does not exceed $|H|$. 
%is linear in the size of $H$.
Indeed, all the tree witnesses for $\omq_H$ fall into two types:
\begin{align*}
& \t_v=(\ti, \tr) \text{ with } \tr =\{z_e \mid v \in e\} \text{ and } \ti = \{y_v\}, \ \ \text{ for } v\in V \text{ that belong to a single } e\in E;\\
%
%``depth 1 tree witnesses corresponding to hyperedges'' of the form $\t^1_e=(\ti, \tr)$ with $\ti = \{z_e\}$ and $\tr =\{z_v \mid v \in e\}$;
%
& \t_e=(\ti, \tr) \text{ with } \tr = \{z_{e'} \mid e \cap e' \neq \emptyset, e \ne e' \} \text{ and } \ti = \{z_e\} \cup\{y_v \mid v \in e\},\quad \text{ for } e\in E.
%
%``depth 1 tree witnesses corresponding to vertices'' of the form $\t^1_v=(\ti, \tr)$ with $\ti = \{z_v\}$ and $\tr =\{z_e \mid v \in e\}$;
%
%\item $\t'_e=(\ti, \tr)$ with  $\tr =\{z_v \mid v \in e\}$ and $\ti = \{z_e\}$;
%
%``depth 2 tree witnesses corresponding to hyperedges'' of the form $\t^e=(\ti, \tr)$ with $\ti = \{z_e\} \cup\{z_v \mid v \in e\}$ and $\tr = \{z_{e'} \mid e \cap e' \neq \emptyset\}$.
\end{align*}

We call a hypergraph $H'$ a \emph{subgraph} of a hypergraph $H = (V,E)$ if $H'$ can be obtained from $H$ by (\emph{i}) removing some of its hyperedges and (\emph{ii}) removing some of its vertices from both $V$ and the hyperedges in $E$.

\begin{theorem}\label{hg-to-query}
\textup{(}i\textup{)} Any hypergraph $H$ is isomorphic to a subgraph of $\HG{\omq_H}$.

\textup{(}ii\textup{)} Any monotone HGP $P$ based on a hypergraph $H$ computes a subfunction of the primitive evaluation function $f^\vartriangle_{\omq_H}$.
\end{theorem}
\begin{proof}
(\emph{i}) An isomorphism between $H$ and a subgraph of $\HG{\omq_H}$ can be established by the map $v \mapsto A_v(y_v)$, for $v \in V$, and $e \mapsto \q_{\t_e}$, for $e \in E$. 
%and the tree witness $\t^e = (\tr^e,\ti^e)$ with $\tr^e = \{z_{e'} \mid e \cap e' \neq \emptyset\}$ and $\ti^e = \{z_e\} \cup \{z_v \mid v \in e\}$, which is generated by $A_e(x)$.

(\emph{ii}) Suppose that $P$ is based on a hypergraph $H= (V, E)$. Given an input $\avec{\alpha}$ for $P$, we define an assignment $\avec{\gamma}$ for the predicates in $\omq_H = (\T,\q)$ by taking $\avec{\gamma}(A_v)$ to be the value of the label of $v$ under  $\avec{\alpha}$, $\avec{\gamma}(B_e) = 1$, %$\avec{\gamma}(A'_e) = 0$, 
$\avec{\gamma}(R_e)=1$ (and of course $\avec{\gamma}(P_\zeta) = 0$, for all normalisation predicates $P_\zeta$).
By the definition of $\T$, for each $e\in E$, the canonical model $\C_{\T,\A(\avec{\gamma})}$ contains labelled nulls  $w_e$ and $w'_{e}$ such that
\begin{equation*}
\C_{\T,\A(\avec{\gamma})} \models \bigwedge_{e \cap e' \neq \emptyset, \ e \ne e'} \hspace*{-1em}R_{e'}(a, w_e) \ \land \ \bigwedge_{v \in e} A_v(w_e) \ \land \  R_e(w'_e, w_e).
\end{equation*}
We now show that $P(\avec{\alpha}) = 1$ iff $f^\vartriangle_{\omq_H}(\avec{\gamma})= 1$ (iff $\T,\A(\avec{\gamma}) \models \q$).
Suppose $P(\avec{\alpha}) = 1$, that is, there is an independent subset $E' \subseteq E$ such that the label of each 
$v \notin \bigcup E'$ evaluates to $1$ under $\avec{\alpha}$. Then the map $h \colon \q \to \C_{\T,\A(\avec{\gamma})}$ defined by taking
\begin{equation*}
h(z_e)=\begin{cases} w'_e,& \mbox{if } e \in E',\\
a,& \mbox{ otherwise,}
\end{cases} 
\qquad
h(y_v)=\begin{cases} w_e,& \mbox{if } v \in e \in E',\\
a,& \mbox{ otherwise}
\end{cases}
\end{equation*}
is a homomorphism witnessing $\C_{\T,\A(\avec{\gamma})} \models \q$, whence $f^\vartriangle_{\omq_H}(\avec{\gamma})= 1$.

Conversely, if $f^\vartriangle_{\omq_H}(\avec{\gamma})= 1$ then there is a homomorphism
$h\colon \q \to \C_{\T,\A(\avec{\gamma})}$. For any hyperedge $e\in E$, there are only two options for $h(z_e)$: either $a$ or  $w'_e$. It follows that the set $E' = \{e \in E \mid  h(z_e) = w'_e\}$ is independent and covers all zeros. Indeed, if $v\notin\bigcup E'$
then $h(y_v) = a$, and so the label of $v$ evaluates to $1$ under $\avec{\alpha}$ because $A_v(y_v) \in \q$.
\end{proof}

Next, we establish a tight correspondence between hypergraphs of degree at most~2 and OMQs with ontologies of depth 1.

%**************

\subsection{Hypergraphs of Degree 2 and OMQs with Ontologies of Depth 1}\label{sec:depth1}

\begin{theorem}\label{depth1}
For any OMQ $\omq = (\T,\q)$ with $\T$ of depth~$1$, the hypergraph $\HG{\omq}$ is of degree at most~2 and $|\HG{\omq}|\leq 2|\q|$. %$|\twset| \leq |\q|$.
\end{theorem}
\begin{proof}
We have to show that every atom in $\q$ belongs to at most two $\q_\t$, $\t\in\twset$.
Suppose $\t = (\tr,\ti)$ is a tree witness for $\omq$ and $y \in \ti$. Since $\T$ is of depth~1, $\ti = \{y\}$ and $\tr$ consists of all the variables in $\q$ adjacent to $y$ in the Gaifman graph $G_\q$ of $\q$.
Thus, different tree witnesses have different internal variables $y$. An atom of the form $A(u)\in\q$ is in~$\q_\t$ iff $u = y$. An atom of the form $P(u,v)\in\q$ is in $\q_\t$ iff either $u = y$ or $v=y$. Therefore, $P(u,v)\in\q$ can only be covered by the tree witness with internal $u$ and by the tree witness with internal $v$.
\end{proof}

Conversely, we show now that any hypergraph $H$ of degree 2 is isomorphic to $\HG{\OMQI{H}}$, for some OMQ $\OMQI{H} = (\T,\q)$ with $\T$ of depth~1.
We can assume that $H = (V, E)$ comes with two fixed maps $i_1,i_2\colon V \to E$ such that for every  $v\in V$, we have $i_1(v) \ne i_2(v)$, $v \in i_1(v)$ and $v \in i_2(v)$.
For any $v \in V$, we fix 
a binary predicate $R_v$, and  let the ontology $\T$ in $\OMQI{H}$ contain the following axioms, for $e \in E$:
\begin{equation*}%\label{eq:TH:exists}
A_e(x) \  \to \  \exists y\, \bigl[ \bigwedge_{\begin{subarray}{c}v\in V\\i_1(v) = e\end{subarray}} R_v(y,x) \ \ \land \bigwedge_{\begin{subarray}{c}v\in V\\i_2(v) = e\end{subarray}} R_v(x,y)\bigr].
\end{equation*}
Clearly, $\T$ is of depth 1. The Boolean CQ $\q$ contains variables $z_e$, for $e\in E$, and is defined by taking 
\begin{equation*}
\q ~=~ \bigl\{\, R_v(z_{i_1(v)}, z_{i_2(v)}) \mid v\in V\,\bigr\}.
\end{equation*}
\begin{example}\label{example1}
Suppose that $H = (V, E)$, where $V = \{v_1, v_2, v_3, v_4\}$, $E = \{e_1, e_2, e_3\}$ and $e_1 = \{v_1, v_2, v_3\}$, $e_2 = \{v_3, v_4\}$, $e_3 = \{v_1, v_2, v_4\}$. Let
\begin{align*}
& i_1\colon v_1\mapsto e_1, \quad v_2\mapsto e_3, \quad v_3 \mapsto e_1, \quad v_4 \mapsto e_2,\\
& i_2\colon v_1\mapsto e_3, \quad v_2 \mapsto e_1, \quad v_3 \mapsto e_2, \quad  v_4 \mapsto e_3.
\end{align*}
The hypergraph $H$ and the query $\q$ are shown in Fig.~\ref{fig:depth1}: each $R_{v_k}$ is represented by an edge, $i_1(v_k)$ is indicated by the circle-shaped end of the edge and $i_2(v_k)$ by the diamond-shaped end of the edge; the $e_j$ are shown as large grey squares.
\begin{figure}[t]%
\centering%
\begin{tikzpicture}[scale=1.2]
\coordinate (c2) at (-1,0.75);
\coordinate (c1) at (-1.5,1.5);
\coordinate (c3) at (0.5,1.5);
\coordinate (c4) at (-0.25,0) {};
\node[draw=black,fill=gray!20,ultra thin,rounded corners=12,fit=(c1) (c2) (c4),inner sep=20] {};
\node[draw=black,ultra thin,rounded corners=12,fit=(c1) (c2) (c4),inner sep=20] {};
\node[draw=black,fill=gray!60,fill opacity=0.5,ultra thin,rounded corners=12,fit=(c3) (c4),inner sep=15] {};
\node[draw=black,ultra thin,rounded corners=12,fit=(c3) (c4),inner sep=15] {};
\node[draw=black,fill=gray!5,fill opacity=0.5,ultra thin,rounded corners=12,fit=(c1) (c2) (c3),inner sep=10] {};
\node[draw=black,ultra thin,rounded corners=12,fit=(c1) (c2) (c3),inner sep=10] {};
\node[point,fill=white,label=left:{$v_2$}] (v2) at (c2) {};
\node[point,fill=white,label=right:{$v_1$}] (v1) at (c1) {};
\node[point,fill=white,label=left:{$v_3$}] (v3) at (c3) {};
\node[point,fill=white,label=right:{$v_4$}] (v4) at (c4) {};
\node at (-1.5,0) {\large $e_3$};
\node at (0.6,0) {\large $e_2$};
\node at (-0.25,1) {\large $e_1$};
\coordinate (v12) at (2,0.3); %[point, fill=white] {}, label=above left:$\due{v_1}$
\coordinate (v22) at (3.3,1.3); % [point, fill=white] {}, label=right:$\due{v_2}$
\coordinate (v42) at (2.6,-0.1); % ,[point, fill=white] {} label=below:$\due{v_4}$
\coordinate (v11) at (3,1.7); % [point, fill=white] {}, label=above left:$\uno{v_1}$
\coordinate (v21) at (2.6,0.3); %  [point, fill=white] {}, label=below:$\uno{v_2}$
\coordinate (v31) at (3.6,1.7); % [point, fill=white] {} , label=above right:$\uno{v_3}$
\coordinate (v32) at (4.5,0.3); % [point, fill=white] {} , label=above right:$\due{v_3}$
\coordinate (v41) at (4.1,-0.1); % [point, fill=white] {} , label=below:$\uno{v_4}$
\coordinate (v61) at (4.6,-0.1); % [point, fill=white] {} , label=above right:$\due{v_3}$
\node[rectangle,draw=black,fill=gray!30,ultra thin,rounded corners=12,fit=(v41) (v61) (v32),inner sep=15]  (e2) {};
\node at ($(e2)+(0.2,-0.3)$) {\large $z_{e_2}$};
\node[rectangle,draw=black,fill=gray!5,ultra thin,rounded corners=12,fit=(v11) (v22) (v31),inner sep=15] (e1) {\large $z_{e_1}$};
\node[rectangle,draw=black,fill=gray!20,ultra thin,rounded corners=12,fit=(v12) (v21) (v42),inner sep=15] (e3) {};
\node at ($(e3)+(-0.1,-0.15)$) {\large $z_{e_3}$};
\begin{scope}[open diamond-*,line width=0.4mm,fill=gray,draw=black!80]
\draw (v12) to node[left, midway] {\small $R_{v_1}$} (v11) ;
\draw (v22) to node[right, pos=0.6] {\small $R_{v_2}$} (v21);
\draw (v42) to node[below, midway] {\small $R_{v_4}$}  (v41);
\draw (v32) to node[right, midway] {\small $R_{v_3}$} (v31);
\end{scope}
\node at (1.2,2) {\large $H$};
\node at (2,1.9) {\large $\q$};
\node at (6.2,-0.25) {\normalsize $\t^{e_1}$};
\node at (8.5,-0.25) {\normalsize $\C_\T^{A_{e_1}(a)}$};
\coordinate (v31c) at (6.05,1.75);
\coordinate (v22c) at (6.2,1.75);
\coordinate (v11c) at (6.35,1.75);
\coordinate (w1c) at (6.2,1.68);
\coordinate (w1c0) at (6.2,1.68);
\coordinate (v32c) at (6.05,0.25);
\coordinate (v21c) at (6.2,0.25);
\coordinate (v12c) at (6.35,0.25);
\coordinate (a1c) at (6.2,0.35);
\coordinate (a1c0) at (6.2,0.35);
\draw[line width=2.5mm,-triangle 60,postaction={draw, line width=6mm, shorten >=4mm, -},draw=gray!20] (6.2,0.25) -- (6.2,2);
\node[rectangle,draw=black,line width=0.3mm,rounded corners=5,fit=(w1c0) (w1c),inner xsep=20, inner ysep=10,label=right:{\small $z_{e_1}$}] (w1) {};
\node[rectangle,draw=black,line width=0.3mm,rounded corners=5,fit=(a1c0) (a1c),inner xsep=20, inner ysep=10,label=right:{\hspace*{-0.5em}\small\begin{tabular}{c}$z_{e_2}$\\$z_{e_3}$\end{tabular}}] (a1) {};
%,label=below:{\small $A_{e_1}$}
%\draw[->,line width=0.5mm,draw=black,decorate,decoration={amplitude=0.3mm,segment length=2mm,post length=1mm}] ($(a1)+(0.74,0.28)$) -- ($(w1)+(0.74,-0.3)$) node[right, midway] {\footnotesize \textcolor{gray}{$P_\zeta$}}; % $S_{e_1}$
%
\begin{scope}[line width=0.25mm,draw=black]
\draw[open diamond-*]  (v32c) to (v31c); %  node[left] {\footnotesize $R_{v_3}^-$}
\draw[*-open diamond] (v21c) to (v22c); %  node[left] {\footnotesize $R_{v_2}$} 
\draw[open diamond-*] (v12c) to (v11c); %  node[left] {\footnotesize $R_{v_1}^-$}
\end{scope}
\draw[densely dotted,bend left,looseness=0.8] (v31) to  (v31c);
\draw[densely dotted,bend left,looseness=0.6] (v22) to  (v22c);
\draw[densely dotted,bend left,looseness=0.8] (v11) to  (v11c);
\draw[densely dotted,bend right,looseness=0.2] (v32) to  (v32c); %
\draw[densely dotted,bend right,looseness=0.2] (v21) to  (v21c);
\draw[densely dotted,bend right,looseness=0.3] (v12) to  (v12c);
\node[bpoint, label=right:{$A_{e_1}$}, label=left:{\footnotesize $a$}] (aa) at (9,0.25) {};
\node[point] (ab) at (9,1.75) {};
\draw[can,->] (aa) to node[right] {\footnotesize\textcolor{gray}{$P_\zeta$}} node[left] {\footnotesize$R_{v_3}^-, R_{v_2}, R_{v_1}^-$} (ab);
\end{tikzpicture}%
\caption{Hypergraph $H$ in Example~\ref{example1}, its CQ $\q$, tree witness $\t^{e_1}$ for $\OMQI{H}$ and canonical model $\C_\T^{A_{e_1}(a)}$.}\label{fig:depth1}
\end{figure}
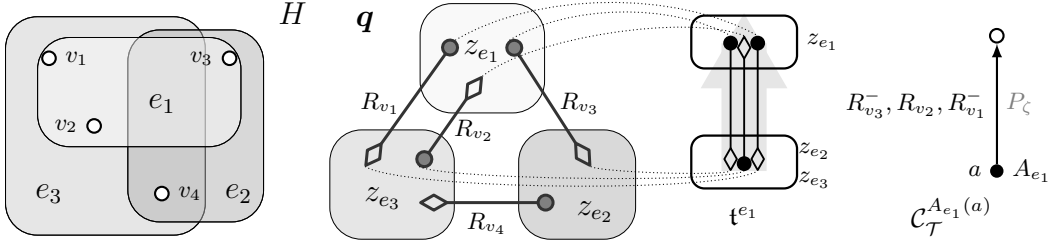
\noindent In this case,
\begin{equation*}
\q \ \ = \ \ \exists z_{e_1},z_{e_2},z_{e_3} \,\bigl(R_{v_1}(z_{e_1},z_{e_3}) \land R_{v_2}(z_{e_3},z_{e_1}) \land
R_{v_3}(z_{e_1},z_{e_2})\land R_{v_4}(z_{e_2},z_{e_3}) \bigr)
\end{equation*}
and $\T$ consists of the following axioms:
\begin{align*}
A_{e_1}(x) &\to %\underbrace{
\exists y\,  \bigl[R_{v_1}(y,x)\land R_{v_2}(x,y)\land R_{v_3}(y,x)\bigr],\\%}_{\zeta(x)},\\
A_{e_2}(x) &\to \exists y\, \bigl[  R_{v_3}(x,y) \land  R_{v_4}(y,x)\bigr], \\
A_{e_3}(x) &\to \exists y\,  \bigl[ R_{v_1}(x,y) \land R_{v_2}(y,x)\land R_{v_4}(x,y)\bigr].
\end{align*}
The canonical model $\C^{A_{\smash{e_1}}(a)}_{\T}$ is shown on the right-hand side of Fig.~\ref{fig:depth1}. Note that each $z_{e}$ determines the tree witness $\t^{e}$ with $\q_{\t^e} = \{R_v(z_{i_1(v)}, z_{i_2(v)}) \mid v \in e\}$; distinct $\t^e$ and $\t^{\smash{e'}}$ are conflicting iff $e \cap e' \ne \emptyset$. It follows that $H$ is isomorphic to $\HG{\OMQI{H}}$.
\end{example}
\begin{theorem}\label{representable}
\textup{(}i\textup{)} Any hypergraph $H$ of degree~2 is isomorphic to $\HG{\OMQI{H}}$.

\textup{(}ii\textup{)} 
Any monotone HGP $P$ based on a hypergraph $H$ of degree~2 computes a subfunction of the primitive evaluation function $f^\vartriangle_{\OMQI{H}}$.
\end{theorem}
\begin{proof}
(\emph{i}) We show that the map $g\colon v\mapsto R_v(z_{i_1(v)}, z_{i_2(v)})$ is an isomorphism between $H$ and $\HG{\OMQI{H}}$. By the definition of $\OMQI{H}$, $g$ is a bijection between $V$ and the atoms of $\q$. For any $e\in E$, there is a tree witness
$\t^e = (\tr^e, \ti^e)$ generated by $A_e(x)$ with
$\ti^e = \{z_e\}$ and $\tr^e = \{z_{e'} \mid e \cap e' \neq\emptyset, e \ne e'\}$,
and $\q_{\t^e}$ consists of the $g(v)$, for $v\in e$. Conversely, every tree witness $\t$ for $\OMQI{H}$ contains $z_e \in \ti$, for some $e\in E$, and so
$\q_\t = \{ g(v) \mid v \in e\}$.

\smallskip

(\emph{ii}) By Proposition~\ref{hyper:program}~(\emph{i}), $P$ computes a subfunction of $f_H$. Thus, it suffices to show that $f_H$ is a subfunction of $f^\vartriangle_{\OMQI{H}}$. Let $H = (V,E)$ be a hypergraph of degree~$2$.  For any $\avec{\alpha}\in \{0,1\}^{\smash{|H|}}$, we define~$\avec{\gamma}$ by taking $\avec{\gamma}(R_v) = \avec{\alpha}(p_v)$ for $v \in V$, $\avec{\gamma}(A_e) = \avec{\alpha}(p_e)$ for $e \in E$ (and $\avec{\gamma}(P_\zeta) = 0$ for all normalisation predicates $P_\zeta$). We prove that  $f_H(\avec{\alpha}) = 1$ iff \mbox{$\T, \A(\avec{\gamma}) \models \q$}. By the definition of $\T$, for each $e\in E$ with $A_e(a) \in\A(\avec{\gamma})$ or, equivalently, $\avec{\alpha}(p_e) = 1$, the canonical model $\C_{\T,\A(\avec{\gamma})}$ contains a labelled null $w_e$ such that 
\begin{equation*}
\C_{\T,\A(\avec{\gamma})} \models \bigwedge_{\begin{subarray}{c}v\in V\\i_1(v) = e\end{subarray}} R_v(w_e,a) \ \ \land \bigwedge_{\begin{subarray}{c}v\in V\\i_2(v) = e\end{subarray}} R_v(a,w_e).
\end{equation*}

$(\Rightarrow)$ Let $E'$ be an independent subset of $E$ such that $\bigwedge_{v \in V \setminus V_{E'}} p_v \land \bigwedge_{e \in E'} p_e$ is true on~$\avec{\alpha}$.
Define $h \colon \q \to \C_{\T, \A(\avec{\gamma})}$
by taking $h(z_e) = a$ if $e\notin E'$ and $h(z_e) = w_e$  otherwise.
One can check that $h$ is a homomorphism, and so $\T, \A({\avec{\gamma}}) \models \q$.

$(\Leftarrow)$ Given a homomorphism $h \colon \q \to \C_{\T, \A(\avec{\gamma})}$, we show that $E' = \{e \in E  \mid h(z_e) \neq a\}$ is independent. Indeed, if $e, e' \in E'$ and $v \in e \cap e'$, then $h$ sends one variable of the $R_v$-atom to the labelled null $w_e$ and the other end to  $w_{e'}$, which is impossible.  We claim  that $f_H(\avec{\alpha}) = 1$. Indeed, for each $v \in V\setminus V_{E'}$, $h$ sends both ends of the $R_v$-atom to $a$, and so $\avec{\alpha}(p_v) = 1$. For each $e \in E'$, we must have $h(z_e) = w_{e}$ because $h(z_e) \neq a$, and so $\avec{\alpha}(p_e) = 1$. It follows that $f_H(\avec{\alpha}) = 1$.
\end{proof}

%******************

\subsection{Tree-Shaped OMQs and Tree Hypergraphs}\label{sec:5.3}

We call an OMQ $\omq = (\T,\q)$ \emph{tree-shaped} if the CQ $\q$ is tree-shaped. We now  establish a close correspondence between tree-shaped OMQs and tree hypergraphs that are defined as follows.\!\footnote{Our definition of tree hypergraph is a minor variant of the notion of (sub)tree hypergraph (aka hypertree) from graph theory~\cite{Flament1978223,Brandstadt:1999:GCS:302970,Bretto:2013:HTI:2500991}.}

Suppose $T=(V_T,E_T)$ is an (undirected) tree. A leaf is a vertex of degree~1. A subtree $T' = (V'_T,E'_T)$ of $T$ is said to be \emph{convex} if, for any non-leaf vertex $u$ in the subtree~$T'$, we have $\{u,v\} \in E'_T$ whenever $\{u,v\} \in E_T$. 
A hypergraph $H = (V,E)$ is called a \emph{tree hypergraph} if there is a tree $T=(V_T,E_T)$ such that $V = E_T$ and every hyperedge $e\in E$ induces a convex subtree $T_e$ of $T$. In this case, we call $T$ the \emph{underlying tree} of~$H$.
The \emph{boundary} of a hyperedge $e$ consists of all leaves of $T_e$; the interior of $e$ is the set of non-leaves of~$T_e$.
%
%If $T$ has exactly 2 leaves then $H$ is called an \emph{interval hypergraph}~\cite{Brandstadt:1999:GCS:302970,Bretto:2013:HTI:2500991}.
%
A \emph{tree hypergraph program} (THGP) is an HGP based on a tree hypergraph. 

%Suppose $T=(V_T,E_T)$ is a tree. For $u,v \in V_T$, the \emph{interval} $[u,v]$ is the set of edges on the (unique simple) path connecting $u$ and $v$. For $v_1, \dots, v_k \in V_T$, $k\ge 2$, the \emph{interval} $[v_1, \dots, v_k]$ is the union of the intervals $[v_i, v_j]$ over all pairs $(i,j)$. We call $v \in V_T$ a \emph{boundary vertex} for an interval $I$ if there exist edges $\{v,u\} \in E_T \cap I$ and $\{v,u'\} \in E_T \setminus I$. A hypergraph $H = (V_H,E_H)$ is a \emph{tree hypergraph}  if there is a tree $T=(V_T,E_T)$ such that $V_H=E_T$ and every hyperedge in $E_H$ is an interval of $T$ all of whose boundary vertices have degree 2 in $T$. If $T$ has exactly 2 leaves then $H$ is called an \emph{interval hypergraph}~\cite{Brandstadt:1999:GCS:302970,Bretto:2013:HTI:2500991}.
%%
%A \emph{tree hypergraph program} (\THP) is an HGP based on a tree hypergraph, and an \emph{interval HGP} (IHGP) is an HGP based on an interval hypergraph.
% As a special case, we have \emph{interval hypergraphs} \cite{Brandstadt:1999:GCS:302970,Bretto:2013:HTI:2500991}
%and \emph{interval HGPs}, whose underlying trees have exactly 2 leaves.
%

\begin{example}\label{ex:hypertree}
Let $T$ be the tree shown in Fig.~\ref{fig:tree}. Any tree hypergraph with underlying tree $T$ has the set of vertices
$\{\{1,2\}, \{2,3\}, \{2,6\}, \{3,4\}, \{4,5\}\}$ (each vertex is an edge of $T$),
 %$\{v_{12},v_{23},v_{26},v_{34},v_{45}\}$, 
and its hyperedges may include 
$\{\{1,2\}, \{2,3\}, \{2,6\}\}$ 
%$\{ v_{12}, v_{23}, v_{26}\}$ 
as the subtree of $T$ induced by these edges is convex, but not $\{\{1,2\},\{2,3\}\}$.
%$\{v_{12}, v_{23}\}$.
%
\begin{figure}[t]%
\centering%
\begin{tikzpicture}[xscale=2,yscale=0.7]
\coordinate (n1) at (0,0);
\coordinate (n3) at (2,0);
\coordinate (n6) at (1,1);
\node[draw,ultra thin,fill=gray!20,rounded corners,inner xsep=5mm,inner ysep=7mm,fit=(n1) (n3) (n6)] {};
\node[draw,ultra thin,fill=gray!5,rounded corners,inner xsep=3mm,inner ysep=5mm,fit=(n1) (n3)] {};
\node (1) at (0,0) [point,label=below:{\footnotesize $1$}]{};
\node (2) at (1,0) [point,label=below:{\footnotesize $2$}]{};
\node (3) at (2,0) [point,label=below:{\footnotesize $3$}]{};
\node (4) at (3,0) [point,label=below:{\footnotesize $4$}]{};
\node (5) at (4,0) [point,label=below:{\footnotesize $5$}]{};
\node (6) at (1,1) [point,label=above:{\footnotesize $6$}]{};
\draw[query] (1) to (2); % node[label=below:{$v_{12}$}]{} 
\draw[query] (2) to (3); %node[label=below:{$v_{23}$}]{} 
\draw[query] (3) to  (4); % node[label=below:{$v_{34}$}]{}
\draw[query] (4) to (5); % node[label=below:{$v_{45}$}]{} 
\draw[query] (2) to  (6); % node[pos=0.7,label=right:{$v_{26}$}]{}
\node (c) at (3.5,1.4) {\small\begin{tabular}{c}non-convex\\[-1pt]subtree\end{tabular}};
\node (d) at (-1.2,1.4) {\small\begin{tabular}{c}convex\\[-1pt]subtree\end{tabular}};
\draw[thin] (0,0.9) -- (d);
\draw[thin] (2,0.3) -- (c);
\end{tikzpicture}%
\caption{Tree $T$ in Example~\ref{ex:hypertree}.}\label{fig:tree}
\end{figure}
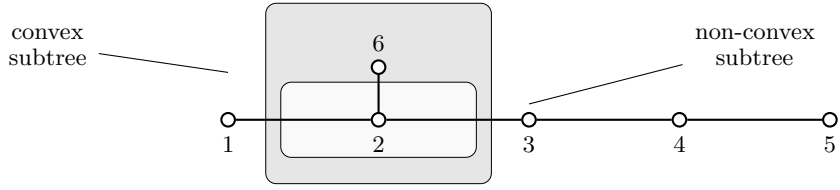
\end{example}

%The following proposition is an immediate consequence of the definitions of tree witness and tree hypergraph.

\begin{theorem}\label{prop:tree-shaped}
If an OMQ $\omq = (\T,\q)$ is tree-shaped, then $\HG{\omq}$ is isomorphic to a tree hypergraph. Furthermore, if $\q$ has at least one binary atom, then the number of leaves in the tree underlying $\HG{\omq}$ is the same as the number of leaves in $\q$.
\end{theorem}
\begin{proof}
The case when $\q$ has no binary atoms is trivial. Otherwise, let $G_\q$ be the Gaifman graph of $\q$ whose vertices $u$ are labelled with the unary atoms $\xi(u)$ in $\q$ of the form $A(u)$ and $P(u,u)$, and whose edges $\{u,v\}$ are labelled with the atoms of the form $P(u,v)$ and $P'(v,u)$ in $\q$.
% such that $\{u,v\} = \{u',v'\}$. 
We replace every edge $\{u,v\}$ labelled with $P_1(u_1',v_1'), \dots, P_n(u_n',v_n')$, for $n \ge 2$, by a sequence of $n$ edges forming a path from $u$ to $v$ and label them with $P_1(u_1',v_1'), \dots, P_n(u_n',v_n')$, respectively. In the resulting tree,  for every vertex $u$ labelled with $n$ unary atoms $\xi_1(u),\dots, \xi_n(u)$,  for $n\ge 1$, we pick an edge $\{u,v\}$ labelled with some $P(u',v')$ and replace it by a sequence of $n+1$ edges forming a path from $u$ to $v$ and label them with $\xi_1(u),\dots, \xi_n(u), P(u',v')$, respectively. The resulting tree $T$ has the same number of leaves as $\q$. It is readily checked that, for any tree witness $\t$ for $\omq$, the set of edges in $T$ labelled with atoms in $\q_\t$ forms a convex subtree of $T$, which gives a tree hypergraph isomorphic to~$\HG{\omq}$.
%
%We start with representing the query as a graph with labels both on its vertices and edges. 
%We take the Gaifman
%graph for $\q$ and then label every edge $(u,v)$ with $R(u,v)$ for all atoms $R(u,v)$ in $\q$. 
%We also label every vertex $u$ by $A(u)$ for all atoms $A(u)$ in $\q$.
%Then we do the following. First, we replace every edge with many labels by a sequence of consecutive
%edges which are labelled by a single variable. 
%Finally, we get rid of labels at vertices by moving them onto fresh edges
%adjacent to the labelled vertices. Formally speaking, we iterate the following surgery: 
%(a) choose some vertex $u$ with a label $A$;
%(b) choose some edge $(u,v)$ with a label $R$;
%(c) replace the edge $(u,v)$ by two edges with a fresh vertex $u'$: 
%$(u, u')$ with a label $A$ and
%$(u', v)$ with a label $R$.
%When the process stops, we arrive at a labelled tree $T$ which will lie under the THGP. We claim that this construction does not increase the number of leaves if $\q$ has at least one binary atom. 
%We also claim that
%for any tree witness $\t$ for $\omq$ the set of all edges of $T$ which are labelled by
%$S(\avec{z})$ for $S(\avec{z}) \in \q_\t$ is a convex subtree of $T$. 
%Thus $H_{\omq}$ is isomorphic to a tree hypergraph.
\end{proof}

%Now we prove the opposite direction: any function given by a THP based on a tree $T$ is a subfunction of
%the $f^p_{\q, \T}$ for an appropriate tree shaped $\q$ and $\T$ of depth 2. The number of leaves in $\q$
%equals to the number of leaves in $T$.

Suppose $H=(V, E)$ is a tree hypergraph whose underlying tree $T = (V_T,E_T)$ has vertices $V_T = \{1,\dots, n\}$,  for  $n > 1$, and $1$ is a leaf of $T$. Let $T^1 = (V_T,E^1_T)$ be the \emph{directed} tree obtained from $T$ by fixing $1$ as the root and orienting the edges away from~$1$. We associate with $H$ a tree-shaped OMQ $\OMQT{H} = (\T,\q)$, in which $\q$ is the Boolean CQ
\begin{equation*}
\q  \ \ = \ \  \bigl\{\, R_{ij}(z_i, y_{ij}), \ \ S_{ij}(y_{ij}, z_{j}) \mid (i,j) \in E^1_{T} \,\bigr\},
\end{equation*}
where the $z_i$, for $i \in V_T$, are the variables for vertices of the tree and the $y_{ij}$, for $(i,j)\in E_T^1$, are the variables for the edges of the tree.
To define $\T$,  %which is the union of $\Tmc_e$ over all hyperedges $e \in E_H$, which are constructed as follows.
suppose a hyperedge $e \in E$ induces a convex directed subtree $T_e = (V_e,E_e)$ of $T^1$ with root $r^e\in V_e$  and leaves $L_e \subseteq V_e$. Denote by $\T$ the ontology that contains the following axiom, for each $e\in E$:
\begin{multline*}
A_e(x) \ \ \to \ \ \exists y\,\bigl[\bigwedge_{(i,j)\in E_e, \ i = r^e}\hspace*{-1em} R_{r^ej}(x,y) \ \ \land \bigwedge_{(i,j)\in E_e, \ j\in L_e} \hspace*{-1.5em}S_{ij}(y,x) \ \ \ \land\\ 
\exists z\,\bigl( \bigwedge_{(i,j)\in E_e, \ i \ne r^e} \hspace*{-1.5em}R_{ij}(z,y) \ \  \ \land \bigwedge_{(i,j)\in E_e, \ j\notin L_e} \hspace*{-2em} S_{ij}(y,z) \bigr)\bigr].
\end{multline*}
%
%\begin{align*}
%& A_e(x) \to \exists y\, R_e(x,y)\quad \text{ and }\quad \exists y\, R_e(y,x) \to  \exists z\, R_e'(x,z),\\
%
%& R_e(x,y) \to S_{i_1 j}(x,y), \text{ for the unique } (i_1,j) \in E_e,\\
%
%& R_e(x,y) \to  S_{j i_k}'(y,x), \text{ for } (j,i_k) \in E_{e} \text{ and } 2 \le k \le m,\\
%
%& R_e'(x,y) \to S'_{ij}(x,y), \text{ for } (i,j) \in E_e \text{ such that  } j \notin \{i_2,\dots,i_m\},\\
%
%& R_e'(x,y) \to S_{ij}(y,x), \text{ for } (i,j) \in E_e \text{ such that } i \ne i_1.
%\end{align*}
% 
Since $T_e$ is convex, its root, $r_e$, has only one outgoing edge, $(r^e,j)$, for some $j$, and so the first conjunct above contains a single atom, $R_{r^ej}(x,y)$. These axioms (together with convexity of hyperedges) 
ensure that  $\OMQT{H}$ has a tree witness $\t^e = (\tr^e,\ti^e)$, for $e\in E$,  with
\begin{align*}
& \tr^e \ \ = \ \ \{\, z_i \mid i \text{ is on the boundary of } e \,\},\\ 
&  \ti^e \ \ = \ \
\{\,z_i \mid i \text{ is in the interior of } e\,\} \ \ \cup \ \ \{\,y_{ij} \mid (i,j) \in e\,\}.
\end{align*}
% 
%We set $\T = \bigcup_{e \in E_H} \T_e$. 
Note that $\T$ is of depth 2, and $\OMQT{H}$ is of polynomial size in $|H|$. 

%Let $e = \la v_{i_1}, \ldots, v_{i_m} \ra \in E_P$ with $v_{i_1}$ the vertex in $e$ that is highest in $T^{\downarrow}$,
%and suppose w.l.o.g.\ that every $v_{i_j}$ is either a boundary vertex of $e$ or a leaf in $T$.
%Then $\Tmc_e$ is defined as follows:
%\begin{align*}
%& \{B_e(x) \rightarrow \exists y R_e(x,y), \exists y R_e(y,x) \rightarrow \exists y R_e'(x,y)\} \ \cup \\
%& \{R_e(x,y) \rightarrow S_{i_1, k}(x,y) \mid \{v_{i_1},v_k\} \in e \}\ \cup \\
%& \{R_e(y,x) \rightarrow  S_{j_\ell, i_\ell}'(x,y) \mid 1 < \ell \leq m, (v_{j_\ell},v_{i_\ell}) \in T^{\downarrow} \}\ \cup \\
%& \{R_e'(x,y) \rightarrow S'_{j, k}(x,y)\mid \{v_j, v_k\} \in e, (v_j, v_k) \in T^{\downarrow}, \\
%&\hspace*{3.75cm} v_k \neq v_{i_\ell} \text{ for all } 1 < \ell \leq m  \}\ \cup \\
%& \{R_e'(x,y) \!\rightarrow\! S_{j,k}(y,x) \! \mid \! \{v_j, v_k\} \in e, (v_j, v_k) \in T^{\downarrow}, v_j \neq v_{i_1} \! \}.
%\end{align*}
%Observe that both $\q_P$ and $\Tmc_P$ are of polynomial size in $|P|$.

\begin{example}\label{ex:tree-hypergraph:char}
Let $H$ be the tree hypergraph whose underlying tree is as in Example~\ref{ex:hypertree} with fixed root $1$ and whose only hyperedge is $e = \{\{1,2\},\{2,3\},\{3,4\},\{2,6\}\}$. The CQ $\q$ and the canonical model $\Cmc^{\smash{A_e(a)}}_{\Tmc}$ for this $H$ are shown  in Fig.~\ref{fig:hypertree-query}.
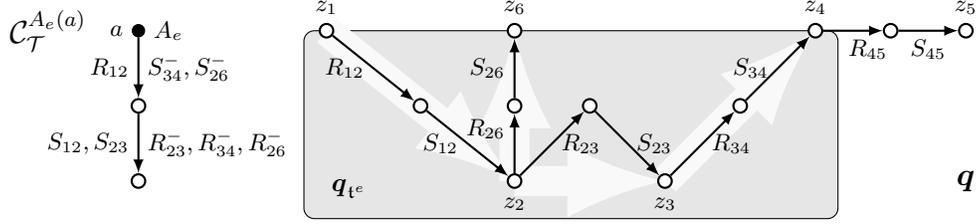
\begin{figure}[t]%
\centering%
\begin{tikzpicture}\footnotesize
\coordinate (z1) at (-1.5,2);
\coordinate (z4) at (5,2);
\coordinate (z2) at (1,-0.5);
\node[draw,thin,fill=gray!20,rounded corners,inner xsep=3mm,inner ysep=0mm,fit=(z1) (z2) (z4)] {};
\begin{scope}[draw=gray!5,fill=gray!5,line width=4mm]
\draw[->] (-1.5,2) -- (1,0);
\draw[->] (1,0) -- (3,0);
\draw[->] (1,0) -- (1,2);
\draw[->] (3,0) -- (5,2);
\end{scope}
\node[point,label=above:{$z_1$}] (y1) at (-1.5,2) {};
\node[point] (y12) at (-0.25,1) {}; % ,label=left:{$y_{12}$}
\node[point,label=below:{$z_2$}] (y2) at (1,0) {};
\node[point] (y26) at (1,1) {}; % ,label=left:{$y_{26}$}
\node[point,label=above:{$z_{6}$}] (y6) at (1,2) {};
\node[point] (y23) at (2,1) {}; % ,label=above:{$y_{23}$}
\node[point,label=below:{$z_3$}] (y3) at (3,0) {};
\node[point] (y34) at (4,1) {}; % ,label=left:{$y_{34}$}
\node[point,label=above:{$z_{4}$}] (y4) at (5,2) {};
\node[point] (y45) at (6,2) {}; % ,label=above:{$y_{45}$}
\node[point,label=above:{$z_{5}$}] (y5) at (7,2) {};
\node at (-1.2, -0.1){\normalsize $\q_{\t^e}$};
\node at (7, 0){\large $\q$};
\draw[->,query] (y1) to node[left] {\footnotesize$R_{12}$} (y12);
\draw[->,query] (y12) to node[left] {\footnotesize$S_{12}$} (y2);
\draw[->,query] (y2) to node[left,near end] {\footnotesize$R_{26}$} (y26);
\draw[->,query] (y2) to node[right] {\footnotesize$R_{23}$} (y23);
\draw[->,query] (y26) to node[left] {\footnotesize$S_{26}$} (y6);
\draw[->,query] (y23) to node[right] {\footnotesize$S_{23}$} (y3);
\draw[->,query] (y3) to node[right] {\footnotesize$R_{34}$} (y34);
\draw[->,query] (y34) to node[left] {\footnotesize$S_{34}$} (y4);
\draw[->,query] (y4) to node[below,near end] {\footnotesize$R_{45}$} (y45);
\draw[->,query] (y45) to node[below] {\footnotesize$S_{45}$} (y5);
\node (a) at (-4,2) [bpoint, label=right:{$A_e$}, label=left:{$a$}]{};
\node (d1) at (-4,1) [point, fill=white]{};
\node (d2) at (-4,0) [point, fill=white]{};
\node at (-5.2, 2){\large $\Cmc^{A_e(a)}_{\Tmc}$};
\draw[can,->] (a)  to node [right]{\footnotesize $S_{34}^-,S_{26}^-$} node[left]{$R_{12}$} (d1);
\draw[can,->] (d1)  to node [left]{\footnotesize $S_{12},S_{23}$} node[right]{$R_{23}^-,R_{34}^-, R_{26}^-$} (d2);
\end{tikzpicture}%
\caption{The canonical model $\C_\T^{\smash{A_e(a)}}$ and the query $\q$ ($y_{ij}$ is the half-way point between $z_i$ and $z_j$) for the tree hypergraph $H$ in Example~\ref{ex:tree-hypergraph:char}.}\label{fig:hypertree-query}
\end{figure}
Note the homomorphism from $\q_{\t^e}$ into $\Cmc^{\smash{A_e(a)}}_{\Tmc}$.
\end{example}

The proofs of the following results (which are THGP analogues of Theorem~\ref{hg-to-query} and Propositions~\ref{hyper:program}~(\emph{ii}) and~\ref{lem:genhgps}, respectively) are given in Appendices~\ref{app:proof:Th5.9} and~\ref{app:proof:Prop5.10}:

\begin{theorem}\label{tree-hg-to-query}
\textup{(}i\textup{)} Any tree hypergraph $H$ is isomorphic to a subgraph of $\HG{\OMQT{H}}$.

\textup{(}ii\textup{)} Any monotone THGP based on a tree hypergraph $H$ computes a subfunction of the primitive evaluation function $f^\vartriangle_{\OMQT{H}}$.
\end{theorem}

\begin{proposition}\label{hyper:thgp}
\textup{(}i\textup{)} For any tree hypergraph $H$ of degree at most  $d$, there is a monotone THGP of size $O(|H|)$ that computes $f_H$ and such that its hypergraph is of degree at most $\max(2,d)$.

\textup{(}ii\textup{)} For every generalised THGP $P$ over $n$ variables, there is a THGP $P'$ such that $|P'| \le n \cdot |P|$ and $P'$ has the same degree and number of leaves as $P$ and computes the same function.
\end{proposition}

%*******************

\subsection{OMQs with Bounded Treewidth CQs and Bounded Depth Ontologies}\label{sec:boundedtw}

Recall (see, e.g.,~\cite{DBLP:series/txtcs/FlumG06}) that a \emph{tree decomposition} of an undirected graph $G=(V,E)$ is a pair $(T,\lambda)$, where $T$ is an (undirected) tree and $\lambda$ a function from the set of nodes of $T$ to $2^V$ such that 
\begin{nitemize}
\item[--] for every $v \in V$, there exists a node $N$ with $v \in \lambda(N)$;

\item[--] for every $e \in E$, there exists a node $N$ with $e \subseteq \lambda(N)$;

\item[--] for every $v \in V$, the nodes $\{N\mid v \in \lambda(N)\} $ induce a (connected) subtree of~$T$.
\end{nitemize}
We call the set $\lambda(N) \subseteq V$ a \emph{bag for} $N$. The \emph{width of a tree decomposition} $(T, \lambda)$ is the size of its largest bag minus one. The \emph{treewidth} of $G$ is the minimum width over all tree decompositions of $G$. The \emph{treewidth of a CQ} $\q$ is the treewidth of its Gaifman graph~$G_\q$.

\begin{example}\label{ex:boundedtree} 
The Boolean CQ
$\q = \bigl\{ R(y_2, y_1),\ R(y_4, y_1), \ S_1(y_3, y_4), \ S_2(y_2, y_4)\bigr\}$
and its tree decomposition $(T, \lambda)$ of width 2 are shown in Fig.~\ref{fig:tree-decomposition}, where $T$ has two nodes, $N_1$ and $N_2$, connected by an edge, with bags $\lambda(N_1)= \{y_1, y_2, y_4\}$ and \mbox{$\lambda(N_2) = \{y_2, y_3, y_4\}$}.%
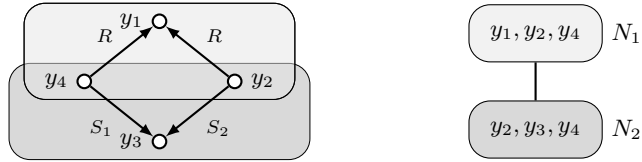
\begin{figure}[t]%
\centering%
\begin{tikzpicture}[bag/.style={draw,rectangle,rounded corners=8,inner sep=8pt},yscale=0.8]
\draw[ultra thin,fill=gray!30,rounded corners=8] (-1,0.7) rectangle +(4,1.6);
\draw[ultra thin,fill=gray!20,fill opacity=0.5,rounded corners=8] (-0.8,1.7) rectangle +(3.6,1.6);
\draw[ultra thin,rounded corners=8] (-0.8,1.7) rectangle +(3.6,1.6);
%Query
\node[point,label=left:{${y_1}$}] (y1) at (1,3) {};
\node[point,label=right:{${y_2}$}] (y2) at (2,2) {};
\node[point,label=left:{${y_3}$}] (y3) at (1,1) {};
\node[point,label=left:{${y_4}$}] (y4) at (0,2) {};
\draw[->,query] (y4) to node[above left] {\scriptsize$R$} (y1);
\draw[->,query] (y2) to node[above right] {\scriptsize$R$} (y1);
\draw[->,query] (y4) to node[below left] {\scriptsize$S_{1}$} (y3);
\draw[->,query] (y2) to node[below right] {\scriptsize$S_{2}$} (y3);
%Tree decomposition
\node[bag,fill=gray!10,label=right:{$N_1$}](N1) at (6,2.8){$y_1, y_2, y_4$};
\node[bag,fill=gray!30,label=right:{$N_2$}](N2) at (6,1.2){$y_2, y_3, y_4$};
\draw[-,query] (N1) to (N2);
\end{tikzpicture}
\caption{Tree decomposition in Example~\ref{ex:boundedtree}.}\label{fig:tree-decomposition}
\end{figure}%
\end{example}

Our aim in this section is to show that, for any OMQ $\omq(\avec{x}) = (\T,\q(\avec{x}))$ with $\q$ of bounded treewidth and a finite fundamental set $\Omega_{\omq}$, the modified tree-witness hypergraph function $\smash{\homfn}$ can be computed using a monotone THGP of size bounded by a polynomial in $|\q|$ and $|\Omega_{\omq}|$.

Let  $(T, \lambda)$ be a  tree decomposition of $G_\q$  of width~\mbox{$\twidth - 1$}. 
In order to refer to the variables of $\q$, for each bag $\lambda(N)$, we fix an order of variables in the bag and define a injection $\nu_N\colon \lambda(N) \to \{1,\dots, m\}$ that gives the index of each $z$ in $\lambda(N)$. 
A  \emph{\textup{(}bag\textup{)} type} is an $\twidth$-tuple of the form $\avec{w} = (\avec{w}[1], \dots, \avec{w}[\twidth])$, where each $\avec{w}[i] \in \Omega_{\omq}$. Intuitively, the $i$th component $\avec{w}[i]$ of $\avec{w}$ indicates that the $i$th variable in the bag is mapped to a domain element of the form $a\avec{w}[i]$ in the canonical model $\C_{\T,\A}$. 
We say that a type $\avec{w}$ is \emph{compatible with a node} $N$ of $T$ if the following conditions hold, for all $z,z'\in\lambda(N)$:
\begin{enumerate}[(1)]
\item if $A(z) \in \q$ and $\avec{w}[\nu_N(z)] \neq \varepsilon$, then $\avec{w}[\nu_N(z)]= w\varrho$ and $\Tmc \models \exists y\, \varrho(y,x) \to A(x)$;

\item if $P(z,z') \in \q$ and either $\avec{w}[\nu_N(z)] \ne\varepsilon$ or  $\avec{w}[\nu_N(z')] \ne \varepsilon$, then 
\begin{nitemize}
\item[--] $\avec{w}[\nu_N(z)] = \avec{w}[\nu_N(z')]$ and $\Tmc \models P(x,x)$, or 

\item[--] $\avec{w}[\nu_N(z')]= \avec{w}[\nu_N(z)] \varrho$ and $\Tmc \models \varrho(x,y) \rightarrow P(x,y)$, or

\item[--] $\avec{w}[\nu_N(z)] = \avec{w}[\nu_N(z')] \varrho^-$ and $\Tmc \models \varrho(x,y) \rightarrow P(y,x)$.
\end{nitemize}
\end{enumerate}
Clearly, the type with all components equal to $\varepsilon$ is compatible with any node $N$ and corresponds to mapping all variables in $\lambda(N)$ to individuals in $\ind(\A)$.

\begin{example}\label{ex:5.11}
Suppose $\T = \{\,A(x) \to \exists y\, R(x, y)\,\}$ and $\q$ is the same as in Example~\ref{ex:boundedtree}. Let $\nu_{N_1}$ and $\nu_{N_2}$ respect the order of the variables in the bags shown in Fig.~\ref{fig:tree-decomposition}. The only types compatible with $N_1$ are $(\varepsilon,\varepsilon,\varepsilon)$ and $(R, \varepsilon,\varepsilon)$, whereas the only type compatible with $N_2$ is $(\varepsilon,\varepsilon,\varepsilon)$.
\end{example}

Let $\avec{w}_1, \dots, \avec{w}_\numtypes$ be all the bag types for $\Omega_\omq$ ($M = |\Omega_\omq|^\twidth$). Denote by $T'$
the tree obtained from $T$ by replacing every edge $\{N_i, N_j\}$ with the following sequence of edges:
\begin{multline*}
\{N_i, u^1_{ij}\},\qquad \{u^k_{ij}, v^k_{ij}\} \text{ and } \{ v^k_{ij}, u^{k+1}_{ij} \}, \text{ for } 1 \leq k < \numtypes, \qquad \{ u^\numtypes_{ij}, v^\numtypes_{ij}\},\qquad \{ v^\numtypes_{ij}, v^\numtypes_{ji}\},\\
\{v^\numtypes_{ji}, u^\numtypes_{ji}\},\qquad \{ u^{k+1}_{ji}, v^k_{ji}  \} \text{ and }  \{ v^k_{ji}, u^k_{ji}\}, \text{ for } 1 \leq k < \numtypes,\qquad \{u^1_{ji}, N_j\}, 
\end{multline*}
for some fresh nodes $u^k_{ij}$, $v^k_{ij}$, $u^k_{ji}$ and $v^k_{ji}$.
We now define a generalised monotone THGP $P_{\omq}$ based on a hypergraph with the underlying tree $T'$. 
Denote by $[L]$  the set of nodes of the minimal convex subtree of $T'$ containing all nodes of $L$. The hypergraph has the following hyperedges:
\begin{nitemize}
\item[--] $E_i^k = [N_i,u_{ij_1}^k, \dots, u_{i j_n}^k]$ if $N_{j_1}, \dots, N_{j_n}$ are the neighbours of $N_i$ in $T$ and $\avec{w}_k$ is compatible with $N_i$;
 
\item[--] $E_{ij}^{k \ell} = [v_{ij}^k, v_{ji}^\ell]$ if $\{N_i, N_j\}$ is an edge in $T$ and $(\avec{w}_k, \avec{w}_\ell)$ is compatible with $(N_i, N_j)$ in the sense that $\avec{w}_k[\nu_{N_i}(z)] = \avec{w}_\ell[\nu_{N_j}(z)]$, for all $z\in \lambda(N_i)\cap \lambda(N_j)$. 
\end{nitemize}
We label the vertices of the hypergraph---that is, the edges of $T'$---in the following way. The edges $\{N_i, u_{ij}^1\}$, $\{v_{ij}^k, u_{ij}^{k +1}\}$ and $\{v_{ij}^\numtypes, v_{ji}^\numtypes\}$ are labelled with $0$, and every edge $\{u_{ij}^k, v_{ij}^k\}$ is labelled with the conjunction of the following variables:
\begin{nitemize}
\item[--] $p_\atom$, whenever $\atom \in \q$, $\avec{z}\subseteq \lambda(N_i)$ and  
$\avec{w}_k[\nu_{N_i}(z)]= \varepsilon$, for all 
$z\in \avec{z}$;

\item[--] $p_{\varrho^*(z)}$, whenever $A(z) \in \q$, $z\in \lambda(N_i)$ and 
$\avec{w}_k[\nu_{N_i}(z)] = \varrho w$;

\item[--] $p_{\varrho^*(z)}$, $p_{\varrho^*(z')}$ and $p_{z=z'}$, whenever $R(z,z') \in \q$ (possibly with $z=z'$), $z,z'\in \lambda(N_i)$,
and either $\avec{w}_k[\nu_{N_i}(z)]= \varrho w$ or $\avec{w}_k[\nu_{N_i}(z')]= \varrho w$.
\end{nitemize}
The following result is proved in Appendix~\ref{app:proof:5.13}:  

\begin{theorem}\label{DL2THP}
For every OMQ $\omq = (\Tmc,\q)$  with a fundamental set $\Omega_{\omq}$ and with $\q$ of treewidth~$t$,
the generalised monotone THGP $P_{\omq}$ computes $\homfn$ and is of size polynomial in $|\q|$ and $|\Omega_{\omq}|^t$. 
\end{theorem}

\begin{example}\label{ex:thgp}
Let $\omq = (\T,\q)$ be the OMQ from Example~\ref{ex:5.11}. As we have seen, there are only two types compatible with nodes in $T$: $\avec{w}_1=(\varepsilon,\varepsilon,\varepsilon)$ and $\avec{w}_2 = (R, \varepsilon, \varepsilon)$.
This gives us the generalised THGP $P_{\omq}$ shown in Fig.~\ref{fig:thgp}, where the omitted labels are all~0. 
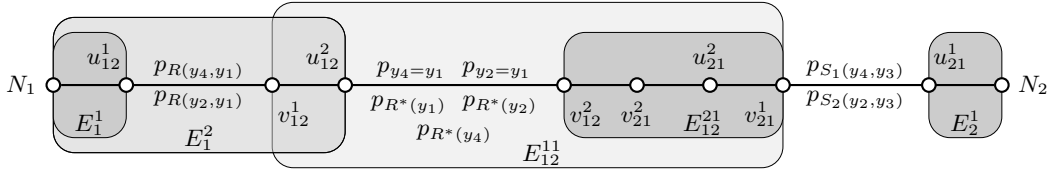
\begin{figure}[t]%
\centering%
\begin{tikzpicture}[xscale=0.97]
\draw[ultra thin,fill=gray!10,rounded corners=8] (3,-1.1) rectangle +(7,2.2);
\node at (6.7,-0.9) {\small $E_{12}^{11}$};
\draw[ultra thin,fill=gray!40,fill opacity=0.5,rounded corners=8] (0,-0.9) rectangle +(4,1.8);
\draw[ultra thin,rounded corners=8] (0,-0.9) rectangle +(4,1.8);
\node at (2,-0.7) {\small $E_1^2$};
\draw[ultra thin,fill=gray!40,rounded corners=8] (0,-0.7) rectangle +(1,1.4);
\node at (0.5,-0.5) {\small $E_1^1$};
\draw[ultra thin,fill=gray!40,rounded corners=8] (12,-0.7) rectangle +(1,1.4);
\node at (12.5,-0.5) {\small $E_2^1$};
\draw[ultra thin,fill=gray!40,rounded corners=8] (7,-0.7) rectangle +(3,1.4);
\node at (8.9,-0.5) {\small $E^{21}_{12}$};
\node[point,label=left:{$N_1$}] (N1) at (0,0) {};
\node[point,label=right:{$N_2$}] (N2) at (13,0) {};
\draw[query] (N1) -- (N2);
\node[point,label=above:{\small $u_{12}^1$\hspace*{1.8em}}] (u121) at (1,0) {};
\node[point,label=below:{\small\hspace*{1.8em}$v_{12}^1$}] (v121) at (3,0) {};
\node[point,label=above:{\small$u_{12}^2$\hspace*{1.8em}}] (u122) at (4,0) {};
\node[point,label=below:{\small\hspace*{1.8em}$v_{12}^2$}] (v122) at (7,0) {};
\node[point,label=below:{\small$v_{21}^2$}] (v212) at (8,0) {};
\node[point,label=above:{\small$u_{21}^2$}] (u212) at (9,0) {};
\node[point,label=below:{\small$v_{21}^1$\hspace*{1.8em}}] (v211) at (10,0) {};
\node[point,label=above:{\small\hspace*{1.8em}$u_{21}^1$}] (u211) at (12,0) {};
\node at (2, 0.2) {\small $p_{R(y_4,y_1)}$};
\node at (2,-0.2) {\small $p_{R(y_2, y_1)}$};
\node at (5.5, 0.2) {\small$p_{y_4 = y_1} \ \ p_{y_2=y_1}$};
\node at (5.5,-0.4) {\small\begin{tabular}{c}$p_{R^*(y_1)} \ \ p_{R^*(y_2)}$\\$p_{R^*(y_4)}$\end{tabular}};
\node at (11, 0.2) {\small $p_{S_1(y_4,y_3)}$};
\node at (11,-0.2) {\small $p_{S_2(y_2, y_3)}$};
\end{tikzpicture}%
\caption{THGP $P_\omq$ in Example~\ref{ex:thgp}: non-zero labels of vertices in $P_\omq$ are given on the edges of the tree.}\label{fig:thgp}
\end{figure}
To explain the meaning of $P_{\omq}$, suppose $\T,\A \models \q$, for some data instance $\A$. Then there is a homomorphism \mbox{$h \colon \q  \to \mathcal{C}_{\T,\A}$}. This homomorphism  defines the type of bag $N_1$, which can be either $\avec{w}_1$ (if $h(z) \in \ind(\A)$ for all $z \in \lambda(N_1)$)  or $\avec{w}_2$ (if \mbox{$h(y_1) = aR$} for some $a \in \ind(\A)$). These two cases are represented by the hyperedges $E^1_1 = [N_1, u^1_{12}]$ and $E^2_1 = [N_1, u^2_{12}]$. Since
$\{N_1, u^1_{12}\}$ is labelled with 0, exactly one of them must be chosen to construct an independent subset of hyperedges covering all zeros. In contrast to that, there is no hyperedge $E^2_2$ because $\avec{w}_2$ is not compatible with $N_2$, and so  \mbox{$E_2^1 = [u_{21}^1, N_2]$} must be present in every covering of all zeros. Both $(\avec{w}_1, \avec{w}_1)$ and $(\avec{w}_2, \avec{w}_1)$ are compatible with $(N_1, N_2)$, which gives $E^{11}_{12}= [v^1_{12}, v^1_{21}]$ and  $E^{21}_{12}=  [v^2_{12}, v^1_{21}]$. Thus, if $N_1$ is of type $\avec{w}_1$, then we include $E^1_1$ and $E^{11}_{12}$ in the covering of all zeros, and so $p_{R(y_4,y_1)}\land p_{R(y_2, y_1)}$ should hold. If $N_1$ is of type $\avec{w}_2$, then instead of $E^{11}_{12}$, we take $E^{21}_{12}$, and so
$p_{y_4 = y_1}\land p_{y_2=y_1}\land p_{R^*(y_1)} \land p_{R^*(y_2)} \land p_{R^*(y_4)}$ should be true. Finally, since  $\{v_{21}^1, u_{21}^1\}$ does not belong to any hyperedge, $p_{S_1(y_4,y_3)}\land p_{S_2(y_2, y_3)}$ must hold in either case.
\end{example}

%SUMMARY
%
%Prop.~\ref{hyper:program}
%\textup{(}i\textup{)} Any monotone HGP based on a hypergraph $H$ computes a subfunction of the hypergraph function $f_H$.
%
%Prop.~\ref{hyper:program}
%\textup{(}ii\textup{)} For any hypergraph $H$ of degree~\mbox{$\le d$}, there is a monotone HGP of size $O(|H|)$ that computes $f_H$ and such that its hypergraph is of degree~$\le \max(2,d)$.

\subsection{Summary}

In Tables~\ref{table:OMQ-to-HGP} and~\ref{table:HG-HGf}, we summarise the results of Section~\ref{sec:OMQs&hypergraphs} that will be used in Section~\ref{sec:7} to obtain lower and upper bounds for the size of OMQ rewritings. 
%
%Thm.~\ref{prop:tree-shaped}
%If an OMQ $\omq = (\T,\q)$ is tree-shaped, then $\HG{\omq}$ is isomorphic to a tree hypergraph. Furthermore, if $\q$ has at least one binary atom, then the number of leaves in the tree underlying $\HG{\omq}$ is the same as the number of leaves in $\q$.
%
Table~\ref{table:OMQ-to-HGP} shows how Theorems~\ref{depth1} and~\ref{prop:tree-shaped}  (on the shape of tree-witness hypergraphs) combined with Proposition~\ref{hyper:program}~(\emph{ii}),  as well as Theorem~\ref{DL2THP} provide us with hypergraph programs computing tree-witness hypergraph functions for OMQs.
\begin{table}[t]
\caption{HGPs computing tree-witness hypergraph functions for OMQs.}
{\centerline{\renewcommand{\arraystretch}{1.5}\begin{tabular}{cccc}\toprule
OMQ $\omq= (\Tmc,\q)$ & $P_\omq$ & of size   & computes \\\midrule
%
%any & $\mHGP$ & $2^{O(|\q|)}$ & $\twfn$\\
%
$\T$ of depth~1 & $\mHGP^2$ & $O(|\q|)$ & $\twfn$\\
tree-shaped $\q$ with $\ell$ leaves & $\mTHGP(\ell)$ & $|\q|^{O(\ell)}$ &  $\twfn$\\
{\renewcommand{\arraystretch}{1}\begin{tabular}{c}$\q$ of treewidth $t$,\\$\Omega_{\omq}$ a fundamental set\end{tabular}} &  $\mTHGP$ & $|\q|^{O(1)} \cdot |\Omega_{\omq}|^{O(t)}$  & $\homfn$\\\bottomrule
\end{tabular}}}%
\label{table:OMQ-to-HGP}
\end{table}
Table~\ref{table:HG-HGf} contains the representation results of Theorems~\ref{hg-to-query},~\ref{representable} and~\ref{tree-hg-to-query} that show how abstract hypergraphs can be embedded into tree-witness hypergraphs of OMQs.
\begin{table}[t]
\caption{Representation results for classes of hypergraphs.}{\centerline{\renewcommand{\arraystretch}{1.6}\begin{tabular}{ccc}\toprule
hypergraph $H$  & is isomorphic to & \rule[-12pt]{0pt}{30pt}\parbox{43mm}{\centering any mHGP based on $H$\\ computes a subfunction of} \\\midrule
any &  a subgraph of $\HG{\omq_H}$ &  $f^\vartriangle_{\omq_H}$ \\
of degree 2 &  $\HG{\OMQI{H}}$ &  $f^\vartriangle_{\OMQI{H}}$ \\
tree hypergraph & a subgraph of $\HG{\OMQT{H}}$ & $f^\vartriangle_{\OMQT{H}}$\\\bottomrule
\end{tabular}}}
\label{table:HG-HGf}
\end{table}
%

%***************

\section{Hypergraph Programs and Circuit Complexity}\label{sec:circuit_complexity}

In the previous section, we saw how different classes of OMQs gave rise to different classes of monotone HGPs. Here we characterise the computational power of HGPs in these classes by relating them to  standard models of computation for Boolean functions. Table~\ref{table:comp:classes} shows some of the obtained results. For example, its first row says that any Boolean function computable by a polynomial-size nondeterministic circuit can also be computed by a polynomial-size HGP of degree at most 3, and the other way round.
\begin{table}
\caption{Complexity classes, models of computation and corresponding classes of HGPs.}{{%
\renewcommand{\arraystretch}{1.6}\tabcolsep=10pt%
\begin{tabular}{lll}\toprule
{\renewcommand{\arraystretch}{0.8}\begin{tabular}{c}complexity\\class\end{tabular}} & \mbox{}\hfil model of computation & \mbox{}\hfil class of HGPs \\\midrule
$\NP/\poly$ & nondeterministic Boolean circuits & $\HGP = \HGP^d$, $d\geq 3$\\
$\P/\poly$ & Boolean circuits & ---\\
{\tabcolsep=0pt\renewcommand{\arraystretch}{1}\begin{tabular}{l}$\LOGCFL/\poly$\\($\SAC^1$)\end{tabular}} & {\tabcolsep=0pt\renewcommand{\arraystretch}{0.8}\begin{tabular}{l}logarithmic-depth circuits with\\\hspace*{1em}unbounded fan-in \AND-gates and\\\hspace*{1em}\NOT-gates only on inputs\end{tabular}} & $\THGP$ \\
$\NL/\poly$ & nondeterministic branching programs & $\HGP^2 = \THGP(\ell)$, $\ell \ge 2$\\
$\NC^1$ & Boolean formulas & $\THGP^d$, $d \ge 3$\\
$\ACz$ & {\tabcolsep=0pt\renewcommand{\arraystretch}{0.8}\begin{tabular}{l}constant-depth circuits with\\\hspace*{1em}unbounded fan-in \AND- and \OR-gates, and\\\hspace*{1em}\NOT-gates only on inputs\end{tabular}} & --- \\
$\Pitr$ & $\ACz$ circuits of depth 3 with output \AND-gate & $\THGP^2 = \THGP^2(2)$\\\bottomrule
\end{tabular}}%
}\label{table:comp:classes}
\end{table}

We remind the reader that the complexity classes in the table form the chain
\begin{equation} \label{eq:inclusions}
\Pitr ~\subsetneqq~ \ACz ~\subsetneqq~ \NC^1 ~\subseteq~ \NL/\poly ~\subseteq~ \LOGCFL/\poly ~\subseteq~ \P/\poly ~\subseteq~ \NP/\poly
\end{equation}
and that whether any of the non-strict inclusions is actually strict remains a major open problem in complexity theory; see, e.g.,~\cite{Arora&Barak09,Jukna12}. 
All these classes are non-uniform in the sense that they are defined in terms
of polynomial-size non-uniform sequences of Boolean circuits of certain shape and depth. The suffix `$/\poly$' comes from an alternative definition of $\mathsf{C}/\poly$ in terms of Turing machines for the class $\mathsf{C}$ with an additional advice input of polynomial size. 

When talking about complexity classes, instead of individual Boolean functions, we consider \emph{sequences} of functions $f = \{f_n\}_{n < \omega}$ with $f_n \colon \zo^n \to \zo$. 
The same concerns circuits, HGPs and the other models of computation we deal with.
For example, we say that a circuit $\Cir = \{\Cir_n\}_{n < \omega}$ \emph{computes} a function $f = \{f_n\}_{n < \omega}$ if  $\Cir_n$ computes $f_n$ for every $n < \omega$. (It will always be clear from context whether $f$, $\Cir$, etc.\ denote an individual function, circuit, etc.\ or a sequence thereof.)   
A circuit $\Cir$ is said to be \emph{polynomial} if there is a polynomial $p \colon \mathbb N \to \mathbb N$ such that $|\Cir_n| \le p(n)$, for every $n < \omega$.  The \emph{depth} of~$\Cir_n$ is the length of the longest directed path from an input to the output of~$\Cir_n$. 

The complexity class $\P/\poly$ can be defined as comprising those Boolean functions that are computed by polynomial circuits, and $\NC^1$ consists of functions  computed by polynomial formulas (that is, circuits  every logic gate in which has at most one output).
Alternatively, a Boolean function is in $\NC^1$ iff it can be computed by a polynomial-size circuit of logarithmic depth, whose \AND- and \OR-gates have two inputs.

$\LOGCFL/\poly$ (also known as $\SAC^1$) is the class of Boolean functions computable by polynomial-size and logarithmic-depth circuits in which \AND-gates have two inputs but \OR-gates can have arbitrarily many inputs (\emph{unbounded fan-in}) and \NOT-gates can only be applied to inputs of the circuit~\cite{circuit}. 
$\ACz$ is the class of functions computable by polynomial-size circuits of constant depth with \AND- and \OR-gates of unbounded fan-in and \NOT-gates only at the inputs; $\Pitr$ is the subclass of $\ACz$ that only allows  circuits of depth 3 (not counting the \NOT-gates) with an output \AND-gate. 

Finally, a Boolean function $f = \{f_n\}_{n < \omega}$ is in the class $\NP/\poly$ if  there is a polynomial $p$ and a polynomial circuit $\Cir = \{\Cir_{n+p(n)}\}_{n < \omega}$ such that, for any $n$ and $\avec{\alpha} \in \{0,1\}^n$,
\begin{equation}\label{eq:NP_poly}
f_n(\avec{\alpha}) = 1 \quad \text{ iff } \quad \text{there is } \avec{\beta} \in \zo^{p(n)} \text{ such that } \Cir_{n+p(n)}(\avec{\alpha},\avec{\beta}) = 1
\end{equation}
(the $\avec{\beta}$-inputs are sometimes called \emph{certificate inputs}).

By allowing only \emph{monotone} circuits or formulas in the definitions of the complexity classes, we obtain their monotone variants: for example, the monotone variant of  $\NP/\poly$ is denoted by $\mNP/\poly$ %. The monotone version $\mNP/\poly$ of $\NP/\poly$ is 
and defined  by restricting the use of \NOT-gates in the circuits to the certificate inputs only. We note in passing that the monotone variants of the classes in~\eqref{eq:inclusions} also form a chain~\cite{Razborov85,AlonB87,KarchmerW88}:
\begin{equation} \label{eq:inclusions_monotone}
\mPitr \subsetneqq \mACz ~\subsetneqq~ \mNC^1 ~\subsetneqq~ \mNL/\poly ~\subseteq~ \mLOGCFL/\poly ~\subsetneqq~ \mP/\poly ~\subsetneqq~ \mNP/\poly.
\end{equation}
Whether the inclusion $\mNL/\poly \subseteq \mLOGCFL/\poly$ is proper remains an open problem. 

We use these facts in the next section to show lower bounds on the size of OMQ rewritings.

%*******************

\subsection{NP/poly and HGP$^\textbf{3}$}

Our first result shows that $\NP/\poly$ and $\mNP/\poly$ coincide with the classes $\HGP^3$ and $\mHGP^3$ of Boolean functions computable by polynomial-size (sequences of) HGPs and monotone HGPs of degree at most~3, respectively.

\begin{theorem}\label{NBC}
$\NP/\poly = \HGP = \HGP^3$ and $\mNP/\poly = \mHGP = \mHGP^3$.
\end{theorem}
\begin{proof}
Suppose $P$ is a (monotone) HGP. We construct a non-deterministic 
circuit~$\Cir$ of size polynomial in $|P|$, whose input variables are the same as the variables in $P$, certificate inputs correspond to the hyperedges of $P$, and such that $\Cir(\avec{\alpha},\avec{\beta}) = 1$ iff $\{e_i \mid \avec{\beta}(e_i) = 1\}$ is an independent set of hyperedges covering all zeros under $\avec{\alpha}$. It will then  follow that
\begin{equation}\label{eq:prop:NBC}
P(\avec{\alpha}) = 1 \quad \text{ iff } \quad \text{there is } \avec{\beta} \text{ such that } \Cir(\avec{\alpha}, \avec{\beta}) = 1.
\end{equation}
First, for each pair of intersecting hyperedges $e_i, e_j$ in $P$, we take the disjunction \mbox{$\neg e_i \vee \neg e_j$}, and, for each vertex in $P$  labelled with a literal $\li$ (that is, $p$ or $\neg p$) and the hyperedges $e_{i_1}, \dots, e_{i_k}$ incident to it, we take the disjunction $\li \vee  e_{i_1} \lor \dots \lor e_{i_k}$. The circuit $\Cir$ is then a conjunction of all such disjunctions. Note that if $P$ is monotone, then~$\neg$ is only applied to the certificate inputs, $\avec{e}$, in $\Cir$.

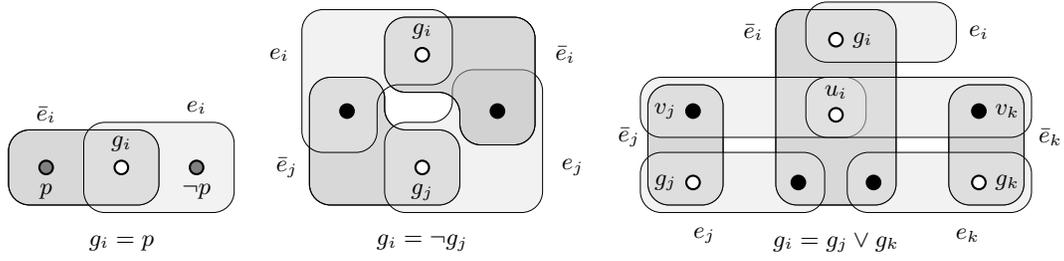
\begin{figure}[t]%
\centering%
\begin{tikzpicture}\footnotesize
\draw[ultra thin,fill=gray!60,fill opacity=0.5,rounded corners=8] (-1.5,-0.5) rectangle +(2,1);
\draw[ultra thin,fill=gray!20,fill opacity=0.5,rounded corners=8] (-0.5,-0.6) rectangle +(2,1.2);
\draw[ultra thin,rounded corners=8] (-1.5,-0.5) rectangle +(2,1);
\begin{scope}\small
\node at (-1,0.7) {$\bar{e}_i$};
\node at (1,0.8) {$e_i$};
\end{scope}
\node[point,label=above:{$g_i$}] (gxi) at (0,0) {}; % ,label=below:{\scriptsize 0}
\node[point,fill=gray,label=below:{$p$}] (xi) at (-1,0) {};
\node[point,fill=gray,label=below:{$\neg p$}] (xi) at (1,0) {};
\node at (0,-1) {\small $g_i = p$};
\begin{scope}[xshift=40mm]
\draw[ultra thin,fill=gray!60,fill opacity=0.5,rounded corners=8] (-1.5,-0.5) -- ++(0,1.7) -- ++(1,0) -- ++(0,-0.7) -- ++(1,0) -- ++ (0,-1) -- cycle;
\draw[ultra thin,fill=gray!20,fill opacity=0.5,rounded corners=8] (-0.5,-0.6) -- ++(2.1,0)  -- ++ (0,1.9)-- ++ (-1.2,0) -- ++(0,-0.7) -- ++(-0.9,0)  --cycle;
\draw[ultra thin,fill=gray!60,fill opacity=0.5,rounded corners=8] (-0.5,1)  -- ++(0,1)  -- ++(2,0) -- ++(0,-1.7) -- ++(-1,0) -- ++(0,0.7)  -- cycle;
\draw[ultra thin,fill=gray!20,fill opacity=0.5,rounded corners=8] (-1.6,0.2) -- ++(0,1.9) -- ++(2,0) -- ++(0,-1) -- ++ (-1,0) -- ++(0,-0.9) -- cycle;
\draw[ultra thin,rounded corners=8] (-1.5,-0.5) -- ++(0,1.7) -- ++(1,0) -- ++(0,-0.7) -- ++(1,0) -- ++ (0,-1) -- cycle;
\draw[ultra thin,rounded corners=8] (-0.5,1)  -- ++(0,1)  -- ++(2,0) -- ++(0,-1.7) -- ++(-1,0) -- ++(0,0.7)  -- cycle;
\begin{scope}\small
\node at (-1.8,0) {$\bar{e}_j$};
\node at (2,0) {$e_j$};
\node at (-1.9,1.5) {$e_i$};
\node at (1.9,1.5) {$\bar{e}_i$};
\end{scope}
\node[point,label=above:{$g_i$}] (gxi) at (0,1.5) {}; % ,label=below:{\scriptsize 0}
\node[point,label=below:{$g_j$}] (gxj) at (0,0) {}; % ,label=below:{\scriptsize 0}
\node[point,fill=black] (xi) at (-1,0.75) {}; % ,label=below:{\scriptsize 1}
\node[point,fill=black] (xi) at (1,0.75) {}; % ,label=below:{\scriptsize 1}
\node at (0,-1) {\small $g_i = \NOTOP g_j$};
\end{scope}
\begin{scope}[xshift=95mm]
\draw[ultra thin,fill=gray!60,fill opacity=0.5,rounded corners=8] (-2.5,-0.5) rectangle +(1,1.6);
\draw[ultra thin,fill=gray!60,fill opacity=0.5,rounded corners=8] (1.5,-0.5) rectangle +(1,1.6);
\draw[ultra thin,fill=gray!60,fill opacity=0.5,rounded corners=8] (-0.8,-0.5) rectangle +(1.6,2.6);
\draw[ultra thin,fill=gray!20,fill opacity=0.5,rounded corners=8] (-2.6,-0.6) rectangle +(2.45,0.8);
\draw[ultra thin,fill=gray!20,fill opacity=0.5,rounded corners=8] (0.15,-0.6) rectangle +(2.45,0.8);
\draw[ultra thin,fill=gray!20,fill opacity=0.5,rounded corners=8] (-2.6,0.4) rectangle +(3,0.8);
\draw[ultra thin,fill=gray!20,fill opacity=0.5,rounded corners=8] (-0.4,0.4) rectangle +(3,0.8);
\draw[ultra thin,fill=gray!20,fill opacity=0.5,rounded corners=8] (-0.4,1.4) rectangle +(2,0.8);
\draw[ultra thin,rounded corners=8] (-2.5,-0.5) rectangle +(1,1.6);
\draw[ultra thin,rounded corners=8] (1.5,-0.5) rectangle +(1,1.6);
\draw[ultra thin,rounded corners=8] (-0.8,-0.5) rectangle +(1.6,2.6);
\node[point,label=right:{$g_i$}] (gxi) at (0,1.7) {}; % ,label=below:{\scriptsize 0}
\node[point,label=above:{$u_i$}] (ui) at (0,0.7) {}; % ,label=below:{\scriptsize 0}
\node[point,label=right:{$g_k$}] (gxjp) at (1.9,-0.2) {}; % ,label=below:{\scriptsize 0}
\node[point,label=left:{$g_j$}] (gxj) at (-1.9,-0.2) {}; % ,label=below:{\scriptsize 0}
\node[point,fill=black] (xi) at (-0.5,-0.2) {}; % ,label=below:{\scriptsize 1}
\node[point,fill=black] (xip) at (0.5,-0.2) {}; % ,label=below:{\scriptsize 1}
\node[point,fill=black,label=left:{$v_j$}] (yi) at (-1.9,0.75) {}; % ,label=below:{\scriptsize 1}
\node[point,fill=black,label=right:{$v_k$}] (yio) at (1.9,0.75) {}; % ,label=below:{\scriptsize 1}
\begin{scope}\small
\node at (-1.75,-0.9) {$e_j$};
\node at (1.75,-0.9) {$e_k$};
\node at (-2.75,0.4) {$\bar{e}_j$};
\node at (2.85,0.4) {$\bar{e}_k$};
\node at (1.9,1.8) {$e_i$};
\node at (-1.1,1.8) {$\bar{e}_i$};
\end{scope}
\node at (0,-1) {\small $g_i = g_j \OROP g_k$};
\end{scope}
\end{tikzpicture}
\caption{HGP in the proof of Theorem~\ref{NBC}: black vertices are labelled with 1 and white vertices with 0.}\label{fig:NBC}
\end{figure}

\smallskip

Conversely, let $\Cir$ be a circuit with certificate inputs. We construct an HGP $P$ of degree at most 3 satisfying~\eqref{eq:prop:NBC} as follows.
For each gate $g_i$ in $\Cir$, the HGP contains a vertex $g_i$ labelled with $0$ and a pair of hyperedges $\bar{e}_i$ and $e_i$, both containing $g_i$. No other hyperedge contains $g_i$, and so either $\bar{e}_i$ or $e_i$ should be present in any cover of zeros.  To ensure this property, for each gate $g_i$, we add the following vertices and hyperedges to $P$ (see Fig.~\ref{fig:NBC}):
\begin{nitemize}
\item[--] if $g_i$ is an input $p$, then we add a vertex labelled with $\neg p$ to $e_i$  and a vertex labelled with $p$ to $\bar{e}_i$;

\item[--] if $g_i$ is a certificate input, then no additional vertices and hyperedges are added;

\item[--] if $g_i = \NOTOP g_j$, then we add a vertex labelled with $1$ to hyperedges $e_i$ and $\bar{e}_j$, and a vertex labelled with $1$ to hyperedges $\bar{e}_i$ and $e_j$;

\item[--] if $g_i = g_j \OROP g_k$, then we add a vertex labelled with $1$ to hyperedges $e_j$ and $\bar{e}_i$,
add a vertex labelled with $1$ to $e_k$ and $\bar{e}_i$;
then, we add vertices $v_j$ and $v_k$ labelled with~$1$ to $\bar{e}_j$ and $\bar{e}_k$, respectively, and a vertex $u_{i}$ labelled with $0$ to $\bar{e}_i$; finally, we add hyperedges $\{v_j, u_i\}$ and $\{v_k, u_i\}$ to $P$;

\item[--] if $g_i = g_j \ANDOP g_k$, then we add the pattern dual to the case of $g_i = g_j \OROP g_k$:
we add a vertex labelled with $1$ to  $\bar{e}_j$ and $e_i$,  a vertex labelled with $1$ to  $\bar{e}_k$ and $e_i$;
then, we add vertices $v_j$ and $v_k$ labelled with $1$ to $e_j$ and $e_k$, respectively, and a vertex $u_{i}$ labelled with $0$ to $e_i$; finally, we add hyperedges $\{v_j, u_i\}$ and $\{v_k, u_i\}$ to $P$.
\end{nitemize}
Finally, we add one more vertex labelled with $0$ to $e_m$ for the output gate $g_m$ of $\Cir$, which ensures that $e_m$ must be included the cover. 
It is easily verified that the constructed HGP is of degree at most 3. 
One can establish~\eqref{eq:prop:NBC} by induction on the structure of $\Cir$.
We illustrate the proof of the inductive step for the case of \mbox{$g_i = g_j \OROP g_k$}: we show that $e_i$ is in the cover iff it contains either $e_j$ or $e_k$.
Suppose the cover contains~$e_j$. Then it cannot contain $\bar{e}_i$, and so it contains $e_i$. The vertex $u_i$ in this case can be covered by $\{v_j, u_i\}$ since $\bar{e}_j$ is not in the cover.
Conversely, if neither $e_j$ nor $e_k$ is in the cover, then it must contain both $\bar{e}_j$ and
$\bar{e}_k$, and so neither $\{v_j, u_i\}$ nor $\{v_k, u_i\}$ can belong to the cover,
and thus we will have to include $\bar{e}_i$ to the cover.

If $\Cir$ is monotone, then we remove from $P$ all vertices labelled with $\neg p$, for an input~$p$, and denote the resulting program by $P'$. We claim that, for any $\avec{\alpha}$, we have $P'(\avec{\alpha})=1$ iff there is $\avec{\beta}$ such that $\Cir(\avec{\alpha}, \avec{\beta})=1$. The implication $(\Leftarrow)$   is trivial: if \mbox{$\Cir(\avec{\alpha}, \avec{\beta})=1$} then, by the argument above, $P(\avec{\alpha})=1$ and, clearly,  $P'(\avec{\alpha}) = 1$. Conversely, suppose $P'(\avec{\alpha})=1$. Each of the vertices $g_i$ in $P'$ corresponding to the inputs is covered by one of the hyperdges $e_i$ or $\bar{e}_i$. Let $\avec{\alpha}'$ be the vector corresponding to these hyperedges; clearly, $\avec{\alpha}' \leq \avec{\alpha}$. This cover of vertices of $P'$ gives us $P(\avec{\alpha}')=1$.
Thus, by the argument above, there is $\avec{\beta}$ such that $\Cir(\avec{\alpha}',\avec{\beta})=1$. Since $\Cir$ is monotone,  $\Cir(\avec{\alpha},\avec{\beta})=1$.
\end{proof}

%*****************

\subsection{NL/poly and HGP$^\textbf{2}$}

A Boolean function belongs to the class $\NL/\poly$ iff it can be computed by a polynomial-size \emph{nondeterministic branching program} (NBP). We remind the reader (consult~\cite{Jukna12} for more details)  that an NBP $B$ is a directed graph $G=(V,E)$,  whose arcs are labelled with the Boolean constants 0 and 1, propositional variables $p_1, \dots, p_n$ or their negations, and which distinguishes two vertices $s,t\in V$. Given an assignment $\avec{\alpha}$ to variables $p_1,\dots,p_n$, we write $s \to_{\avec{\alpha}}t$ if there is a path in $G$ from~$s$ to~$t$ all of whose labels evaluate to $1$ under $\avec{\alpha}$. We say that an NBP $B$ \emph{computes} a Boolean function $f$ if $f(\avec{\alpha}) = 1$ iff $s \to_{\avec{\alpha}}t$, for any $\avec{\alpha} \in \zo^n$.
The \emph{size} $|B|$ of $B$ is the size of the underlying graph, $|V| + |E|$.
An NBP is \emph{monotone} if there are no negated variables among its labels. The class of Boolean functions computable by polynomial-size monotone NBPs is denoted by $\mNL/\poly$; the class of functions $f$ whose \emph{duals} $f^*(p_1,\dots,p_n) = \neg f(\neg p_1,\dots,\neg p_n)$ are in $\mNL/\poly$ is denoted by $\comNL/\poly$.

%Just as for the functions and circuits, from now on by a branching program we mean a sequence of branching programs $P_n$ with $n$ variables for all $n \in \bb{N}$.
%
%A branching program $P$ is a \emph{polynomial size branching program} if there is a polynomial $p \in \bb{Z}[x]$ such that the size of $P_n$ is at most $p(n)$.
%
%We denote by $\NBPc$ the minimal size of a nondeterministic branching program $P$ for $f$.
%A Boolean function $f$ lies in the class $\NBPc$ iff there is a polynomial size branching program $P$ computing $f$.
%It is known that $\NBPc$ coincides with nonuniform analog of nondeterministic logarithmic space $\NL$, that is $\NBPc= \NL/\poly$~\cite{Jukna12,Razborov91}.
%
%For every complexity class $\K$ introduced above, we denote by $\mK$ its monotone counterpart.
%For a monotone Boolean function $f$ and its complexity measures $\NBC(f)$ and $\NBP(f)$ we denote by $\mNBC(f)$ and $\mNBP(f)$
%their monotone counterpart.
%
%
%Similar inclusions hold for monotone case:
%\begin{equation} \label{eq:inclusions_monotone}
%\mNC^1 \subseteq \mNBPc \subseteq \mSAC^1 \subseteq \mP/\poly \subseteq \mNP/\poly.
%\end{equation}

%It is also known that $\mP/\poly \neq \mNP/\poly$~\cite{Razborov85,AlonB87} and $\mNBP \neq \mNC^1$~\cite{Karchmer88}.
%We will use these facts to prove lower bounds on the rewriting size.

\begin{theorem}\label{thm:deg_2}
$\NL/\poly = \HGP^2$ and $\comNL/\poly = \mHGP^2$.
\end{theorem}
\begin{proof}
As follows from~\cite{szelepcsenyi88,immerman88}, if a function $f$ is computable by a polynomial-size NBP, then $\neg f$ is also computable by a polynomial-size NBP. So  suppose $\neg f$ is computed by an NBP $B$. We construct an HGP $P$ computing~$f$ of degree at most~2 and polynomial size in $|B|$ as follows (see Fig.~\ref{fig:deg_2}). 
\begin{figure}[t]%
\centering%
\begin{tikzpicture}\footnotesize
\node[point,label=left:{$v_0$}] (v0) at (-1.5,0) {};
\node[point,label=above:{$v_1$}] (v1) at (0,0) {};
\node[point,label=above:{$v_2$}] (v2) at (1.5,0.5) {};
\node[point,label=below:{$v_3$}] (v3) at (1.5,-0.5) {};
\draw[->,thick] (v0) -- (v1) node[below,midway,sloped] {\small${\scriptstyle e_0\colon} q$};
\draw[->,thick] (v1) -- (v2) node[above,midway,sloped] {\small${\scriptstyle e_1\colon} \neg q$};
\draw[->,thick] (v1) -- (v3) node[below,midway,sloped] {\small${\scriptstyle e_2\colon} p$};
\begin{scope}[xshift=30mm]
\draw[ultra thin,fill=gray!60,fill opacity=0.5,rounded corners=8] (2.1,-1) -- ++(0,2) -- ++(2.4,0) -- ++(0,-2) -- cycle;
\draw[ultra thin,rounded corners=8] (2.1,-1) -- ++(0,2) -- ++(2.4,0) -- ++(0,-2) -- cycle;
\draw[ultra thin,fill=gray!20,fill opacity=0.5,rounded corners=8] (0.3,-0.4) -- ++(0,0.8) -- ++(2.9,0) -- ++(0,-0.8) -- cycle;
\draw[ultra thin,rounded corners=8] (0.3,-0.4) -- ++(0,0.8) -- ++(2.9,0) -- ++(0,-0.8) -- cycle;
\draw[ultra thin,fill=gray!20,fill opacity=0.5,rounded corners=8] (3.3,0) -- ++(0,0.7) -- ++(2.9,0.25) -- ++(0,-0.7) -- cycle;
\draw[ultra thin,rounded corners=8] (3.3,0) -- ++(0,0.7) -- ++(2.9,0.25) -- ++(0,-0.7) -- cycle;
\draw[ultra thin,fill=gray!20,fill opacity=0.5,rounded corners=8] (3.3,0) -- ++(0,-0.7) -- ++(2.9,-0.25) -- ++(0,0.7) -- cycle;
\draw[ultra thin,rounded corners=8] (3.3,0) -- ++(0,-0.7) -- ++(2.9,-0.25) -- ++(0,0.7) -- cycle;
\draw[ultra thin,rounded corners=8] (2.1,-1) -- ++(0,2) -- ++(2.4,0) -- ++(0,-2) -- cycle;
\node at (3.3,1.3) {$v_1$-hyperedge};
\node at (1,0.7) {$e_0$-hyperedge};
\node[rotate=5] at (5.5,1.2) {$e_1$-hyperedge};
\node[rotate=-5] at (5.5,-1.2) {$e_2$-hyperedge};
\node[point,fill=gray,label=left:{\scriptsize$e_0^0$},label=below:{\small $\neg q$}] (e00) at (1,0) {};
\node[point,fill=gray,label=left:{\scriptsize$e_1^0$},label=above:{\small $q$}] (e10) at (4,0.25) {};
\node[point,fill=gray,label=left:{\scriptsize$e_2^0$},label=below:{\small $\neg p$}] (e20) at (4,-0.25) {};
\node[point,fill=black,label=right:{\scriptsize$e_0^1$}] (e01) at (2.5,0) {};
\node[point,fill=black,label=right:{\scriptsize$e_1^1$}] (e11) at (5.5,0.5) {};
\node[point,fill=black,label=right:{\scriptsize$e_2^1$}] (e21) at (5.5,-0.5) {};
\end{scope}
\end{tikzpicture}
\caption{HGP in the proof of Theorem~\ref{thm:deg_2}: black vertices are labelled with 1.}\label{fig:deg_2}
\end{figure}
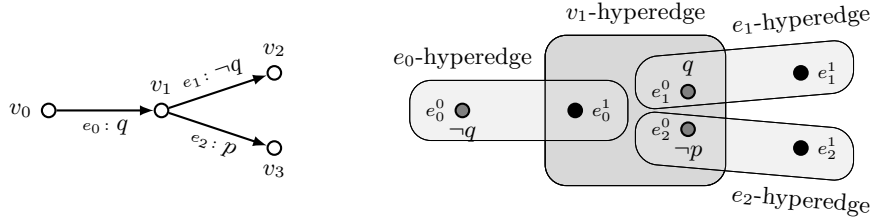
For each arc $e$ in $B$, the HGP $P$ has two vertices $e^0$ and $e^1$, which represent the beginning and the end of $e$, respectively. The vertex $e^0$ is labelled with the \emph{negated} label of $e$ in $B$ and $e^1$  with $1$. For each arc $e$ in $B$, the HGP $P$ has an $e$-hyperedge $\{e^0, e^1\}$.
For each vertex $v$ in $B$ but $s$ and $t$, the HGP $P$ has a $v$-hyperedge comprising all vertices $e^1$ for the arcs $e$ leading to $v$, and all vertices $e^0$ for the arcs $e$ leaving $v$. We also add to the HGP $P$ a vertex $w$ labelled with~$0$  and
a hyperedge, $\bar e_w$, that consists of $w$ and all vertices $e^1$ for the arcs $e$ in $B$ leading to $t$. 
We claim that the constructed HGP $P$ computes $f$.
Indeed, if $s \not\to_{\avec{\alpha}} t$ then the following subset of hyperedges is independent and covers all zeros: all $e$-hyperedges, for the arcs~$e$ reachable from $s$ and labelled with 1 under $\avec{\alpha}$, and all $v$-hyperedges with $s\not\to_{\avec{\alpha}} v$ (including $\bar e_w$).
Conversely, if $s\to_{\avec{\alpha}} t$ then one can show by induction that, for each arc $e$ of the path, the $e$-hyperedge must be in the cover of all zeros. Thus, no independent set can cover $w$, which is labelled with 0.

Conversely, suppose $f$ is computed by an HGP $P$ of degree 2 with hyperedges $e_1, \dots, e_k$. We first provide a graph-theoretic characterisation of independent sets covering all zeros based on the implication graph~\cite{AspvallPlassTarjan79}. 
With every hyperedge $e_i$ we associate a propositional variable $u_i$ and with every assignment $\avec{\alpha}$ we associate the following set $\Phi_{\avec{\alpha}}$ of  propositional binary clauses:
\begin{align*}
& \neg u_i  \lor \neg u_j, && \text{if } \ e_i\cap e_j \ne \emptyset,\\
& u_i \lor u_j, && \text{if there is } \ v\in e_i\cap e_j \text{ with } \avec{\alpha}(v) = 0.
\end{align*}
Informally, the former means that  intersecting hyperedges cannot be chosen at the same time and the latter that all zeros must be covered; note that all vertices have at most two incident edges.
By definition, $X$ is an independent set covering all zeros iff  $X = \{ e_i \mid  \avec{\gamma}(u_i) = 1\}$, for some assignment $\avec{\gamma}$ satisfying $\Phi_{\avec{\alpha}}$. Let $C_{\avec{\alpha}} = (V, E_{\avec{\alpha}})$ be the implication graph of $\Phi_{\avec{\alpha}}$, that is, a directed graph with
\begin{align*}
V & ~=~ \bigl\{ u_i, \bar{u}_i \mid 1\leq i \leq k\bigr\},\\
E_{\avec{\alpha}} & ~=~ \bigl\{ (u_i, \bar{u}_j) \mid e_i \cap e_j \ne \emptyset \bigr\} \ \cup
\bigl\{ (\bar{u}_i,u_j) \mid \text{there is } v\in e_i\cap e_j \text{ with } \avec{\alpha}(v) = 0 \bigr\}.
\end{align*}
($V$ is the set of all `literals' for the variables of $\Phi_{\avec{\alpha}}$ and $E_{\avec{\alpha}}$ is the arcs for the implicational form of the clauses of $\Phi_{\avec{\alpha}}$.) Note that $\neg u_i \lor \neg u_j$ gives rise to two implications, $u_i \to \neg u_j$ and $u_j \to \neg u_i$, and so to two arcs in the graph; similarly, for $u_i \lor u_j$.
%
%By~Lemma~8.3.1 in~\cite{Borgeretal97}, 
By~\cite[Theorem~1]{AspvallPlassTarjan79},
$\Phi_{\avec{\alpha}}$ is satisfiable iff there is no $u_i$ with a (directed) cycle going through $u_i$ and $\bar{u}_i$.
It will be convenient for us to regard the $C_{\avec{\alpha}}$, for assignments~$\avec{\alpha}$, as a single labelled directed graph $C$ with arcs of the form $(u_i, \bar{u}_j)$ labelled with $1$ and arcs of the form $(\bar{u}_i, u_j)$ labelled with the negation of the literal labelling the uniquely defined $v\in e_i\cap e_j$ (recall that the hypergraph of $P$ is of degree~2). It should be clear that $C_{\avec{\alpha}}$ has a cycle going through $u_i$ and $\bar{u}_i$ iff we have both $\bar{u}_i \to_{\avec{\alpha}} u_i$ and $u_i \to_{\avec{\alpha}} \bar{u}_i$ in $C$.
The required NBP $B$ contains distinguished vertices $s$ and $t$, and, for each hyperedge $e_i$ in $P$, two copies, $C_i^0$ and $C_i^1$, of $C$ with additional arcs from $s$ to the $\bar{u}_i$ vertex of $C_i^0$, from the $u_i$ vertex of $C_i^0$ to the $u_i$ vertex of $C_i^1$, and from the $\bar{u}_i$ vertex of $C_i^1$ to $t$; see Fig.~\ref{fig:deg_2:NBP}. By construction, $s\to_{\avec{\alpha}} t$ iff there is a hyperedge $e_i$ in $P$ such that $C_{\avec{\alpha}}$ contains a cycle going through $u_i$ and $\bar{u}_i$.
We have thus constructed a polynomial-size NBP $B$ computing $\neg f$, and thus $f$ must also be computable by a polynomial-size NBP.

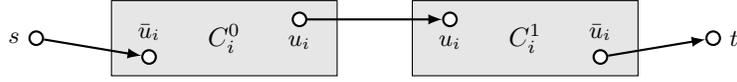
\begin{figure}[t]%
\centering%
\begin{tikzpicture}\footnotesize
\node[point,label=left:{$s$}] at (0,0) (s) {}; 
\filldraw[fill=gray!20] (1,-0.5) rectangle +(3,1);
\filldraw[fill=gray!20] (5,-0.5) rectangle +(3,1);
\node[point,label=right:{$t$}] at (9,0) (t) {}; 
\node[point,label=above:{$\bar{u}_i$}] at (1.5,-0.25) (bui) {}; 
\node at (2.5,0) {\normalsize $C_i^0$};
\node[point,label=below:{$u_i$}] at (3.5,0.25) (ui) {}; 
\draw[thick,->] (s) -- (bui);
\node[point,label=below:{$u_i$}] at (5.5,0.25) (bui2) {}; 
\node at (6.5,0) {\normalsize $C_i^1$};
\node[point,label=above:{$\bar{u}_i$}] at (7.5,-0.25) (ui2) {}; 
\draw[thick,->] (ui2) -- (t);
\draw[thick,->] (ui) -- (bui2);
\end{tikzpicture}
\caption{The NBP in the proof of Theorem~\ref{thm:deg_2}.}\label{fig:deg_2:NBP}
\end{figure}

As to $\comNL/\poly = \mHGP^2$,  observe that the first construction, if applied to a monotone NBP for $f^*$, produces a polynomial-size HGP of degree~2 computing $\neg f^*$, all of whose labels are negative. By removing negations from labels, we obtain a monotone HGP computing~$f$.
The second construction allows us to transform a monotone HGP of degree~2 for $f$ into an NBP with only negative literals that computes $\neg f$. By changing the polarity of the literals in the labels, we obtain a monotone NBP computing $f^*$. 
\end{proof}

\subsection{NL/poly and THGP($\ell$)}\label{sec:NLpoly-THGP}

For any natural $\ell \ge 2$, we denote by $\THGP(\ell)$ and $\mTHGP(\ell)$ the classes of Boolean functions computable by (sequences of) polynomial-size THGPs and, respectively, monotone THGPs whose underlying trees have at most $\ell$ leaves.

\begin{theorem}\label{thm:linear_hgp}
$\NL/\poly = \THGP(\ell)$ and $\mNL/\poly = \mTHGP(\ell)$, for any $\ell \ge 2$.
\end{theorem}
\begin{proof}
Suppose a polynomial-size THGP $P$ computes a Boolean function $f$. 
For simplicity,  we consider only $\ell =2$ here and prove the general case in Appendix~\ref{app:circuit_complexity}. Thus, we can assume that the vertices $v_1, \dots, v_n$ of $P$ are consecutive edges of the path graph underlying $P$, and therefore, every hyperedge in $P$ is of the form \mbox{$[v_i,v_{i+m}]=\{v_i,\dots,v_{i+m}\}$}, for some $m \ge 0$. We add to $P$ two extra vertices, $v_0$ and $v_{n+1}$ (thereby extending the underlying 2-leaf tree to $v_0,v_1, \dots, v_n,v_{n+1}$) and label them with 0; we also add two hyperedges $s = \{v_0\}$ and $t = \{v_{n+1}\}$ to $P$. Clearly, the resulting THGP $P'$ computes the same $f$.
To construct a polynomial-size NBP $B$ computing $f$, we take a directed graph whose vertices are hyperedges  of $P'$ and which contains an arc from  $e_i = [v_{i_1}, v_{i_2}]$ to $e_j = [v_{j_1}, v_{j_2}]$ iff 
$i_2 < j_1$; we label this arc with $\bigwedge_{i_2 <  k < j_1} \li_k$, where $\li_k$ is the label of $v_k$ in HGP $P$.
It is not hard to see that a path from $s$ to $t$ evaluated to~1 under given assignment~$\avec{\alpha}$ corresponds to a cover of zeros in $P'$ under $\avec{\alpha}$. Finally, to get rid of conjunctive labels on edges, we replace every arc with a label $\li_{i_1}\land \dots\land \li_{i_k}$ by a sequence of $k$ arcs consequently labelled with $\li_{i_1}, \dots, \li_{i_k}$.

Conversely, suppose a Boolean function $f$ is computed by an NBP $B$ based on a directed graph with vertices $V = \{v_1, \dots, v_n\}$, edges $E = \{e_1, \dots, e_m\}$, $s=v_1$ and $t=v_n$. Without loss of generality, we assume that $e_m$ is a loop from $t$ to $t$ labelled with~1. Thus, if there is a path from $s$ to $t$ whose labels evaluate to 1, then there is such a path of length $n-1$. 
We now construct a polynomial-size THGP computing $f$ whose underlying tree $T$ has two leaves. The vertices of the tree $T$ are arranged into $n$ vertex blocks and $n-1$ edge blocks, which alternate. The $k$th vertex (edge) block contains two copies $v_i^k, \bar v_i^k$ (respectively, $e_i^k, \bar e_i^k$) of every  $v_i \in V$
(respectively, $e_i \in E$):
\begin{multline*}
\underbrace{\textcolor{gray!50}{v_1^1}, \bar{v}_1^1, v_2^1,\bar{v}_2^1,\dots,v_n^1,\bar{v}_n^1,}_{\text{1st vertex block}} \ \ 
\underbrace{e_1^1, \bar{e}_1^1, e_2^1,\bar{e}_2^1,\dots,e_m^1,\bar{e}_m^1,}_{\text{1st edge block}} \ \ 
\underbrace{v_1^2, \bar{v}_1^2, v_2^2,\bar{v}_2^2,\dots,v_n^2,\bar{v}_n^2,}_{\text{2nd vertex block}} \ \ \dots\\
\underbrace{e_1^{n-1}, \bar{e}_1^{n-1}, e_2^{n-1},\bar{e}_2^{n-1},\dots,e_m^{n-1},\bar{e}_m^{n-1},}_{\text{$(n-1)$th edge block}} \ \ 
\underbrace{v_1^n, \bar{v}_1^n, v_2^n,\bar{v}_2^n,\dots,v_n^n,\textcolor{gray!50}{\bar{v}_n^n}}_{\text{$n$th vertex block}}.
\end{multline*}
We remove the first, $v^1_1$, and last vertex, $\bar v_n^n$ (shown in grey in the formula above), and connect the adjacent vertices by edges to construct the undirected tree $T$. 
Consider now a hypergraph~$H$ whose vertices are the edges of $T$ and hyperedges are of the form $h^k_i = [\bar v^k_j, e^k_i]$ and $g^k_i = [{\bar e}^k_i, v^{k+1}_{j'}]$, for $e_i =(v_j,v_{j'}) \in E$ and $1 \leq k < n$. The vertices of $H$ of the form $\{e^k_i,\bar e^k_i\}$, which separate hyperedges $h^k_i$ and $g^k_i$, are labelled with the label of $e_i$ in the given NBP $B$, and all other vertices of $H$ with $0$. 
%
%
%
%Conversely, consider an NBP $P = (V_P, E_P, v_1, v_n, \l_P)$ \nb{$\l$ should be different}, where $V_P= \{v_1, \ldots, v_n\}$
%and $E_P= \{e_1, \ldots, e_m\}$. We may assume w.l.o.g.\ that $e_m=(v_n,v_n)$ and $\l_P(e_m)=1$. 
%This assumption ensures that if there is a path from $v_1$ to $v_n$ whose labels evaluate to $1$, then there is 
%a (possibly non-simple) path with the same properties 
% whose length is exactly $n-1$. 
%
%We now construct an interval \HP\ $(H, \l_H)$ 
%that computes the function $f_P$. In Figure~\ref{fig:interval-hypergraph}, we display the graph 
%$G=(V_G,E_G)$ that underlies the interval hypergraph $H$. %the construction. 
%Its vertices are arranged into $n$ \emph{vertex blocks} and $n-1$ \emph{edge blocks} which alternate. 
%The $\ell$th vertex block (resp.\ edge block) contains two copies, $v_i^\ell, \bar v_i^\ell$ (resp.\ $e_i^\ell, \bar e_i^\ell$), 
%of every vertex $v_i \in V_P$
%(resp.\ edge $e_i \in E_P$). 
%We remove the first and last vertices $v^1_1$ and $\bar v_n^n$ and 
%connect the remaining vertices as shown in Figure~\ref{fig:interval-hypergraph}. 
%The hypergraph $H=(V_H,E_H)$ is defined by setting $V_H=E_G$ and
%letting $E_H$ be the set of all hyperedges 
%$\zeta_{i,\ell} = \langle \bar v^\ell_j, e^\ell_i \rangle$ and $\zeta'_{i,\ell} = \langle  {\bar e}^\ell_i, v^{\ell+1}_k \rangle$
%where $e_i =(v_j,v_k) \in E_P$ and $1 \leq \ell < n$. 
%The function $\l_H$ labels $\{e^\ell_i,\bar e^\ell_i\}$ with $\l_P(e_i)$ and all other vertices of $H$ (i.e.\ edges of $G$)
% with $0$. 
%
% 
%
We show now that the constructed THGP $P$ computes $f$. Indeed, if $f(\avec{\alpha})=1$,
then there is a path $e_{i_1}, \dots, e_{i_{n-1}}$ from $v_1$ to $v_n$ whose labels evaluate to 1 under $\avec{\alpha}$. It follows that $\{h^k_{i_k}, g^k_{i_k} \mid 1 \leq k< n\}$ is an independent set in $H$ covering all zeros. 
Conversely, if $E'$ is an independent set in~$H$ and covers all zeros under $\avec{\alpha}$, then it must contain exactly one pair of hyperedges $h^k_{i_k}$ and $g^k_{i_k}$ for every $k$ with $1 \leq k <n$, and the corresponding sequence of edges $e_{i_1}, \dots, e_{i_{n-1}}$ defines a path from $v_1$ to $v_n$. 
Moreover, since $E'$ does not cover vertices $\{e^k_{i_k},\bar e^k_{i_k}\}$, for $1 \leq k <n$, their labels (that is, the labels of the  $e_{i_k}$ in $B$)  evaluate to 1 under $\avec{\alpha}$. 
\end{proof}

To prove Theorem~\ref{bbcq-ndl} below, we shall require a somewhat different variant of Theorem~\ref{thm:linear_hgp}. The proof of the following result is given in Appendix~\ref{app:circuit_complexity}:

\begin{theorem}\label{thm:linear_hgp1}
Fix $\ell \ge 2$. For any tree hypergraph $H$ based on a  tree with at most $\ell$ leaves, the function $f_H$ can be computed by an NBP of size polynomial in $|H|$.
\end{theorem}

Note that Theorem~\ref{thm:linear_hgp1} does not immediately follow  from Theorem~\ref{thm:linear_hgp} and Proposition~\ref{hyper:thgp}~(\emph{i})  because the transformation of $H$ into a monotone HGP computing $f_H$ given in the proof of Proposition~\ref{hyper:thgp}~(\emph{i}) does not preserve the number of leaves.

%****************

\subsection{LOGCFL/poly and THGP}

$\THGP$ and $\mTHGP$ are the classes of functions computable by polynomial-size THGPs and, respectively, monotone THGPs.

\begin{theorem} \label{thm:thp_vs_sac}
$\LOGCFL/\poly = \THGP$ and $\mLOGCFL/\poly = \mTHGP$.
%
%{\rm (\emph{i})} For any Boolean function $f$,
%$\HGP_t(f)$ is polynomial iff $f \in \SAC^1$.
%
%\noindent {\rm (\emph{ii})} For any monotone Boolean function $f$,
%$\HGP_{+,t}(f)$ is polynomial iff $f \in \mSAC^1$.
\end{theorem}
\begin{proof}
\begin{figure}[t]%
\centering%
\begin{tikzpicture}
\node at (-3.2,3.7) {{\small (a)}};
\node at (3,3.7) {{\small (b)}};
\draw[dashed] (-3.5,1.6) -- (2,1.6);
\node[rotate=90] at (-3.3,2.4) {\scriptsize $\AND$-depth 1};
\node[rotate=90] at (-3.3,0.7) {\scriptsize $\AND$-depth 0};
\node[input,label=below:{\footnotesize $g_{1}$}] (g1) at (-2.8,-0.2) {\small $x_1$}; % ,label=right:{\!\footnotesize 0}
\node[input,label=below:{\footnotesize $g_{2}$}] (g2) at (-1.4,-0.2) {\small $x_2$}; % ,label=right:{\!\footnotesize 0}
\node[input,label=below:{\footnotesize $g_{3}$}] (g3) at (-.6,1) {\small\!\!$\NOTOP x_3$}; % ,label=right:{\!\footnotesize 0}
\node[input,label=below:{\footnotesize $g_{4}$}] (g5) at (.6,1) {\small $x_4$}; % ,label=right:{\!\footnotesize 0}
\node[or-gate,label=left:{\footnotesize $g_5$}] (g4) at (-2.2,1) {$\OR$}; % ,label=right:{\!\footnotesize 0}
\node[fill=gray!60,and-gate,label=left:{\footnotesize $g_{6}$}] (g6) at (-1.5,2.2) {$\AND$}; % ,label=right:{\!\footnotesize 1}
\node[fill=gray!20,and-gate,label=left:{\footnotesize $g_{7}$}] (g7) at (0.6,2.2) {$\AND$}; % ,label=right:{\!\footnotesize 1}
\node[fill=white,or-gate,label=left:{\footnotesize $g_{8}$}] (g8) at (-0.4,3.4) {$\OR$}; %,label=right:{\!\footnotesize 1}
\draw[->] (g1) to (g4);
\draw[->] (g2) to (g4);
\draw[->] (g4) to (g6);
\draw[->] (g3) to (g6);
\draw[->] (g4) to (g7);
\draw[->] (g5) to (g7);
\draw[->] (g6) to (g8);
\draw[->] (g7) to (g8);
\begin{scope}[yshift=5mm,xshift=-5mm]\footnotesize
\begin{scope}[rounded corners=2mm]
%g6 = g5 land g3
\filldraw[fill=gray!60,ultra thin]  (6,1.9) -- ++(-3,0) -- ++(0,-0.3) -- ++(1.35,-0.9) -- ++(2.2,0) -- ++(0.9,-1) -- ++(1.2,0) -- ++(0,0.9) -- cycle; 
%g7 = g4 land g5
\filldraw[fill=gray!40,ultra thin,fill opacity=0.5] (8,2.5) -- ++(-4,0) -- ++(0,-0.9) -- ++(1.35,-0.9) -- ++(5.6,0) -- ++(0,0.9) -- cycle; 
\end{scope}
\node[rotate=-15] at (8.8,1.8) {$g_7 = g_5 \ANDOP g_4$};
\node[rotate=-35] at (4,1.3) {$g_6 = g_3 \ANDOP g_4$};
\node[point,draw=gray!40,label=above:{\scriptsize\textcolor{gray!40}{$w_8$}}] (w8) at (9,3.4) {}; 
\node[point,fill=gray,label=right:{\scriptsize$v_8$}] (v8) at (9,3.1) {}; 
\node[point,label=right:{\scriptsize$u_8$},label=below:{$8$}] (u8) at (9,2.8) {}; 
\node[point] (w7) at (7,2.8) {}; % ,label=above:{$w_7$} 
\node[point,fill=gray] (v7) at (7,2.5) {}; % ,label=right:{$v_7$} 
\node[point,label=below:{$7$}] (u7) at (7,2.2) {}; 
\node[point] (w6) at (5,2.2) {};  % ,label=above:{$w_6$}
\node[point,fill=gray] (v6) at (5,1.9) {};  % ,label=left:{$v_6$}
\node[point,label=below:{$6$}] (u6) at (5,1.6) {}; 
\node[point] (w5) at (6,1) {};  % ,label=above:{$w_5$}
\node[point,fill=gray] (v5) at (6,0.7) {};  % ,label=right:{$v_5$}
\node[point,label=below:{$5$}] (u5) at (6,0.4) {}; 
\node[point,label=above:{\scriptsize$w_2$}] (w2) at (4,0.4) {}; % 
\node[point,fill=gray,label=left:{\scriptsize$v_2$}] (v2) at (4,0.1) {};  % 
\node[point,label=right:{\scriptsize$u_2$},label=below:{$2$}] (u2) at (4,-0.6) {}; 
\node[point,label=above:{\scriptsize$w_1$}] (w1) at (2,-0.6) {}; % 
\node[point,fill=gray,label=left:{\scriptsize$v_1$}] (v1) at (2,-0.9) {}; % 
\node[point,label=left:{\scriptsize$u_1$},label=below:{$1$}] (u1) at (2,-1.6) {}; 
\node[point] (w4) at (10,1) {}; % ,label=above:{$w_4$}
\node[point,fill=gray] (v4) at (10,0.7) {};  % ,label=right:{$v_4$}
\node[point,label=below:{$4$}] (u4) at (10,0) {}; 
\node[point] (w3) at (8,0) {}; % ,label=above:{$w_3$}
\node[point,fill=gray] (v3) at (8,-0.3) {}; % ,label=left:{$v_3$}
\node[point,label=below:{$3$}] (u3) at (8,-1) {}; 
\begin{scope}[thick]\tiny
\draw (v8) -- (u8);
\draw (u8) -- node[midway,fill=black,rectangle,draw,inner sep=1pt] {\textcolor{white}{$1$}} (w7);
\draw (w7) -- (v7);
\draw (v7) -- (u7);
\draw (u7) -- node[midway,fill=black,rectangle,draw,inner sep=1pt] {\textcolor{white}{$1$}} (w6);
\draw (w6) -- (v6);
\draw (v6) -- (u6);
\draw (u6) -- node[midway,fill=black,rectangle,draw,inner sep=1pt,sloped] {\textcolor{white}{$1$}} (w5);
\draw (u6) -- node[midway,fill=black,rectangle,draw,inner sep=1pt,sloped] {\textcolor{white}{$1$}} (w4);
\draw (w5) -- (v5);
\draw (v5) --  (u5);
\draw (u5) -- node[midway,fill=black,rectangle,draw,inner sep=1pt] {\textcolor{white}{$1$}} (w2);
\draw (w2) -- (v2);
\draw (v2) -- node[fill=black,rectangle,draw,inner sep=2pt] {\small\textcolor{white}{$x_2$}} (u2);
\draw (u2) -- node[midway,fill=black,rectangle,draw,inner sep=1pt] {\textcolor{white}{$1$}} (w1);
\draw (w1) -- (v1);
\draw (v1) -- node[fill=black,rectangle,draw,inner sep=2pt] {\small\textcolor{white}{$x_1$}} (u1);
\draw (w4) -- (v4);
\draw (v4) -- node[fill=black,rectangle,draw,inner sep=2pt] {\small\textcolor{white}{$x_4$}} (u4);
\draw (u4) -- node[midway,fill=black,rectangle,draw,inner sep=1pt] {\textcolor{white}{$1$}} (w3);
\draw (w3) -- (v3);
\draw (v3) -- node[fill=black,rectangle,draw,inner sep=2pt] {\textcolor{white}{\small $\NOTOP x_3$}} (u3);
\end{scope}
\end{scope}
\end{tikzpicture}
\caption{(a) A circuit $\Cir$. (b) The labelled tree $T$ for $\Cir$: the vertices in the $i$th triple are $u_i,v_i,w_i$ and the omitted edge labels are 0s. The vertices of THGP are the edges of $T$ (with the same labels) and the hyperedges are sets of edges of $T$ (two of them are shown).}
\label{fig:7}
\end{figure}
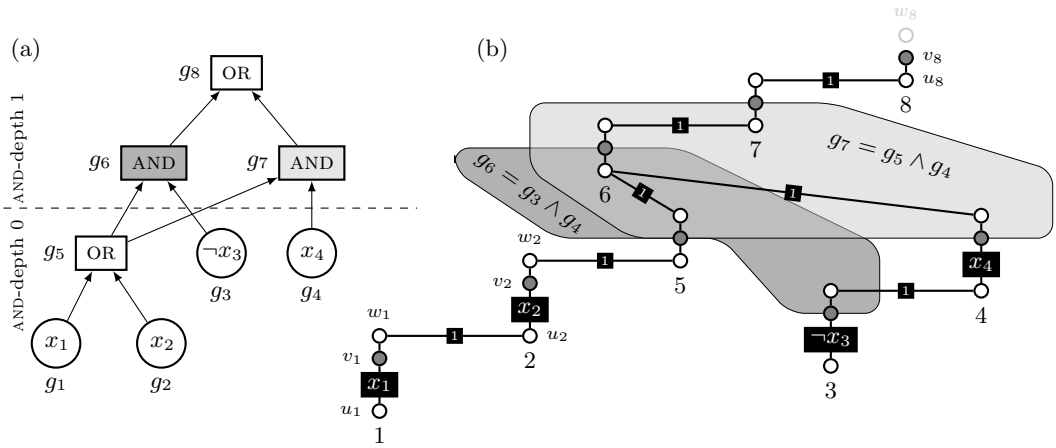
To show $\LOGCFL/\poly \subseteq \THGP$, consider a $\SAC^1$-circuit $\Cir$ of depth \mbox{$d \le \log |\Cir|$}.
It will be convenient to think of $\Cir$ as containing no $\NOT$-gates but having \emph{literals} as inputs.
By the $\AND$-\emph{depth} of a gate $g$ in $\Cir$ we mean the maximal number of $\AND$-gates in a path from an input of $\Cir$ to $g$ (it does not exceed $d$). Let $S_n$ be the set of $\AND$-gates in $\Cir$ of $\AND$-depth $n$.
We denote by $\leftt(g)$ and $\rightt(g)$ the sub-circuits of $\Cir$ computing the left and right inputs of an $\AND$-gate $g$, respectively.  Without loss of generality (see Lemma~\ref{l:6.5} in Appendix~\ref{app:thp_vs_sac}) we can assume that, for any $n \le d$,  
\begin{equation*}
\bigcup\nolimits_{g \in S_n} \leftt(g) \quad \cap \quad \bigcup\nolimits_{g \in S_n} \rightt(g) \ \ = \ \ \emptyset.
\end{equation*}
Our aim is to transform $\Cir$ into a polynomial-size THGP $P$. We construct its underlying tree $T$ by associating with each gate $g_i$  three vertices $u_i, v_i, w_i$ and arranging them into a tree as shown in Fig.~\ref{fig:7}. More precisely, we first  arrange the vertices associated with the gates of maximal $\AND$-depth, $n$, into a path following the order of the gates in~$\Cir$ and the alphabetic order for $u_i, v_i, w_i$. Then we fork the path into two branches one of which is associated with the sub-circuit $\bigcup_{g \in S_n} \leftt(g)$ and the other with $\bigcup_{g \in S_n} \rightt(g)$, and so forth. We obtain the tree $T$ by removing the vertex $w_{m}$ from the result, where $m = |\Cir|$ and $g_m$ is the output gate of $\Cir$; it has $v_{m}$ as its root and contains $3|\Cir| -1$ vertices.
The THGP $P$ is based on the hypergraph whose vertices are the edges of $T$ and whose hyperedges 
comprise the following (see Fig.~\ref{fig:7}):
\begin{nitemize}
\item $[w_i, u_i]$, for each $i < m$ (pairs of edges in each triple of vertices in Fig.~\ref{fig:7});

\item $[v_j,v_k,v_i]$, for each $g_i = g_j \ANDOP g_k$ (shown in Fig.~\ref{fig:7} by shading);

\item $[v_{j_1}, v_i], \dots, [v_{j_k}, v_i]$, for each $g_i = g_{j_1} \OROP \cdots \OROP g_{j_k}$,
\end{nitemize}
where $[L]$ is the minimal convex subtree of $T$ containing the vertices in $L$.
Finally, if an input gate $g_i$ is a literal $\li$, we label the edge $\{u_i, v_i\}$ with $\li$; we label all other $\{u_j,v_j\}$- and $\{w_j,v_j\}$-edges with 0, and the remaining ones with 1. Clearly, the size of $P$ is polynomial in $|\Cir|$. By Lemma~\ref{D2}, for any input $\avec{\alpha}$, the output of $g_i$ is 1 iff the subtree with root $v_i$ can be covered, i.e., there is an independent set of hyperedges wholly inside and covering all zeros. Thus, $P$ computes the same function as $\Cir$.

To show $\THGP \subseteq \LOGCFL/\poly$, suppose a THGP $P$ is based on a hypergraph~$H$ with an underlying tree $T$. 
By a \emph{subtree} of $T$ we understand a (possibly empty) connected subset of edges in $T$.
Given an input $\avec{\alpha}$ for $P$ and a nonempty subtree $D$ of $T$, we set  $\Cover_D$ true iff there exists an independent subset of hyperedges in $H$ that lie in $D$ and cover all zeros in $D$. We also set $\Cover_{\emptyset}$ true. Note that, for any edge $e$ of $T$, $\Cover_{\{e\}}$ is true if $\{e\}$ is a hyperedge of $H$; otherwise $\Cover_{\{e\}}$ is the value of $e$'s label in $P$ under~$\avec{\alpha}$. 

Our aim is to construct recursively a polynomial-size $\SAC^1$-circuit $\Cir$ computing the function $\Cover_T$. 
Observe that, if $D$ is a subtree of $T$ and a vertex $v$ splits $D$ into subtrees $D_1, \dots, D_k$, then 
\begin{equation} \label{eq:sac_construction}
\Cover_D \ \ = \ \ \bigwedge_{1 \leq j \leq k} \Cover_{D_j} \ \ \lor \ \   
\bigvee_{v \in h \subseteq D} \hspace{3mm} \bigwedge_{1 \leq j \leq k_h} \Cover_{D_j^h},
\end{equation}
where $h$ ranges over the hyperedges in $H$, and $D_1^h,\dots,D_{k_h}^h$ are the maximal convex subtrees of $T$ that lie in $D \setminus h$.
We call a vertex $v$ of $D$ \emph{boundary} if $T$ has an edge $\{v,u\}$ with $u$ not in $D$, and define the \emph{degree} $\degree(D)$ of $D$ to be the number of its boundary vertices. Note that $T$ itself  is the only subtree of $T$ of degree $0$. 
The following lemma shows that to compute $\Cover_T$ we only need subtrees of degree 1 and 2 and the depth of recursion $O(\log |P|)$.

\begin{lemma}\label{l:6.8}
Let $D$ be a subtree of $T$ with $m$ vertices and $\degree(D) \leq 2$. If $\degree(D) \leq 1$, then there is a vertex $v$ splitting $D$ into subtrees of size at most $m/2 + 1$ and degree at most $2$. If $\degree(D) =2$, then there is $v$ splitting $D$ into subtrees of size at most $m/2+1$ and degree at most $2$ and, possibly, one subtree of size less than $m$ and degree~$1$.
\end{lemma}
\begin{proof}
Let $\degree(D) \leq 1$. Suppose some vertex $v_1$ splits $D$ into subtrees one of which, say $D_1$, is larger than $m/2+1$. Let $v_2$ be the (unique) vertex in $D_1$ adjacent to $v_1$. The splitting of $D$ by $v_2$ consists of the subtree $D_2 = (D \setminus D_1) \cup \{v_1,v_2\}$ of size at most $m/2$ and some other subtrees lying inside $D_1$; all of them are of degree at most 2. We repeat this process until the size of the largest subtree becomes at most $m/2+1$.

Let $\degree(D) =2$, with $b_1$ and $b_2$ being the boundary vertices. We proceed as above starting from $v_1=b_1$, but stop when either the largest subtree has $\le m/2+1$ vertices or $v_{i+1}$ leaves the path between $b_1$ and $b_2$, in which case $v_i$ splits $D$ into subtrees of degree at most $2$ and one subtree of degree $1$ with more than $m/2 + 1$ vertices. 
\end{proof}

By applying~\eqref{eq:sac_construction} to $T$ recursively and choosing the splitting vertices $v$ as prescribed by Lemma~\ref{l:6.8}, we obtain a circuit $\Cir$ whose inputs are the labels of some vertices of~$H$. Since any tree has polynomially many subtrees of degree $1$ or $2$, the size of $\Cir$ is polynomial in $|P|$. We now show how to make the depth of $\Cir$ logarithmic in $|P|$. 

Suppose $D$ is a subtree with $m$ edges constructed on the recursion step $i$. To compute $\Cover_D$ using~\eqref{eq:sac_construction}, we need one $\OR$-gate of unbounded fan-in and a number of $\AND$-gates of fan-in 2. We show by induction that we can make the $\AND$-depth of these $\AND$-gates at most $\log m + i$. Suppose $D_j$ in~\eqref{eq:sac_construction} has $m_j$ edges, and so $m = m_1 + \dots + m_k$.
By the induction hypothesis, we can compute each $\Cover_{D_j}$ within the $\AND$-depth at most $\log m_j + i-1$. Assign  the probability $m_j/m$ to $D_j$. As shown by Huffman~\cite{huf52}, there is a prefix binary code such that each $D_j$ is encoded by a word of length $\lceil \log(m/m_j)\rceil$. This encoding can be represented as a binary tree whose leaves are labelled with the $D_j$ so that the length of the branch ending at $D_j$ is $\lceil \log(m/m_j)\rceil$. By replacing each non-leaf vertex of the tree with an $\AND$-gate, we obtain a circuit for the first conjunction in~\eqref{eq:sac_construction} whose depth does not exceed
\begin{equation*}
\max_j\{\log m_j + (i-1) + \log(m/m_j) + 1\} \ \ \ = \ \  \log m + i.
\end{equation*}
The second conjunction is considered analogously.\qed
\end{proof}

%******************

\subsection{NC$^{\boldsymbol{1}}$, $\boldsymbol{\mathsf{\Pi}}_{\boldsymbol{3}}$ and THGP$^d$}

The proof of the following theorem, given in Appendix~\ref{app:thp_vs_sac}, is a simplified version of the proof of Theorem~\ref{thm:thp_vs_sac}:

\begin{theorem}\label{thm:nc1-thgp3}
$\NC^1 =  \THGP^d$ and $\mNC^1 =  \mTHGP^d$, for any $d \geq 3$.
\end{theorem}

THGPs of degree 2 turn out to be less expressive:

\begin{theorem}\label{thm:pi3-thgp2}
$\Pitr = \THGP^2 = \THGP^2(2)$ and $\mPitr = \mTHGP^2 = \mTHGP^2(2)$.
\end{theorem}
\begin{proof}
To show $\THGP^2 \subseteq \Pitr$, take a THGP $P$ of degree 2. Without loss of generality we can assume that it contains no hyperedges $e, e^\prime$ with $e \subseteq e^\prime$, for otherwise the vertices in $e$ would  not be covered by any other hyperedges, and so could be removed  from $P$ together with $e$.

Consider the graph $D$ whose vertices are the hyperedges of $P$, with two vertices being connected if the corresponding hyperedges intersect. Clearly, $D$ is a forest. We label an edge $\{e_1, e_2\}$ with the conjunction of the labels of the vertices in $e_1\cap e_2$, and  label a vertex $e$ with the conjunction of the labels of the vertices in $P$ contained exclusively in $e$.
It is easy to see that, for any given input, an independent cover of zeros in $P$ corresponds to an independent set in $D$ covering all zeros in the vertices and such that each edge labelled with 0 has precisely one endpoint in that independent set.

We claim that there is no such an independent set $I$ in $D$ iff there is a path $e_0, e_1, \ldots, e_k$ in $D$ with odd $k$ (in particular, $k=1$) such that $e_0$ and $e_k$ are labelled with $0$ and `even' edges $\{e_{i-1},e_{i}\}$ with even $i$ are labelled with $0$.
To see $(\Leftarrow)$, observe that we have to include $e_0$ and $e_k$ in $I$. Then, the edge $\{e_1, e_2\}$ labelled with $0$ makes us to include $e_2$ to $I$ ($e_1$~is adjacent to $e$ and cannot be included in $I$). Next, the edge $\{e_3,e_4\}$ makes us to include~$e_4$ in $I$ and so on. In the end we will have to include $e_{k-1}$ to $I$ and, since $e_k$ is also in $I$, this gives a contradiction with independence of $I$.

To show $(\Rightarrow)$, suppose there is no such a pair of vertices.  Then we can construct a desired independent set $I$. Add to $I$ all vertices labelled with $0$. If there is a triple of consecutive vertices $e, e_1, e_2$ in $D$ such that $e$ is already in $I$ and an edge $\{e_1,e_2\}$ is labelled with $0$, then we add $e_2$ to $I$. Note that, if we have add some vertex $e^\prime$ to $e$ in this process, then  there is a path $e=e_0, e_1, \ldots, e_k=e^\prime$ with even $k$ such that vertex $e$ is labelled with $0$ and every edge $\{e_{i-1},e_i\}$ for even $i$ in this path is labelled with $0$.

In this process we never add two connected vertices $e$ and $e^\prime$ of $D$ to $I$, for otherwise the union of the paths described above for these two vertices  would result in a path of odd length with endpoints labelled with $0$ and with every second edge labelled with $0$. This directly contradicts our assumption.

If there are still edges labelled with $0$ in $D$ with no endpoints in $I$, then add any endpoint of such an edge to $I$ and repeat the process above. This also will not lead to a pair of connected vertices in $I$. Indeed, if as a result we add to $I$ a vertex $e_1$ connected to a vertex $e$ which was added to $I$ previously, then there is an edge $\{e_2, e_1\}$ labelled with $0$ (that was the reason for adding $e_1$ to $I$), and so we should have added $e_2$ to $I$ before.
By repeating this process, we obtain an independent set $I$ covering all vertices and edges labelled with $0$.

The established claim means that an independent set $I$ in $D$ exists iff, for any simple path $e_0, e_1, \dots, e_k$ with an odd $k$, the label of $e_0$ or $e_k$ evaluates to $1$, or the label of at least one $\{e_{i-1},e_{i}\}$, for even $i$, evaluates to $1$.
This property is computed by a $\Pitr$-circuit where, for each simple path $e_0, e_1, \dots, e_k$ with an odd $k$, we take $(k+3)/2$-many 
$\AND$-gates whose inputs are the literals in the labels of $e_0$, $e_k$ and the $\{e_{i-1},e_{i}\}$ for even $i$; then we send the outputs of those $\AND$-gates to an $\OR$-gate; and, finally, we collect the outputs of all the $\OR$-gates as inputs to an $\AND$-gate.

\smallskip

To show $\Pitr \subseteq \THGP^2(2)$, suppose we are given a $\Pitr$-circuit $\Cir$. We can assume $\Cir$ to be a conjunction of DNFs. So, we first construct a generalised HGP $P$ from $ \THGP^2(2)$ computing the same function as $\Cir$. 
Denote the $\OR$-gates of $\Cir$ by $g^1, \dots, g^k$ and the inputs of $g^i$ by $h^i_{1}, \dots, h^i_{l_i}$, where each $h^i_{j}$ is an $\AND$-gate. 
Now, we define a tree hypergraph whose underlying path graph has the  following edges (in the given order)
\begin{equation*}
v^1_{0},\dots, v^1_{2l_1-2}, \ \ \ v^2_{0},\dots, v^2_{2l_2-2},\ \ \ \dots, \ \ \ v^k_{0},\dots, v^k_{2l_k-2}
\end{equation*}
and whose hyperedges are of the form $\{v^i_{j}, v^i_{j+1}\}$. We label $v^i_{2m}$ with a conjunction of the inputs of $h^i_{m+1}$ and the remaining vertices with $0$.
By the previous analysis  for a given~$i$ and an input for $\Cir$, we can cover all zeros among $v^i_{0},\dots, v^i_{2l_i-2}$ with an independent set of hyperedges iff at least one of the gates $h^i_{1},\dots, h^i_{l_i}$ outputs 1. For different $i$, the corresponding $v^i_{0},\dots, v^i_{2l_i-2}$ are covered independently. Thus, $P$ computes the same function as $\Cir$. We convert $P$ to a THGP from $\smash{\THGP^2(2)}$ using   Proposition~\ref{hyper:thgp}~(\emph{ii}).
\end{proof}

%*******************

\section{The Size of OMQ Rewritings}\label{sec:7}

In this section, by an OMQ $\omq = (\T,\q)$ we mean a sequence $\{\omq_n = (\T_n,\q_n)\}_{n < \omega}$ of OMQs whose size is polynomial in $n$; by a rewriting $\q'$ of $\omq$ we mean a sequence $\{\q'_n\}_{n<\omega}$, where each $\q'_n$ is a rewriting of $\omq_n$, for $n < \omega$. 

By putting together the results of the previous three sections and  some known facts from circuit complexity, we obtain the upper and lower bounds on the size of PE-, NDL- and FO-rewritings for various OMQ classes that are collected in Table~\ref{table:rewritings},
\begin{table}
\caption{The size of OMQ rewritings.}{%
\renewcommand{\tabcolsep}{7pt}%
\begin{tabular}{lccc}\toprule
OMQ $\omq = (\T,\q)$ & PE & NDL & FO \\\midrule
$\T$ of depth 2 & \textcolor{gray}{$\mathsf{exp}$ {\footnotesize (Th.~\ref{Depth2:Clique})}} & $\mathsf{exp}$ {\footnotesize (Th.~\ref{Depth2:Clique})} & 
\begin{tabular}{c}$>\mathsf{poly} ~\text{ if }~ \NP \not\subseteq \P/\poly$ {\footnotesize (Th.~\ref{Depth2:Clique})}\\
$\mathsf{poly} ~\text{ iff }~ \NP/\poly \subseteq \NC^1$ {\footnotesize (Th.~\ref{nppoly})}\end{tabular}\\\midrule
$\T$ of depth 1 & $> \mathsf{poly}${\footnotesize (Th.~\ref{Depth1:PE})} & $\mathsf{poly}${\footnotesize (Th.~\ref{polyNDL})} & $\mathsf{poly} ~\text{ iff }~ \NL/\poly \subseteq \NC^1${\footnotesize (Th.~\ref{thm:nc1cond})}\\
\ \& $\q$ of treewidth $t$ & $\mathsf{poly}${\footnotesize (Th.~\ref{depth-one-btw})} &  \textcolor{gray}{$\mathsf{poly}${\footnotesize (Th.~\ref{depth-one-btw})}} & \textcolor{gray}{$\mathsf{poly}${\footnotesize (Th.~\ref{depth-one-btw})}} \\
\ \& $\q$ tree & $\mathsf{poly}\text{-}\mathsf{\Pi}_4${\footnotesize (Th.~\ref{depth-one-tree})} & \textcolor{gray}{$\mathsf{poly}${\footnotesize (Th.~\ref{depth-one-tree})}} & \textcolor{gray}{$\mathsf{poly}\text{-}\mathsf{\Pi}_4${\footnotesize (Th.~\ref{depth-one-tree})} } \\\midrule
$\q$ tree with $\ell$ leaves & \begin{tabular}{c}$> \mathsf{poly}${\footnotesize (Th.~\ref{linear-lower})}\\\footnotesize ($\T$ is of depth 2)\end{tabular} & $\mathsf{poly}${\footnotesize (Th.~\ref{bbcq-ndl})} & $\mathsf{poly} ~\text{ iff }~ \NL/\poly \subseteq \NC^1${\footnotesize (Th.~\ref{nbps-conditional})}\\\midrule
\begin{tabular}{l}$\omq$ with PFSP\\
$\q$ of treewidth $t$\end{tabular}  &  & $\mathsf{poly}${\footnotesize (Th.~\ref{btw-ndl})} & \begin{tabular}{c}$\mathsf{poly} ~\text{ iff }~ \LOGCFL/\poly \subseteq \NC^1$\\\mbox{}\hfill\footnotesize{(Th.~\ref{btw-fo})}\end{tabular}\\\bottomrule
\end{tabular}}
\label{table:rewritings}
\end{table}
where $\mathsf{exp}$ means an exponential lower bound, $> \mathsf{poly}$ a superpolynomial lower bound, $\mathsf{poly}$ a polynomial upper bound, $\mathsf{poly}\text{-}\mathsf{\Pi}_4$ a polynomial-size $\mathsf{\Pi}_4$-rewriting (that is, a PE-rewriting with the matrix of the form ${\land}{\lor}{\land}{\lor}$), and $\ell$ and $t$ are any fixed constants. It is to be noted that, in case of  polynomial upper bounds, we actually provide polynomial algorithms for constructing rewritings.

\subsection{Rewritings for OMQs with ontologies of depth 2}

By Theorem~\ref{NBC}, OMQs with ontologies of depth~2 can compute any $\NP$-complete monotone Boolean function, in particular, the function $\cli$ with $n(n-1)/2$ variables $e_{jj'}$, $1 \leq j < j'\le n$, that returns 1 iff the graph with vertices $\{1,\dots,n\}$ and  edges $\{ \{j,j'\} \mid e_{jj'}=1\}$ contains a $k$-clique, for some fixed $k$.
A series of papers,  started by Razborov~\cite{Razborov85}, gave an exponential lower bound for the size of monotone circuits computing $\cli$, namely, $2^{\Omega(\sqrt{k})}$ for $k \leq \smash{\frac{1}{4}} (n/ \log n)^{2/3}$~\cite{AlonB87}. For monotone formulas, an even better lower bound is known: $2^{\Omega(k)}$ for $k = 2n/3$~\cite{RazW92}. Thus, we obtain:

\begin{theorem}\label{Depth2:Clique}
There is an OMQ with ontologies of depth $2$, any \PE- and \NDL-rewritings of which are of exponential size, while any \FO-rewriting is of superpolynomial size unless $\NP \subseteq \P/\poly$.
\end{theorem}
\begin{proof}
In view of $\cli \in \NP \subseteq \NP/\poly$ and Theorem~\ref{NBC}, there is a polynomial-size monotone HGP $P$ computing $\cli$. Suppose $P$ is based on a hypergraph $H$ and $\omq_H$ is the OMQ for $H$ constructed in Section~\ref{sec5.1}. By Theorem~\ref{hg-to-query}~(\emph{ii}), $\cli$ is a subfunction of the primitive evaluation function $f^\vartriangle_{\omq_H}$. By Theorem~\ref{rew2prim}, if $\q'$ is a \PE- or \NDL-rewriting of $\omq_H$, then $f^\vartriangle_{\omq_H}$---and so $\cli$---can be computed by a monotone formula or, respectively, circuit of size $O(|\q'|)$. Thus, $\q'$ must be of exponential size. If $\q'$ is an \FO-rewriting of $\omq_H$ then, by Theorem~\ref{rew2prim}, $\cli$ is computable by a Boolean formula of size $O(|\q'|)$. If $\NP \not\subseteq \P/\poly$ then $\cli$ cannot be computed by a polynomial circuit, and so $\q'$ must be of superpolynomial size.
\end{proof}

Our next theorem gives a complexity-theoretic characterisation of  the existence of FO-rewritings for OMQs with ontologies of depth $2$.

\begin{theorem}\label{nppoly}
The following conditions are equivalent\textup{:}
\begin{enumerate}[(1)]
\item all OMQs with ontologies of depth $2$ have polynomial-size \FO-rewritings\textup{;}

\item all OMQs with ontologies of depth $2$ and polynomially many tree witnesses have polynomial-size \FO-rewritings\textup{;}

\item $\NP/\poly \subseteq \NC^1$.
\end{enumerate}
\end{theorem}
\begin{proof}
The implication (1) $\Rightarrow$ (2) is obvious. To show that (2) $\Rightarrow$ (3), suppose there is a polynomial-size \FO-rewriting for the OMQ $\omq_H$ from the proof of Theorem~\ref{Depth2:Clique}, which has polynomially many tree witnesses. Then $\cli$ is computed by a polynomial-size Boolean formula. Since $\cli$ is $\NP/\poly$-complete under $\NC^1$-reductions, we have $\NP/\poly \subseteq \NC^1$. Finally, to prove (3) $\Rightarrow$ (1), assume $\NP/\poly \subseteq \NC^1$. Let $\omq$ be an arbitrary OMQ with ontologies of depth 2. As observed in Section~\ref{hyper-functions}, the function $\homfn$ is in $\NP \subseteq \NP/\poly$. 
Therefore, by our assumption, $\homfn$ can be computed by a polynomial-size formula, and so, by Theorem~\ref{Hom2rew}, $\omq$ has a polynomial-size \FO-rewriting.
\end{proof}

%\begin{theorem}\label{nppoly}
%There exist polynomial-size \FO-rewritings for all OMQs with  ontologies of depth $2$ and polynomially-many tree witnesses
%iff all functions in $\NP/\poly$ are com\-puted by polynomial-size formulas, that is, iff $\NP/\poly \subseteq \NC^1$.
%\end{theorem}
%%
%\begin{proof}
%$(\Leftarrow)$ Suppose $\NP/\poly \subseteq \NC^1$. Consider an arbitrary OMQ $\omq$ with ontologies of depth 2 and polynomially-many tree witnesses. Then the hypergraph $H_{\omq}$ is of polynomial size. The hypergraph function $f_{H_{\omq}}$ is in $\NP/\poly$ because the problem whether there exists an independent set of hyperedges in a hypergraph covering all zeros is in $\NP$. Therefore, by our assumption, $f_{H_{\omq}}$ can be computed by a polynomial-size formula, and so, by Theorem~\ref{TW2rew}, $\omq$ has a polynomial-size \FO-rewriting.
%
%$(\Rightarrow)$ Suppose there is a polynomial-size \FO-rewriting for the OMQ $\omq_H$ from the proof of Theorem~\ref{Depth2:Clique}, which has polynomially-many tree witnesses.\nb{mention in sec 5.1} Then $\cli$ is computed by a polynomial-size Boolean formula. Since $\cli$ is $\NP/\poly$-complete under $\NCo$ reductions, we then obtain $\NP/\poly \subseteq \NC^1$.
%\end{proof}

%*****************

\subsection{Rewritings for OMQs with ontologies of depth 1}

%The restriction of ontology depth to 1 results in a polynomial upper bound on the size of \NDL-rewritings:

\begin{theorem}\label{polyNDL}
Any OMQ $\omq$ with ontologies of depth $1$ has a polynomial-size $\NDL$-rewriting.
\end{theorem}
\begin{proof}
By Theorem~\ref{depth1},  the hypergraph $\HG{\omq}$ is of degree at most~2, and so, by Proposition~\ref{hyper:program}~(\emph{ii}), there is a polynomial-size monotone HGP of degree at most~2 computing $\twfn$. By Theorem~\ref{thm:deg_2}, $\comNL/\poly = \mHGP^2$, and so we have a polynomial-size monotone NBP computing the dual ${\twfn}^*$ of $\twfn$. 
Since $\mNL/\poly \subseteq \mP/\poly$,  
we also have a polynomial-size monotone Boolean circuit that  computes ${\twfn}^*$.  
By swapping $\AND$- and $\OR$-gates in that circuit, we obtain a polynomial-size monotone circuit computing $\twfn$. It remains to apply Theorem~\ref{TW2rew}~(\emph{ii}).
\end{proof}

However, this upper bound cannot be extended to \PE-rewritings:

\begin{theorem}\label{Depth1:PE}
There is an OMQ $\omq$ with ontologies of depth $1$, any \PE-rewriting of which is of superpolynomial size \textup{(}$n^{\Omega(\log n)}$, to be more precise\textup{)}.
\end{theorem}
\begin{proof}
Consider the monotone function $\reach$ that takes the adjacency matrix of a  directed graph $G$ with two distinguished vertices $s$ and $t$ and returns~1 iff the graph $G$ contains a directed path from $s$ to $t$. It is known~\cite{KarchmerW88,Jukna12} that $\reach$ is computable by a polynomial-size monotone NBP (that is, belongs to $\mNL/\poly$), but any monotone Boolean formula for $\reach$ is of size $n^{\Omega(\log n)}$. 
Let $f = \reach$. By Theorem~\ref{thm:deg_2}, there is a polynomial-size monotone HGP that is based on a hypergraph $H$ of degree~$2$ and computes the dual $f^*$ of $f$. Consider now the OMQ $\OMQI{H}$ for $H$ defined in Section~\ref{sec:depth1}. By Theorem~\ref{representable}~(\emph{ii}), %and  Proposition~\ref{hyper:program}~(\emph{ii}), 
$f^*$ is a subfunction of $f^\vartriangle_{\OMQI{H}}$. By Theorem~\ref{rew2prim}~(\emph{i}),  no PE-rewriting of the OMQ $\OMQI{H}$ can be shorter than $\smash{n^{\Omega(\log n)}}$.
\end{proof}

%The function $f$ in the proof above  checks whether two vertices in a given undirected graph are connected by a path. Note that undirected graphs can be safely replaced by directed ones. Indeed, since the undirected case reduces to the directed one, we have the same lower bound for computing directed reachability by monotone formulas. 

%On the other hand, it is known~\cite{Razborov91} that directed reachability can be computed by polynomial-size monotone circuits. This observation gives the following complexity-theoretic characterisation of the existence of FO-rewritings for OMQs with ontologies of depth $1$:

%As reachability in directed graphs is $\NLogSpace/\poly$-complete under $\NC^1$-reductions, the argument in the proof of Theorem~\ref{Depth1:PE} shows that the existence of short \FO-rewritings for OMQs with ontologies of depth~1 is equivalent to a well-known open problem in computational complexity:
%
\begin{theorem}\label{thm:nc1cond}
All OMQs with ontologies of depth $1$ have polynomial-size \FO-rewritings iff $\NL/\poly \subseteq \NC^1$.
\end{theorem}
\begin{proof}
Suppose $\NL/\poly \subseteq \NC^1$. Let $\omq$ be an OMQ with ontologies of depth $1$. By Theorem~\ref{depth1}, its hypergraph $\HG{\omq}$ is of degree~$2$ and polynomial size. By Proposition~\ref{hyper:program}~(\emph{ii}), there is a polynomial-size HGP of degree~2 that computes  $\twfn$. By Theorem~\ref{thm:deg_2}, 
$\twfn\in\NL/\poly$. Therefore, by our assumption,  $\twfn$ can be computed by a polynomial-size Boolean formula.
Finally, Theorem~\ref{TW2rew}~(\emph{i}) gives a polynomial-size \FO-rewriting of $\omq$.

Conversely, suppose there is a polynomial-size \FO-rewriting for any OMQ with ontologies of depth~$1$. Let $f = \reach$. Since $f \in \NL \subseteq \NL/\poly$, 
by Theorem~\ref{thm:deg_2},  
we obtain a polynomial-size HGP computing $f$ and based on a hypergraph $H$ of degree~2. Consider the OMQ $\OMQI{H}$ with ontologies of depth~1 defined in Section~\ref{sec:depth1}. By Theorem~\ref{representable}~(\emph{ii}),
$f$ a subfunction of $f^{\vartriangle}_{\OMQI{H}}$. 
By our assumption, $\OMQI{H}$ has a polynomial-size \FO-rewriting; hence, by Theorem~\ref{rew2prim}~(\emph{i}), $f^{\vartriangle}_{\OMQI{H}}$ (and so $f$) are computed by polynomial-size Boolean formulas. Since $f$ is $\NL/\poly$-complete under $\smash{\NC^1}$-reductions~\cite{Razborov91}, we  obtain $\NL/\poly \subseteq \smash{\NC^1}$.
\end{proof}

%*******************

\subsection{Rewritings for tree-shaped OMQs with a bounded number of leaves}

Since, by Theorem~\ref{thm:linear_hgp1}, the hypergraph function of a leaf-bounded OMQ can be computed by a polynomial-size NBP, we have:

\begin{theorem}\label{bbcq-ndl}
For any fixed $\ell \ge 2$, all tree-shaped OMQs with at most $\ell$ leaves have polynomial-size NDL-rewritings.
\end{theorem}

The superpolynomial lower bound below is proved in exactly the same way as Theorem~\ref{Depth1:PE} using Theorems~\ref{thm:linear_hgp} and~\ref{tree-hg-to-query} instead of Theorems~\ref{thm:deg_2} and~\ref{representable}.

\begin{theorem}\label{linear-lower}
There is an OMQ with ontologies of depth 2 and linear CQs any \PE-rewriting of which is of superpolynomial size \textup{(}$n^{\Omega(\log n)}$, to be more precise\textup{)}.
\end{theorem}
%
%\begin{proof}
%Apply Theorem \ref{connectivity-to-hp} to the sequence $f_n$ mentioned above to obtain % gives us
%a sequence of interval hypergraph programs $P_n$ based on interval hypergraphs $H_n$  which compute the functions $f^n$.
%By Theorem~\ref{lin-hg-to-query}, there exist linear CQs $\q_n$ and ontologies
%$\Tmc_n$ of depth 2 such that $f_{H_n}$ is a subfunction of $f^P_{\q_n, \Tmc_n}$. By the construction,
%$\q_n$ and $\Tmc_n$ are of polynomial size in $n$.
%Since  $f_n$ is obtained from  $f^P_{\q_n, \Tmc_n}$ through a simple substitution, the lower bound
%$n^{\Omega(\log n)}$ still holds for $f^P_{\q_n, \Tmc_n}$. It remains to apply Theorem~\ref{rew2prim}
%to transfer this lower bound to PE-rewritings of $\q_n$ and~$\Tmc_n.$\qed
%\end{proof}
%Finally, we use Theorems~\ref{TW2rew},~\ref{rew2prim},~\ref{nbp-to-ihp}, and \ref{lin-hg-to-query} to show that the existence of polysize FO-rewritings is equivalent to the open problem of whether $\nlpoly\subseteq \ncone$.

Our next result is similar to Theorem~\ref{thm:nc1cond}:

\begin{theorem}\label{nbps-conditional}
The following are equivalent\textup{:}
\begin{enumerate}[(1)]
\item there exist polynomial-size \FO-rewritings for all OMQs with linear CQs and ontologies of depth $2$\textup{;}

\item for any fixed $\ell$, there exist polynomial-size \FO-rewritings for all tree-shaped OMQs with at most $\ell$ leaves\textup{;}

\item $\NL/\poly \subseteq \NC^1$.
\end{enumerate}
\end{theorem}
\begin{proof} $(1) \Rightarrow (3)$ Suppose every OMQ $\omq = (\T,\q)$ with linear $\q$ and $\Tmc$ of depth~2 has an FO-rewriting of size $p(|\omq|)$, for some fixed polynomial $p$. Consider $f = \reach$.  As $f$ is monotone and $f\in\NL$, we have $f\in \mNL/\poly$. Thus, Theorem~\ref{thm:linear_hgp} gives us an HGP $P$ from $\mTHGP(2)$ that computes $f$. Let $P$ be based on a hypergraph~$H$, and let $\OMQT{H}$ be the OMQ with a linear CQ and an ontology of depth~2 constructed in Section~\ref{sec:5.3}.  By Theorem~\ref{tree-hg-to-query}~(\emph{ii}), $f$ is a subfunction of $f^\vartriangle_{\OMQT{H}}$.
By our assumption, however,  $\OMQT{H}$ has a polynomial-size FO-rewriting, and so, by Theorem~\ref{rew2prim}~(\emph{i}), it is computed by a polynomial-size Boolean formula. Since $f$ is $\NL/\poly$-complete under $\NC^1$-reductions~\cite{Razborov91}, we  obtain $\NL/\poly \subseteq \NC^1$.
The implication $(3) \Rightarrow (2)$ follows from Theorems~\ref{thm:linear_hgp1} and~\ref{TW2rew}~(\emph{i}), and $(2) \Rightarrow (1)$ is trivial.
\end{proof}

%***********************

\subsection{Rewritings for OMQs with PFSP and bounded treewidth}\label{sec:7.4}

Since OMQs with the polynomial fundamental set property (PFSP, see Section~\ref{sec:TW}) and CQs of bounded treewidth can be polynomially translated into monotone THGPs and $\mTHGP = \mLOGCFL/\poly \subseteq \mP/\poly$, we obtain:
\begin{theorem}\label{btw-ndl}
For any fixed $t >0$, all OMQs with the PFSP and CQs of  treewidth at most $t$ have polynomial-size NDL-rewritings.
\end{theorem}

Using Theorem~\ref{role-inc} and the fact that OMQs with ontologies of bounded depth  enjoy the PFSP, we obtain:
\begin{corollary}
The following OMQs have polynomial-size \NDL-rewritings\textup{:}
\begin{nitemize}
\item[--] OMQs with ontologies of bounded depth and CQs of bounded treewidth\textup{;}

\item[--] OMQs with ontologies not containing axioms of the form $\varrho(x,y) \to \varrho'(x,y)$ \textup{(}and~\eqref{eq:sugar}\textup{)} and CQs of bounded treewidth.
\end{nitemize}
\end{corollary}

Whether all OMQs without axioms of the form $\varrho(x,y) \to \varrho'(x,y)$ have polynomial-size rewritings remains open.\!\footnote{A positive answer to this question given by Kikot et al.~\cite{DBLP:conf/dlog/KikotKZ11} is based on a flawed proof.}

\begin{theorem}\label{btw-fo}
The following are equivalent\textup{:}
\begin{enumerate}[(1)]
\item there exist polynomial-size FO-rewritings for all tree-shaped OMQs with ontologies of depth 2\textup{;}

\item there exist polynomial-size FO-rewritings for all OMQs with the PFSP and CQs of treewidth at most $t$ \textup{(}for any fixed $t$\textup{)}\textup{;} 

\item $\LOGCFL/\poly \subseteq \NC^1$.
\end{enumerate}
\end{theorem}
\begin{proof} The implication $(2) \Rightarrow (1)$ is trivial, and $(1) \Rightarrow (3)$ is proved similarly to the corresponding case of Theorem~\ref{nbps-conditional} using Theorems~\ref{thm:thp_vs_sac},~\ref{tree-hg-to-query} and~\ref{rew2prim}.
%
%Suppose all OMQs $\omq$ in (1) have FO-rewritings of size $p(|\omq|)$, for some polynomial $p$. 
%
%Take any Boolean function $f$ computed by a $\SAC^1$-circuit. 
%By Theorem~\ref{thm:thp_vs_sac}, $f$ is computed by a $\THP$ $P$ based on a tree hypergraph $H$ of size at most $p'(\size)$, for some polynomial $p'$. The OMQ $\omq^{tr}_H = (\T_H,\q_H)$ is tree-shaped and of depth 2. By Theorem~\ref{tree-hg-to-query}, $f$ is a subfunction of $f^{\sf p}_{\omq^{tr}_H}$. By the assumption, there exists a FO rewriting for
%$\q_H$ and $\Tmc_H$ of size at most $p''(|\q|+|\Tmc|)$. This number is polynomial in $\size$
%(take the composition of $p$, the polynomial function from Theorem \ref{tree-hg-to-query} and $p''$).
%Now by Theorem \ref{rew2prim}, there exists a polysize first-order formula for computing
%$f^P_{\q, \Tmc}$ ,and hence also for $f_P$ and $f$.\smallskip\\
%%
%
$(3) \Rightarrow (2)$ follows from Theorems~\ref{DL2THP},~\ref{thm:thp_vs_sac} and~\ref{Hom2rew}.
\end{proof}

%******************

\subsection{Rewritings for OMQs with ontologies of depth 1 and CQs of  bounded treewidth}\label{sec:7.5}

We show finally that polynomial PE-rewritings are guaranteed to exist for OMQs with ontologies of depth 1 and CQs of bounded treewidth. By Theorem~\ref{thm:nc1-thgp3}, it suffices to show that $\homfn$ is computable by a THGP of bounded degree. However, since  tree witnesses can be initiated by multiple roles, the THGPs constructed in Section \ref{sec:boundedtw} do not enjoy this property and require a minor modification. 

Let $\omq = (\T,\q)$ be an OMQ with $\T$ of depth 1.  For every tree witness $\t = (\tr,\ti)$, we take a fresh binary predicate $P_\t$ 
(which cannot occur in any data instance) 
and extend~$\T$ with the following axioms:
\begin{align*}
\tau(x) \to \exists y\, P_\t(x,y), & \qquad  \text{ if } \tau \text{ generates } \t,\\
P_\t(x,y) \to \varrho(x,y), & \qquad \text{ if } \varrho(u,v) \in \q_\t, u \in \tr \text{ and } v \in \ti.  
\end{align*}
Denote the resulting ontology by $\T'$ and set $\omq' = (\T',\q)$. 
By Theorem~\ref{depth1}, the number of tree witnesses for $\omq$ does not exceed $|\q|$, and so 
the size of $\omq'$ is polynomial in $|\omq|$. It is easy to see that any rewriting of $\omq'$ (with $P_\t$ replaced by $\bot$) is also a rewriting for~$\omq$.  
Thus, it suffices to consider OMQs of the form $\omq'$, which will be called \emph{explicit}.

Given an explicit OMQ $\omq = (\T,\q)$, we construct a THGP $P'_{\omq}$ in the same way as $P_{\omq}$ in Section~\ref{sec:boundedtw} except that in the definition of $E^k_i$, instead of considering all types $\avec{w}_k$ of $N_i$, we only use $\avec{w}_k=(\avec{w}[1],\dots,\avec{w}[m])$ in which $\avec{w}[j]$ is either $\varepsilon$ or $P_\t$ for the unique tree witness $\t = (\tr,\ti)$ with $\ti = \{ \lambda_j(N_i)\}$. (Since $\T$ is of depth 1, every tree witness~$\t$ has $\ti=\{z\}$, for some variable $z$, and $\ti \neq \ti'$ whenever $\t \neq \t'$.)
This modification guarantees that, for every $i$, the number of distinct $E^k_i$ is bounded by $2^m$. It follows that the hypergraph of $\smash{P_{\omq}'}$ is of bounded degree, $2^m + 2^{2m}$ to be more precise. 
To establish the correctness of the modified construction, we can prove an analogue of Theorem~\ref{DL2THP}, in which the original THGP $P_{\omq}$ is replaced by $P_{\omq}'$, and the function $\homfn$ is replaced by 
\begin{equation*}\label{hyper-func''}
\homfnprime \ \ = \ \ \hspace*{-1em}\bigvee_{\substack{\Theta \subseteq \twset
\\ \text{ independent}}}\hspace*{-2mm}
\Big(\bigwedge_{\atom \in \q \setminus \q_\Theta} \hspace*{-1em} p_\atom \hspace*{1em}
 \wedge  \hspace*{0.5em}\bigwedge_{\t \in \Theta} \big(\bigwedge_{R(z,z')\in \q_\t} p_{z=z'} \hspace*{0.5em} \land  \hspace*{0.5em} \bigwedge_{z \in \tr\cup\ti} p_{\exists y P_\t(z,y)}\big)\Big) 
\end{equation*}
(which is obtained from $\homfn$ by always choosing $P_\t$ as the predicate that initiates $\t$). It is easy to see that Theorem~\ref{TW2rew} holds also for $\homfnprime$ (with explicit $\omq$), which gives us:
\begin{theorem} \label{depth-one-btw}
For any fixed $t>0$, all OMQs with ontologies of depth 1 and CQs of treewidth at most $t$ have polynomial-size PE-rewritings.
\end{theorem}

For tree-shaped OMQs, we obtain an even better result. Indeed, by Theorem~\ref{prop:tree-shaped}, $\HG{\omq}$ is a tree hypergraph; by Theorem~\ref{depth1}, it is of degree at most~2, and so, by Theorem~\ref{thm:pi3-thgp2}, $\twfn$ is computed by a polynomial-size $\Pitr$-circuit (which is monotone by definition). Thus, Theorem~\ref{TW2rew}~(\emph{i}) gives us the following ($\Pitr$ turns into $\mathsf{\Pi}_4$ because of the disjunction in the formula $\tw_\t$):

\begin{theorem} \label{depth-one-tree}
All tree-shaped OMQs with ontologies of depth 1 have polynomial-size $\mathsf{\Pi}_4$-rewritings. 
\end{theorem}

%*****************

\section{Combined Complexity of OMQ answering}\label{sec:complexity}

The size of OMQ rewritings we investigated so far is crucial for  classical OBDA, which relies upon a reduction to standard database query evaluation (under the assumption that it is efficient in real-world applications). However, this way of answering OMQs may not be optimal, and so understanding the size of OMQ rewritings does not shed much light on how hard OMQ answering actually is. For example, 
answering the OMQs from the proof of Theorem~\ref{Depth1:PE} via PE-rewriting requires superpolynomial time, while the graph reachability problem encoded by those OMQs is \NL-complete. On the other hand, the existence of a short rewriting does not obviously imply tractability.

In this section, we analyse the \emph{combined} complexity of answering OMQs classified according to the depth of ontologies and the shape of CQs. More precisely, our concern is the following decision problem: given an OMQ $\omq(\avec{x}) = (\T,\q(\avec{x}))$, a data instance $\A$ and a tuple $\avec{a}$ from $\ind(\A)$ (of the same length as $\avec{x}$), decide whether $\T,\A \models \q(\avec{a})$. Recall from Section~\ref{sec:TW} that $\T,\A \models \q(\avec{a})$ iff $\canmod \models \q(\avec{a})$ iff there exists a homomorphism from $\q(\avec{a})$ to $\canmod$.

The combined complexity of CQ evaluation has been thoroughly investigated in relational database theory. In general, evaluating CQs  is \NP-complete~\cite{Chandra&Merlin77}, but becomes tractable for tree-shaped CQs~\cite{DBLP:conf/vldb/Yannakakis81} and bounded treewidth CQs~\cite{DBLP:journals/tcs/ChekuriR00,DBLP:conf/stoc/GroheSS01}---\LOGCFL-complete, to be more precise~\cite{DBLP:journals/jacm/GottlobLS01}.

The emerging combined complexity landscape for OMQ answering is summarised in Fig.~\ref{pic:results}~(b) in Section~\ref{sec:results}. The \NP{} and \LOGCFL{} lower bounds for arbitrary OMQs and tree-shaped OMQs with ontologies of bounded depth are inherited from the corresponding CQ evaluation problems. The \NP{} upper bound for all OMQs was shown by~\cite{CDLLR07} and \cite{ACKZ09}, while the matching lower bound for tree-shaped OMQs by~\cite{DBLP:conf/dlog/KikotKZ11} and \cite{DBLP:journals/ai/GottlobKKPSZ14}. By reduction of the reachability problem for directed graphs, one can easily show that evaluation of tree-shaped CQs with a bounded number of leaves (as well as answering OMQs with unary predicates only) is \NL-hard.
%the \NL{} lower bound The \NL{} lower bound for OMQs with ontologies of depth 0 (consisting solely of axioms of the form $A(x) \to B(x)$) and tree-shaped CQs with a bounded number of leaves is trivial.  
We now establish the remaining results.

%**********

\subsection{OMQs with bounded-depth ontologies}

We begin by showing that the \LOGCFL{} upper bound for CQs of bounded treewidth~\cite{DBLP:journals/jacm/GottlobLS01} is preserved even in the presence of ontologies of bounded depth.

\begin{theorem}\label{logcfl-btw}
For any fixed $d\geq 0$ and $t>0$, answering OMQs with ontologies of depth at most $d$ and CQs of treewidth at most $t$ is \LOGCFL-complete. 
\end{theorem}
\begin{proof}
Let $\omq(\avec{x}) = (\T,\q(\avec{x}))$ be an OMQ with $\T$ of depth at most~$d$ and $\q$ of treewidth at most~$t$. As $\Tmc$ is of finite depth, $\canmod$ is finite for any $\A$.
As \LOGCFL{} is closed under \llred\ reductions~\cite{DBLP:conf/icalp/GottlobLS99} and evaluation of CQs of bounded treewidth is \LOGCFL-complete, it suffices to show that $\canmod$ can be computed by an \llred-transducer (a deterministic logspace Turing machine with a \LOGCFL{} oracle). 
Clearly, we need only logarithmic space to represent any predicate name or individual constant from $\T$ and $\A$, as well as any word $aw \in \Delta^{\canmod}$ (since $|w| \leq d$ and $d$ is fixed). Finally, as entailment in \OWLQL{} is in \NL~\cite{ACKZ09}, each of the following problems can be decided by making a call to an \NL\ (hence \LOGCFL) oracle:
\begin{nitemize}
\item decide whether $a\varrho_1 \dots \varrho_n \in \Delta^{\canmod}$, for any $n \le d$ and roles $\varrho_i$ from $\T$;

\item decide whether $u \in \Delta^{\canmod}$ belongs to $A^{\canmod}$, for a unary $A$ from $\T$ and $\A$;

\item decide whether $(u_1,u_2) \in \Delta^{\canmod} \times \Delta^{\canmod}$ is in $P^{\canmod}$, for a binary $P$ from $\T$ and $\A$. \qed
\end{nitemize}
\end{proof}

If we restrict the number of leaves in tree-shaped OMQs, then the \LOGCFL{} upper bound can be reduced to \NL:

\begin{theorem}\label{nl-bb}
For any fixed $d\geq 0$ and $\ell\geq 2$, answering OMQs with ontologies of depth at most $d$ and tree-shaped CQs with at most $\ell$ leaves is \NL-complete.
\end{theorem}
\begin{proof}
Algorithm~\ref{algo:tree-entail} defines a non-deterministic procedure \bbbdalgo{} for deciding whether a tuple $\avec{a}$ is a certain answer to a tree-shaped OMQ $(\T,\q(\avec{x}))$ over $\A$. The procedure views $\q$ as a directed tree (we pick one of its variables $z_0$ as a root) and constructs a homomorphism from $\q(\avec{x})$ to~$\canmod$ on-the-fly by traversing the tree from root to leaves. 
The set $\frontier$ is initialised with a pair $z_0\mapsto u_0$  representing the choice of where to map $z_0$.
The possible choices for $z_0$ include $\ind(\Amc)$ and $aw \in \Delta^{\canmod}$ such that  \mbox{$|w| \leq 2|\Tmc|+|\q|$}, which are enough to find a homomorphism if it exists~\cite{ACKZ09}. This set of possible choices is denoted by $U$ in Algorithm~\ref{algo:tree-entail}. Note that $U$ occurs only in statements  of the form `\Guess $u \in U$' and need not be materialised. Instead, we assume that the sequence $u$ is guessed element-by-element and the condition $u\in U$ is verified along the sequence of guesses. We use the subroutine call \canMap{$z_0$, $u_0$} to check whether the guessed $u_0$ is compatible with $z_0$.\!\footnote{The operator \Check{} immediately returns $\false$ if the condition is not satisfied.}\ It first ensures that, if $z_0$ is an answer variable of $\q(\avec{x})$, then $u_0$ is the individual constant corresponding to $z_0$ in $\avec{a}$. Next, if $z_0  \in \ind(\A)$, then it verifies that $u_0$ satisfies all atoms in $\q(\avec{x})$ that involve only~$z_0$. If $u_0  \not \in \ind(\A)$, then $u_0$ must take the form $a w \varrho$ and the subroutine checks whether $\T \models \exists y\, \varrho(y,x) \rightarrow A(x)$ (equivalently, $a w \varrho \in A^{\canmod}$) for every $A(z_0) \in \q$ and whether $\T\models P(x,x)$ for every $P(z_0,z_0)\in \q$.  The remainder of the algorithm consists of a while loop, in which we remove $z\mapsto u$ from $\frontier$, and if $z$ is not a leaf node, guess where to map its children. We must then check that the guessed element $u'$ for child $z'$ is compatible with (\emph{i}) the binary atoms linking $z$ to $z'$ and (\emph{ii}) the atoms that involve only $z'$; the latter is done by \canMap{$z'$, $u'$}. If the check succeeds, we add  $z' \mapsto u'$ to $\frontier$, for each child $z'$ of $z$; otherwise, $\false$ is returned. We exit the while loop when $\frontier$ is empty, i.e., when an element of $\canmod$ has been assigned to each variable in~$\q(\avec{x})$.

\begin{algorithm}[t]\SetAlgoVlined
\KwData{a tree-shaped OMQ $(\T,\q(\avec{x}))$, a data instance $\A$ and 
a tuple $\avec{a}$ from $\ind(\A)$}
\KwResult{\true{} if $\T,\Amc\models\q(\avec{a})$ and \false{} otherwise}
\BlankLine
fix a directed tree $T$ compatible with the Gaifman graph of $\q$ and let $z_0$ be its root\; 
let $U = \bigl\{aw \in\Delta^{\canmod} \mid a\in \ind(\A) \text{ and }  |w| \leq 2|\Tmc|+|\q|\bigr\}$\tcc*[r]{not computed}
\Guess{$u_0 \in U$}\tcc*[r]{use the definition of $U$ to check whether the guess is allowed}
\Check{\canMap{$z_0$,$u_0$}}\;
$\frontier \assign \{z_0 \mapsto u_0 \}$\;
\While{$\frontier \ne \emptyset$}{
remove some $z \mapsto u$ from $\frontier$\;
\ForEach{child $z'$ of $z$ in $T$}{%
\Guess{$u' \in U$}\tcc*[r]{use the def. of $U$ to check whether the guess is allowed}
\Check{$(u,u')\in P^{\C_{\T,\A}}$, for all $P(z,z')\in\q$, \KwAnd \canMap{$z'$,$u'$}}\;
$\frontier \assign \frontier \cup \{ z' \mapsto u' \}$
}
}
\Return \true\;
\BlankLine
\func{\canMap{$z$, $u$}}{
\lIf{$z$ is the $i$th answer variable \KwAnd $u\ne a_i$}{\Return \false}
\uIf(\tcc*[f]{the element $u$ is in the tree part of the canonical model}){$u = aw\varrho$}{%
\Check{$\T\models \exists y\,\varrho(y,x)\to A(x)$, for all $A(z)\in \q$, \KwAnd $\T\models P(x,x)$, for all $P(z,z)\in \q$}}
\Else(\tcc*[f]{otherwise, $u\in \ind(\A)$}){%
\Check{$u\in A^{\canmod}$, for all $A(z)\in \q$, \KwAnd $(u,u)\in P^{\canmod}$, for all $P(z,z)\in \q$}
}
\Return \true;
}
\caption{Non-deterministic procedure \bbbdalgo{}
for answering tree-shaped OMQs}\label{algo:tree-entail}
\end{algorithm}

Correctness and termination of the algorithm are straightforward and hold for tree-shaped OMQs with arbitrary ontologies. Membership in \NL{} for bounded-depth ontologies and bounded-leaf queries follows from the fact that the number of leaves of $\q$ does not exceed  $\ell$, in which case the cardinality of $\frontier$ is bounded by $\ell$, and the fact that the depth of $\T$ does not exceed~$d$, in which case every element of $U$ requires only a fixed amount of space to store. So, since each variable $z$ can be stored in logarithmic space, the set $\frontier$ can also be stored in logarithmic space. Finally, it should be clear that the subroutine \canMap{$z$, $u$} can also be implemented in \NL~\cite{ACKZ09}.
\end{proof}

\subsection{OMQs with bounded-leaf CQs}

It remains to settle the complexity of answering OMQs with arbitrary ontologies and bounded-leaf CQs, for which neither the upper bounds from the preceding subsection nor the \NP{} lower bound by~\cite{DBLP:conf/dlog/KikotKZ11} are applicable. 

\begin{theorem}\label{logcfl-c-arb}
For any fixed $\ell\geq 2$, answering OMQs with tree-shaped CQs having at most $\ell$ leaves is \LOGCFL-complete.
\end{theorem} 
\begin{proof} 
First, we establish the upper bound using a characterisation of the class \LOGCFL{} in 
terms of non-deterministic auxiliary pushdown automata (NAuxPDAs). 
An NAuxPDA~\cite{DBLP:journals/jacm/Cook71} is a non-deterministic Turing machine with an additional work tape constrained
to operate as a pushdown store. \cite{sudborough78} showed that \LOGCFL{} coincides with the class of problems that can be solved by NAuxPDAs running in logarithmic space and polynomial time (note that the space on the pushdown tape is not subject to the logarithmic space bound). Algorithms~\ref{algo:bbqueries} and~\ref{algo:bbqueries:2} give a procedure \bbarbalgo{} for answering OMQs with bounded-leaf CQs that can be implemented by an NAuxPDA.

\begin{algorithm}[t]\SetAlgoVlined%
\KwData{a bounded-leaf OMQ $(\T,\q(\avec{x}))$, a data instance $\A$ and
a tuple $\avec{a}$ from $\ind(\A)$}
\KwResult{\true{} if $\T,\Amc\models\q(\avec{a})$ and \false{} otherwise}
\BlankLine
fix a directed tree $T$ compatible with the Gaifman graph of $\q$ and let $z_0$ be its root\; 
\Guess{$a_0 \in\ind(\A)$}\tcc*[r]{guess the ABox element}
\Guess{$n_0 < 2|\Tmc|+|\q|$}\tcc*[r]{maximum distance from ABox of relevant elements}
\ForEach(\tcc*[f]{guess the initial element in a step-by-step fashion} ){$n$ in $1,\dots,n_0$}{
\Guess{a role $\varrho$ in $\T$ \KwSuchThat \isGenerated{$\varrho$, $a_0$, $\topof(\stack)$}}\;
push $\varrho$ on $\stack$}
\Check{\canMapTail{$z_0$, $a_0$, $\topof(\stack)$}}\;%{\Return \false}
$\frontier\assign \bigl\{(z_0 \mapsto (a_0, |\stack|), z_i) \mid z_i \text{ is a child of } z_0 \text{ in } \T \bigr\}$\;
\While{$\frontier \ne \emptyset$}{
\Guess{one of the 4 options}\;
\uIf(\tcc*[f]{take a step in $\ind(\A)$}){Option 1}{
remove some $(z \mapsto (a,0), z')$ from $\frontier$\;
\Guess{$a' \in \ind(\A)$}\; 
\Check{$(a,a')\in P^{\C_{\T,\A}}$, for all $P(z,z')\in\q$, \KwAnd \canMapTail{$z'$, $a'$, $\varepsilon$}}\;%{\Return \false}
$\frontier \assign \frontier \cup \{ (z' \mapsto (a',0),z_i')\mid z_i' \text{ is a child of } z' \text{ in } T \}$
}
\uElseIf(\tcc*[f]{a step `forward' in the tree part}){Option 2 \KwAnd $|\stack| < 2|\Tmc|+|\q|$}{
remove some $(z \mapsto (a,|\stack|), z')$ from $\frontier$\;
\Guess{a role $\varrho$ in $\T$ \KwSuchThat \isGenerated{$\varrho$, $a$, $\topof(\stack)$}}\;
push $\varrho$ on $\stack$\;
\Check{$\T\models \varrho(x,y)\to P(x,y)$, for all $P(z,z')\in \q$, \KwAnd \canMapTail{$z'$, $a$, $\topof(\stack)$}}\;
$\frontier \assign \frontier \cup \{ (z' \mapsto (a,|\stack|),z_i')\mid z_i' \text{ is a child of } z' \text{ in } T \}$
}
\uElseIf(\tcc*[f]{take a step `backward' in the tree part}){Option 3 \KwAnd $|\stack| > 0$}{
let  $\deepest = \{(z \mapsto (a,n), z') \in \frontier \mid n = |\stack| \}$\tcc*[r]{may be empty}
remove all $\deepest$ from $\frontier$\; 
pop $\varrho$ from $\stack$\;
\ForEach{$(z \mapsto (a,n), z') \in \deepest$}{
\Check{$\T\models \varrho(x,y) \to P(x,y)$, for all $P(z',z)\in\q$, \KwAnd \canMapTail{$z'$, $a$, $\topof(\stack)$}}\;
$\frontier \assign \frontier \cup \{ (z' \mapsto (a,|\stack|),z_i')\mid z_i' \text{ is a child of } z' \text{ in } T \}$
}
}
\uElseIf(\tcc*[f]{take a `loop'-step in the tree part of $\canmod$}){Option 4}{ 
remove some $(z \mapsto (a,|\stack|), z')$ from $\frontier$\;
\Check{$\T\models P(x,x)$, for all $P(z,z')\in\q$,  \KwAnd \canMapTail{$z'$, $a$, $\topof(\stack)$}}\;
$\frontier \assign \frontier \cup \{ (z' \mapsto (a,|\stack|),z_i')\mid z_i' \text{ is a child of } z' \text{ in } T \}$
}
\lElse{\Return \false}
}
\Return \true\;
\caption{Non-deterministic procedure \bbarbalgo{} 
for answering bounded-leaf OMQs.}
\label{algo:bbqueries}
\end{algorithm}

\begin{algorithm}[t]\SetAlgoVlined%
\func{\canMapTail{$z$, $a$, $\sigma$}}{
\lIf{$z$ is the $i^{\text{th}}$ answer variable \KwAnd either $a\ne a_i$ or $\sigma\ne\varepsilon$}{\Return \false}
\uIf(\tcc*[f]{an element of the form $a\ldots\sigma$ in the tree part}){$\sigma \ne\varepsilon$}{%
\Check{$\T\models \exists y\,\sigma(y,x)\to A(x)$, for all $A(z)\in \q$, \KwAnd $\T\models P(x,x)$, for all $P(z,z)\in\q$}
}
\Else(\tcc*[f]{otherwise, in $\ind(\A)$}){%
\Check{$a\in A^{\canmod}$, for all $A(z)\in \q$, \KwAnd $(a,a)\in P^{\canmod}$, for all $P(z,z)\in\q$}
}
\Return \true;
}
\BlankLine
\func{\isGenerated{$\varrho$, $a$, $\sigma$}}{
\uIf(\tcc*[f]{an element of the form $a\ldots\sigma$ in the tree part}){$\sigma \ne\varepsilon$}{%
\Check{$\T\models \exists y\,\sigma(y,x)\to \exists y\,\varrho(x,y)$}
}
\Else(\tcc*[f]{otherwise, in $\ind(\A)$}){%
\Check{$(a,b)\in \varrho(x,y)^{\canmod}$, for some $b\in\Delta^{\canmod}\setminus\ind(\A)$}
}
\Return \true;
}
\caption{Subroutines for \bbarbalgo{}.}
\label{algo:bbqueries:2}
\end{algorithm}

Similarly to \bbbdalgo, the idea is to view the input CQ $\q(\avec{x})$ as a tree and iteratively construct a homomorphism from $\q(\avec{x})$ to $\canmod$, working from root to leaves. We begin by guessing an element $a_0w$ to which the root variable $z_0$ is mapped and checking that $a_0w$ is compatible with $z_0$. However, instead of storing directly $a_0w$ in $\frontier$, we guess it element-by-element and push the word $w$ onto the stack, $\stack$. We assume that we have access to the top of the $\stack$, denoted by $\topof(\stack)$, and  the call $\topof(\stack)$ on empty $\stack$ returns $\varepsilon$. During execution of \bbarbalgo, the height of the stack will never exceed $2|\T| + |\q|$, and so we assume that the height of the stack, denoted by $|\stack|$, is also available as, for example, a variable whose value is updated by the push and pop operations on $\stack$.  

After having guessed $a_0w$, we check that $z_0$ can be mapped to it, which is done by calling \canMapTail{$z_0$, $a_0$, $\topof(\stack)$}. If the check succeeds, we initialise $\frontier$ to the set of $4$-tuples of the form $(z_0 \mapsto (a_0, |\stack|), z_i)$, for all children $z_i$ of $z_0$ in $T$. 
Intuitively, a tuple $(z \mapsto (a,n),z')$ records that the variable $z$ is mapped to the element $a \,\stack_{\leq n}$ and that the child $z'$ of $z$ remains to be mapped (in the explanations we use $\stack_{\leq n}$ to refer to the word comprising the first $n$ symbols of $\stack$; the algorithm, however, cannot make use of it). 

In the main loop, we remove one or more tuples from $\frontier$, choose where to map the variables and update $\frontier$ and $\stack$ accordingly. There are four options. Option~1 is used for tuples $(z \mapsto (a,0),z')$ where both $z$ and $z'$ are mapped to individual constants, Option~2 (Option~3) for tuples $(z\mapsto(a,n),z')$ in which we map $z'$ to a child (respectively, parent) of the image of $z$ in $\canmod$, while Option~4 applies when $z$ and $z'$ are mapped to the same element (which is possible if $P(z,z')\in \q$, for some $P$ that is reflexive according to~$\T$). Crucially, however, the order in which tuples are treated matters due to the fact that several tuples `share' the single stack. Indeed, when applying Option~3, we pop a symbol from $\stack$, and may therefore lose some information that is needed for processing other tuples. To avoid this, 
Option~3 may only be applied to tuples \mbox{$(z \mapsto (a,n),z')$} with maximal $n$, and it must be applied to \emph{all} such tuples at the same time. For Option~2, we  require that the selected tuple $(z\mapsto(a,n), z')$ is such that $n=|\stack|$:
since $z'$ is being mapped to an element $a \, \stack_{\leq n} \,\varrho$, we need to access the $n$th symbol in $\stack$ to determine the possible choices for $\varrho$ and to record the symbol chosen by pushing it onto $\stack$. 

The procedure terminates and returns \true{} when $\frontier$ is empty, meaning that we have successfully constructed a homomorphism witnessing that the input tuple is an answer. Conversely, given a homomorphism from $\q(\avec{a})$ to $\canmod$, we can define a successful execution of \bbarbalgo. We prove in Appendix~\ref{app:complexity}  that  \bbarbalgo{} terminates (Proposition~\ref{logcfl-upper-prop:termination}), is correct (Proposition~\ref{logcfl-upper-prop:correctness}) 
and can be implemented by an NAuxPDA (Proposition \ref{nauxpda}).
The following example illustrates the  construction.  
 
\begin{example}\label{ex-second-algo}
Suppose $\T$ has the following axioms:
\begin{align*}
A(x) & \rightarrow \exists y\, P(x,y),  & 
P(x,y) & \rightarrow U(y,x),  \\ 
\exists y \, P(y,x) & \rightarrow \exists y \,S(x,y), &
\exists y\, S(y,x) & \rightarrow \exists y\, T(y,x), &
\exists y\, P(y,x) & \rightarrow B(x).
 \end{align*}
the query is as follows:
\begin{multline*}
\q(x_1,x_2) \ \ = \ \ \exists y_1y_2y_3y_4y_5\, \bigl(R(y_2,x_1) \ \land \ P(y_2,y_1) \ \land \ S(y_1,y_3) \ \land {}\\T(y_5,y_3)  \ \land \ S(y_4,y_3) \ \land \ U(y_4,x_2)  \bigr)
\end{multline*}
and $\Amc = \{A(a), R(a,c)\}$.
Observe that $\Cmc_{\Tmc, \Amc} \models \q(c,a)$. We show how to define an execution of \bbarbalgo{} that returns \true{} on $((\Tmc,\q), \Amc, \q, (c,a))$ and the homomorphism it induces.
We fix some variable, say $y_1$, as the root of the query tree. We then guess the constant~$a$ and the word $P$, push $P$ onto $\stack$ and check using \mbox{\canMapTail{$y_1$, $a$, $P$}} that our choice is compatible with $y_1$.  
At the start of the while loop, we have
\begin{equation*}\tag{\textsf{w-1}}
\frontier = \{(y_1\mapsto (a,1),y_2),(y_1\mapsto(a,1), y_3)\} \ \ \ \text{ and } \ \ \ \stack = P,
\end{equation*}
where the first tuple, for example, records that $y_1$ has been mapped to $a\, \stack_{\leq 1} = aP$ and %the edge between $y_1$ and 
$y_2$ remains to be mapped. 
We are going to use Option~3 for $(y_1\mapsto(a,1),y_2)$ and  Option~2 for $(y_1\mapsto(a,1), y_3)$. 
We (have to) start with Option~2 though: we remove $(y_1\mapsto (a,1), y_3)$ from $\frontier$, guess $S$, push it onto $\stack$, and add $(y_3\mapsto(a,2), y_4)$ and $(y_3\mapsto(a,2),y_5)$ to $\frontier$. Note that the tuples in $\frontier$ allow us to  read off the elements $a\, \stack_{\leq 1}$ and $a\, \stack_{\leq 2}$ to which $y_1$ and $y_3$ are mapped. Thus, 
\begin{equation*}\tag{\textsf{w-2}}
\frontier = \{(y_1\mapsto(a,1),y_2),(y_3\mapsto(a,2),y_4),(y_3\mapsto(a,2),y_5)\} \ \ \ \text{ and } \ \ \ \stack= PS
\end{equation*}
at the start of the second iteration of the while loop. We are going to use Option~3 for $(y_3\mapsto(a,2),y_4)$ and Option~2 for $(y_3\mapsto(a,2),y_5)$. Again, we have to start with Option~2: we remove $(y_3\mapsto(a,2),y_5)$ from $\frontier$, and guess $T^-$ and push it onto $\stack$. As $y_5$ has no children, we leave $\frontier$ unchanged. At the start of the third iteration, 
\begin{equation*}\tag{\textsf{w-3}}
\frontier = \{(y_1\mapsto(a,1),y_2),(y_3\mapsto(a,2),y_4)\}  \ \ \ \text{ and } \ \ \ \stack= PST^-;
\end{equation*}
see Fig.~\ref{ex-fig}~(a). We apply Option~3 and, since $\deepest = \emptyset$, we pop $T^-$ from  $\stack$ but make no other changes. In the fourth iteration, we again apply Option~3. Since \mbox{$\deepest = \{(y_3\mapsto(a,2),y_4)\}$}, we remove this tuple from
$\frontier$ and pop $S$ from $\stack$. As the checks succeed for $S$, we add $(y_4\mapsto(a,1), x_2)$ to $\frontier$. 
Before the fifth iteration, 
\begin{equation*}\tag{\textsf{w-5}}
\frontier = \{(y_1\mapsto(a,1),y_2),(y_4\mapsto(a,1),x_2)\}   \ \ \ \text{ and } \ \ \ \ \stack = P;
\end{equation*}
see Fig.~\ref{ex-fig}~(b). We apply Option 3 with $\deepest = \{(y_1\mapsto(a,1),y_2),(y_4\mapsto(a,1),x_2)\}$. This leads to both tuples being removed
from $\frontier$ and $P$ popped from $\stack$. We next perform the required checks and, in particular,  verify that the choice of where to map the answer variable $x_2$ agrees with the input vector $(c,a)$ (which is indeed the case).  Then, we add $(y_2\mapsto(a,0), x_1)$ to $\frontier$. The final, sixth, iteration begins with 
\begin{equation*}\tag{\textsf{w-6}}
\frontier = \{(y_2\mapsto(a,0),x_1)\} \ \ \ \text{ and } \ \ \ \ \stack= \varepsilon;
\end{equation*}
see Fig.~\ref{ex-fig}~(c). We choose Option 1,  remove $(y_2\mapsto (a,0),x_1)$ from $\frontier$, guess $c$, and perform the required compatibility checks. As $x_1$ is a leaf, no new tuples are added to $\frontier$; see Fig.~\ref{ex-fig}~(d).  We are thus left with $\frontier= \emptyset$, and return \true{}. 
\begin{figure}[t]%
\centering% 
\begin{tikzpicture}[yscale=0.9,xscale=1]\footnotesize
\node at (-0.5,3.8) {\small (a)};
\node[point,label=left:{$y_1$}, label=above:$B$] (y1) at (1,3) {};
\node at (1.25,3.8) {$\q(x_1,x_2)$};
\node[lrgpoint,label=left:{$y_2$}] (y2) at (0,2) {};
\node[point,label=right:{$y_3$}] (y3) at (2,2) {};
\node[bpoint,label=left:{$x_1$}] (x1) at (0,1) {};
\node[lrgpoint,label=left:{$y_4$}] (y4) at (1.5,1) {};
\node[point,label=below:{$y_5$}] (y5) at (2.5,1) {};
\node[bpoint,label=left:{$x_2$}] (x2) at (1.5,0) {};
\begin{scope}\footnotesize
\draw[->,query] (y2) to node[above,sloped,pos=0.4] {$P$} (y1);
\draw[->,query] (y1) to node[below,sloped] {$S$} (y3);
\draw[->,query] (y2) to node[left] {$R$} (x1);
\draw[->,query] (y4) to node[pos=0.3,above,sloped] {$S$} (y3);
\draw[->,query] (y5) to node[pos=0.3,above,sloped] {$T$} (y3);
\draw[->,query] (y4) to node[left] {$U$} (x2);
\end{scope}
\begin{scope}[xshift=25mm]
\node[bpoint, label=left:{$A$}, label=above:{$a$}] (a) at (1,3) {};
\node at (1.75,3.8) {$\canmod$};
\node[point, label=above:{$c$}] (c) at (2.5,3) {};
\node[point, label=right:{\scriptsize $aP$}, label=left:{$B$}] (d1) at (1,2) {};
\node[point, label=right:{\scriptsize $aPS$}] (d2) at (1,1) {};
\node[point, label=right:{\scriptsize $aPST^-$}] (d3) at (1,0) {};
\draw[->,can] (a) to node[above] {$R$} (c);
\draw[->,can] (a)  to node [right]{$P, U^-$} (d1); %R, 
\draw[->,can] (d1)  to node [right]{$S$} (d2);
\draw[->,can] (d2)  to node [right]{$T^-$} (d3);
\draw [line width=1mm] (2,0.7) -- ++(1,0);
\node[rectangle,fill=black,minimum width=7mm] at (2.5,1) {\textcolor{white}{$\boldsymbol{P}$}}; 
\node[rectangle,fill=black,minimum width=7mm] at (2.5,1.5) {\textcolor{white}{$\boldsymbol{S}$}}; 
\node[rectangle,fill=black,minimum width=7mm] at (2.5,2) {\textcolor{white}{$\boldsymbol{T^-}$}}; 
\end{scope}
\draw[hom] (y5) to node[above,midway,sloped,circle,fill=black,inner sep=1pt] {\sffamily\bfseries\scriptsize\textcolor{white}{2}}  (d3);
\draw[hom] (y3) to node[above,midway,sloped,circle,fill=black,inner sep=1pt] {\sffamily\bfseries\scriptsize\textcolor{white}{1}} (d2);
\draw[hom] (y1) to node[above,midway,sloped,circle,fill=black,inner sep=1pt] {\sffamily\bfseries\scriptsize\textcolor{white}{0}} (d1);
%\draw[hom] (y2) -- (a);
%\draw[hom] (x1) -- (c);
%\draw[hom] (y4) -- (d1);
%\draw[hom] (x2) -- (a);
%
\begin{scope}[xshift=70mm]
\node at (-0.5,3.8) {\small (b)};
\node[point,label=left:{$y_1$}, label=above:$B$] (y1) at (1,3) {};
\node at (1.25,3.8) {$\q(x_1,x_2)$};
\node[lrgpoint,label=left:{$y_2$}] (y2) at (0,2) {};
\node[point,label=right:{$y_3$}] (y3) at (2,2) {};
\node[bpoint,label=left:{$x_1$}] (x1) at (0,1) {};
\node[point,label=left:{$y_4$}] (y4) at (1.5,1) {};
\node[point,label=below:{$y_5$}] (y5) at (2.5,1) {};
\node[lrgbpoint,label=left:{$x_2$}] (x2) at (1.5,0) {};
\begin{scope}\footnotesize
\draw[->,query] (y2) to node[above,sloped,pos=0.4] {$P$} (y1);
\draw[->,query] (y1) to node[below,sloped] {$S$} (y3);
\draw[->,query] (y2) to node[left] {$R$} (x1);
\draw[->,query] (y4) to node[pos=0.3,above,sloped] {$S$} (y3);
\draw[->,query] (y5) to node[pos=0.25,below,sloped] {$T$} (y3);
\draw[->,query] (y4) to node[left] {$U$} (x2);
\end{scope}
\begin{scope}[xshift=25mm]
\node[bpoint, label=left:{$A$}, label=above:{$a$}] (a) at (1,3) {};
\node at (1.75,3.8) {$\canmod$};
\node[point, label=above:{$c$}] (c) at (2.5,3) {};
\node[point, label=right:{\scriptsize $aP$}, label=left:{$B$}] (d1) at (1,2) {};
\node[point, label=right:{\scriptsize $aPS$}] (d2) at (1,1) {};
\node[point, label=right:{\scriptsize $aPST^-$}] (d3) at (1,0) {};
\draw[->,can] (a) to node[above] {$R$} (c);
\draw[->,can] (a)  to node [right]{$P, U^-$} (d1); %R, 
\draw[->,can] (d1)  to node [right]{$S$} (d2);
\draw[->,can] (d2)  to node [right]{$T^-$} (d3);
\draw [line width=1mm] (2,0.7) -- ++(1,0);
\node[rectangle,fill=black,minimum width=7mm] at (2.5,1) {\textcolor{white}{$\boldsymbol{P}$}}; 
%\node[rectangle,fill=black,minimum width=7mm] at (2.5,1) {\textcolor{white}{$S$}}; 
%\node[rectangle,fill=black,minimum width=7mm] at (2.5,0.5) {\textcolor{white}{$T^-$}}; 
%
\end{scope}
%
%
%\draw[hom] (y5) -- (d3);
\draw[hom] (y3) to node[below,pos=0.7,sloped,circle,fill=gray,inner sep=1pt] {\sffamily\bfseries\scriptsize\textcolor{white}{1}} (d2);
%\draw[hom] (y1) -- (d1);
%\draw[hom] (y2) -- (a);
%\draw[hom] (x1) -- (c);
\draw[hom] (y4) to node[above,pos=0.7,sloped,circle,fill=black,inner sep=1pt] {\sffamily\bfseries\scriptsize\textcolor{white}{4}} (d1);
%\draw[hom] (x2) -- (a);
\end{scope}
\begin{scope}[xshift=0mm,yshift=-42mm]
\node at (-0.5,3.3) {\small (c)};
\node[point,label=left:{$y_1$}, label=above:$B$] (y1) at (1,3) {};
%\node at (1.25,3.8) {$\q(x_1,x_2)$};
\node[point,label=left:{$y_2$}] (y2) at (0,2) {};
\node[point,label=right:{$y_3$}] (y3) at (2,2) {};
\node[lrgbpoint,label=left:{$x_1$}] (x1) at (0,1) {};
\node[point,label=left:{$y_4$}] (y4) at (1.5,1) {};
\node[point,label=below:{$y_5$}] (y5) at (2.5,1) {};
\node[bpoint,label=left:{$x_2$}] (x2) at (1.5,0) {};
\begin{scope}\footnotesize
\draw[->,query] (y2) to node[above,sloped,pos=0.4] {$P$} (y1);
\draw[->,query] (y1) to node[below,sloped] {$S$} (y3);
\draw[->,query] (y2) to node[left] {$R$} (x1);
\draw[->,query] (y4) to node[pos=0.3,above,sloped] {$S$} (y3);
\draw[->,query] (y5) to node[pos=0.3,above,sloped] {$T$} (y3);
\draw[->,query] (y4) to node[left] {$U$} (x2);
\end{scope}
\begin{scope}[xshift=25mm]
\node[bpoint, label=left:{$A$}, label=above:{$a$}] (a) at (1,3) {};
%\node at (1.75,3.8) {$\canmod$};
\node[point, label=above:{$c$}] (c) at (2.5,3) {};
\node[point, label=right:{\scriptsize $aP$}, label=left:{$B$}] (d1) at (1,2) {};
\node[point, label=right:{\scriptsize $aPS$}] (d2) at (1,1) {};
\node[point, label=right:{\scriptsize $aPST^-$}] (d3) at (1,0) {};
\draw[->,can] (a) to node[above] {$R$} (c);
\draw[->,can] (a)  to node [right]{$P, U^-$} (d1); %R, 
\draw[->,can] (d1)  to node [right]{$S$} (d2);
\draw[->,can] (d2)  to node [right]{$T^-$} (d3);
\draw [line width=1mm] (2,0.7) -- ++(1,0);
%\node[rectangle,fill=black,minimum width=7mm] at (2.5,1.5) {\textcolor{white}{$P$}}; 
%\node[rectangle,fill=black,minimum width=7mm] at (2.5,1) {\textcolor{white}{$S$}}; 
%\node[rectangle,fill=black,minimum width=7mm] at (2.5,0.5) {\textcolor{white}{$T^-$}}; 
%
\end{scope}
%
%
%\draw[hom] (y5) -- (d3);
%\draw[hom] (y3) -- (d2);
\draw[hom] (y2) to node[above,pos=0.75,sloped,circle,fill=black,inner sep=1pt] {\sffamily\bfseries\scriptsize\textcolor{white}{5}} (a);
\draw[hom] (y1) to node[above,pos=0.25,sloped,circle,fill=gray,inner sep=1pt] {\sffamily\bfseries\scriptsize\textcolor{white}{0}} (d1);
%\draw[hom] (x1) -- (c);
\draw[hom] (y4) to node[below,pos=0.75,sloped,circle,fill=gray,inner sep=1pt] {\sffamily\bfseries\scriptsize\textcolor{white}{4}} (d1);
\draw[hom] (x2) to node[above,pos=0.85,sloped,circle,fill=black,inner sep=1pt] {\sffamily\bfseries\scriptsize\textcolor{white}{5}} (a);
\end{scope}
\begin{scope}[xshift=70mm,yshift=-42mm]
\node at (-0.5,3.3) {\small (d)};
\node[point,label=left:{$y_1$}, label=above:$B$] (y1) at (1,3) {};
%\node at (1.25,3.8) {$\q(x_1,x_2)$};
\node[point,label=left:{$y_2$}] (y2) at (0,2) {};
\node[point,label=right:{$y_3$}] (y3) at (2,2) {};
\node[bpoint,label=left:{$x_1$}] (x1) at (0,1) {};
\node[point,label=left:{$y_4$}] (y4) at (1.5,1) {};
\node[point,label=below:{$y_5$}] (y5) at (2.5,1) {};
\node[bpoint,label=left:{$x_2$}] (x2) at (1.5,0) {};
\begin{scope}\footnotesize
\draw[->,query] (y2) to node[above,sloped,pos=0.4] {$P$} (y1);
\draw[->,query] (y1) to node[below,sloped] {$S$} (y3);
\draw[->,query] (y2) to node[left] {$R$} (x1);
\draw[->,query] (y4) to node[pos=0.3,above,sloped] {$S$} (y3);
\draw[->,query] (y5) to node[pos=0.3,above,sloped] {$T$} (y3);
\draw[->,query] (y4) to node[left] {$U$} (x2);
\end{scope}
\begin{scope}[xshift=25mm]
\node[bpoint, label=left:{$A$}, label=above:{$a$}] (a) at (1,3) {};
%\node at (1.75,3.8) {$\canmod$};
\node[point, label=above:{$c$}] (c) at (2.5,3) {};
\node[point, label=right:{\scriptsize $aP$}, label=left:{$B$}] (d1) at (1,2) {};
\node[point, label=right:{\scriptsize $aPS$}] (d2) at (1,1) {};
\node[point, label=right:{\scriptsize $aPST^-$}] (d3) at (1,0) {};
\draw[->,can] (a) to node[above] {$R$} (c);
\draw[->,can] (a)  to node [right]{$P, U^-$} (d1); %R, 
\draw[->,can] (d1)  to node [right]{$S$} (d2);
\draw[->,can] (d2)  to node [right]{$T^-$} (d3);
\draw [line width=1mm] (2,0.7) -- ++(1,0);
%\node[rectangle,fill=black,minimum width=7mm] at (2.5,1.5) {\textcolor{white}{$P$}}; 
%\node[rectangle,fill=black,minimum width=7mm] at (2.5,1) {\textcolor{white}{$S$}}; 
%\node[rectangle,fill=black,minimum width=7mm] at (2.5,0.5) {\textcolor{white}{$T^-$}}; 
%
\end{scope}
%
%
%\draw[hom] (y5) -- (d3);
%\draw[hom] (y3) -- (d2);
%\draw[hom] (y1) -- (d1);
\draw[hom] (y2) to node[above,pos=0.7,sloped,circle,fill=gray,inner sep=1pt] {\sffamily\bfseries\scriptsize\textcolor{white}{5}} (a);
\draw[hom] (x1) to node[above,pos=0.6,sloped,circle,fill=black,inner sep=1pt] {\sffamily\bfseries\scriptsize\textcolor{white}{6}} (c);
%\draw[hom] (y4) -- (d1);
%\draw[hom] (x2) -- (a);
\end{scope}
\end{tikzpicture}%
\caption{Partial homomorphisms from a  tree-shaped CQ $\q(x_1,x_2)$ to the canonical model 
$\Cmc_{\Tmc, \Amc}$ and the contents of $\stack$ in Example~\ref{ex-second-algo}: (a) before the third iteration,  (b) before the fifth iteration, (c) before and (d) after the final (sixth) iteration. Large nodes indicate the last component of the tuples in $\frontier$.}
\label{ex-fig}
\end{figure}
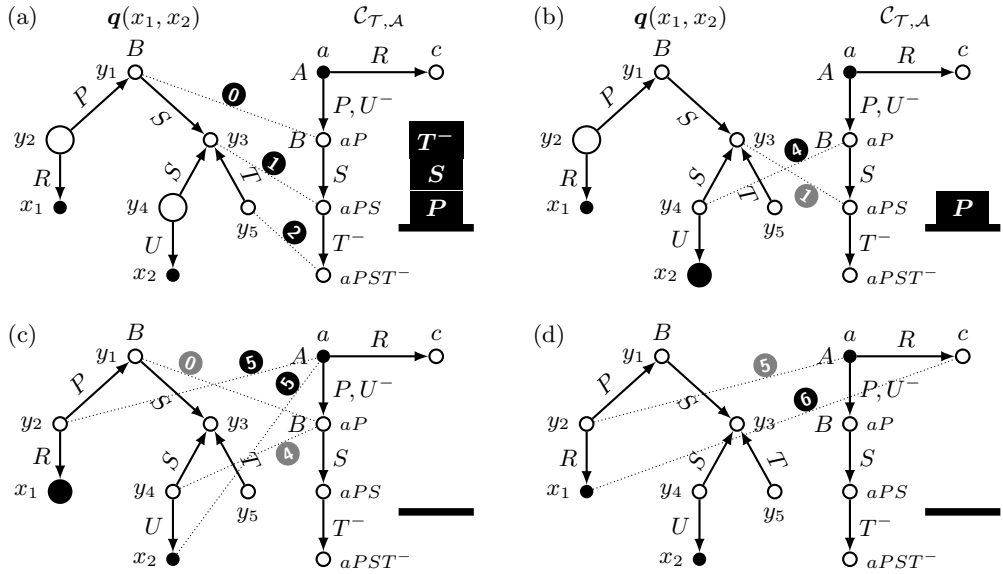%
\end{example}

The proof of \LOGCFL-hardness is by reduction of the following problem:  decide whether an input of length $n$ is accepted by the $n$th circuit of a \emph{logspace-uniform} family of $\SAC^1$ circuits, which is known to be \LOGCFL-hard~\cite{DBLP:journals/jcss/Venkateswaran91}. This problem was used by \cite{DBLP:journals/jacm/GottlobLS01} to show \LOGCFL-hardness of evaluating tree-shaped CQs. We follow a similar approach, but with one crucial difference: using an ontology, we `unravel' the circuit into a tree, which allows us to replace tree-shaped CQs by linear ones. 
Following~\cite{DBLP:journals/jacm/GottlobLS01}, we assume without loss of generality that the considered $\SAC^1$ circuits adhere to the following \emph{normal form}: 
\begin{nitemize}
\item fan-in of all $\AND$-gates is 2;

\item nodes are assigned to levels, with gates on level $i$ only receiving inputs from gates on level $i-1$, the input gates on level 1 and the output gate on the greatest level;

\item the number of levels is odd, all even-level gates are $\OR$-gates, and all odd-level non-input gates are $\AND$-gates.
\end{nitemize}
It is well known~\cite{DBLP:journals/jacm/GottlobLS01,DBLP:journals/jcss/Venkateswaran91} that a circuit in normal form 
accepts an input $\avec{\alpha}$ iff there is a labelled rooted tree (called a \emph{proof tree}) such that
\begin{nitemize}
\item the root node is labelled with the output $\AND$-gate;

\item if a node is labelled with an $\AND$-gate $g_i$ and $g_i = g_j \ANDOP g_k$,
then it has two children labelled with $g_j$ and $g_k$, respectively;

\item if a node is labelled with an $\OR$-gate $g_i$ and $g_i = g_{j_1} \OROP \dots \OROP g_{j_k}$,
then it has a unique child that is labelled with one of $g_{j_1},\dots,g_{j_k}$;

\item every leaf node is labelled with an input gate whose literal evaluates to 1 under $\avec{\alpha}$.
\end{nitemize}
For example, the circuit in Fig.~\ref{fig:5a}~(a) accepts $(1,0,0,0,1)$, as witnessed by the proof tree in Fig.~\ref{fig:5a}~(b). 
While a circuit-input pair may admit multiple proof trees, they are all isomorphic modulo the labelling. Thus, with every circuit $\Cir$, we can associate a \emph{skeleton proof tree} $T$ such that $\Cir$ accepts $\avec{\alpha}$ iff some labelling of $T$ is a proof tree for $\Cir$ and~$\avec{\alpha}$. Note that $T$ depends only on the number of levels in $\Cir$. The reduction~\cite{DBLP:journals/jacm/GottlobLS01}, which is for presentation purposes reproduced here with minor modifications, encodes $\Cir$ and $\avec{\alpha}$ in the database and uses a Boolean tree-shaped CQ based on the skeleton proof tree. Specifically, the database $\dcx$ uses the gates of $\Cir$ as constants and consists of the following facts:
\begin{align*}
& \leftand(g_j, g_i) \text{ and } \rightand(g_k,g_i), && \text{ for every $\AND$-gate $g_i$ with } g_i = g_j \ANDOP g_k;\\
& \aor(g_{j_1},g_i),\dots,\aor(g_{j_k},g_i), && \text{ for every $\OR$-gate $g_i$ with } g_i = g_{j_1} \OROP \cdots \OROP g_{j_k};\\
& \trueleaf(g_i), && \text{ for every input gate $g_i$ whose value is $1$ under $\avec{\alpha}$}. 
\end{align*}
The CQ $\q$ uses the nodes of $T$ as variables, 
has an atom $\aor(z_j,z_i)$ ($\leftand(z_j,z_i)$, $\rightand(z_j,z_i)$) for every node $z_i$ with unique (left, right) child $z_j$,
and has an atom $\trueleaf(z_i)$ for every leaf node $z_i$.
These definitions guarantee that $\dcx\models \q$ iff $\Cir$ accepts $\avec{\alpha}$; moreover, both $\q$ and $\dcx$ can be constructed by logspace transducers.

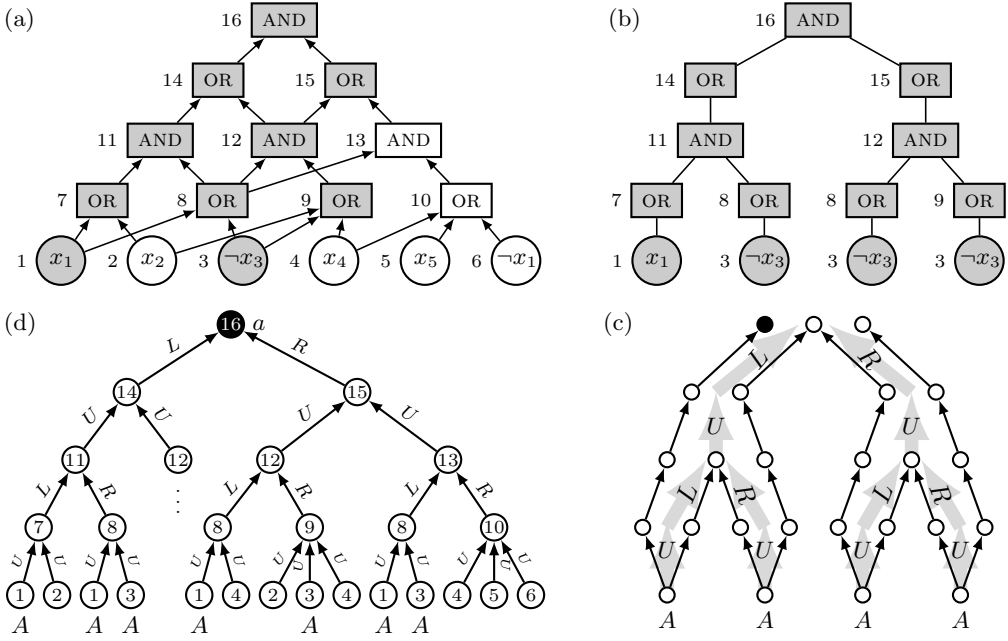
\begin{figure}[t]%
\begin{tikzpicture}[xscale=1.1,yscale=0.8]\small%
%\scriptsize
\node at (-3.2,4) {{\small (a)}};
\node[fill=gray!40,input,label=left:{\scriptsize $1$}] (in1) at (-2.7,0) {$x_1$}; 
\node[input,label=left:{\scriptsize $2$}] (in2) at (-1.6,0) {$x_2$}; 
\node[fill=gray!40,input,label=left:{\scriptsize $3$}] (in3) at (-.5,0) {$\!\!\neg x_3$}; 
\node[input,label=left:{\scriptsize $4$}] (in4) at (.6,0) {$x_4$}; 
\node[input,label=left:{\scriptsize $5$}] (in5) at (1.7,0) {$x_5$}; 
\node[input,label=left:{\scriptsize $6$}] (in6) at (2.8,0) {$\!\!\neg x_1$}; 
%\node[fill=gray!40,input,label=left:$g_7$] (in7) at (3.3,0) {$\!\!\neg x_2$}; 
%
%
\node[fill=gray!40,or-gate,label=left:{\scriptsize $7$}] (or1-1) at (-2.2,1) {$\OR$}; 
\node[fill=gray!40,or-gate,label=left:{\scriptsize $8$}] (or1-2) at (-.75,1) {$\OR$}; 
%\node[or-gate,label=left:$g_{10}$] (or1-3) at (0,1) {OR}; 
\node[fill=gray!40,or-gate,label=left:{\scriptsize $9$}] (or1-4) at (.75,1) {$\OR$}; 
\node[or-gate,label=left:{\scriptsize $10$}] (or1-5) at (2.2,1) {$\OR$}; 
\node[fill=gray!40,and-gate,label=left:{\scriptsize $11$}] (and1-1) at (-1.5,2) {$\AND$};
\node[fill=gray!40,and-gate,label=left:{\scriptsize $12$}] (and1-2) at (0,2) {$\AND$};
\node[and-gate,label=left:{\scriptsize $13$}] (and1-3) at (1.5,2) {$\AND$};
%\node[fill=gray!40,and-gate,label=left:$g_{16}$] (and1-4) at (2.2,2) {AND};
%
\node[fill=gray!40,or-gate,label=left:{\scriptsize $14$}] (or2-1) at (-.8,3) {$\OR$}; 
\node[fill=gray!40,or-gate,label=left:{\scriptsize $15$}] (or2-2) at (.8,3) {$\OR$}; 
\node[fill=gray!40,and-gate,label=left:{\scriptsize $16$}] (and2) at (0,4) {$\AND$}; 
\begin{scope}[semithick]
\draw[->] (in1) to (or1-1);
\draw[->] (in1) to (or1-2);
\draw[->] (in2) to (or1-1);
\draw[->] (in2) to (or1-4);
\draw[->] (in3) to (or1-2);
\draw[->] (in3) to (or1-4);
\draw[->] (in4) to (or1-4);
\draw[->] (in4) to (or1-5);
\draw[->] (in5) to (or1-5);
\draw[->] (in6) to (or1-5);
\draw[->] (or1-1) to (and1-1);
\draw[->] (or1-2) to (and1-1);
\draw[->] (or1-2) to (and1-3);
\draw[->] (or1-2) to (and1-2);
\draw[->] (or1-4) to (and1-2);
\draw[->] (or1-5) to (and1-3);
\draw[->] (and1-1) to (or2-1);
\draw[->] (and1-2) to (or2-1);
\draw[->] (and1-2) to (or2-2);
\draw[->] (and1-3) to (or2-2);
\draw[->] (or2-1) to (and2);
\draw[->] (or2-2) to (and2);
\end{scope}
\begin{scope}[xshift=-20mm,xscale=1.3]
\node at (4.7,4) {{\small (b)}};
\node[input,fill=gray!40,label=left:{\scriptsize $1$}] (v10) at (5,0)  {$x_1$}; %[label=left:$a^1_3$]{};
\node[or-gate,fill=gray!40,label=left:{\scriptsize $7$}] (v6) at (5,1)    {$\OR$}; %[label=left:$b^1_2$]{};
\node[input,fill=gray!40,label=left:{\scriptsize $3$}] (v11) at (6,0)   {$\!\!\neg x_3$}; % [label=left:$a^2_3$]{};
\node[or-gate,fill=gray!40,label=left:{\scriptsize $8$}] (v7) at (6,1)    {$\OR$}; %[label=left:$b^2_2$]{};
\node[and-gate,fill=gray!40,label=left:{\scriptsize $11$}] (v4) at (5.5,2)  {$\AND$}; % [label=left:$a^1_2$]{};
\node[or-gate,fill=gray!40,label=left:{\scriptsize $14$}] (v2) at (5.5,3)  {$\OR$}; %[label=left:$b^1_1$]{};
\node[and-gate,fill=gray!40,label=left:{\scriptsize $16$}] (v1) at (6.5,4)  {$\AND$}; %[label=left:$a^1_1$]{};
\node[or-gate,fill=gray!40,label=left:{\scriptsize $15$}] (v3) at (7.5,3)  {$\OR$}; %[label=left:$b^2_1$]{};
\node[and-gate,fill=gray!40,label=left:{\scriptsize $12$}] (v5) at (7.5,2)  {$\AND$}; % [label=left:$a^2_2$]{};
\node[or-gate,fill=gray!40,label=left:{\scriptsize $8$}] (v8) at (7,1)   {$\OR$}; % [label=left:$b^3_2$]{};
\node[or-gate,fill=gray!40,label=left:{\scriptsize $9$}] (v9) at (8,1)   {$\OR$}; % [label=left:$b^4_2$]{};
\node[input,fill=gray!40,label=left:{\scriptsize $3$}] (v12) at (7,0) {$\!\!\neg x_3$}; %  [label=left:$a^3_3$]{};
\node[input,fill=gray!40,label=left:{\scriptsize $3$}] (v13) at (8,0)  {$\!\!\neg x_3$}; % [label=left:$a^4_3$]{};
\begin{scope}[semithick]
\draw[-] (v10) to (v6);
\draw[-] (v11) to (v7);
\draw[-] (v6) to (v4);
\draw[-] (v7) to (v4);
\draw[-] (v4) to (v2);
\draw[-] (v2) to (v1);
\draw[-] (v12) to (v8);
\draw[-] (v13) to (v9);
\draw[-] (v8) to (v5);
\draw[-] (v9) to (v5);
\draw[-] (v5) to (v3);
\draw[-] (v3) to (v1);
\end{scope}
\end{scope}
\end{tikzpicture}\\[6pt]
\begin{tikzpicture}[yscale=0.9,
point/.style={circle,draw=black,thick,minimum size=3.5mm,inner sep=0pt,fill=white},
ipoint/.style={circle,draw=black,thick,minimum size=2mm,inner sep=0pt,fill=white}]
\footnotesize
\begin{scope}[xscale=0.7]
\node at (-4,4) {{\small (d)}};
\begin{scope}\scriptsize
\node[point,label=below:{\footnotesize $\trueleaf$}] (in1) at (-4,0) {1}; 
\node[point] (in2) at (-3.3,0) {2}; 
\node[point,label=below:{\footnotesize $\trueleaf$}] (in3c2) at (-2.6,0) {1}; 
\node[point,label=below:{\footnotesize $\trueleaf$}] (in4c2) at (-1.9,0) {3}; 
\node[point,label=below:{\footnotesize $\trueleaf$}] (in3) at (-0.6,0) {1}; 
\node[point] (in4) at (0.1,0) {4}; 
\node[point] (in2d) at (0.8,0) {2}; 
\node[point,label=below:{\footnotesize $\trueleaf$}] (in3d) at (1.5,0) {3}; 
\node[point] (in4d) at (2.2,0) {4}; 
\node[point,label=below:{\footnotesize $\trueleaf$}] (in3c) at (2.9,0) {1}; 
\node[point,label=below:{\footnotesize $\trueleaf$}] (in4c) at (3.6,0) {3}; 
\node[point] (in5) at (4.3,0) {4}; 
\node[point] (in6) at (5,0) {5}; 
\node[point] (in7) at (5.7,0) {6}; 
\node[point] (or1-1) at (-3.65,1) {7}; 
\node[point] (or1-2) at (-2.25,1) {8};
\node[point] (or1-3) at (-0.25,1) {8}; 
\node[point] (or1-4) at (1.5,1) {9}; 
\node[point] (or1-5) at (3.25,1) {8}; 
\node[point] (or1-6) at (5,1) {10}; 
\node[point] (and1-1) at (-2.25-0.7,2) {11};
\node[point] (and1-2) at (-1,2) {12};
\node at (-1, 1.5) {$\vdots$};
\node[point] (and1-3) at (.75,2) {12};
\node[point] (and1-4) at (3.25+0.875,2) {13};
\node[point] (g2) at (-1-1.95/2,3) {14}; 
\node[point] (g3) at (0.75+3.3/2,3) {15}; 
\node[point,fill=black,label=right:{\footnotesize $a$}] (a) at (0,4) {\textcolor{white}{16}}; 
\end{scope}
\begin{scope}\scriptsize
\draw[->,can] (g2)  to node [above,sloped]{$L$} (a);
\draw[->,can] (g3)  to node [above,sloped]{$R$} (a);
\draw[->,can] (and1-1)  to node [above,sloped]{$U$} (g2);
\draw[->,can] (and1-2)  to node [above,sloped]{$U$} (g2);
\draw[->,can] (and1-3)  to node [above,sloped]{$U$} (g3);
\draw[->,can] (and1-4)  to node [above,sloped]{$U$} (g3);
\draw[->,can] (or1-1)  to node [above,sloped,pos=0.4]{$L$} (and1-1);
\draw[->,can] (or1-2)  to node [above,sloped,pos=0.4]{$R$} (and1-1);
\draw[->,can] (or1-3)  to node [above,sloped,pos=0.4]{$L$} (and1-3);
\draw[->,can] (or1-4)  to node [above,sloped,pos=0.4]{$R$} (and1-3);
\draw[->,can] (or1-5)  to node [above,sloped,pos=0.4]{$L$} (and1-4);
\draw[->,can] (or1-6)  to node [above,sloped,pos=0.4]{$R$} (and1-4);
\begin{scope}\tiny
\draw[->,can] (in1)  to node [above,sloped,pos=0.4]{$U$} (or1-1);
\draw[->,can] (in2)  to node [above,sloped,pos=0.4]{$U$} (or1-1);
\draw[->,can] (in3c2)  to node [above,sloped,pos=0.4]{$U$} (or1-2);
\draw[->,can] (in4c2)  to node [above,sloped,pos=0.4]{$U$} (or1-2);
\draw[->,can] (in3c)  to node [above,sloped,pos=0.4]{$U$} (or1-5);
\draw[->,can] (in4c)  to node [above,sloped,pos=0.4]{$U$} (or1-5);
\draw[->,can] (in2d)  to node [above,sloped,pos=0.4]{$U$} (or1-4);
\draw[->,can] (in3d)  to node [above,sloped,pos=0.4]{$U$} (or1-4);
\draw[->,can] (in4d)  to node [above,sloped,pos=0.4]{$U$} (or1-4);
\draw[->,can] (in3)  to node [above,sloped,pos=0.4]{$U$}  (or1-3);
\draw[->,can] (in4)  to node [above,sloped,pos=0.4]{$U$} (or1-3);
\draw[->,can] (in5)  to node [above,sloped,pos=0.4]{$U$} (or1-6);
\draw[->,can] (in6)  to node [above,sloped,pos=0.4]{$U$} (or1-6);
\draw[->,can] (in7)  to node [above,sloped,pos=0.4]{$U$} (or1-6);
\end{scope}
\end{scope}
\end{scope}
\begin{scope}[xshift=-7mm,xscale=1.3]
\node at (4.5,4) {{\small (c)}};
%\node[fill=black] (a) at (1,3) [point, label=above:{$A$}, label=left:{$a$}]{};
%\node[fill=black] (c) at (2.5,3) [point, label=right:{$c$}]{};
%\node (d1) at (1,2) [point, fill=white, label=left:{$aP$}, label=right:{$B$}]{};
%
\begin{scope}[line width=1.8mm,gray!30] 
\draw[->] (5,0) to node[pos=0.6] {\color{black}$\aor$} (5,1.1);
\draw[->] (6,0)  to node[pos=0.6] {\color{black}$\aor$} (6,1.1);
\draw[->] (5,1)  to node[midway,sloped] {\normalsize\color{black}$\leftand$} (5.4,2);
\draw[->] (6,1)  to node[midway,sloped] {\normalsize\color{black}$\rightand$} (5.6,2);
\draw[->] (7,0)  to node[pos=0.6] {\color{black}$\aor$} (7,1.1);
\draw[->] (8,0)  to node[pos=0.6] {\color{black}$\aor$} (8,1.1);
\draw[->] (7,1)  to node[midway,sloped] {\normalsize\color{black}$\leftand$} (7.4,2);
\draw[->] (8,1)  to node[midway,sloped] {\normalsize\color{black}$\rightand$} (7.6,2);
\draw[->] (5.5,2)  to node[pos=0.5] {\color{black}$\aor$} (5.5,3);
\draw[->] (7.5,2)  to node[pos=0.5] {\color{black}$\aor$} (7.5,3);
\draw[->] (5.5,3)  to node[midway,sloped] {\normalsize\color{black}$\leftand$} (6.35,4);
\draw[->] (7.5,3)  to node[midway,sloped] {\normalsize\color{black}$\rightand$} (6.65,4);
\end{scope}
\begin{scope}\footnotesize
\node[ipoint,fill=black] (u1) at (6,4) {};
\node[ipoint] (u2) at (5.25,3) {};
\node[ipoint] (u3) at (5,2) {};
\node[ipoint] (u4) at (4.75,1) {};
\node[ipoint,label=below:{$\trueleaf$}] (u5) at (5,0) {};
\node[ipoint] (u6) at (5.25,1) {};
\node[ipoint] (u7) at (5.5,2) {};
\node[ipoint] (u8) at (5.75,1) {};
\node[ipoint,label=below:{$\trueleaf$}] (u9) at (6,0) {};
\node[ipoint] (u10) at (6.25,1) {};
\node[ipoint] (u11) at (6,2) {};
\node[ipoint] (u12) at (5.75,3) {};
\node[ipoint] (u13) at (6.5,4) {};
\node[ipoint] (u14) at (7.25,3) {};
\node[ipoint] (u15) at (7,2) {};
\node[ipoint] (u16) at (6.75,1) {};
\node[ipoint,label=below:{$\trueleaf$}] (u17) at (7,0) {};
\node[ipoint] (u18) at (7.25,1) {};
\node[ipoint] (u19) at (7.5,2) {};
\node[ipoint] (u20) at (7.75,1) {};
\node[ipoint,label=below:{$\trueleaf$}] (u21) at (8,0) {};
\node[ipoint] (u22) at (8.25,1) {};
\node[ipoint] (u23) at (8,2) {};
\node[ipoint] (u24) at (7.75,3) {};
\node[ipoint] (u25) at (7,4) {};
\end{scope}
\begin{scope}[thick]\normalsize
\draw[<-] (u1) to (u2);
\draw[<-] (u2) to (u3);
\draw[<-] (u3) to node[left] {} (u4);
\draw[<-] (u4) to  (u5);
\draw[->] (u5) to  (u6);
\draw[->] (u6) to (u7);
\draw[<-] (u7) to  (u8);
\draw[<-] (u8) to (u9);
\draw[->] (u9) to (u10);
\draw[->] (u10) to (u11);
\draw[->] (u11) to (u12);
\draw[->] (u12) to (u13);
\draw[<-] (u13) to (u14);
\draw[<-] (u14) to  (u15);
\draw[<-] (u15) to (u16);
\draw[<-] (u16) to (u17);
\draw[->] (u17) to (u18);
\draw[->] (u18) to  (u19);
\draw[<-] (u19) to  (u20);
\draw[->] (u21) to (u20);
\draw[->] (u21) to (u22);
\draw[->] (u22) to (u23);
\draw[->] (u23) to (u24);
\draw[->] (u24) to (u25);
\end{scope}
\end{scope}
\end{tikzpicture}
\caption{
(a) A circuit $\Cir$ of 5 levels with input $\avec{\alpha}\colon x_1\mapsto 1, \  x_2\mapsto 0, \ x_3\mapsto 0,\  x_4\mapsto 0, \  x_5\mapsto 0$ (the gate number is indicated on the left and gates with value $1$ under $\avec{\alpha}$ are shaded); (b) a proof tree for $\Cir$ and $\avec{\alpha}$; (c) CQs $\q$ (thick gray arrows) and $\qclin$ (black arrows);
(d) canonical model of $(\kbC)$ with the subscript of $G_i$ inside the nodes.}
\label{fig:5a}
\end{figure}

To adapt this reduction to our setting, we replace $\q$ by a linear CQ $\qclin$, which is obtained by a depth-first traversal of $\q$. When evaluated on $\dcx$, the CQs $\qclin$ and $\q$ may give different answers, but the answers coincide if the CQs are evaluated on the \emph{unravelling of $\dcx$ into a tree}. Thus, we define $(\kbC)$ whose canonical model 
induces a tree isomorphic to the unravelling of $\dcx$. 
To formally introduce $\qclin$,
consider the sequence of words  defined inductively as follows:
\begin{equation*}
w_0 = \varepsilon\quad\text{ and }\quad  w_{j+1} = \leftand^-\,\aor^-\,w_j\,\aor\,\leftand\,\rightand^-\,\aor^-\,w_j\,\aor\,\rightand, \text{ for } j > 0.
\end{equation*}
Suppose $\Cir$ has $2d+1$ levels, $d \geq 0$. Consider the $d$th word $w_d =  \varrho_1 \varrho_2 \dots \varrho_k$ and take
\begin{equation*}
\qclin(y_0) \ \  \ = \ \ \ \exists y_1, \dots, y_k\,\Bigl[\ \ \bigwedge_{i=1}^{k} \varrho_i(y_{i-1}, y_i) \ \ \ \land\bigwedge_{
\varrho_i\varrho_{i+1} =\aor^-\,\aor} \hspace*{-1.5em}\trueleaf(y_i)\ \Bigr];
\end{equation*}
see Fig.~\ref{fig:5a}~(c). 
We now define $(\kbC)$.
Suppose $\Cir$ has gates $g_1, \dots, g_m$, with $g_m$ 
the output gate. In addition to predicates $\aor$, $\leftand$, $\rightand$, $\trueleaf$, we introduce  
a unary predicate $G_i$ for each gate $g_i$.
We set $\Amc = \{G_m(a)\}$ and include  the following axioms in $\Tmc_{\avec{\alpha}}$:
\begin{align*}
& G_{i}(x) \rightarrow \exists y\, \bigl(S(x,y)\land G_j(y)\bigr), && \text{ for every } S(g_j,g_i) \in \dcx, \ S\in \{\aor, \leftand,\rightand\},\\
& G_i(x) \rightarrow \trueleaf(x), && \text{ for every } \trueleaf(g_i) \in \dcx;
\end{align*}
see Fig.~\ref{fig:5a}~(d) for an illustration. 
When restricted to predicates $\aor$, $\leftand$, $\rightand$, $\trueleaf$, the canonical model of $(\kbC)$ is isomorphic to the unravelling of $\dcx$ starting from $g_m$.

We show in Appendix~\ref{app:logcfl-hardness} that $\qclin$ and $(\kbC)$ can be constructed by logspace transducers (Proposition~\ref{prop:f.2}), and that $\Cir$ accepts $\avec{\alpha}$ iff $\Tmc_{\avec{\alpha}}, \Amc \models \qclin(a)$ (Proposition~\ref{prop:f.3}). 
\end{proof}
%******************

\section{Conclusions and open problems}\label{sec:conclusions}

Our aim in this work was to understand how the size of OMQ rewritings and the combined complexity of OMQ answering depend on (\emph{i}) the existential depth of \OWLQL{} ontologies, (\emph{ii}) the treewidth of CQs or the number of leaves in tree-shaped CQs, and (\emph{iii}) the type of rewriting: PE, NDL or arbitrary FO. 

We tackled  the succinctness problem by representing OMQ rewritings as (Boolean) hypergraph functions and  establishing an unexpectedly tight correspondence between the size of OMQ rewritings and the size of various computational models for computing these  functions.
It turned out that polynomial-size \emph{PE-rewritings} can only be constructed for OMQs with ontologies of depth 1 and CQs of bounded treewidth. Ontologies of larger depth require, in general, PE-rewritings of super-polynomial size. 
The good and surprising news, however, is that, for classes of OMQs with ontologies of bounded depth and CQs of bounded treewidth, we can always (efficiently) construct polynomial-size \emph{NDL-rewritings}. The same holds if we consider OMQs obtained by pairing ontologies of depth 1 with arbitrary CQs or coupling arbitrary ontologies with bounded-leaf queries; see Fig.~\ref{pic:results} for details. 
The existence of polynomial-size \emph{FO-rewritings} for different classes of OMQs was shown to be equivalent to major open problems in computational and circuit complexity such as `$\NL/\poly \subseteq \smash{\NC^1}?$'\!, `$\LOGCFL/\poly \subseteq \smash{\NC^1}?$' and `$\NP/\poly \subseteq \smash{\NC^1}?$'

We also determined the combined complexity of answering OMQs from the considered classes. In particular, we showed that OMQ answering is tractable---either \NL- or \LOGCFL-complete---for  bounded-depth ontologies coupled with bounded treewidth CQs, as well as for arbitrary ontologies paired with tree-shaped queries with a bounded number of leaves. We point out that membership in \LOGCFL{} implies that answering OMQs from the identified tractable classes can be `profitably parallelised' (for details, consult~\cite{DBLP:journals/jacm/GottlobLS01}).

Comparing the two sides of Fig.~\ref{pic:results}, we remark that the class of tractable OMQs nearly coincides with the OMQs admitting polynomial-size \NDL-rewritings (the only exception being OMQs with ontologies of depth 1 and arbitrary CQs). However, the \LOGCFL{} and \NL{} membership results cannot be immediately inferred from the existence of polynomial-size \NDL-rewritings, since evaluating polynomial-size NDL-queries is a \PSpace-complete problem in general. In fact, 
much more work is required to construct NDL-rewritings that can be evaluated in \LOGCFL{} and \NL{}, which will be done in a follow-up publication; see technical report~\cite{MeghynDL16}.

Although the present work gives comprehensive solutions to the succinctness and combined complexity problems formulated in Section~\ref{intro}, it also raises some interesting and challenging questions: 
\begin{enumerate}[(1)]
\item What is the size of rewritings of OMQs with a \emph{fixed ontology}?

\item What is the size of rewritings of OMQs with ontologies in a \emph{fixed signature}?

\item Is answering OMQs with CQs of bounded treewidth and ontologies of finite depth fixed-parameter tractable if the \emph{ontology depth} is the parameter? 

\item What is the size of rewritings for OMQs whose ontologies do not contain role inclusions, that is, axioms of the form $\varrho(x,y) \to \varrho'(x,y)$?
\end{enumerate}
Answering these questions would provide further insight into the difficulty of OBDA and could lead to the identification of new  classes of well-behaved OMQs.  

As far as practical OBDA is concerned, our experience with the query answering engine Ontop~\cite{ISWC13,DBLP:conf/semweb/KontchakovRRXZ14}, which employs the tree-witness rewriting, shows that mappings and database constraints together with semantic query optimisation techniques can drastically reduce the size of rewritings and produce efficient SQL queries over the data. The role of mappings and data constraints in OBDA 
is yet to be fully investigated~\cite{DBLP:conf/kr/Rodriguez-MuroC12,DBLP:conf/esws/Rosati12,DBLP:conf/semweb/LemboMRST15,DBLP:conf/dlog/BienvenuR15}
and constitutes another promising avenue for future work.

Finally, the focus of this paper was on the ontology language \OWLQL{} that has been designed specifically for OBDA via query rewriting. However, in practice ontology designers often require constructs that are not available in \OWLQL. Typical examples are axioms such as $A(x) \to B(x) \lor C(x)$ and $P(x,y) \land A(y) \to B(x)$. The former is a standard covering constraint in conceptual modelling, while the latter occurs in ontologies such as SNOMED CT. 
There are at least two ways of extending the applicability of rewriting techniques to a wider class of ontology languages.
A first approach relies upon the observation that although many ontology languages do not guarantee the existence of rewritings for all ontology-query pairs,
it may still be the case that the queries and ontologies typically encountered in practice do admit rewritings.
This has motivated the development of diverse methods for identifying particular ontologies and OMQs for which (first-order or Datalog) rewritings exist \cite{DBLP:conf/ijcai/LutzPW11,DBLP:conf/ijcai/BienvenuLW13,DBLP:journals/tods/BienvenuCLW14,DBLP:conf/aaai/KaminskiNG14,DBLP:conf/ijcai/LutzHSW15}. 
A second approach consists in replacing an ontology formulated in a complex ontology language (which lacks efficient query answering algorithms)
by an ontology written in a simpler language, for which query rewriting methods can be employed.
Ideally, one would show that the simpler ontology is equivalent to the original with regards to query answering \cite{BotoevaCSSSX16}, and thus provides the exact set of answers.
Alternatively, one can use a simpler ontology to approximate the answers for the full one \cite{DBLP:conf/semweb/ConsoleMRSS14,BotoevaCSSSX16}
(possibly employing a more costly complete algorithm to decide the status of the remaining candidate answers \cite{DBLP:journals/jair/ZhouGNKH15}).

\appendix

%\medskip

%\section{This is an example of Appendix section head}

\section{Proof of Theorem~\ref{Hom2rew}}\label{app:proof:Th4.5}

\indent

\textsc{Theorem~\ref{Hom2rew}}\ \ {\itshape
\textup{(}i\textup{)} For any OMQ $\omq(\avec{x})$, the formulas $\qtw(\avec{x})$ and $\qtw'(\avec{x})$ are equivalent, and so $\qtw'(\avec{x})$ is a PE-rewriting of $\omq(\avec{x})$ over complete data instances.

\textup{(}ii\textup{)} Theorem~\ref{TW2rew} continues to hold for $\twfn$ replaced by $\homfn$.}

\smallskip

\noindent\begin{proof} 
Let $\omq(\avec{x}) = (\T,\q(\avec{x}))$ and $\q(\avec{x}) = \exists \avec{y}\, \varphi (\avec{x},\avec{y})$.
We begin by showing that for every tree witness $\t$ for $\omq(\avec{x})$, we have the following chain of equivalences: 
\begin{multline*}
\bigwedge_{R(z,z') \in \q_\t} \hspace*{-0.75em}(z=z')\ \  \land \hspace*{-0.5em}\bigvee_{\t \text{ is $\varrho$-initiated}} \,\, 
\bigwedge_{z \in \tr \cup \ti} \hspace*{-0.25em}\varrho^*(z)
\quad \equiv 
\bigwedge_{z,z' \in \tr \cup \ti} \hspace*{-0.75em}(z=z') \ \ \land \hspace*{-0.5em}\bigvee_{\t \text{ is $\varrho$-initiated}}  \,\, \bigwedge_{z \in \tr \cup \ti} \hspace*{-0.25em}\varrho^*(z)
\\
\equiv \ \  
\exists z_0 \, \Bigl (\bigwedge_{z \in \tr \cup \ti} \hspace*{-0.75em} (z=z_0) \ \ \land \bigvee_{\t \text{ is $\varrho$-initiated}}  \hspace*{-1em}\varrho^*(z_0)
\Bigr ) %\label{third-eq}
\quad\equiv \quad 
\exists z_0 \, \Bigl (\bigwedge_{z \in \tr \cup \ti} (z=z_0)\ \ \land \bigvee_{\t \text{ generated by } \tau}\hspace*{-1em}\tau(z_0)\Bigr), 
\end{multline*}
where $z_0$ is a fresh variable.
The first equivalence 
follows from the transitivity of equality and the fact that every pair of variables $z,z'$ in a tree witness must be linked by a sequence of binary atoms. The following equivalence can be readily verified using first-order semantics. For the final equivalence, 
we use the fact that
if~$\t$ is $\varrho$-initiated and $\T\models\tau(x) \to \exists y \,\varrho(x,y)$, then $\t$ is generated by~$\tau$, and conversely, if $\t$~is generated by~$\tau$,
then there is some $\varrho$ that initiates $\t$ and is such that $\T\models\tau(x) \to \exists y\, \varrho(x,y)$. 

By the above equivalences, the query $\qtw'(\avec{x})$ can be equivalently expressed as follows:
\begin{align*}
&\exists \avec{y} \!\!\!\!\!\bigvee_{\begin{subarray}{c}\Theta \subseteq \twset\\ \text{ independent}\end{subarray}}\!\!\!\!\!\!\!
\Big(\bigwedge_{\atom \in \q \setminus \q_\Theta}\!\!\! \rsz \,\,\wedge \bigwedge_{\t \in \Theta} \big(\exists z_0 \, (\bigwedge_{z \in \tr \cup \ti}\hspace*{-0.25em} (z=z_0) \ \ \land \hspace*{-0.5em}\bigvee_{\t \text{ is generated by } \tau}\hspace*{-1em} \tau(z_0)) \big) \Big).
\end{align*}
Finally, we observe that, for every independent $\Theta \subseteq \twset$, the variables that occur in some $\ti$, for $\t \in \Theta$, do not occur in $\ti'$ for any other $\t' \in \Theta$. It follows that if $z \in \ti$ and $\t \in \Theta$, then the only occurrence of $z$ in the disjunct for $\Theta$ is in the equality atom $z = z_0$. We can thus drop all such atoms, 
while preserving equivalence, which gives us precisely the tree-witness rewriting $\q_\tw(\avec{x})$. In particular, this means that $\q_\tw'(\avec{x})$ is a rewriting of $\omq(\avec{x})$ over complete data instances.

To establish the second statement, let $\Phi$ be a Boolean formula that computes
\begin{equation*}
\homfn \ \ = \bigvee_{\substack{\Theta \subseteq \twset
\\ \text{ independent}}}\hspace*{-2mm}
\Bigl(\bigwedge_{\atom \in \q \setminus \q_\Theta} \hspace*{-1em}p_\atom
 \ \ \wedge \ \  \bigwedge_{\t \in \Theta} \big(\bigwedge_{R(z,z')\in \q_\t}\hspace*{-1em} p_{z=z'} \ \ \wedge \bigvee_{\t \text{ is $\varrho$-initiated}} \,\, \bigwedge_{z \in \tr\cup\ti} p_{\varrho^*(z)}\big)\Bigr),
\end{equation*}
 and let $\q'(\avec{x})$ be the FO-formula obtained by replacing each $p_{z=z'}$ in $\Phi$ with $z = z'$, each $p_\atom$ with~$\atom$, each $p_{\varrho^*(z)}$ with $\bigvee_{\T \models \tau(x) \to \exists y\, \varrho(x,y)}\tau(z)$, and prefixing the result with  $\exists \avec{y}$. 
Recall that the modified rewriting $\qtw'(\avec{x})$
was obtained by applying this same transformation to the original monotone Boolean formula for $\homfn$.
Since $\Phi$ computes $\homfn$,
$\q'(\avec{x})$ and $\qtw'(\avec{x})$ are equivalent FO-formulas. As we have already established that $\qtw'(\avec{x})$ is a rewriting of $\omq(\avec{x})$, 
the same must be true of $\q'(\avec{x})$. The statement regarding \NDL-rewritings can be proved 
similarly to the proof of Theorem~\ref{TW2rew}~(\emph{ii}).
\end{proof}

\section{Proof of Theorem~\ref{tree-hg-to-query}}\label{app:proof:Th5.9}

\indent

\textsc{Theorem~\ref{tree-hg-to-query}} \ {\itshape
\textup{(}i\textup{)} Any tree hypergraph $H$ is isomorphic to a subgraph of $\HG{\OMQT{H}}$.

\textup{(}ii\textup{)} Any monotone THGP based on a tree hypergraph $H$ computes a subfunction of the primitive evaluation function $f^\vartriangle_{\OMQT{H}}$.}

\begin{proof}
(i) Fix a tree hypergraph $H = (V, E)$ whose underlying tree $T=(V_T, E_T)$  has vertices 
$V_T = \{1, \dots, n\}$, for $n > 1$, and $1$ is a leaf of $T$.   The \emph{directed} tree obtained from $T$ by fixing $1$ as the root and orienting the edges away from~$1$ is denoted by $T^1 = (V_T,E^1_T)$. 
By definition, each  $e\in E$ induces a convex subtree $T^e = (V_e, E_e)$ of $T^1$. Since, for each subtree $T^e$, the OMQ $\OMQT{H}$ has  a tree-witness $\t^e$ with
\begin{align*}
\tr^e & = \{\, z_i \mid i \text{ in on the boundary of } e \,\},\\
\ti^e & = 
\{\,z_i \mid i \text{ is in the  interior of } e\,\} \cup \{\,y_{ij} \mid (i,j) \in e \, \},
\end{align*}
it follows that $H$ is 
isomorphic to the subgraph of $\HG{\OMQT{H}}$  
obtained by removing all superfluous hyperedges and 
all vertices corresponding to atoms with $S_{ij}$. 

\bigskip

(ii)  Suppose that $P$ is bssed on a tree hypergraph $H$.
Given an input~$\avec{\alpha}$ for $P$, we define an assignment $\avec{\gamma}$ for the 
predicates in $\OMQT{H} = (\T,\q)$ by taking each $\avec{\gamma}(R_{ij})$ and $\avec{\gamma}(S_{ij})$ to be the value of the label of $(i,j) \in E^1_{T}$ under $\avec{\alpha}$  and $\avec{\gamma}(A_e) = 1$ for all $e \in E$ ($\avec{\gamma}(R_\zeta) = 0$, for all normalisation predicates $R_\zeta$). 
We show that for all $\avec{\alpha}$ we have
\begin{equation*}
P(\avec{\alpha}) = 1 \quad \text{iff} \quad \primfnP(\avec{\gamma}) = 1.
\end{equation*}
Observe that the canonical model $\smash{\C_{\T,\A(\avec{\gamma})}}$  contains two labelled nulls, $w_e$ and $w_e'$, for each $e\in E$, satisfying 
\begin{multline*}
\C_{\T,\A(\avec{\gamma})}\models \bigwedge_{(i,j)\in E_e, \ i = r^e}\hspace*{-1em} R_{r^ej}(a,w_e) \ \ \land \bigwedge_{(i,j)\in E_e, \ j\in L_e} \hspace*{-1.5em}S_{ij}(w_e,a) \ \ \ \land\\ 
\bigwedge_{(i,j)\in E_e, \ i \ne r^e} \hspace*{-1.5em}R_{ij}(w'_e,w_e) \ \  \ \land \bigwedge_{(i,j)\in E_e, \ j\notin L_e} \hspace*{-2em} S_{ij}(w_e,w_e').
\end{multline*}

\smallskip

\noindent ($\Rightarrow$)\  Suppose that $P(\avec{\alpha}) = 1$.
Then there exists an independent $E' \subseteq E$ that covers all zeros of $\avec{\alpha}$.
We show $\Tmc, \Amc(\avec{\gamma}) \models  \q$ (that is, $\primfnP(\avec{\gamma})= 1$).
Define a mapping $h$ as follows:
\begin{equation*}
h(z_i)= \begin{cases}
w_e', & \text{if } i \text{ is in the interior of } e \in E', \\
a, &  \text{otherwise},
\end{cases}
\quad
h(y_{ij})= \begin{cases}
w_e, & \text{if }(i, j) \in e \in E', \\
a & \text{otherwise}.
\end{cases}
\end{equation*}
Note that $h$ is well-defined: since $E'$ is independent, its hyperedges share no interior, and there can be at most one hyperedge $e \in E'$ containing any given vertex $(i, j)$.

It remains to show that $h$ is a homomorphism from $\q$ to   $\Cmc_{\Tmc, \Amc(\avec{\gamma})}$.
Consider a pair of atoms $R_{ij}(z_i, y_{ij})$  and $S_{ij}(y_{ij}, z_{j})$ in $\q$.
Then $(i, j) \in \smash{E^1_T}$. 
If there is $e\in E'$ with$(i, j) \in e$ then there are four possibilities to consider:
\begin{nitemize}
\item if 
neither $i$ nor $j$ is in the interior then, since $T_e$ is a tree and $(i,j)$ is its edge,  
the only possibility is $e=\{( i, j )\}$, whence  $h(z_i)=h(z_j)=a$ and $h(y_{ij})=w_e$; 

\item 
$i$ is on the boundary and $j$ is internal, then
$h(z_i)=a$, $h(y_{ij})=w_e$, and $h(z_{j})=w_e'$;

\item 
if $j$ is on the boundary and $i$ is internal, then this case is the mirror image;

\item   
if both $i$ and $j$ are in the interior,  then $h(z_i)=h(z_j)=w_e'$ and $h(y_{ij})=w_e$.
\end{nitemize}
Otherwise, the label of $(i, j)$ must evaluate to 1 under $\avec{\alpha}$, whence $\Amc(\avec{\gamma})$ contains $R_{ij}(a,a)$ and $S_{ij}(a,a)$ and we set $h(z_i)=h(y_{ij})=h(z_j)=a$.
In all cases, $h$ preserves the atoms $R_{ij}(z_i, y_{ij})$ and $S_{ij}(y_{ij}, z_{j})$, and so $h$ is indeed a homomorphism.

\medskip

\noindent ($\Rightarrow$)\  Suppose that $\primfnP(\avec{\gamma}) = 1$. Then $\Tmc, \Amc(\avec{\gamma}) \models \q$, and so there is a homomorphism
$h\colon \q \rightarrow  \Cmc_{\Tmc, \Amc(\avec{\gamma})}$. We show that there is
an independent $E'\subseteq E$ that covers all zeros of~$\avec{\alpha}$.
Let  $E' $ be the set of all $e \in E$ such that $h^{-1}(w_e) \neq \emptyset$ (that is, $w_e$ is in the image of~$h$).
To show that $E'$ is independent, we need the following claim:

\medskip

\noindent\textbf{Claim}. If $h^{-1}(w_e) \neq \emptyset$, %$y_{ij}$ is a variable of $\q$,, 
then $h(y_{ij})= w_e$ for all $(i, j) \in e$. 

\smallskip

\noindent\emph{Proof of claim}. 
%Suppose that $h^{-1}(w_e) \neq \emptyset$. 
Let $r^e$ be the root of $T^e$ and  $L_e$ its leaves.
%where $e = [ {k_1}, \ldots, {k_m} ] \in E'$. We may assume w.l.o.g.\
%that $k_1$ is the highest vertex in $e$ according to $T^1$.
Pick some variable $z \in h^{-1}(w_e)$
such that there is no $z' \in h^{-1}(w_e)$ higher than $z$ in $\q$ (we use the ordering
of variables induced by the tree $T^1$).
Observe that $z$ cannot be of the form~$z_j$, because then $\q$ would contain some atom $R_{j \ell}(z_j, y_{j\ell})$ or $S_{\ell j}(y_{\ell j}, z_j)$, 
but $w_e$ has no outgoing $R_{j \ell}$ or $S_{\ell j}^-$ arcs in $\smash{\Cmc_{\Tmc, \Amc(\avec{\gamma})}}$. 
It follows that $z$ is of the form $y_{j \ell}$, for some $j,\ell$. By considering
the available arcs leaving~$w_e$ again, we conclude that $(j, \ell) \in e$. We next
show that $j= r^e$. Suppose that this is not the case. Then, 
there must be $(p, j) \in e$ with $(p, j) \in T^1$. A simple examination of the axioms in $\Tmc$ shows that the only
way for $h$ to satisfy the atom $R_{j \ell}(z_j, y_{j \ell})$ is to map $z_j$ to $w_e'$.
It follows that to satisfy the atom $S_{p j}(y_{p j}, z_j)$, we must put $h(y_{p j})=w_e$ contrary to the assumption that $z=y_{j \ell}$ was a highest vertex in~$h^{-1}(w_e)$. Thus, $j= r^e$.
Now, using a simple inductive argument on the distance from $z_{r^e}$, and considering the possible ways of
mapping the atoms of $\q$, we can show that
$h(y_{ij})= w_e$ for every $(i, j) \in e$. \hfill(\emph{end proof of claim})

\medskip

\noindent Suppose that there are two distinct hyperedges $e,e' \in E'$ that have a non-empty intersection:
$(i, j) \in e \cap e'$. We know that either $y_{ij}$ or $y_{ji}$ occurs in $\q$, and we can assume the former without loss of generality. 
By the claim, we obtain $h(y_{ij})= w_e = w_{e'}$, a contradiction.
Therefore, $E'$ is independent. We now show that it covers all zeros. 
Let $(i,j)$ be such that its label evaluates to $0$ under $\avec{\alpha}$, and  assume again without loss of generality 
that $y_{ij}$ occurs in  $\q$. Then $\Amc(\avec{\gamma})$ does not contain
$R_{ij}(a,a)$, so the only way $h$ can satisfy the atom $R_{ij}(z_i, y_{ij})$
is by mapping~$y_{ij}$ to some $w_e$ with $(i,j) \in e$. 
It follows that there is an $e \in E'$ such that
$(i,j) \in e$, so all zeros of $\avec{\alpha}$ are covered by $E'$.
We have thus shown that $E'$ is an independent subset of $E$
that covers all zeros of $\avec{\alpha}$, and hence, $P(\avec{\alpha}) = 1$.
\end{proof}

\section{Proof of Proposition~\ref{hyper:thgp}}\label{app:proof:Prop5.10}

\indent

\textsc{Proposition}~\ref{hyper:thgp}.  \ \ \emph{ 
\textup{(}i\textup{)} For any tree hypergraph $H$ of degree~\mbox{$\le d$}, there is a monotone THGP of size $O(|H|)$ that computes $f_H$ and such that its hypergraph is of degree~\mbox{$\le \max(2,d)$}.}

\emph{\textup{(}ii\textup{)} For every generalised THGP $P$ over $n$ variables, there is a THGP $P^\prime$ computing the same function and such that $|P'| \le n \cdot |P|$.}

\begin{proof}
(\emph{i})  
Consider a hypergraph $H = (V,E)$ based on a tree $T=(V_T, E_T)$ with $V = E_T$. 
We label each $v \in V$ with a variable $p_v$ and, for each $e \in E$,
we choose some $v_e \in \bigcup e$, add fresh vertices $a_e$ and $b_e$ with edges $\{v_e, a_e\}$ and $\{a_e, b_e\}$ to $T$ as well as a new hyperedge $e' = [v_e, b_e]$ to $E$. We label the segment $[v_e, a_e]$ with $1$ and the segment $[a_e, b_e]$ with $p_e$. We also extend $e$ to include the segment $[v_e, a_e]$.
We claim that the resulting THGP $P$ computes $f_H$. Indeed, for any input $\avec{\alpha}$ with $\avec{\alpha}(p_e) = 0$, we have to include the edge $e'$ into the cover, and so cannot include the edge $e$ itself. Thus, $P(\avec{\alpha}) = 1$ iff there is an independent set $E$ of hyperedges with $\avec{\alpha}(p_e)=1$, for all $e\in E$, covering all zeros of the variables $p_v$. It follows that $P$ computes $f_H$.

\smallskip

(\emph{ii}) Let P be a generalised THGP based on a hypergraph $H = (V,E)$ with the underlying tree  $T=(V_T, E_T)$ such that $V= E_T$.  To construct $P'$, we split every vertex $v\in V$ 
(which is an edge of $T$) labelled with $\bigwedge_{i=1}^{k} \li_i$ into $k$ new edges $v_1, \dots, v_k$ and label $v_i$ with $\li_i$, for $1 \leq i \leq k$; each hyperedge containing $v$ will now contain all the $v_i$. It is easy to see that $P(\avec{\alpha}) = P'(\avec{\alpha})$, for any valuation $\avec{\alpha}$. Since $k \leq n$, we have $|P'| \le n \cdot |P|$. It should be clear that the degree of $P'$ and the number of leaves in it are the same as in $P$.
\end{proof}

\section{Proof of  Theorem~\ref{DL2THP}}\label{app:proof:5.13}

\indent

\textsc{Theorem~\ref{DL2THP}.} \ \ {\itshape
For every OMQ $\omq(\avec{x}) = (\Tmc,\q(\avec{x}))$  with a fundamental set $\Omega_{\omq}$ and with $\q$ of treewidth~$t$,
the generalised monotone THGP $P_{\omq}$ computes $\homfn$ and is of size polynomial in $|\q|$ and $|\Omega_{\omq}|^t$.}

\begin{proof}
By~\cite[Lemma~11.9]{DBLP:series/txtcs/FlumG06}, we can assume that the tree $T$ in the tree decomposition of $\q$ has at most $N$, $N \le |\q|$, nodes. Recall that $M = |\Omega_{\omq}|^t$ is the number of bag types. We claim that $P_{\omq}$
\begin{nitemize}
\item contains at most $(2M +1)N$ vertices and at most $N(M+M^2)$ hyperedges;
\item and has labels with at most $3|\q|$ conjuncts.
\end{nitemize}
The vertices of the hypergraph of $P_{\omq}$ correspond to the edges of $T'$, and there can be at most $N \cdot (2M+1)$ of them, because there can be no more than $N$ edges in $T$, and each is replaced by a sequence of $2M + 1$ new edges.
The hyperedges are of two types: $E^k_i$ (where $1 \leq i \leq N$ and $1 \leq k \leq M$) and $E^{k \ell}_{i j}$ (where $(i, j)$ correspond to an edge in $T$ and $1 \leq k,\ell \leq M$). It follows that the total number of hyperedges cannot exceed $N(M+M^2)$. Finally, a simple examination of the labelling function shows that there can be at most $3|\q|$ conjuncts in each label.
Indeed, given $i$, $j$ and $k$, each atom $S(\avec{z})$ with $\avec{z} \subseteq \lambda(N_i)$ generates either 1 or 3 propositional variables
in the label of $\{u^k_{ij}, v^k_{ij}\}$, and $|\q|$ is the upper bound for the number of such atoms.

%By $f_P^\circ$ we denote the function $f_P(\avec{\sigma}(\avec{v}))$ where $\avec{\sigma}(\avec{v})$ is the valuation of the variables in $f_P$
%obtained by setting $$\avec{\sigma}(p_{\exists y \varrho(z,y)}) = \max \{\avec{v}(p_{\tau(z)}) \mid \Tmc \models \tau(x) \to \exists y \varrho (x,y) \}$$ and
%setting $\avec{\sigma}(p)= \avec{v}(p)$ for all other variables $p$.

\smallskip

To complete the proof, we show that $P_{\omq}$ computes $\homfn$:  for any valuation $\avec{\alpha}$, 
\begin{equation*}
\homfn(\avec{\alpha}) = 1\qquad \text{ iff }\qquad P_{\omq}(\avec{\alpha}) = 1.
\end{equation*}

$(\Rightarrow)$ Let $\avec{\alpha}$ be such that $\homfn(\avec{\alpha})=1$.
Then we can find an independent $\Theta \subseteq \twset$ such that $\avec{\alpha}$ satisfies the
corresponding disjunct of $\homfn$:
\begin{equation}\label{disjunct}
\bigwedge_{\atom \in \q \setminus \q_\Theta} \hspace*{-1em} p_\atom \ \ \ 
 \wedge  \bigwedge_{\t \in \Theta} \Bigl(\bigwedge_{R(z,z')\in \q_\t}\hspace*{-0.5em}p_{z=z'} \,\, \wedge \bigvee_{\t \text{ is $\varrho$-initiated}}
 %\substack{\varrho \text{ strongly} \\ \text{ generates } \t}} 
 \,\, \bigwedge_{z \in \tr\cup\ti}\hspace*{-0.3em} p_{\varrho^*(z)}\Bigr).
\end{equation}
For every $\t \in \Theta$, let $\varrho_\t$ be a role that makes the  disjunction hold.
Since $\t$ is $\varrho_\t$-initiated, we can 
choose a homomorphism $h_\t\colon \q_\t \rightarrow \Cmc^{\varrho_{\t}(a)}_\Tmc$
such that, for every $z \in \ti$,  $h_\t(z)$ is of the form $a \varrho_\t w$, for some $w$ . 

With each node $N$ in the tree decomposition $(T,\lambda)$
we associate the \emph{type $\avec{w}$ of $N$} by taking, for all $z\in \lambda(N)$:
\begin{equation*}
\avec{w}[\nu_N(z)] = \begin{cases}
w, & \text{if } z \in \ti \text{ and } h_\t(z) = a w, \text{ for some } \t \in \Theta,\\
\emptyword, & \text{otherwise}. 
\end{cases}
\end{equation*}
Observe that $\avec{w}$ is well-defined since
the independence of $\Theta$ guarantees that every variable in $\q$ can appear in $\ti$ for
at most one $\t \in \Theta$. 
We show that $\avec{w}$ is compatible with~$N$.
Consider a unary atom \mbox{$A(z) \in \q$}
such that $z\in\lambda(N)$ and $\avec{w}[\nu_N(z)] \neq \emptyword$. Then there must be $\t \in \Theta$
such that $z \in \ti$, in which case  %$\avec{w}_N[i] \neq \emptyword$
$h_\t (z)=a \avec{w}[\nu_N(z)]$.
Let $\varrho$ be the final symbol in~$h_\t (z)$. Since $h_\t\colon \q_\t\to\Cmc^{\exists y \varrho_\t(a,y)}_\Tmc$ is a homomorphism,
we have $\Tmc \models \exists y\, \varrho(y,x) \to A(x)$.
Consider now a binary atom $P(z,z') \in \q$ such that $z,z'\in\lambda(N)$ and either $\avec{w}[\nu_N(z)] \neq \emptyword$
or $\avec{w}[\nu_N(z')] \neq \emptyword$. 
We assume w.l.o.g.\ that the former is true (the other case is handled analogously).
By definition, there is $\t \in \Theta$
such that $z \in \ti$ and $h_\t (z)=a \avec{w}[\nu_N(z)]$.
Since $z \in \ti$ and $P(z,z') \in \q$, by the definition of tree witnesses,
$z' \in \tr \cup \ti$.
Since $h_\t\colon \q_\t\to \Cmc^{\exists y \varrho_\t(a,y)}_\Tmc$ is a homomorphism, one of the following holds:
\begin{nitemize}
%\item $z' \in \tr$, $\avec{w}_N[\nu_N(z')]=\emptyword$ and $\avec{w}_N[\nu_N(z)]=\varrho$ for some $\varrho$ with $\Tmc \models \varrho(y,x) \to  P(x,y)$;
%
\item %$z'\in \ti$, 
$\avec{w}[\nu_N(z')] = \avec{w}[\nu_N(z)]$ and  $\Tmc \models P(x,x)$; 
\item %$z'\in\ti$ and 
$\avec{w}[\nu_N(z)]=\avec{w}[\nu_N(z')] \cdot \varrho$ for some $\varrho$
with $\Tmc \models \varrho(y,x) \to P(x,y)$;
\item %$z' \in \ti$ and 
$\avec{w}[\nu_N(z')]=\avec{w}[\nu_N(z)] \cdot \varrho$ for some $\varrho$
with $\Tmc \models \varrho(x,y) \to P(x,y)$.
\end{nitemize}
This establishes the second part of the compatibility condition.
Next, we show that the pairs associated with different nodes in $T$ are compatible.
Consider a pair of nodes $N$ and $N'$ and their types
$\avec{w}$ and  $\avec{w}'$.
It is clear  
that, by construction,  $\avec{w}[\nu_{N}(z)] = \avec{w}'[\nu_{N'}(z)]$, for all $z \in \lambda(N) \cap \lambda(N')$. 

\smallskip

Let $\avec{w}_1,\dots,\avec{w}_\numtypes$ be all the bag types. Consider now the tree hypergraph $P_{\omq}$, and let~$E'$ be the set consisting of the following hyperedges:
\begin{nitemize}
\item for every $N_i$ in $T$, the hyperedge $E_i^k = [ N_i, u_{ij_1}^k, \ldots, u_{i j_n}^k] $, where $k$ is such that
$\avec{w}_k$ is the type of $N_i$, and $N_{j_1}, \ldots, N_{j_n}$ are the neighbours of $N_i$;
\item for every pair of adjacent nodes $N_i, N_j$ in $T$, the hyperedge
$E_{ij}^{k \ell} = [ v_{ij}^k, v_{ji}^\ell] $, where $k$ and $\ell$ are such that $\avec{w}_k$ and $\avec{w}_\ell$ are the types of $N_i$ and $N_j$, respectively.
\end{nitemize}
Note that all these hyperedges are present in the hypergraph of $P_{\omq}$ because we have shown that
the type of each node $N_i$ is compatible with it and that the pairs of types of $N_i$ and $N_j$ are compatible with the pair $(N_i,N_j)$.
It is easy to see that $E'$ is independent, since whenever we include $E_i^k$ or $E_{ij}^{k \ell}$, we do not include any $E_i^{k'}$ or $E_{ij}^{k' \ell}$ for $k' \neq k$.
It remains to show that every vertex of the hypergraph of $P_{\omq}$  that is not covered by $E'$ evaluates to 1 under~$\avec{\alpha}$. 
Observe first that most of the vertices are covered by $E'$. Specifically: 
\begin{nitemize}
\item $\{N_i, u_{ij}^1\}$ is covered by $E_i^k$;
\item $\{v_{ij}^k, u_{ij}^{k +1}\}$ is  covered either by $E_i^n$ (if $n \leq k+1$) or by 
$E_{ij}^{n \ell}$ (if $n > k +1$);
\item $\{v_{ij}^\numtypes, v_{ji}^\numtypes\}$ is covered by $E_{ij}^{k \ell}$;
\item $\{u_{ij}^k, v_{ij}^k\}$ is covered by $E_i^n$ if $k< n$, and by $E_{ij}^{n \ell}$ if $n > k$.
\end{nitemize}
Thus, the only type of vertex not covered by $E'$ is of the form $\{u_{ij}^k, v_{ij}^k\}$, where $\avec{w}_k$ is the type of $N_i$.
In this case, by definition, $\{u_{ij}^k, v_{ij}^k\}$ is labelled by the following variables:
\begin{nitemize}
\item $p_\atom$, if  $\atom \in \q$,
$\avec{z}\subseteq \lambda(N_i)$
and $\avec{w}_k[\nu_{N_i}(z)]= \emptyword$, for all $z\in\avec{z}$;
\item $p_{\varrho^*(z)}$, if $A(z) \in \q$,  $z \in \lambda(N_i)$ and $\avec{w}_k[\nu_{N_i}(z)]= \varrho w$;
\item $p_{\varrho^*(z)}$, $p_{\varrho^*(z')}$ and $p_{z=z'}$, if 
$S(z,z') \in \q$ (possibly with $z=z'$),  $z,z' \in \lambda(N_i)$
and either $\avec{w}_k[\nu_{N_i}(z)]= \varrho w$ or $\avec{w}_k[\nu_{N_i}(z')]= \varrho w$.
\end{nitemize}
First suppose that $p_\atom$ appears in the label of $\{u_{ij}^k, v_{ij}^k\}$.
Then $\avec{w}_{k}[\nu_{N_i}(z)]= \emptyword$, for all $z\in \avec{z}$, 
and hence there is no variable in $\atom$ that belongs to any $\ti$ for $\t \in \Theta$. 
It follows that $\atom \in \q \setminus \q_\Theta$, and since~\eqref{disjunct} is satisfied,
the variable $p_\atom$ evaluates to 1 under~$\avec{\alpha}$. 
Next suppose that one of $p_{\varrho^*(z)}$, $p_{\varrho^*(z')}$ and $p_{z=z'}$ is part of the label. 
We focus on the case where these variables came from a binary atom (third item above), but the proof is  similar for the case of a unary atom (second item above).
We know that there is some atom $S(z,z') \in \q$ with
 $z,z'\in\lambda(N_i)$ 
and either $\avec{w}_{k}[\nu_{N_i}(z)]= \varrho w$ or $\avec{w}_{k}[\nu_{N_i}]= \varrho w$.
It follows that there is a tree witness $\t \in \Theta$ such that $z,z' \in \tr\cup\ti$.
This means that the atom $p_{z=z'}$ is a conjunct of~\eqref{disjunct}, and so it
is satisfied under $\avec{\alpha}$. 
Also, either $\avec{w}_{k}[\nu_{N_i}(z)]=h_\t(z)$ or $\avec{w}_{k}[\nu_{N_i}(z')]=h_\t(z')$
is of the form $\varrho w$, and, since all non-empty words in the image of $h_\t$ begin by $\varrho_\t$, 
we obtain $\varrho = \varrho_\t$. Since $\varrho_\t$ was chosen so that 
$\bigwedge_{z \in \tr\cup\ti} p_{\varrho^*(z)}$ is satisfied under $\avec{\alpha}$, 
both $p_{\varrho^*(z)}$ and $p_{\varrho^*(z')}$ evaluate to~1 under~$\avec{\alpha}$. 
Therefore, $E'$ is independent and covers all zeros under $\avec{\alpha}$,
which means that $P_{\omq}(\avec{\alpha})=1$. 

\bigskip

($\Leftarrow$) Suppose $P_{\omq}(\avec{\alpha})=1$, i.e., there is an independent subset $E'$
of the hyperedges in~$P_{\omq}$ that covers all vertices evaluated to~$0$ under $\avec{\alpha}$.
It is clear from the construction of $P_{\omq}$ that $E'$ contains exactly one hyperedge
of the form $E_i^k$ for every node $N_i$ in $T$, and so we can associate with every node $N_i$ the unique index $\mu(N_i) = k$.
We also know that $E'$ contains exactly one hyperedge of the form
$E_{ij}^{k \ell}$ for every edge $\{N_i, N_j\}$ in~$T$. Moreover, if we have hyperedges
$E_i^k$ and $E_{ij}^{\smash{k'} \ell}$ (respectively, $E_j^\ell$ and $E_{ij}^{k \smash{\ell'}}$), then $k=k'$ (respectively, $\ell=\ell'$).
It also follows from the definition of  $P_{\omq}$ that
every $\avec{w}_{\mu(N_i)}$ is compatible with~$N_i$, and 
pairs $(\avec{w}_{\mu(N_j)},\avec{w}_{\mu(N_j)})$ are compatible for adjacent nodes $N_i$, $N_j$.
Using the compatibility properties and the connectedness condition of tree
decompositions, we can conclude that the pairs assigned to \emph{any two nodes} $N_i$ and $N_j$
in $T$ are compatible. Since every variable must appear in at least one node label,
it follows that we can associate a \emph{unique} word $w_z$ with every variable $z$ in $\q$. 

Since all zeros are covered by $E'$, we know that for every node $N_i$,
the following variables are assigned to $1$ by $\avec{\alpha}$:
\begin{nitemize}
\item $p_\atom$, if  $\atom \in \q$,
$\avec{z} \subseteq \lambda(N_i)$
and $w_z= \emptyword$, for $z\in \avec{z}$;
\item $p_{\varrho^*(z)}$, if $A(z) \in \q$, $z \in \lambda(N_i)$, and $w_z= \varrho w$; \hfill ($\star$)
\item $p_{\varrho^*(z)}$, $p_{\varrho(z')}$ and $p_{z=z'}$, 
if  $S(z,z') \in \q$ (possibly with $z=z'$),  
 $z,z' \in\lambda(N_i)$ and either $w_z= \varrho w$ or $w_{z'}= \varrho w$.
\end{nitemize}
Now let $\equiv$ be the smallest equivalence relation on the atoms of $\q$ that satisfies
the following condition, for every  variable $z$ in $\q$, 
\begin{equation*}
\text{if } w_z \neq \emptyword \text{ and } z \text{ occurs in both } S_1(\avec{z}_1) \text{ and } S_2(\avec{z}_2), \text{ then } S_1(\avec{z}_1)\equiv S_2(\avec{z}_2). 
\end{equation*}
Let $\q_1, \ldots, \q_n$ be the subqueries corresponding to the equivalence classes of $\equiv$. It is easily verified that the $\q_i$
are pairwise disjoint. Moreover, if $\q_i$ contains only variables~$z$ with $w_z= \emptyword$, then $\q_i$ consists of a single atom.
We can show that the remaining $\q_i$ correspond to tree witnesses.

\smallskip

\noindent\textbf{Claim}. For every $\q_i$ that contains a variable $z$ with $w_z \neq \emptyword$:
\begin{enumerate}[(1)]
\item there is a role $\varrho_{i}$ such that every $w_z \neq \emptyword$ (with $z$ a variable in $\q_i$) begins by $\varrho_{i}$;
\item there is a homomorphism $h_{i}\colon \q_i\to \Cmc^{\exists y \varrho_{i}(a,y)}_\Tmc$ such that $h_i(z)=a w_z$ for every variable $z$ in $\q_i$;
\item there is a tree witness $\t^i$ for $\omq$ that is $\varrho_i$-initiated and such that $\q_i= \q_{\t^i}$
\end{enumerate}
\emph{Proof of claim}.
By the definition of $\q_i$, there exists a sequence $Q_0, \ldots, Q_n$ of subsets of~$\q$
such that $Q_0 = \{S_0(\avec{z}_0)\} \subseteq \q_i$ contains a variable $z_0$ with $w_{z_0}\neq \emptyword$,
$Q_n= \q_i$, and for every $0 \leq \ell < n$,
$Q_{\ell +1}$ is obtained from $Q_\ell$ by adding an atom %$\atom \in \q \setminus Q_\ell$
that contains a variable $z$ that appears in $Q_\ell$ and is such that $w_z \neq \emptyword$.
By construction, every atom in $\q_i$ contains a variable $z$ with $w_z \neq \emptyword$.
Let $\varrho_i$ be the first letter of the word $w_{z_0}$, and for every $0 \leq \ell \leq n$,
let $h_\ell$ be the function that maps every variable $z$ in $Q_\ell$ to~$a w_z$.

Statements 1 and 2 can be shown by induction. The base case is trivial.
For the induction step, suppose that at stage $\ell$, we know that every variable $z$ in $Q_\ell$ with
$w_z \neq \emptyword$ begins by $\varrho_i$, and that $h_\ell$ is a homomorphism of $Q_\ell$ into
the canonical model $\Cmc^{\exists y \varrho_{i}(a,y)}_\Tmc$ that satisfies $h_\ell(y)=a w_z$. % $C_{\Tmc_i, \Amc_i}$.
We let $\atom$ be the unique atom in $Q_{\ell+1} \setminus Q_\ell$.
Then $\atom$ contains a variable $z$ that appears in $Q_\ell$ and is such that
$w_z \neq \emptyword$. 
If  $\atom= B(z)$ or $\atom=R(z,z)$, then Statement~1 for $w_z$ is immediate.
For Statement~2, we let $N$ be a node in $T$ such that 
 $z \in \lambda(N)$.
 Since $\avec{w}_N$ is compatible with $N$, 
 it follows that, if $\atom= B(z)$, then 
 $w_z$ ends by a role $\varrho$ with $\Tmc \models \exists y\, \varrho(y,x) \to B(x)$, 
 and, if $\atom=R(z,z)$, then $\Tmc \models  R(x,x)$, 
which proves Statement~2.
Next, consider the case when $\atom$ contains two variables, that is, it is of the form
$R(z,z')$ or $R(z,z')$.
We give the argument for the former (the the latter is analogous).
Let $N$ be a node in $T$ such that $\{z,z'\} \subseteq \lambda(N)$.
Since $\avec{w}_N$ is compatible with $N$,
either
\begin{nitemize}
\item $w_{z'}= w_{z} \varrho$ with $\Tmc \models \varrho(x,y) \to R(x,y)$, or
\item $w_{z} = w_{z'} \varrho$ with $\Tmc \models \varrho(y,x) \to R(x,y)$.
\end{nitemize}
Since $w_z$ begins with $\varrho_i$, the same holds for $w_{z'}$ unless $w_{z'} = \emptyword$,
which proves Statement~1. It is also clear from the way we defined
$h_{\ell+1}$ that it is homomorphism from $Q_{\ell+1}$ to $\Cmc^{\exists y \varrho_{i}(a,y)}_\Tmc$, so Statement~2 holds.

Statement 3 now follows from Statements~1 and~2, the definition of $\q_i$ and the definition of tree witnesses.\hfill
\emph{(end proof of claim)}
\smallskip

Let $\Theta$ consist of all the tree witnesses $\t^i$ obtained in the~claim.
As the $\q_i$ are disjoint, the set $\{\q_{\t^i} \mid \t^i \in \Theta\}$ is independent.
We show that $\avec{\alpha}$ satisfies the disjunct of $\homfn$ that corresponds to $\Theta$; cf.~\eqref{disjunct}.
First, consider some $\atom \in \q \setminus \q_\Theta$. Then, for every variable~$z$ in $\atom$,
we have $w_z = \emptyword$.
Let $N$ be a node in $T$ such that $\avec{z} \subseteq \lambda(N)$.
Then $w_z= \emptyword$, for all $z \in \avec{z}$.
It follows from $(\star)$ that $\avec{\alpha}(p_\atom) =1$.
Next, consider a variable $p_{z=z'}$ such that there is an atom $\atom \in \q_{\t^i}$
with $\avec{z}=\{z,z'\}$. Since $\atom \in \q_i$,
either $w_{z} \neq \emptyword$ or
$w_{z'} \neq \emptyword$.  It follows from $(\star)$ that $\avec{\alpha}(p_{z=z'})=1$.
Finally, let us consider a tree witness $\t^i \in \Theta$, and let $\varrho_i$ be the role
from the claim. % that initiates $\t^{i}$. 
We show that 
$p_{\varrho_i^*(z)}=1$ for every variable~$z$  in  $\t^i$, which will imply that 
the final disjunction is satisfied by $\avec{\alpha}$. Consider a variable~$z$ in $\t^{i}$. 
By the construction of the query $\q_i = \q_{\t^{i}}$, it contains a binary atom 
$S(\avec{z})$ with $z,z' \in \avec{z}$ and either $w_z \neq \emptyword$ or $w_{z'} \neq \emptyword$. 
By the definition of tree decompositions, there is a node $N$ in $T$ with 
 $z,z' \in \lambda(N)$. Then, by Statement~1 of the claim, 
 either $w_z= \varrho_i w$ or $w_{z'}= \varrho_i w$. 
 Now we can apply $(\star)$ to obtain $p_{\varrho^*_i(z)}=1$, as required. 
\end{proof}

\section{Proofs of Theorems~\ref{thm:linear_hgp} and ~\ref{thm:linear_hgp1}}\label{app:circuit_complexity}

{\sc Theorem~\ref{thm:linear_hgp} (general case).}
$\NL/\poly = \THGP(\ell)$ and $\mNL/\poly = \mTHGP(\ell)$, for any $\ell \ge 2$.
\begin{proof} 
Suppose a polynomial-size THGP $P$ based on a tree hypergraph $H$ with at most $\ell$ leaves computes a Boolean function $f$. We show how to construct a polynomial-size NBP that computes the same $f$. By Theorem~\ref{thm:linear_hgp1} (to be proved below),  $f_H$ can be computed by a polynomial-size NBP $B$. We replace the vertex variables $p_v$ in labels of $B$ by the corresponding vertex labels in $P$ and fix all the edge variables $p_e$ to $1$; see the proof of Proposition~\ref{hyper:program}~(\emph{i}). Clearly, the resulting NBP $B'$ is as required.

\smallskip

The converse direction is given in Section~\ref{sec:NLpoly-THGP}.
\end{proof}

{\sc Theorem~\ref{thm:linear_hgp1}.}{\it\ 
Fix $\ell \ge 2$. For any tree hypergraph $H$ based on a  tree with at most $\ell$ leaves, the function $f_H$ can be computed by an NBP of size polynomial in $|H|$.
}
\begin{proof} 
Let $H = (V, E)$ be a tree hypergraph  and  $T = (V_T, E_T)$ its underlying tree ($V = E_T$ and each $e\in E$ induces a convex subtree $T_e$ of $T$). 
Pick some vertex $r \in V_T$ and fix it as a root of $T$.
We call an independent subset $F \subseteq E$ of hyperedges \emph{flat}
if every simple path in $T$ with endpoint $r$ intersects at most one of
$T_e$, for $e\in F$. 
Note that every flat subset can contain at most $\ell$ hyperedges, 
so the number of flat subsets is  bounded by a polynomial in $|H|$.
%By $\cup F$ \nb{different notation?} we denote the union of all hyperedges in $F$: $\cup F = \bigcup_{e\in F} e \subseteq V = E_T$.
We denote by $[F]$ the union of subtrees $T_e$, for all $e\in F$ ($[F]$ is a possibly disconnected subgraph of $T$).
Flat subsets can be partially ordered by taking $F \preceq F'$ if every simple path between the root $r$ 
and a vertex of $[F']$ intersects $[F]$. As usual, $F \prec F'$ if $F \preceq F'$ but $F\ne F'$. 

The required NBP $P$ 
is based on the graph $G = (V_P, E_P)$ with 
\begin{align*} 
V_P & = \bigl\{u_F, \bar{u}_F \mid F \text{ is flat}\bigr\} \cup \bigl\{s,t\bigr\},\\ 
E_P & =\bigl\{(s, u_F), (\bar{u}_F, t), (u_F, \bar{u}_F) \mid F
 \text{ is flat}\bigr\}  \cup \bigl\{(\bar{u}_F, u_{F'}) \mid F, F' \text{ are flat and } F \prec F'\bigr\}.
\end{align*}
% 
%We label $(u_F, v_F)$ with $\bigwedge_{e \in F}p_e$ and label 
%edges $(s, u_F)$, $(v_F, t)$, and $(v_F, u_{F'})$ 
%by conjunctions of variables $p_w$ ($w \in V = E_T$) corresponding respectively to the edges of $T$ 
%that occur `before' $F$, `after' $F$,
%and `between' $F$ and~$F'$, respectively. 
%
%To be more formal, we must first introduce some notation. 
To define labels, we introduce some notation first for sets of edges of $T$ (which are sets of vertices of $H$). For a flat~$F$, 
let $\mathsf{before}(F)$ be the edges of $T$ that lie outside $[F]$ and are accessible from the root $r$ via paths not passing through $[F]$;
we denote by $\mathsf{after}(F)$ the edges of $T$
outside $[F]$ that are accessible from $r$ only via paths passing through $[F]$.
Finally, for flat $F$ and $F'$ with $F \prec F'$, we denote by $\mathsf{between}(F, F')$ the  set of edges in $T$ `between' $[F]$ and $[F']$, that is those edges of $T$ outside $[F]$ and $[F']$ that are accessible from $[F]$ via paths not passing through $[F']$ but are not accessible from the root $r$ via a path not passing through $[F]$; see Fig.~\ref{fig:8}.
Now we are ready to define the labelling for edges of $G$:
\begin{nitemize}
\item each $(u_F, \bar{u}_F)$ is labelled with the conjunction of $p_e$ for $e\in F$;
\item each $(s, u_{F})$ is labelled with the conjunction of $p_v$ for  $v\in \mathsf{before}(F)$;
\item each $(\bar{u}_{F}, u_{F'})$ is labelled with the conjunction of $p_v$ for $v\in\mathsf{between}(F, F')$;
\item each $(\bar{u}_{F}, t)$ is labelled with the conjunction of $p_v$ for $v\in\mathsf{after}(F)$.
\end{nitemize}
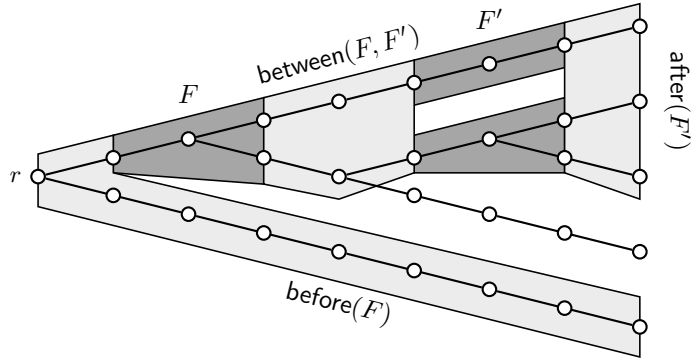
\begin{figure}[t]%
\centering%
%\scalebox{1.2}{\includegraphics{Figure_S4_T5.pdf}}
%\caption{In the above diagram, $F$ consists of one hyperedge, and $F'$ consists of two tree hyperedges. Both are independent and flat, and $F \prec F'$.  
%Here $\mathsf{before}(F)$ is the segment between $v_0$
%and $F$,
%$\mathsf{between}(F, F')$  comprises the segment between $F$ and $F'$
%as well as the downwards branch going out from $F$, and $\mathsf{after}(F')$ consists of the two branches that leave $F'$ on the side furthest from 
%$v_0$.}
\begin{tikzpicture}[xscale=2]
\filldraw[fill=gray!70,semithick] (0.5,0.05) -- ++(0,0.5) -- ++(1,0.5) -- ++(0,-1.15) -- cycle;
\filldraw[fill=gray!70,semithick] (2.5,0.05) -- ++(0,0.5) -- ++(1,0.5) -- ++(0,-1) -- cycle;
\filldraw[fill=gray!70,semithick] (2.5,0.95) -- ++(0,0.6) -- ++(1,0.5) -- ++(0,-0.6) -- cycle;
\filldraw[fill=gray!15,semithick] (0.5,0.55) -- ++(-0.5,-0.25) -- ++(0,-0.7) -- ++(4,-2) -- ++(0,0.8) -- (0.5,0.05) -- cycle;
\filldraw[fill=gray!15,semithick] (1.5,-0.1) -- ++(0,1.15) -- ++(1,0.5) -- ++(0,-1.5) -- ++(-0.5,-0.35) -- cycle;
\filldraw[fill=gray!15,semithick] (3.5,0.05) -- ++(0,2) -- ++(0.5,0.25) -- ++(0,-2.6) -- cycle;
\begin{scope}\normalsize
\node at (1,1.1) {$F$};
\node at (3,2.1) {$F'$};
\node[rotate=15] at (2,1.6) {$\mathsf{between}(F,F')$};
\node[rotate=-90] at (4.25,1) {$\mathsf{after}(F')$};
\node[rotate=-14] at (2,-1.7) {$\mathsf{before}(F)$};
\end{scope}
\node[point,label=left:{$r$}] (r) at (0,0) {};
\node[point] (v1) at (1,0.5) {};
\node[point] (v2) at (1,-0.5) {};
\node[point] (v3) at (2,1) {};
\node[point] (v4) at (2,0) {};
\node[point] (v5) at (2,-1) {};
\node[point] (v6) at (3,1.5) {};
\node[point] (v7) at (3,0.5) {};
\node[point] (v8) at (3,-0.5) {};
\node[point] (v9) at (3,-1.5) {};
\node[point] (v10) at (4,2) {};
\node[point] (v11) at (4,1) {};
\node[point] (v12) at (4,0) {};
\node[point] (v13) at (4,-1) {};
\node[point] (v14) at (4,-2) {};
\begin{scope}[thick]
\draw (r) -- (v1);
\draw (r) -- (v2);
\draw (v1) -- (v3);
\draw (v1) -- (v4);
\draw (v2) -- (v5);
\draw (v3) -- (v6);
\draw (v4) -- (v7);
\draw (v4) -- (v8);
\draw (v5) -- (v9);
\draw (v6) -- (v10);
\draw (v7) -- (v11);
\draw (v7) -- (v12);
\draw (v8) -- (v13);
\draw (v9) -- (v14);
\end{scope}
\node[point] (v01) at (0.5,0.25) {};
\node[point] (v02) at (0.5,-0.25) {};
\node[point] (v03) at (1.5,0.75) {};
\node[point] (v04) at (1.5,0.25) {};
\node[point] (v05) at (1.5,-0.75) {};
\node[point] (v6) at (2.5,1.25) {};
\node[point] (v7) at (2.5,0.25) {};
\node[point] (v8) at (2.5,-0.25) {};
\node[point] (v9) at (2.5,-1.25) {};
\node[point] (v10) at (3.5,1.75) {};
\node[point] (v11) at (3.5,0.75) {};
\node[point] (v12) at (3.5,0.25) {};
\node[point] (v13) at (3.5,-0.75) {};
\node[point] (v14) at (3.5,-1.75) {};
\end{tikzpicture}
\caption{Parts of the underlying tree in the proof of Theorem~\ref{thm:linear_hgp1}: $F \prec F'$ for $F$ with one hyperedge and $F'$ with two hyperedges.}\label{fig:8}
\end{figure}
We claim that under any valuation $\avec{\alpha}$ of $p_e$ and $p_v$, there is a path from $s$ to 
$t$ in $G$ all of whose 
labels evaluate to 1 under~$\avec{\alpha}$ iff $f_H(\avec{\alpha}) = 1$, that is, iff
there is an independent (not necessarily flat) subset $E' \subseteq E$  such that
$\avec{\alpha}(p_e) = 1$ for all $e \in E'$ and $\avec{\alpha}(p_v) = 1$ for all $v\in V\setminus V_{E'}$.
Indeed, any such $E'$ splits  into flat `layers' $F^1, F^2, \dots F^m$ that form a path
\begin{equation*}
s\to u_{F^1} \to \bar{u}_{F^1} \to u_{F^2} \to \cdots \to \bar{u}_{F^m} \to t
\end{equation*}
in $G$ and whose edge labels evaluate to 1:
take $F^1$ to be the set of all hyperedges from $E'$ that are accessible from $r$ via paths
which do not cross (that is come in and go out) any hyperedge of $E'$; take  $F^2$ to be the set of all edges from $E' \setminus F^1$ that are accessible from $r$ via paths
that do not cross any hyperedge of $E' \setminus F^1$, and so on.
Conversely, any path leading from $s$ to $t$ gives us a covering $E'$, which is the union
of all flat sets that occur in the subscripts of vertices on this path.
\end{proof}

\section{Proof of Theorems~\ref{thm:thp_vs_sac} and~\ref{thm:nc1-thgp3}}\label{app:thp_vs_sac}

\begin{lemma}\label{l:6.5}
Any semi-unbounded fan-in circuit $\Cir$ of $\AND$-depth $d$ is equivalent to a semi-unbounded fan-in circuit 
$\Cir'$
of size $2^d |\Cir|$ and $\AND$-depth $d$ such that,
for each $n \leq d$, $\Cir'$ satisfies 
\begin{equation*}
\bigcup\nolimits_{g \in S_n} \leftt(g) \quad \cap \quad \bigcup\nolimits_{g \in S_n} \rightt(g) \ \ = \ \ \emptyset.
\end{equation*}
\end{lemma}
\begin{proof}
We show by induction on $n$ that we can reconstruct the circuit in such a way that the property holds for all $i \leq n$,
the $\AND$-depth of the circuit does not change and the size of the circuit increases at most by the factor of $2^n$.

Consider a subcircuit $\bigcup_{g \in S_n} \leftt(g)$ of $\Cir$, take its copy $\Cir''$ and feed
the outputs of $\Cir''$ as left inputs to $\AND$-gates in $S_n$. This at most doubles the size of the circuit and ensures the property for $S_n$.
Now apply the induction hypothesis to both $\Cir''$ and $\bigcup_{g \in S_n} \rightt(g)$  (which do not intersect).
The size of the resulting circuit will increase at most by the factor of $2^{n-1}$ and the property for $S_i$ for $i < n$ will be ensured.\qed
\end{proof}

Let $g_i$ be a gate in $\Cir$. We denote by $T_i$ the subtree of $T$ with the root $v_i$ and, given an input $\avec{\alpha}$, we say that $T_i$ can be \emph{covered} under $\avec{\alpha}$ if the hypergraph with the underlying tree $T_i$ has an independent subset of hyperedges that 
are wholly in $T_i$ and cover all zeros under $\avec{\alpha}$.

\begin{lemma}\label{D2}
For a given input $\avec{\alpha}$ and any $i$, the gate $g_i$ outputs $1$ iff $T_i$ can be covered.
\end{lemma}

\begin{proof}
We prove the claim by induction on $i$. If $g_i$ is an input gate and outputs~1, then the label of the edge 
$\{v_i, u_i\}$ is evaluated into 1 under $\avec{\alpha}$, and the remainder of $T_i$ can be covered by a set of
$[w_j, u_j]$-hyperedges. Conversely, if an input gate $g_i$ outputs 0, then no hyperedge can cover $\{v_i, u_i\}$.

If $g_i = g_j \ANDOP g_k$ is an $\AND$-gate and outputs 1, then both its inputs output $1$. We cover both subtrees corresponding to the inputs (by induction hypothesis) and add to the covering the hyperedge
$[v_i, v_j, v_k]$, which covers $T_i$. Conversely, any covering of zeros in
$T_i$ must include the hyperedge $[v_i, v_j, v_k]$, and so the subtrees $T_j$ and $T_k$
must be covered. Thus, by the induction hypothesis, $g_j$ and $g_k$ should output 1, and so does $g_i$.

If $g_i = g_{j_1} \OROP \cdots \OROP g_{j_k}$ is an $\OR$-gate and outputs 1, then one of its inputs, say, $g_j$, is 1. By the induction hypothesis, we cover its subtree and add the hyperedge $[v_i, v_j]$, which forms a covering of $T_i$. Conversely, since $\{v_i, u_i\}$ is labelled by 0, 
any covering of $T_i$ must include a hyperedge of the form $[v_i, v_j]$ for some $j\in \{j_1,\dots, j_k\}$.
Thus $T_j$ must also be covered. By the induction hypothesis, $g_j$ outputs~1 and so does $g_i$.
\qed \end{proof}

%\section{Proof of Theorem}

\medskip

{\sc Theorem~\ref{thm:nc1-thgp3}.}{\it
$\NC^1 =  \THGP^d$ and $\mNC^1 =  \mTHGP^d$, for any $d \geq 3$.
}

\begin{proof}
To prove $\NC^1 \subseteq \THGP^3$, consider a polynomial-size formula $\Cir$, which we represent as a tree of gates $g_1, \dots, g_m$ enumerated so that $j < i$ whenever  $g_j$ is an input of $g_i$. We assume that $\Cir$ has negated variables in place of $\NOT$-gates. We now construct the tree, $T$, underlying the THGP $P$ we are after: $T$  contains triples of vertices $u_i,v_i,w_i$ partially ordered in the same way as the $g_i$ in $\Cir$. We then remove vertex $w_m$ and make $v_m$ the root of $T$.
% and on each edge $\{u_i,u_j\}$ with $i<j$, we insert two new vertices $v_j$ and $v'_j$ so that we obtain the edges $\{u_i,v'_j\}$, $\{v'_j,v_j\}$ and $\{v_j,u_j\}$. 
The THGP $P$ is based on the hypergraph whose vertices are the edges of $T$ and whose hyperedges comprise the following:
\begin{nitemize}
\item $[u_i, u_j]$, for each $i < m$, where $j<i$ and $g_j$ is the input of $g_i$;

\item $[v_i,v_j,v_k]$, for each $g_i = g_j \ANDOP g_k$;

\item $[v_i, v_j, w_k], [v_i, w_j, v_j]$, for each $g_i = g_{j} \OROP g_{k}$.
\end{nitemize}
Finally, if an input gate $g_i$ is a literal $\li$, we label the edge $\{u_i, v_i\}$ with $\li$; we label all other edges with $0$. It is not hard to check that $P$ is of degree~3, has size polynomial in~$|\Cir|$, and computes the same function as $\Cir$.

\smallskip

The inclusion $\NC^1 \supseteq \THGP^d$ follows from the proof of $\LOGCFL/\poly \subseteq \THGP$ in Theorem~\ref{thm:thp_vs_sac}. Indeed, if the degree of the THGP is at most $d$, then the disjunction in~\eqref{eq:sac_construction} has at most $d+1$ disjuncts, and so the constructed circuit has depth $O(\log s)$.
\end{proof}

\section{Proofs for Section~\ref{sec:7}}

\hspace*{1em} {\sc Theorem~\ref{bbcq-ndl}.} {\it For any fixed $\ell \ge 2$, all tree-shaped OMQs with at most $\ell$ leaves have polynomial-size NDL-rewritings.}

\begin{proof}
Fix $\ell \ge 2$ and let $\omq$ be a tree-shaped OMQ with at most $\ell$ leaves. By Theorem~\ref{prop:tree-shaped}, $\HG{\omq}$ is a tree hypergraph whose underlying tree has at most $\ell$ leaves. By Theorem~\ref{thm:linear_hgp1}, $\twfn$ is computable by a polynomial-size monotone NBP, and so,
since $\mNL/\poly \subseteq \mP/\poly$, $\twfn$ can be computed by a polynomial-size monotone Boolean circuit. It  remains to apply Theorem~\ref{TW2rew}~(\emph{ii}). 
\end{proof}

{\sc Theorem~\ref{linear-lower}.} {\it 
There is an OMQ with ontologies of depth 2 and linear CQs any PE-rewriting of which is of superpolynomial size $n^{\Omega(\log n)}$.
}

\begin{proof}
We consider the function $f = \reach$. Since $f\in \mNL/\poly$, by Theorem~\ref{thm:linear_hgp}, there is a polynomial-size monotone HGP that is based on a hypergraph $H$ with underlying tree with~2 leaves and computes $f$. Consider now the OMQ $\OMQT{H}$ for $H$ defined in Section~\ref{sec:5.3}, which has an ontology of depth~2. By Theorem~\ref{tree-hg-to-query}~(\emph{ii}), $f$ is a subfunction of $f^\vartriangle_{\OMQT{H}}$. By Theorem~\ref{rew2prim}~(\emph{i}),  no PE-rewriting of the OMQ $\OMQT{H}$ can be shorter than $n^{\Omega(\log n)}$.
%
%
%Apply  the construction from Section \ref{sec:5.3} to the sequence $f_n$ computing 
%$s$-$t$-connectivity in a given graph with  a $n^{\Omega(\log n)}$-lower bound for mononotone formulas \cite{KarchmerW88}.
%This yields a sequence of interval hypergraph programs $P_n$ based on interval hypergraphs $H_n$  which compute the functions $f^n$. The construction from Section \ref{sec:5.3} and Theorem~\ref{tree-hg-to-query}
%give us a sequence of OMQs $\omq_n = \omq^{tr}_{H_n}$ with 
%linear CQs and ontologies of depth 2 
%such that $f_{H_n}$ is a subfunction of $f^\primsuper_{\omq_n}$. 
%
%By the construction,
%$\omq_n$ is of polynomial size in $n$.
%Since  $f_n$ is obtained from  $f^\primsuper_{\omq_n}$ through a simple substitution, the lower bound
%$n^{\Omega(\log n)}$ still holds for $f^\primsuper_{\omq_n}$. It remains to apply Theorem~\ref{TW2rew}
%to transfer this lower bound to PE-rewritings of $\omq_n$.
\qed
\end{proof}

{\sc Theorem~\ref{btw-ndl}.} {\it 
For any fixed $t >0$, all OMQs with the PFSP and CQs of  treewidth at most $t$ have polynomial-size NDL-rewritings.
}

\begin{proof}
Fix a $t > 0$ and a class of OMQs with PFSP. 
Take an OMQ $\omq$  of treewidth at most $t$ from the class. By  Theorem~\ref{DL2THP}, there is 
a polynomial-size monotone THGP that computes~$\homfn$. Since
$\mTHGP \subseteq \mLOGCFL/\poly \subseteq \mP/\poly$ 
(Theorem~\ref{thm:thp_vs_sac}), $\homfn$ can be computed by a polynomial-size 
monotone Boolean circuit. It remains to apply Theorem~\ref{Hom2rew}~(\emph{ii}).
\end{proof}

{\sc Theorem~\ref{depth-one-btw}.}  {\it
For any fixed $t>0$, all OMQs with ontologies of depth 1 and CQs of treewidth at most $t$ have polynomial-size PE-rewritings.
}

\smallskip

The main argument underlying Theorem~\ref{depth-one-btw} was given in the body of the paper. To complete the proof, we
give the following two lemmas, which are the modified versions of Theorems~\ref{TW2rew} and~\ref{DL2THP} mentioned in the body. 

\begin{lemma}
Theorem~\ref{TW2rew} continues to hold if $\twfn$ is replaced by $\homfnprime$. 
\end{lemma}
\begin{proof}
The proof proceeds similarly to the proof of Theorem~\ref{Hom2rew}. The key step in the proof is showing that the 
FO-formula 
$$
\exists \avec{y} \!\!\!\!\!\bigvee_{\substack{\Theta \subseteq \twset\\ \text{ independent}}}\!\!\!\!\!\!\!
\Big(\bigwedge_{\atom \in \q \setminus \q_\Theta}\!\!\! \rsz \,\,\wedge \bigwedge_{\t \in \Theta} \big(\bigwedge_{R(z,z')\in\q_\t}\hspace{0cm} \!\!\!\!\!\!z=z' \,\wedge \bigwedge_{\substack{\phantom{y}\\z \in \tr \cup \ti}} \bigvee_{\t \text{ generated by } \tau} \hspace*{-1em}\tau(z) \big)\Big)
$$
obtained from $\homfnprime$ by replacing variables $p_{\atom}$, $p_{z=z'}$, and $p_{\exists y P_\t(z,y)}$ by $\atom$, $z=z'$, 
\\and $\bigvee_{\t \text{ generated by } \tau} \tau(z)$ respectively is equivalent the tree-witness rewriting $\qtw$. 
\end{proof}

\begin{lemma} In the setting of Section~\ref{sec:7.5}, for the modified hypergraph program $P'_{\omq}$ we still have $f_{P'_{\omq}}(\avec{v}) = \homfnprime(\avec{v})$.
\end{lemma}
\begin{proof}
The proof closely follows that of Theorem~\ref{DL2THP}. For the first direction of the proof, the only notable difference is that instead of selecting a role $\varrho_\t$
that satisfies the disjunct corresponding to the tree witness $\t$, we must take the special role $P_\t$. For the second direction, we use the assumption that $\T$ is of depth $1$ to show that every query~$\q_j$ (constructed according to the equivalence relation) has a single variable $v_j$
such that $w_{v_j} \neq \varepsilon$. This allows us to prove a stronger version of the claim in which $\q_j = \q_{\t^j}$, with $\t^j$ the unique tree witness with $\ti^j = \{v_j\}$, and the selected role $\varrho_j$ is equal to the special predicate $P_{\t^j}$ associated with $\t^j$. \end{proof}

%The part of the proof ``$f_{P'_Q}(\avec{v}) = 1$ implies $\homfn(\avec(v)) = 1$''
%where we go from a consistent system of bag types to an independent set of tree witnesses 
%goes through without changes.
%
%The other way round where we go from an independent set of tree witnesses 
%to a consistent system of bag types  also goes through without changes ? 

{\sc Theorem \ref{depth-one-tree}. }{\it
All tree-shaped OMQs with ontologies of depth 1 have polynomial-size $\mathsf{\Pi}_4$-rewritings.
}
\begin{proof}
Take an OMQ $\omq = (\T, \q)$ with $\T$ of depth 1 and a tree-shaped $\q$. 
By Theorems~\ref{depth1} and \ref{prop:tree-shaped}, $\HG{\omq}$ is a polynomial-size 
tree hypergraph of degree at most~2. By Proposition \ref{hyper:thgp}~(\emph{i}), 
$\twfn$ can be computed by a polynomial-size THGP $P$ of degree at most~2.
By Theorem \ref{thm:pi3-thgp2}, there is a polynomial-size 
monotone $\Pitr$-circuit computing  $f^\triangledown_\omq$.  By a simple 
unravelling argument, it follows that there is polynomial-size monotone Boolean formula
computing $f^\triangledown_\omq$. It remains to apply Theorem \ref{TW2rew}~(\emph{i}) and
conclude that there is a polynomial-size positive existential $\mathsf{\Pi}_4$-rewriting for $\omq$.
\end{proof}

\section{Proof of  LOGCFL membership in Theorem~\ref{logcfl-c-arb}}\label{app:complexity}

%To show the \LOGCFL{} upper bound for bounded-leaf queries, we begin with a number of
We say that an iteration of the \textbf{while} loop is \emph{successful} if the procedure \bbarbalgo{}  does not return \false{}; in particular, if none of the \Check{} operations returns \false{}. The following properties can be easily seen to hold by examination of  \bbarbalgo{} and straightforward induction:  
\begin{align}
\label{inv:B1}
&\text{For every tuple  } (z\mapsto(a,n),z')\in \frontier,  \ z'  \text{ is a child of } z \text{ in } T.\\
\label{inv:B2}
&\text{For every tuple } (z\mapsto(a,n),z')\in\frontier,   \text{ we have } n \leq |\stack|.\\
\label{inv:B3}
&\text{All tuples } (z\mapsto(a,n),z')\in\frontier \text{ with } n> 0 \text{ share the same } a.\\ 
\label{inv:B4}
& \text{Once } (z\mapsto(a,n), z') \text{ is added to } \frontier, \text{ no tuple of the form } (z\mapsto(a',n'), z')\\[-3pt] 
\notag{}&\hspace*{4em} \text{can ever be added to } \frontier.\\ 
\label{inv:B5}
&\text{In every successful iteration, either at least one tuple is removed from } \frontier\\[-3pt]
\notag{}&\hspace*{4em}\text{or $\frontier$ is unchanged but one } \varrho \text{ is popped from the } \stack.\\
\label{inv:B6}
&\text{If } (z\mapsto(a,n), z') \text{ is removed from } \frontier \text{ in a successful iteration,}\\[-3pt]
\notag{}&\hspace*{4em}\text{then a tuple of the form } (z'\mapsto(a',n'),z'') \text{ is added to } \frontier,\\[-3pt]
\notag{}&\hspace*{6em} \text{ for every child } z'' \text{ of } z' \text{ in } T. 
\end{align}

%start by establishing termination and correctness of the procedure \bbarbalgo. 

\begin{proposition}\label{logcfl-upper-prop:termination}
Every execution of \bbarbalgo\ terminates.
\end{proposition}
\begin{proof}
%We first prove termination for each execution of \bbarbalgo{}.
A simple examination of \bbarbalgo{} shows that the only possible source of non-termination 
is the \textbf{while} loop, which continues as long as $\frontier$ is non-empty. 
By~\eqref{inv:B1} and~\eqref{inv:B4}, the total number of tuples that may appear in $\frontier$ at any point cannot exceed the number of edges in~$T$, 
which is itself bounded by $|\q|$. By~\eqref{inv:B4} and~\eqref{inv:B5}, every tuple is added at most once and is eventually removed from $\frontier$. Thus, either the algorithm will exit the \textbf{while} loop by returning \false{} (if one of the \Check operations fails), or it will eventually exit the  loop after reaching an empty $\frontier$. 
\end{proof}

\begin{proposition}\label{logcfl-upper-prop:correctness}
There exists an execution of \bbarbalgo{} that returns \true{} on input $((\Tmc,\q),\Amc, \avec{a})$ if and only if $\Tmc, \Amc \models \q(\avec{a})$. 
\end{proposition}
\begin{proof}
($\Leftarrow$) Suppose that $\Tmc, \Amc \models \q(\avec{a})$. Then there exists a homomorphism
\mbox{$h\colon \q \to \canmod$} such that $h(\avec{x}) = \avec{a}$. Without loss of generality we may choose $h$ so that the image of $h$ consists
of elements $aw$ with $|w| \leq 2 |\Tmc| + |\q|$~\cite{ACKZ09}. We use $h$ to specify an execution of \bbarbalgo($(\Tmc, \q), \Amc,  \avec{a})$ that 
returns \true. First, we fix an arbitrary variable $z_0$ as root, and then, we choose the element $h(z_0)=a_0w_0$.
Since $h$ defines a homomorphism of $\q(\avec{a})$ into $\canmod$, the call \canMapTail{$z_0$, $a_0$, $\topof(\stack)$} 
returns \true.
We initialise $\stack$ to $w_0$ and $\frontier$ to $\{(z_0 \mapsto (a_0, |\stack|), v_i \mid v_i \text{ is a child of } v_0 \}$.
Next, we enter the \textbf{while} loop. Our aim is to make the non-deterministic choices to satisfy the following invariant: 
\begin{equation}\label{eq:logcfl:inv}
\text{If } \ \ (z\mapsto (a,m),z')\in\frontier, \ \ \ \ \text{ then } \ \ h(z)=a\, \stack_{\leq m}.
\end{equation}
Recall that $\stack_{\leq m}$ denotes the word obtained by concatenating the first $m$ symbols of $\stack$. 
 Observe that before the \textbf{while} loop,  property~\eqref{eq:logcfl:inv} is satisfied. 
At the start of each iteration of the while loop, we proceed as follows.

\medskip

\noindent\textsc{[Case 1.]} If $\frontier$ contains $(z\mapsto(a,0),z')$ such that $h(z') \in \ind(\Amc)$, then we choose Option~1. We remove the tuple from $\frontier$ and choose the individual $a' = h(z')$ for the guess. 
As $a=h(z)$ (by~\eqref{eq:logcfl:inv}) and $h$ is a homomorphism, we have $(a,a')\in P^{\canmod}$, for all $P(z,z')\in\q$, and the call \canMapTail{$z'$, $a'$, $\varepsilon$} returns \true. 
We  thus add \mbox{$(z'\mapsto(a',0),z'')$} to $\frontier$ for every child $z''$ of $z'$ in $T$. These additions to $\frontier$ clearly preserve the invariant. 

\medskip

\noindent\textsc{[Case 2.]} If Case 1 does not apply and $\frontier$ contains $(z\mapsto (a,|\stack|),z')$ such that $h(z') = h(z)$, then we choose Option 4 and remove the tuple from $\frontier$. 
Since $h$ is homomorphism, we have $\Tmc \models P(x,x)$, for all $P(z,z') \in \q$, and \mbox{\canMapTail{$z'$, $a$, $\topof(\stack)$}} returns \true{}. 
%If $v_2$ has at least one child in $T$, then 
Then, for every child $z''$ of $z'$ in $T$, we add $(z'\mapsto(a,|\stack|),z'')$ to $\frontier$. 
Observe that since $h(z) = h(z')$ and~\eqref{eq:logcfl:inv} holds for $z$, property~\eqref{eq:logcfl:inv} also holds for the newly added tuples. 
%and it continues to hold for existing tuples since no changes have been made to $\stack$. 
%If $v_2$ is a leaf in $T$, then no additions are made to $\frontier$, but we pop $\delta$ symbols from $\stack$ and decrement $\stackheight$
%by $\delta$, where $\delta$ is the difference between $\stackheight$ and the maximal current value appearing in any tuple of $\frontier$. 
%Since there are no additions, and the relevant initial segment of $\stack$ remains unchanged,~\eqref{eq:logcfl:inv} continues to hold. 

\medskip

\noindent\textsc{[Case 3.]} If neither Case~1 nor Case~2 applies
and $\frontier$ contains $(z\mapsto(a,|\stack|),z')$ such that $h(z') = h(z) \varrho$, then 
we choose Option~2 and remove the tuple from $\frontier$. Note that in this case, $|\stack| < 2|\Tmc|+|\q|$ since 
(i) by~\eqref{eq:logcfl:inv}, $h(z)=a w$, for $w=\stack_{\leq |\stack|}$, and (ii) by the choice of homomorphism $h$, 
we have $|w \varrho| \leq 2|\Tmc|+|\q|$. 
So, we continue and choose $\varrho$ for the guess. 
By~\eqref{eq:logcfl:inv}, since $h$ is a homomorphism and $h(z') = h(z) \varrho$, the call \isGenerated{$\varrho$, $a$, $\topof(\stack)$} returns \true{}, $\T\models \varrho(x,y)\to P(x,y)$, for all $P(z,z')\in\q$ and the call \canMapTail{$z'$, $a$, $\topof(\stack)$} returns \true{}. 
So, %If $v_2$ has some child, 
we push $\varrho$ onto $\stack$ and add $(z'\mapsto (a,|\stack|), z'')$ to $\frontier$ 
for every child $z''$ of $z'$ in $T$. As $\stack$ contains the word component of $h(z')$,
invariant~\eqref{eq:logcfl:inv} holds for the newly added tuples.  
%and it still holds for existing tuples since we have not modified the relevant portion of the stack.
% If $v_2$ is a leaf in $T$, $\frontier$ is left untouched, but we pop $\delta = \stackheight - \mathbf{max} \{\ell \mid (v,v',d,\ell) \in \frontier\}$ symbols from $\stack$ and decrement $\stackheight$ by $\delta$. It can be argued, similarly as in Case 2, that~\eqref{eq:logcfl:inv} continues to hold. 

\medskip
 
\noindent\textsc{[Case 4.]} If none of Case~1, Case~2 or Case~3 is applicable, 
then we choose Option~3 and remove all elements in $\deepest = \{(z\mapsto(a,n),z') \in \frontier \mid n = |\stack|\}$ from $\frontier$.
Since neither Case~1 nor Case~3 applies, $|\stack| > 0$. So,
we pop the top symbol~$\varrho$ from $\stack$.  Suppose first that $\deepest \ne\emptyset$. By~\eqref{inv:B3},
all tuples in $\deepest$ share the same individual $a$. By~\eqref{eq:logcfl:inv}, 
every tuple $(z\mapsto(a,n),z') \in \deepest$ is such that $h(z)=aw\varrho$, where $w =\stack_{\leq|\stack|}$. 
Moreover, since Case~3 is not applicable, for every such tuple 
$(z\mapsto(a,n),z')$, we have $h(z')=aw$. 
Using the fact that $h$ is a homomorphism, one can show that $\T\models\varrho(x,y)\to P(x,y)$, for all $P(z',z)\in\q$, and \canMapTail{$z'$, $a$, $\topof(\stack)$} returns \true{}. So, we add to $\frontier$ all tuples $(z'\mapsto (a,|\stack|), z'')$, a child $z''$ of $z'$ in $T$.
Note that invariant~\eqref{eq:logcfl:inv} is satisfied by all the new tuples. Moreover, since we only removed the last symbol in $\stack$, 
all the remaining tuples in $\frontier$ continue to satisfy~\eqref{eq:logcfl:inv}. Finally, if $\deepest$ was empty, then we do nothing but the tuples in $\frontier$ continue to satisfy~\eqref{eq:logcfl:inv}.

%If $\children$ is empty, then we pop $\delta$ symbols 
%from $\stack$ and decrement $\stackheight$ by $\delta$, where $\delta = \stackheight - \mathbf{max} \{\ell \mid (v,v',d,\ell) \in \frontier\}$.
%We can use the same reasoning as in Option 2 to show that~\eqref{eq:logcfl:inv} is preserved. % continues to hold. 
%\item[Case 4] \todo{Case for loop steps, decide when to handle these (between forward and backward, or before forward?) Should we change ordering of steps in algorithm to reflect ordering here in proof?}

\medskip

It is easily verified that so long as $\frontier$ is non-empty, one of these four cases applies. 
Since we have shown how to make the non-deterministic choices in the \textbf{while} loop without returning \false, by Proposition~\ref{logcfl-upper-prop:termination},
the procedure eventually leaves the \textbf{while} loop and returns \true{}.

\bigskip

($\Rightarrow$) Consider an execution of \bbarbalgo($(\Tmc,\q), \Amc, \avec{a})$ that returns \true. 
It follows that the \textbf{while} loop is successfully exited after reaching an empty $\frontier$. 
Let $L$ be the total number of iterations of the \textbf{while} loop. We inductively define a sequence $h_0, h_1, \ldots, h_L$ of partial functions 
from the variables of $\q$ to $\Delta^{\canmod}$ by considering the guesses made during the different iterations of the \textbf{while} loop. The domain of $h_i$ will be denoted by  $\dom(h_i)$.
We will ensure that the following properties hold for every $0 \leq i < L$:
\begin{align}
\label{prop:one} 
& \text{If } i>0, \text{ then } \dom(h_{i-1}) \subseteq \dom(h_i), \text{ and } h_i(z) = h_{i-1}(z), \text{ for } z \in \dom(h_{i-1}).\\
\label{prop:two} 
& \text{If } (z\mapsto (a,n),z')\in\frontier \text{ at the end of iteration } i, \text{ then}\\
\tag{\ref{prop:two}a}\label{prop:two:a}
&\hspace*{6em} h_i(z)=aw, \text{ where } w= \stack_{\leq n},\\
\tag{\ref{prop:two}b}\label{prop:two:b}
&\hspace*{6em} \text{and neither } z' \text{ nor any of its descendants belongs to } \dom(h_i).\\
\label{prop:three} 
& h_i \text{ is a homomorphism } \q_i\to\canmod, \text{ where } \q_i \text{ is the restriction of } \q \text{ to } \dom(h_i).
\end{align}

We begin by setting $h_0(z_0)=a_0w_0$, where $w_0$ is the word in $\stack$ (and leaving $h_0$ undefined for all other variables). Property~\eqref{prop:one} is vacuously satisfied. 
Property~\eqref{prop:two} holds because of the initial values of $\frontier$ and $\stack$ because only $z_0 \in \dom(h_0)$, and $z_0$ cannot be its own child (hence, it cannot appear in the last component of 
a tuple in $\frontier$). To see why~\eqref{prop:three} is satisfied, first suppose that $w_0 = \varepsilon$ and so 
$a_0w_0 \in \ind(\Amc)$. Then, the call \canMapTail{$z_0$, $a_0$, $\topof(\stack)$} returns \true. It follows that 
\begin{align*}
& \text{if } z_0 \text{ is the $j$th answer variable then } a_0 = a_j;\\ 
& a_0\in A^{\canmod}, \text{ for each } A(z_0)\in\q, \ \ \ \text{ and } \ \ \ (a_0,a_0)\in P^{\canmod}, \text{ for each } P(z_0,z_0)\in \q;
 \end{align*}
and hence, $h_0$ defines a homomorphism of $\q_0$ into $\canmod$. 
Otherwise, $w_0$ is non-empty and $w_0 = w_0' \varrho$.  
It follows that 
\begin{align*}
& z_0 \text{ is not an answer variable of } \q;\\ 
& \Tmc \models \exists y\, \varrho(y,x) \rightarrow A(x), \text{ for each } A(z_0)\in\q,  \ \ \ \text{ and } \ \ \ \T\models P(x,x), \text{ for each } P(z_0,z_0)\in\q;
 \end{align*}
and hence $h_0$ homomorphically maps all atoms of $\q_0$ into $\canmod$. Thus, the initial partial function $h_0$
satisfies~\eqref{prop:one}--\eqref{prop:three}. 

Next we show how to inductively define $h_{i}$ from $h_{i-1}$ while preserving~\eqref{prop:one}--\eqref{prop:three}.
The variables that belong to $\dom(h_{i}) \setminus \dom(h_{i-1})$ are precisely those 
variables that appear in the last position of tuples removed from $\frontier$ during iteration $i$ (since these are the
variables for which we guess a domain element). 
The choice of where to map these variables depends on which of the four options was selected.
In what follows, we will use $\stack^i$ to denote the contents of $\stack$ at the end of iteration $i$.

\medskip

\noindent\textsc{Option 1:} we remove a tuple $(z\mapsto (a,0),z')$ and guess $a' \in \ind(\Amc)$. 
So, we set \mbox{$h_i(z')=a'$} and $h_i(v)=h_{i-1}(v)$ for all $v\in\dom(h_{i-1})$ (all other variables remain undefined). 
Property~\eqref{prop:one} is by definition. 
For property~\eqref{prop:two}, consider a tuple \mbox{$\tau = (v\mapsto(c,m),v')$} that belongs to $\frontier$ at the end of iteration $i$. 
Suppose first $\tau$ was added to $\frontier$ during 
iteration $i$, in which case $\tau = (z'\mapsto(a',0),z'')$ for some child $z''$ of $z'$.
Property~\eqref{prop:two:a} is  satisfied because $\stack^i_{\leq 0}=\varepsilon$. 
Since $h_{i-1}$ satisfies~\eqref{prop:two}, $z''$ (a descendant of $z'$) is not in $\dom(h_{i-1})$, 
which satisfies~\eqref{prop:two:b}. 
The remaining possibility is that $\tau$ was already in $\frontier$ at the beginning of iteration $i$. Since  
$h_{i-1}$ satisfies~\eqref{prop:two}, we have
$h_{i-1}(v)=c w$ for $w= \stack^{i-1}_{\leq n}$ and  neither $v'$ nor any of its descendants belongs to $\dom(h_{i-1})$.  
Since $\stack^i = \stack^{i-1}$ and $h_i(v)=h_{i-1}(v)$, property~\eqref{prop:two:a} holds for~$\tau$. 
Moreover, as $\tau$ was not removed from $\frontier$ during iteration $i$, 
we have $\tau \neq (z\mapsto (a,0),z')$, and so, by~\eqref{inv:B4}, $v' \neq z'$.
Thus, neither $v'$ nor any of its descendants is in  $\dom(h_{i})$.

For property~\eqref{prop:three}, we first note that since $h_{i}$ agrees with $h_{i-1}$ on $\dom(h_i)$ and $h_{i-1}$ satisfies~\eqref{prop:three},
it is only necessary to consider the atoms in $\q_{i}$ that do not belong to $\q_{i-1}$. There are three kinds of such atoms:
\begin{nitemize}
\item[--] if $A(z') \in \q_i$, then, since \canMapTail{$z'$, $a'$, $\varepsilon$} returns \true{}, 
$h_i(z') = a' \in A^{\canmod}$;
\item if $P(z',z') \in \q_i$, then, again,  since \canMapTail{$z'$, $a'$, $\varepsilon$} returns \true{}, we have
$(h_i(z'),h_i(z')) = (a',a') \in P^{\canmod}$;
\item if $P(z',v)\in \q_i$ with $v \neq z'$, then $v\in\dom(h_i)$,
so $v$ must coincide with $z$, the parent of~$z'$ (rather than being one of the children of $z'$); the \textbf{check} operation 
in the algorithm  then guarantees $(h_{i}(z'),h_{i}(v))= (a',a) \in P^{\canmod}$.
\end{nitemize}
Thus,~\eqref{prop:three} holds for $h_i$. 

\smallskip

\noindent\textsc{Option 2:}  
a tuple $(z\mapsto (a,n),z')$ was removed from $\frontier$, $n = |\stack|$ and  a role~$\varrho$ was guessed.  
We set $h_i(z')=h_{i-1}(z) \varrho$. By~\eqref{prop:two}, $h_{i-1}(z)$ is defined. 
Moreover, the call \isGenerated{$\varrho$, $a$, $\topof(\stack)$} ensures that $h_{i-1}(z) \varrho\in\Delta^{\canmod}$. 
We also set $h_i(v)=h_{i-1}(v)$ for all $v\in\dom(h_{i-1})$ and leave the remaining variables undefined. 
Property~\eqref{prop:one} is immediate from the definition of $h_i$, and~\eqref{prop:two:b} can be shown exactly as for Option~1. To show~\eqref{prop:two:a}, consider 
a tuple $\tau=(v\mapsto (c,m),v')$ that belongs to $\frontier$ at the end of iteration $i$. 
Suppose first that $\tau$ was added to $\frontier$
during iteration $i$, in which case $\tau=(z'\mapsto (a,n+1),z'')$ for some child $z''$ of $z'$. 
Since $h_{i-1}$ satisfies~\eqref{prop:two}, $h_{i-1}(z)= a \, \stack^{i-1}_{\leq n}$. Property~\eqref{prop:two:a} follows then 
from $h_i(z')= h_{i-1}(z) \varrho$ and $\stack^i = \stack^{i-1} \, \varrho$. 
The other possibility is that $\tau$ was present in $\frontier$ at the beginning of iteration $i$. Since $h_{i-1}$ satisfies~\eqref{prop:two}, we have $h_{i-1}(v)=a \, \stack^{i-1}_{\leq m}$.
Property~\eqref{prop:two:a} continues to hold for $\tau$ because $\stack^i = \stack^{i-1} \, \varrho$ and $m \leq |\stack^{i-1}|$ and $h_i(v)=h_{i-1}(v)$.

We now turn to property~\eqref{prop:three}. As explained in the proof for Option~1, it is sufficient to consider the atoms in 
$\q_{i} \setminus \q_{i-1}$, which can be of three types:
\begin{nitemize}
\item[--]  if $A(z') \in \q_i$, then, since \canMapTail{$z'$, $a$, $\varrho$} returns \true{}, we have
\mbox{$\Tmc \models \exists y\ \varrho(y,x)\rightarrow A(x)$}, hence $h_{i}(z')=h_{i-1}(z) \varrho\in A^{\canmod}$.
\item[--]  if $P(z',z')\in\q_i$, then, again, since \canMapTail{$z'$, $a$, $\varrho$} returns \true{}, we have
\mbox{$\Tmc \models P(x,x)$}, hence $(h_{i}(z'),h_i(z'))\in P^{\canmod}$.
\item[--]  if $P(z',v)\in\q_i$ with $v \neq z'$ then $v=z$ (see Option 1);
so, $\Tmc \models \varrho(x,y) \rightarrow P(y,x)$, whence $(h_i(z'),h_i(v)) = (h_{i-1}(z) \varrho, h_{i-1}(z)) \in P^{\canmod}$.
\end{nitemize}
Therefore, $h_i$ is a homomorphism from $\q_i$ into $\canmod$, which is required by~\eqref{prop:three}.

\medskip

\noindent\textsc{Option 3:}  tuples in $\deepest = \{(z\mapsto (a,n),z') \in \frontier \mid n = |\stack|\}$ are removed 
from $\frontier$, and role $\varrho$ is popped from $\stack$. By~\eqref{inv:B3}, all tuples in 
$\deepest$ share the same individual $a$.  Let $V = \{z' \mid (z\mapsto (a,n), z')\in\deepest \}$.
For every $v \in V$, we set $h_i(v) = a \,\stack^i$; we also set $h_{i}(v)=h_{i-1}(v)$ for all $v\in\dom(h_{i-1})$ and leave the remaining variables undefined. 
Property~\eqref{prop:one} is again immediate, and the argument for~\eqref{prop:two:b} is the same as in Option~1. 
For property~\eqref{prop:two:a}, 
take any tuple $\tau=(v\mapsto (c,m),v')$ in $\frontier$ at the end of iteration $i$. 
If the tuple  was added
to $\frontier$ during this iteration, then  $v \in V$, $a=c$, $m=|\stack^i|$, and 
$h_i(v)= a \, \stack^i$, whence~\eqref{prop:two:a}. 
The other possibility is that $\tau$ was present in $\frontier$ at the beginning of iteration $i$. 
Then $h_{i-1}(v)=c \,\stack^{i-1}_{\leq m}$ and $m<|\stack^{i-1}|$.
Since $\stack^i$ is obtained from $\stack^{i-1}$ by popping one role, 
%and that at the end of iteration $i$, $\stackheight$ is equal to the largest value appearing in a tuple of $\frontier$. 
%We thus know that 
%at the end of iteration $i+1$, 
we have $m \leq |\stack^i|$,  and so~\eqref{prop:two:a} holds for $\tau$.

For property~\eqref{prop:three}, the argument is similar to Options~1 and~2 and involves considering
the different types of atoms that may appear in $\q_{i} \setminus \q_{i-1}$: 
\begin{nitemize}
\item[--]  if $A(z') \in \q_i$ with $z'\in V$ then, since \canMapTail{$z'$, $a$, $\topof(\stack)$} returns \true{}, we have $h_{i}(z') \in A^{\canmod}$ (see Options~1 and~2);
\item[--] if $P(z',z')\in \q_i$ with $z'\in V$ then, since  \canMapTail{$z'$, $a$, $\topof(\stack)$} returns \true{}, we have  $(h_i(z'),h_i(z'))=(a,a) \in P^{\canmod}$;
\item   if $P(z',v) \in \q_{i}$ with $v \neq z'$ and $z'\in V$, then $v$ is the parent of $z$ (see Option~1) and, since
$\Tmc \models \varrho(y,x) \rightarrow P(x,y)$, we obtain $(h_i(z'),h_i(v)) = (a \, \stack^i, a \, \stack^i\, \varrho) \in P^{\canmod}$.
\end{nitemize}
Thus,~\eqref{prop:three} holds for $h_i$. 

\medskip

\noindent\textsc{Option 4:}  a tuple $(z\mapsto (a,n),z')$ was removed from $\frontier$ with $n=|\stack|$.
We set $h_i(z')=h_i(z)$, $h_i(v)=h_{i-1}(v)$ for every $v \in \dom(h_{i-1})$, and leave all other variables unmapped. 
Again, it is easy to see that properties~\eqref{prop:one} and~\eqref{prop:two:b} are satisfied by $h_i$.
For property~\eqref{prop:two:a}, 
let  $\tau=(v\mapsto (c,m),v')$ be a tuple in $\frontier$ at the end of iteration~$i$. 
If the tuple was added during iteration $i$, then $v=z'$, $a=c$, and $m=n$. Since \mbox{$(z\mapsto (a,n),z')$} was present at the end of iteration $i-1$ and $\stack^i = \stack^{i-1}$,
we have $h_i(z') = a\, \stack^{i-1}_{\leq n}$, hence $h_i(z)= c\, \stack^i_{\leq m}$. 
As $h_i(z')=h_i(z)$, we have $h_i(z')= a \,\stack^i_{\leq m}$, so $\tau$ satisfies~\eqref{prop:two:a}.
If $\tau$ was already present at the beginning of iteration $i$, 
then we can use the fact that $\stack^{i} = \stack^{i-1}$  
and all tuples in $\frontier$ satisfy~\eqref{prop:two:a}. 

To show~\eqref{prop:three}, we consider the three types of atoms that may appear in $\q_{i} \setminus \q_{i-1}$: 
\begin{nitemize}
\item[--] if $A(z') \in \q_i$ then, since \canMapTail{$z'$, $a$, $\topof(\stack)$} returns \true{},  then \mbox{$\Tmc \models \exists y\, \varrho(y,x)\rightarrow A(x)$,} where $\varrho = \topof(\stack)$, and so $h_{i}(z')\in A^{\canmod}$;
\item[--] if $P(z',z') \in \q_i$ then, since \canMapTail{$z'$, $a$, $\topof(\stack)$} returns \true{},  then \mbox{$\Tmc \models P(x,x)$}, and so $(h_i(z'), h(z'))\in P^{\canmod}$;
\item[--]  if $P(z',v) \in \q_{i}$  with $v \neq z'$, then $v=z$ (see Option~1), and so, since $\Tmc \models P(x,x)$, we have $(h_i(z'),h_i(z)) \in P^{\canmod}$.
\end{nitemize}

\medskip

We claim that the final partial function $h_L$ is a homomorphism of $\q$ to $\canmod$. 
Since $h_L$ is a homomorphism of $\q_L$ into $\canmod$, it suffices to show that $\q=\q_L$, 
or equivalently, that all variables of $\q$ are in $\dom(h_L)$. This follows from the tree-shapedness of $\q$ (which in particular means that $\q$ is connected), invariants~\eqref{inv:B1}, and \eqref{inv:B6}, 
and the 
fact that $\dom(h_{i+1}) = \dom(h_i) \cup \{z' \mid (z\mapsto (a,n),z') \text { is removed from } \frontier \text{ during iteration } i\}$. 
\end{proof}

\begin{proposition}\label{nauxpda}
\bbarbalgo\ can be implemented by an NAuxPDA.
\end{proposition}
\begin{proof}
It suffices to show that \bbarbalgo{} runs in non-deterministic logarithmic 
space and polynomial time (the size of $\stack$ does not have to be bounded). 

First, we non-deterministically fix a root variable $z_0$, but do not 
actually need to store the induced directed tree $T$ in memory. Instead, it suffices   
to decide, given two variables $z$ and $z'$, whether $z'$ is a child of $z$ in $T$, which clearly belongs to~\NL. 

Next, we need only logarithmic space to store the individual $a_0$. The initial word $w_0 = \varrho_1 \ldots \varrho_{n_0}$ 
is guessed symbol by symbol and pushed onto $\stack$. 
%We recall that $a_0w_0 \in \Delta^{\canmod}$ just in the case that:
%\begin{nitemize}
%\item $\Tmc, \Amc \models \exists y \varrho_1(a,y)$ and  $\Tmc, \Amc \not \models \varrho_1(a,b)$ for any $b \in \ind(\Amc)$; 
%\item for every $1 \leq i < N$: $\mathcal{T} \models \exists y\, \varrho_i(y,x) \rightarrow \exists y \, \varrho_{i+1}(x,y)$
%  and $\Tmc \not \models \varrho_i(x,y) \rightarrow \varrho_{i+1}(y,x)$. %and $R_i^- \ne R_{i+1}$.
%\end{nitemize}
%Thus, it is possible to perform the required entailment checks incrementally as the symbols of $w_i$ are guessed. 
%Finally, to ensure that the guessed word $w_0$ does not exceed the length bound, 
%each time we push a symbol onto $\stack$, we increment $\stackheight$ by $1$.
%If $\stackheight$ reaches $2 |\Tmc| + |\q|$, then no more symbols may be guessed. 
We note that both subroutines, \isGenerated and \canMapTail, can be made to run in 
non-deterministic logarithmic space. % (cf.\ proof of Theorem \ref{nl-bb}). 
%
%The initialisations of $\stack$ and $\stackheight$ in Step 3 were already handled in our discussion of Step 2. 
Then, since the children of a node in $T$ can be identified in \NL, % (again we refer back to the proof of Theorem \ref{nl-bb}),
we can decide in non-deterministic  logarithmic space whether a tuple $(z_0\mapsto (a_0,|\stack|,z_i)$ should be included
in $\frontier$. Moreover, since the input query $\q$ is a tree-shaped query with a bounded number of leaves,
we know that only constantly many tuples can be added to $\frontier$ by each such operation. Moreover, it is clear that 
every tuple can be stored using in logarithmic space. More generally, by~\eqref{inv:B1} and~\eqref{inv:B4}, one can show that 
$|\frontier|$ is bounded by a constant throughout the execution of the procedure, 
and the tuples added during the \textbf{while} loop can also be stored in logarithmically space. 

Next observe that every iteration of the \textbf{while} loop involves a polynomial number of the
following elementary operations such as
\begin{nitemize}
\item remove a tuple from $\frontier$, or add a tuple to $\frontier$;
\item pop a role from $\stack$, or push a role onto $\stack$;
%\item increment or decrement $\stackheight$ by a number bounded by $2 |\Tmc| + |q|$
%\item test whether $\stackheight$ is equal to $0$ or to $2 |\Tmc| + |q|$
\item guess a single individual constant or symbol;
\item identify the children of a given variable;
%\item locate an atom in $\q$;
\item test whether $\Tmc \models \alpha$, for some inclusion $\alpha$ involving symbols from $\Tmc$;
\item make a call to one of the subroutines \isGenerated or \canMapTail. 
\end{nitemize}
For each of the above operations, it is either easy to see, or has already been explained, that the 
operation can be performed in non-deterministic logarithmic space.

To complete the proof, observe that, by~\eqref{inv:B5}, each iteration of the while loop involves removing a tuple from $\frontier$ or popping a role from $\stack$. 
By~\eqref{inv:B1} every tuple in $\frontier$ corresponds to an edge in $T$, and, by~\eqref{inv:B4}, we create at most one tuple per edge. Thus, there can be at most $|\q|$ iterations involving the removal of a tuple. The total number of roles added to $\stack$ is bounded by  at most $\leq 2|\T| +|\q|$ roles in the initial stack, plus the at most $|q|$ roles added in later iterations, yielding at most $2|\T| +2|\q|$ iterations involving only the popping of a role. Thus, the total number of iterations of the while loop cannot exceed can  $2|\T| +3|\q|$.
\end{proof}

\section{Proof of LOGCFL-hardness in Theorem~\ref{logcfl-c-arb}}\label{app:logcfl-hardness}

\begin{proposition}\label{prop:f.2}
The query $\qclin$ and KB $(\kbC)$ can be computed from $\Cir$ by logspace transducers. 
\end{proposition}
\begin{proof}
Consider a  circuit $\Cir$ in normal form with $2d+1$ layers of gates, where $d$ is logarithmic in number of its inputs $n$. We show that $(\kbC)$ and $\qclin$ can be constructed 
using $O(\log(n))$ worktape memory.

To produce the query $\qclin$, we can generate the word $w_d$ letter-by-letter and insert the corresponding variables. 
This can be done by a simple recursive procedure of depth $d$, using the worktape to remember the current position in the recursion tree
as well as the index of the current variable $y_i$. Note that $|w_d|$ (hence the largest index of the query variables) may be exponential in $d$, 
but is only polynomial in $n$, and so we need only logarithmic space to store the index of the current variable. 

The ontology $\Tmc_{\avec{\alpha}}$ is obtained by 
making a single pass over a (graph representation) of the circuit and generating the 
axioms that correspond to the gates of $\Cir$ and the links between them. 
To decide which axioms of the form $G_i(x) \rightarrow A(x)$ to include,  we must also look up the value 
of the variables associated to the input gates under the valuation $\avec{\alpha}$. 
%doing a single pass over
%the circuit $\Cir$ and the database $D_{\Cir}^\avec{x}$. On this pass, we map the 
%conductors  of the circuit and database atoms into ontology axioms and do not 
%require any worktape memory.
%
Finally, $\Amc$ consists of a single constant atom.
\end{proof}

\begin{proposition}\label{prop:f.3}
$\Cir$ accepts $\avec{\alpha}$ iff 
$\Tmc_{\avec{\alpha}}, \Amc \models \qclin(a)$.
\end{proposition}
\begin{proof}
Denote by $e$ the natural homomorphism from $\qclin$ to $\q$, and
by $e'$ the natural homomorphism from $\C_{\kbC}$ to $\dcx$.
Since
$\Cir$ accepts input $\avec{\alpha}$ iff there is a homomorphism $h$ from  $\q$ to $\dcx$~\cite{DBLP:journals/jacm/GottlobLS01}, it 
suffices to show that there exists a homomorphism $f$ from $\qclin$ to $\C_{\kbC}$ 
iff there is a homomorphism $h$ from $\q$ to $\dcx$:\\
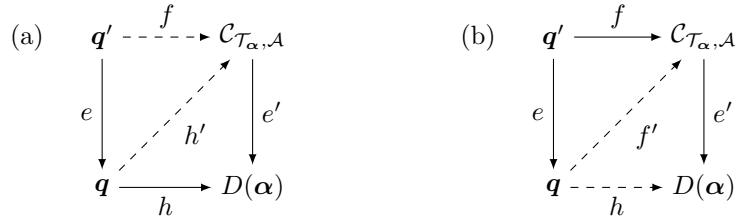
\centerline{%
\begin{tikzpicture}[auto]
\node at (-1,2) {(a)};
\node (qclin) at (0,2) {$\qclin$};
\node (qc) at (0,0) {$\q$};
\node (dc) at (2,0) {$\dcx$};
\node (cm) at (2,2) {$\C_{\kbC}$};
\draw[->] (qclin) to node [left] {$e$} (qc);
\draw[->] (cm) to node {$e'$} (dc);
\draw[->] (qc) to node [below]  {$h$} (dc);
\draw[->,dashed] (qc) to node [below right] {$h'$} (cm);
\draw[->,dashed] (qclin) to node [above] {$f$} (cm);
\begin{scope}[xshift=2cm]
\node at (3,2) {(b)};
\node (qclin1) at (4,2) {$\qclin$};
\node (qc1) at (4,0) {$\q$};
\node (dc1) at (6,0) {$\dcx$};
\node (cm1) at (6,2) {$\C_{\kbC}$};
\draw[->] (qclin1) to node [left] {$e$} (qc1);
\draw[->] (cm1) to node {$e'$} (dc1);
\draw[->] (qclin1) to node [above] {$f$} (cm1);
\draw[->,dashed] (qc1) to node [below right] {$f'$} (cm1);
\draw[->,dashed] (qc1) to node [below]  {$h$} (dc1);
\end{scope}
\end{tikzpicture}}

($\Rightarrow$) Suppose that $h$ is a homomorphism  from $\q$ to $\dcx$.
We define a homomorphism $h'\colon \q \to \C_{\kbC}$ inductively moving from
the root $n_1$ of $\q$ to its leaves. For the basis of induction,  we set $h'(n_1) = a$; note that
$\C_{\kbC} \models G_m(a)$. For the inductive step,  
suppose that $n_j$ is a child of $n_i$, $h'(n_i)$ is defined, $\C_{\kbC} \models G_{i'}(h(n_i))$
and  $h(n_j) = g_{j'}$. In this case, we set $h'(n_j) = h'(n_i) P^-_{i'j'}$.  It follows from the definition of
$\Tmc_{\avec{\alpha}}$ that $\C_{\kbC} \models G_{j'}(h'(n_j))$, 
which enables us to continue the induction. It should be clear that
$h'$ is indeed a homomorphism from $\q$ into $\C_{\kbC}$. %Since the composition of homomorphisms is again a homomorphism, we can 
%obtain 
The desired homomorphism $f\colon \qclin \rightarrow \C_{\kbC}$ can be obtained as the composition of $e$ and $h'$, as illustrated in diagram (a).  

\smallskip

($\Leftarrow$) Suppose that $f$ is a homomorphism  from $\qclin$ to $\C_{\kbC}$. 
We prove,  by induction on $|j - i|$, that for all its variables $y_i, y_j$, %(with $i < j$) 
\begin{equation}\label{eq:match}
e(y_i) = e(y_j) \ \ \text{ implies } \ \ f(y_i) = f(y_j). 
\end{equation}
The base case ($|j - i| = 0$) is trivial.
For the inductive step, we may assume without loss of generality that $i < j$ and
between $y_i$ and $y_j$ there is no intermediate variable $y_k$ with $e(y_i) = e(y_k) =e(y_j)$
(otherwise, we can simply use the induction hypothesis together with the transitivity of equality). It follows that 
$e(y_{i+1}) = e(y_{j-1})$, and the atom between $y_{j-1}$ and $y_{j}$ is oriented from $y_{j-1}$ towards $y_{j}$,
while the atom between $y_i$ and $y_{i+1}$ goes from $y_{i+1}$ to $y_i$. 
 Indeed, this holds if the node $n = e(y_i) = e(y_j)$ is an $\OR$-node since there are exactly two variables in $\qclin$ which are mapped to $n$, 
 and they bound the subtree in $\q$ generated by $n$. For an $\AND$-node, this also holds because of our assumption about intermediate variables.
 By the induction hypothesis, we have
$f(y_{i+1}) =  f(y_{j-1}) = aw\varrho$ for some word $aw\varrho$. Since the only parent of $aw\varrho$ in $\C_{\kbC}$ is $aw$, 
all arrows in relations $U$, $L$ and $R$ are oriented towards the root, 
and $f$ is known to be a homomorphism, it follows that $f(y_{i}) =  f(y_{j}) = aw$. This concludes the inductive argument.

Next, we define $f'\colon \q \to \C_{\kbC}$ by setting $f'(x)=f(y)$, where $y$ is such that \mbox{$e(y)=x$}. By~\eqref{eq:match},  $f'$ is well-defined, and because $f$ is a homomorphism, the same holds for $f'$. To obtain the desired 
homomorphism $h\colon\q\to\dcx$, it suffices to consider the composition of $f'$ and $e'$; see diagram (b).
\end{proof}

\end{document}